\documentclass[reqno]{amsart}
\usepackage{comment,amssymb,latexsym,upref,enumerate,fouridx}
\usepackage{mathrsfs,color}
\usepackage{centernot}
\usepackage[colorlinks,linkcolor=blue,citecolor=blue]{hyperref}

\allowdisplaybreaks

\numberwithin{equation}{section}

\theoremstyle{plain}
\newtheorem{thm}{Theorem}[section]
\newtheorem{lem}[thm]{Lemma}
\newtheorem{prop}[thm]{Proposition}

\theoremstyle{definition}
\newtheorem{Def}[thm]{Definition}

\theoremstyle{remark}
\newtheorem{rem}[thm]{Remark}

\DeclareMathAlphabet{\mathpzc}{OT1}{pzc}{m}{it}

\DeclareMathOperator{\tr}{tr}

\DeclareMathOperator{\diag}{diag}
\DeclareMathOperator{\Div}{div}

\DeclareMathOperator{\Ld}{L}

\DeclareMathOperator{\op}{op}

\newcommand{\id}{\mathord{{\mathrm 1}\kern-0.27em{\mathrm I}}\kern0.35em}

\newcommand{\del}[1]{\partial_{#1}}

\newcommand{\delt}[1]{\tilde{\partial}_{#1}}
\newcommand{\delh}[1]{\hat{\partial}_{#1}}

\DeclareMathOperator{\Ord}{O}

\newcommand{\oset}[2]{\overset{#1}{#2}{}}

\newcommand{\dsp}{\displaystyle}

\newcommand{\AND}{{\quad\text{and}\quad}}

\newcommand{\norm}[1]{\|#1\|}

\newcommand{\bnorm}[1]{\bigl\|#1\bigr\|}
\newcommand{\Bnorm}[1]{\Bigl\|#1\Bigr\|}
\newcommand{\ipe}[2]{ ( #1 | #2 )}

\newcommand{\Ac}{\mathcal{A}{}}

\newcommand{\Bc}{\mathcal{B}{}}
\newcommand{\bc}{\mathpzc{b}{}}
\newcommand{\Cc}{\mathcal{C}{}}

\newcommand{\Dc}{\mathcal{D}{}}

\newcommand{\ec}{\mathpzc{e}{}}

\newcommand{\Gc}{\mathcal{G}{}}
\newcommand{\gc}{\mathpzc{g}{}}

\newcommand{\Ic}{\mathcal{I}{}}

\newcommand{\Jc}{\mathcal{J}{}}

\newcommand{\Nc}{\mathcal{N}{}}

\newcommand{\Rc}{\mathcal{R}{}}

\newcommand{\Uc}{\mathcal{U}{}}

\newcommand{\Zc}{\mathcal{Z}{}}

\newcommand{\Gb}{\bar{G}{}}
\newcommand{\gb}{\bar{g}{}}

\newcommand{\Rb}{\bar{R}{}}

\newcommand{\Tb}{\bar{T}{}}
\newcommand{\tb}{\bar{t}{}}

\newcommand{\xb}{\bar{x}{}}

\newcommand{\Gammab}{\bar{\Gamma}{}}

\newcommand{\nablab}{\bar{\nabla}{}}

\newcommand{\Jch}{\check{J}{}}

\newcommand{\gbr}{\breve{g}{}}

\newcommand{\Pbr}{\breve{P}{}}

\newcommand{\qbr}{\breve{q}{}}

\newcommand{\rbr}{\breve{r}{}}

\newcommand{\Gammabr}{\breve{\Gamma}{}}

\newcommand{\taubr}{\breve{\tau}{}}

\newcommand{\At}{\tilde{A}{}}

\newcommand{\Ft}{\tilde{F}{}}
\newcommand{\gt}{\tilde{g}{}}
\newcommand{\Gt}{\tilde{G}{}}

\newcommand{\It}{\tilde{I}{}}
\newcommand{\jt}{\tilde{j}{}}
\newcommand{\Jt}{\tilde{J}{}}
\newcommand{\kt}{\tilde{k}{}}
\newcommand{\Kt}{\tilde{K}{}}
\newcommand{\lt}{\tilde{l}{}}
\newcommand{\Lt}{\tilde{L}{}}
\newcommand{\mt}{\tilde{m}{}}
\newcommand{\Mt}{\tilde{M}{}}

\newcommand{\Qt}{\tilde{Q}{}}
\newcommand{\rt}{\tilde{r}{}}
\newcommand{\Rt}{\tilde{R}{}}

\newcommand{\ttl}{\tilde{t}{}}

\newcommand{\ut}{\tilde{u}{}}

\newcommand{\Wt}{\tilde{W}{}}
\newcommand{\xt}{\tilde{x}{}}

\newcommand{\ellt}{\tilde{\ell}{}}

\newcommand{\psit}{\tilde{\psi}{}}

\newcommand{\betat}{\tilde{\beta}{}}

\newcommand{\Gammat}{\tilde{\Gamma}{}}
\newcommand{\nablat}{\tilde{\nabla}{}}
\newcommand{\kappat}{\tilde{\kappa}{}}

\newcommand{\xit}{\tilde{\xi}{}}

\newcommand{\Bh}{\hat{B}{}}

\newcommand{\Gh}{\hat{G}{}}
\newcommand{\gh}{\hat{g}{}}

\newcommand{\hh}{\hat{h}{}}

\newcommand{\nh}{\hat{n}{}}

\newcommand{\Qh}{\hat{Q}{}}

\newcommand{\Rh}{\hat{R}{}}

\newcommand{\vh}{\hat{v}{}}

\newcommand{\wh}{\hat{w}{}}
\newcommand{\Xh}{\hat{X}{}}
\newcommand{\xh}{\hat{x}{}}

\newcommand{\zh}{\hat{z}{}}

\newcommand{\nablah}{\hat{\nabla}{}}

\newcommand{\tauh}{\hat{\tau}{}}

\newcommand{\bfr}{\mathfrak{b}{}}

\newcommand{\ef}{\mathfrak{e}{}}

\newcommand{\Gf}{\mathfrak{G}{}}
\newcommand{\gf}{\mathfrak{g}{}}
\newcommand{\Hf}{\mathfrak{H}{}}

\newcommand{\Kf}{\mathfrak{K}{}}
\newcommand{\kf}{\mathfrak{k}{}}
\newcommand{\Lf}{\mathfrak{L}{}}

\newcommand{\Mf}{\mathfrak{M}{}}

\newcommand{\Pf}{\mathfrak{P}{}}

\newcommand{\Qf}{\mathfrak{Q}{}}

\newcommand{\rf}{\mathfrak{r}{}}

\newcommand{\Tf}{\mathfrak{T}{}}
\newcommand{\tf}{\mathfrak{t}{}}

\newcommand{\Ebb}{\mathbb{E}{}}

\newcommand{\Mbb}[1]{\mathbb{M}_{#1}{}}

\newcommand{\Pbb}{\mathbb{P}{}}

\newcommand{\Rbb}{\mathbb{R}{}}
\newcommand{\Sbb}[1]{\mathbb{S}_{#1}{}}
\newcommand{\Tbb}{\mathbb{T}{}}

\newcommand{\Zbb}{\mathbb{Z}{}}

\newcommand{\Asc}{\mathscr{A}{}}
\newcommand{\Bsc}{\mathscr{B}{}}
\newcommand{\Csc}{\mathscr{C}{}}

\newcommand{\Hsc}{\mathscr{H}{}}

\newcommand{\Jsc}{\mathscr{J}{}}
\newcommand{\Ksc}{\mathscr{K}{}}
\newcommand{\Lsc}{\mathscr{L}{}}
\newcommand{\Msc}{\mathscr{M}{}}

\newcommand{\btt}{\mathtt{b}{}}

\newcommand{\Dtt}{\mathtt{D}{}}

\newcommand{\ett}{\mathtt{e}{}}

\newcommand{\gtt}{\mathtt{g}{}}

\newcommand{\Ktt}{\mathtt{K}{}}

\newcommand{\Ntt}{\mathtt{N}{}}
\newcommand{\ntt}{\mathtt{n}{}}

\newcommand{\Ptt}{\mathtt{P}{}}

\newcommand{\Qtt}{\mathtt{Q}{}}

\newcommand{\Rtt}{\mathtt{R}{}}

\newcommand{\ep}{\epsilon}

\newcommand{\tauac}{\acute{\tau}{}}
\newcommand{\gac}{\acute{g}{}}

\newcommand{\Qft}{\tilde{\mathfrak{Q}}{}}

\newcommand{\taut}{\tilde{\tau}{}}

\newcommand{\Jcch}{\check{\mathcal{J}}{}}

\newcommand{\chihu}{\underline{\hat{\chi}}{}}
\newcommand{\ghu}{\underline{\gh}{}}
\newcommand{\whu}{\underline{\wh}{}}
\newcommand{\vhu}{\underline{\vh}{}}

\newcommand{\zhu}{\underline{\zh}{}}
\newcommand{\hhu}{\underline{\hh}{}}
\newcommand{\Bhu}{\underline{\Bh}{}}
\newcommand{\Qhu}{\underline{\Qh}{}}
\newcommand{\Gammahu}{\underline{\hat{\Gamma}}{}}

\newcommand{\Gammah}{\hat{\Gamma}{}}

\newcommand{\gchat}{\hat{\gc}{}}

\newcommand{\chih}{\hat{\chi}{}}

\newcommand{\taur}{\mathring{\tau}{}}
\newcommand{\chir}{\mathring{\chi}{}}

\newcommand{\taug}{\grave{\tau}{}}
\newcommand{\er}{\mathring{e}{}}
\newcommand{\gr}{\mathring{g}{}}
\newcommand{\lr}{\mathring{l}{}}
\newcommand{\Jcr}{\mathring{\mathcal{J}}{}}
\newcommand{\vr}{\mathring{v}{}}

\newcommand{\rhor}{\mathring{\rho}{}}

\newcommand{\Kttt}{\tilde{\Ktt}{}}
\newcommand{\gttt}{\tilde{\gtt}{}}
\newcommand{\Nttt}{\tilde{\Ntt}{}}
\newcommand{\bttt}{\tilde{\btt}{}}

\newcommand{\varep}{\varepsilon}

\newcommand{\ggr}{\grave{g}{}}
\newcommand{\taugr}{\grave{\tau}{}}

\newcommand{\Lambdat}{\tilde{\Lambda}{}}

\newcommand{\epv}{\boldsymbol{\epsilon}}

\newcommand{\udim}{n^4+2 n^3-3 n^2+n+1}

\newcounter{mnotecount}[section]
\let\oldmarginpar\marginpar
\renewcommand\marginpar[1]{\-\oldmarginpar[\raggedleft\footnotesize #1]%
 {\raggedright\footnotesize #1}}

\begin{document}

\title[Localized big bang stability]{Localized big bang stability for the Einstein-scalar field equations}

\author[F. Beyer]{Florian Beyer}
\address{Dept of Mathematics and Statistics\\
730 Cumberland St\\
University of Otago, Dunedin 9016\\ New Zealand}
\email{fbeyer@maths.otago.ac.nz }

\author[T.A. Oliynyk]{Todd A. Oliynyk}
\address{School of Mathematical Sciences\\
9 Rainforest Walk\\
Monash University, VIC 3800\\ Australia}
\email{todd.oliynyk@monash.edu}

\begin{abstract}
\noindent 
We prove the nonlinear stability in the contracting direction of Friedmann-Lema\^itre-Robertson-Walker (FLRW) solutions to the Einstein-scalar field equations in $n\geq 3$ spacetime dimensions that are defined on spacetime manifolds of the form $(0,t_0]\times \Tbb^{n-1}$, $t_0>0$. Stability is established under the assumption that the initial data is \textit{synchronized}, which means that on the initial hypersurface $\Sigma= \{t_0\}\times \Tbb^{n-1}$ the scalar field $\tau= \exp\bigl(\sqrt{\frac{2(n-2)}{n-1}}\phi\bigr) $ is constant, that is, $\Sigma=\tau^{-1}(\{t_0\})$. As we show that all initial data sets that are sufficiently close to FRLW ones can be evolved via the Einstein-scalar field equation into new initial data sets that are \textit{synchronized}, no generality is lost by this assumption. By using $\tau$ as a time coordinate, we establish that the perturbed FLRW spacetime manifolds are of the form $M = \bigcup_{t\in (0,t_0]}\tau^{-1}(\{t\})\cong (0,t_0]\times \Tbb^{n-1}$, the perturbed FLRW solutions are asymptotically pointwise Kasner as $\tau \searrow 0$, and a big bang singularity, characterised by the blow up of the scalar curvature, occurs at $\tau=0$. An important aspect of our past stability proof is that we use a hyperbolic gauge reduction of the Einstein-scalar field equations. As a consequence, all of the estimates used in the stability proof can be localized and we employ this property to establish a corresponding localized past stability result for the FLRW solutions. 
\end{abstract}

\maketitle

\section{Introduction}
Within the class of spatially homogeneous and isotropic spacetimes, the  Friedmann-Lema\^itre-Robertson-Walker (FLRW) spacetimes generically develop curvature singularities in the contracting time direction along spacelike hypersurfaces, known as \textit{big bang singularities}, both in vacuum and for a wide range of matter models. 
More generally, the Penrose and Hawking singularity theorems \cite{hawkingLargeScaleStructure1973} guarantee that cosmological spacetimes will be geodesically incomplete in the contracting direction for a large class of matter models and initial data sets that can include highly anisotropic ones. The geodesics incompleteness is widely expected to be due to the formation of curvature singularities, and it is an outstanding problem in mathematical cosmology to rigorously establish the conditions under which this expectation is true and to understand the dynamical behavior of cosmological solutions near singularities. 

The well known BKL-conjecture \cite{belinskii1970,lifshitz1963} suggests that singularities should generically be big bang singularities that are spacelike and oscillatory. However, currently there is little in the way of rigorous arguments beyond the work of \cite{beguin:2010,liebscher_et_al:2011,Ringstrom:2001} in the spatially homogeneous setting to support the BKL-picture. While it is worth noting that there is some numerical support \cite{Andersson:2005, curtis2005, garfinkle2002,garfinkle2002a,garfinkle2004,garfinkle2007,Weaver:2001}, recent work on \emph{spikes} \cite{berger1993,coley2015, lim2008, lim2009, rendall2001} and \emph{weak null singularities} \cite{dafermos2017,luk2017} indicate that the BKL-picture is incomplete.
At the moment, understanding cosmological solutions near generic singularities seems out of reach. However, the situation improves considerably for cosmological spacetimes that exhibit \emph{asymptotically velocity term dominated} (AVTD) behaviour \cite{Eardley:1972,Isenberg:1990} near the singularity. By definition, AVTD singularities are a special type of big bang type singularities (more details in Section~\ref{sec:AVTDAPK}) that are spacelike and  non-oscillatory.
AVTD behavior has been shown to occur generically in the class of vacuum Gowdy spacetimes \cite{CIM1990,Isenberg:1990,ringstrom2009a}, and for infinite dimensional families of cosmological spacetimes with prescribed asymptotics near the singularity in a variety of settings using Fuchsian methods \cite{ames2013a,andersson2001,beyer2017,choquet-bruhat2006, choquet-bruhat2004,ChruscielKlinger:2015,Clausen2007,damour2002,Fournodavlos:2020,heinzle2012,isenberg1999,isenberg2002,kichenassamy1998,stahl2002}. 

Recently, in the article \cite{Fournodavlos_et_al:2023}, a significant advance has been made in understanding AVTD behavior near big bang singularities for solutions of the Einstein-scalar field equations\footnote{See Section \ref{indexing} for our indexing conventions. Note also that \eqref{ESF.1} can be expressed in the more familiar form $\Gb_{ij}=2\Tb_{ij}$ where $\Tb_{ij} = \nablab_i\phi\nablab_j\phi - \frac{1}{2}|\nablab\phi|^2_g \gb_{ij}$ is the scalar field stress-energy tensor.} 
\begin{align}
\Rb_{ij}&=2\nablab_i\phi\nablab_j \phi, 
\label{ESF.1}\\
\Box_{\gb} \phi &=0, \label{ESF.2}
\end{align}
where here, $\nablab_i$ denotes the Levi-Civita connection of $\gb_{ij}$, $\Rb_{ij}$ is the Ricci tensor of $\gb_{ij}$ and $\Box_{\gb} =\gb^{ij}\nablab_i\nablab_j$ is the wave operator.
Early work on the Einstein-scalar field equations was carried out in \cite{andersson2001,barrow1978,berger1999}. Building on these results and the more recent work \cite{RodnianskiSpeck:2018b,RodnianskiSpeck:2018c,RodnianskiSpeck:2021,Speck:2018}, the authors of \cite{Fournodavlos_et_al:2023} show
that small nonlinear perturbations of the Kasner family of spacetimes are stable in the contracting direction and terminate in spacelike, non-oscillatory big bang singularities in the following settings and spacetime dimensions $n$:  Einstein-scalar field equations $(n\geq 4)$, the polarized $\mathbb{U}(1)$-symmetric vacuum Einstein equations $(n=4)$, and the vacuum Einstein equations $(n\geq 11)$. Remarkably, these stability results hold for  the full range of Kasner exponents where stable singularity formation is expected.

The past stability proof established in  \cite{Fournodavlos_et_al:2023} relies, in an essential way, on a foliation of spacetime by level sets of a function $t$ that are spacelike and have constant mean curvature (CMC). The importance of the function $t$ is that it defines a time coordinate for which the singularity can be shown to occur along the hypersurface $t=0$. In this sense, this choice of time coordinate \textit{synchronizes} the singularity, which is important because it allows statements to be made about the behavior of the physical fields as the singularity is approached, i.e.\ in the limit $t\searrow 0$, that are uniform across the whole singular surface.

One consequence of the use of CMC foliations in \cite{Fournodavlos_et_al:2023} is that the stability results obtained there depend \textit{non-locally} on the initial data, which implies, in particular, that these stability results cannot be localized to truncated cone shaped domains that reach the singularity. On the other hand, it is well known that there exists gauge choices, e.g.\ wave gauges, for which solutions of the Einstein equations depend on the initial data in a local fashion. 
These considerations lead to the natural question of whether or not it is possible to establish a local version of the stable singularity formation result from \cite{Fournodavlos_et_al:2023}. One of the main purposes of this article is to show that this is indeed possible in the restricted setting of nonlinear perturbations of FLRW solutions to the Einstein-scalar field equations in $n\geq 3$ spacetime dimensions. As the FLRW solutions to Einstein-scalar field equation are special cases of the family of Kasner-scalar field solutions (see Section~\ref{sec:Kasnerscalarfield}),  our results only establish the nonlinear stability of Kasner-scalar field solutions that are sufficiently close to FLRW solutions. It would, of course, be interesting to determine if the stability results established in this article continue to hold for the full range of expected Kasner exponents as in \cite{Fournodavlos_et_al:2023}. However, this would require a more detailed analysis of certain terms in the Einstein-scalar field equations than we carry out in this article. We plan on returning to this question in future work.

\subsection{Conformal Einstein-scalar field equations}
Before we discuss the main results of this article further, we first reformulate the Einstein-scalar field equations in a way that will be more favorable to the analysis that we carry out in this article. The reformulation begins with replacing the physical metric $\gb_{ij}$ with a \textit{conformal metric} $g_{ij}$ defined by
\begin{equation}\label{confmet}
\gb_{ij} = e^{2\Phi}g_{ij},
\end{equation}
where here and throughout this article, we
assume that the spacetime dimension $n$ satisfies $n\geq 3$.
Under the conformal transformation \eqref{confmet},
it is well known that the Ricci tensor transforms according to 
\begin{equation}\label{confRicci}
\Rb_{ij}=R_{ij}-(n-2)\nabla_i \nabla_j \Phi + (n-2)\nabla_i \Phi \nabla_j \Phi -(\Box_g \Phi + (n-2)|\nabla\Phi|^2_g)g_{ij}
\end{equation}
where $\nabla_i$ is the Levi-Civita connection of the conformal metric $g_{ij}$ and as above,
$\Box_g=g^{ij}\nabla_i\nabla_j$.
For use below, we recall that the connection coefficents of the
metrics $\gb_{ij}$ and $g_{ij}$ are related by
\begin{equation}\label{confChrist}
\Gammab_{i}{}^k{}_j-\Gamma_{i}{}^k{}_j=g^{kl}(g_{il}\nabla_j \Phi + g_{jl}\nabla_i\Phi -g_{ij}\nabla_l\Phi).
\end{equation}
For now, the connection coefficients $\Gammab_i{}^k{}_j$ and $\Gamma_i{}^k{}_j$ 
are to be understood as being determined by an arbitrary frame $e_i=e_i^\mu\del{\mu}$. Later on different choices for the frame will be made; see Sections \ref{sec:locexist_cont} and \ref{Fermi} for details. 

Now, using \eqref{confRicci}, we can express the Einstein equations \eqref{ESF.1} in the form
\begin{equation}\label{confESFA}
-2R_{ij}=-2(n-2)\nabla_i \nabla_j \Phi + 2(n-2)\nabla_i \Phi \nabla_j \Phi -2(\Box_g \Phi + (n-2)|\nabla\Phi|^2_g)g_{ij}-4\nabla_{i}\phi\nabla_{j} \phi.
\end{equation}
Also, by introducing a background metric $\gc_{ij}$ and letting $\Dc_i$ and $\gamma_{i}{}^k{}_{j}$ denote the associated Levi-Civita
connection and connection coefficients, we can write the scalar field equation \eqref{ESF.2} as
\begin{equation*}
\gb^{ij}\Dc_{i}\Dc_{j}\phi-\gb^{ij}(\Gammab_{i}{}^k{}_j-\Gamma_i{}^k{}_j+\Cc_{i}{}^k{}_{j})\Dc_{k}\phi = 0
\end{equation*}
where
\begin{equation}\label{Ccdef}
\Cc_{i}{}^k{}_j := \Gamma_{i}{}^k{}_j-\gamma_{i}{}^k{}_j=\frac{1}{2}g^{k l}\bigl(\Dc_i g_{j l}+\Dc_j g_{i l}-\Dc_l g_{ij}\bigr).
\end{equation}
It is then not difficult to verify using \eqref{confmet} and \eqref{confChrist} that the scalar field equation can be expressed as
\begin{equation} \label{confESFB}
g^{ij}\Dc_{i}\Dc_{j}\phi= X^k\Dc_{k}\phi -(n-2)g^{ij}\Dc_{i}\Phi\Dc_{j}\phi,
\end{equation}
where
\begin{equation} \label{Xdef}
X^k := g^{ij}\Cc_{i}{}^k{}_j=\frac{1}{2}g^{ij}g^{k l}\bigl(2\Dc_i g_{j l}-\Dc_l g_{ij}\bigr),
\end{equation}
or equivalently as
\begin{equation} \label{confESFC}
\Box_g \phi = -(n-2)\nabla^i\Phi\nabla_i\phi.
\end{equation}
Note that, here and below, all indices will be raised and lowered using the conformal metric, e.g.  $\nabla^k\Phi = g^{kl}\nabla_l\Phi$. 

So far the scalar field $\Phi$ used to define the conformal transformation \eqref{confmet} has been arbitrary. We now fix it by setting
\begin{equation}\label{gaugefixA}
\Phi =\lambda \phi
\end{equation}
where
\begin{equation} \label{lambdafix}
 \lambda = \sqrt{\frac{2}{(n-2)(n-1)}}.
\end{equation}
We also define a scalar field $\tau$ via
\begin{equation} \label{taudef}
 \tau=e^{-\alpha\phi}  \quad \Longleftrightarrow \quad \phi = -\frac{1}{\alpha}\ln(\tau)
\end{equation}
where
\begin{equation}\label{alphafix}
\alpha= \frac{\lambda^2(n-2)-2}{\lambda (n-2)}=-\sqrt{\frac{2(n-2)}{n-1}}.
\end{equation}
Then, by \eqref{confESFC},  $\Phi$ satisfies
\begin{equation*}
\Box_g \Phi +(n-2)\nabla^i\Phi\nabla_i\Phi=0.
\end{equation*}
Substituting this into \eqref{confESFA}, we find after  using \eqref{gaugefixA}-\eqref{alphafix} to replace $\Phi$ and $\phi$ with $\tau$ that
\begin{equation} \label{confESFAa}
R_{ij}=\frac{1}{\tau}\nabla_i \nabla_j \tau.
\end{equation}
We see also from a straightforward calculation that the wave equation \eqref{confESFC}, when expressed in terms of the scalar field $\tau$, becomes
\begin{equation*}
    \Box_g \tau = 0.
\end{equation*}
Collecting together the above two equations, we can can express them as
\begin{align}
G_{ij}&=\frac{1}{\tau}\nabla_i \nabla_j \tau,
\label{confESFAb} \\
\Box_g \tau &= 0, \label{confESFCb}
\end{align}
where $G_{ij}$ is the Einstein tensor of the conformal metric $g_{ij}$. We will refer to these equations as the \textit{conformal Einstein-scalar field equations}. It follows from \eqref{confmet}  and \eqref{gaugefixA}-\eqref{alphafix} that each solution $\{g_{ij},\tau\}$ of the conformal Einstein-scalar field equations yields the solution
\begin{equation}
  \label{eq:conf2phys}
  \biggl\{\gb_{ij}=\tau^{\frac {2}{n -2}}g_{ij},\,\phi=\sqrt{\frac{n-1}{2(n-2)}}\ln(\tau)\biggr\}
\end{equation}
of the \textit{physical Einstein-scalar field equations} \eqref{ESF.1}-\eqref{ESF.2}. We remark that since \eqref{ESF.1}-\eqref{ESF.2} are invariant under the transformation $\phi\mapsto -\phi$, our sign convention for $\phi$ in \eqref{eq:conf2phys} incurs no loss of generality.

\subsection{Kasner-scalar field spacetimes}
\label{sec:Kasnerscalarfield}
In our conformal picture, the  \textit{Kasner-scalar field spacetimes}, which are solutions of the conformal Einstein-scalar field equations
\eqref{confESFAb}-\eqref{confESFCb}, are determined
by the following conformal metric and scalar
field
\begin{equation}\label{gtau-Kasner}
\gbr = -t^{\rbr_0}dt\otimes dt + \sum_{\Lambda=1}^{n-1} t^{\rbr_\Lambda} dx^\Lambda \otimes dx^\Lambda
\AND
\taubr= t,
\end{equation}
respectively, which are defined on the spacetime manifold $M^{(K)}=\Rbb_{>0}\times \Tbb^{n-1}$; see Section \ref{indexing} for our coordinate and indexing conventions.
In the above expressions, the constants
$\rbr_\mu$ are called \emph{Kasner exponents} and are defined by
\begin{equation}
    \rbr_0 = \frac{1}{\Pbr}\sqrt{\frac{2(n-1)}{n-2}}-\frac{2(n-1)}{n-2} \AND
    \rbr_\Lambda =\frac{1}{\Pbr} \sqrt{\frac{2(n-1)}{n-2}}\qbr_\Lambda
    -\frac{2}{n-2}, \label{rmu-def}
\end{equation}
where $0<\Pbr\le \sqrt{(n-2)/(2(n-1))}$ and the $\qbr_\Lambda$ satisfy the \textit{Kasner relations} 
\begin{equation}\label{qLambda-def}
    \sum_{\Lambda=1}^{n-1}\qbr_\Lambda =1 \AND 
    \sum_{\Lambda=1}^{n-1}\qbr_\Lambda^2 = 1-2 \Pbr^2. 
\end{equation}
From \eqref{gtau-Kasner}-\eqref{qLambda-def} as well as \eqref{confmet} and \eqref{gaugefixA}-\eqref{alphafix}, it follows 
from the curvature expressions
\begin{equation*}
  \Rb_{\mu\nu}\Rb^{\mu\nu}=\left(\frac{n-1}{n-2}\right)^2 t^{-4 \frac{n-1}{n-2}-2\rbr_0},\quad
  \Rb=-\frac{n-1}{n-2} t^{-2 \frac{n-1}{n-2}-\rbr_0},
\end{equation*}
for the physical  metric $\gb_{\mu\nu}=t^{\frac {2}{n -2}}\gbr_{\mu\nu}$ that the \textit{Kasner big bang singularity} occurs along
the spacelike hypersurface $\{0\}\times \Tbb^{n-1}$.

Noting from \eqref{rmu-def} and \eqref{qLambda-def} that
\begin{equation*}
 \sum_{\Lambda=1}^{n-1} \rbr_\Lambda=\rbr_0,
\end{equation*}
which in turn implies by \eqref{gtau-Kasner} that
\begin{equation*}
    -\gbr^{00}\sqrt{|\det(\gbr_{\alpha\beta})|} = t^{\frac{1}{2}\bigl(\rbr_0+\sum_{\Lambda=1}^{n-1} \rbr_\Lambda\bigr)-\rbr_0} = 1,
\end{equation*}
we observe from  \eqref{gtau-Kasner} and a short calculation
that
\begin{equation*}
\Box_{\gbr}x^\gamma = \frac{1}{\sqrt{|\det(\gbr_{\alpha\beta})|}}\del{\mu}\Bigl(\sqrt{|\det(\gbr_{\alpha\beta})|}\gbr^{\mu\nu}\del{\nu}x^\gamma\Bigr) = 0. 
\end{equation*}
This shows the  $(x^\mu)$ are \textit{wave coordinates} and  that
\begin{equation} \label{Kasner-wave-gauge}
\gbr^{\mu\nu}\Gammabr_{\mu\nu}^\gamma = 0
\end{equation}
is satisfied, where, here, $\Gammabr_{\mu\nu}^\gamma$ denotes the Christoffel symbols of the conformal Kasner metric $\gbr_{\mu\nu}$. It is also useful to observe that the lapse $\Ntt$ and the Weingarten map induced on the $t=const$-hypersurfaces by the conformal Kasner metric $\gbr_{\mu\nu}$ are
\begin{equation}
  \label{eq:KSFlapse2ndFF}
  \Ntt=t^{\frac{\rbr_0}{2}} \AND (\Ktt_\Lambda{}^\Omega)=\frac 12t^{-1-\frac{\rbr_0}{2}}\diag(\rbr_1,\ldots,\rbr_{n-1}),
\end{equation}
respectively, and that when we express a Kasner-scalar field solution with respect to the time coordinate
\[\tb=\frac{1}{\frac{n-1}{n-2}+\frac{\rbr_0}2} t^{\frac{n-1}{n-2}+\frac{\rbr_0}2},\]
and appropriately rescale all the spatial coordinates $(x^\Lambda)$, then the physical solution \eqref{eq:conf2phys} corresponding to \eqref{gtau-Kasner} takes the more conventional form
\begin{equation}\label{gtau-Kasnerphys}
\overline{\gbr} = -d\bar t\otimes d\bar t + \sum_{\Lambda=1}^{n-1}
{\bar t}^{2\qbr_\Lambda} d\xb^\Lambda \otimes d\xb^\Lambda
\AND
\breve\phi= \Pbr\ln(\bar t)+\Pbr\ln\Bigl(\frac{n-1}{n-2}+\frac{\rbr_0}2\Bigr),
\end{equation}
where $\Pbr$ and $\qbr_\Lambda$ are all related to $\rbr_0$ and $\rbr_\Lambda$ by \eqref{rmu-def}. In particular, the constant $\Pbr$ can be interpreted as the \emph{asymptotic scalar field strength}. Noticing that since $\Pbr=0$ is not allowed, the special case of vacuum Kasner solutions are not covered by our conformal representation of the Kasner-scalar field solutions.

\subsubsection{FLRW-scalar field spacetimes}
\label{sec:KasnerscalarfieldFLRW}
Kasner spacetimes where the constants $q_\Lambda$ are all the same coincide with FLRW spacetimes. In this situation, we have, by \eqref{qLambda-def}, that 
\begin{equation*}
    \qbr_\Lambda = \frac{1}{n-1}.
\end{equation*}
This corresponds to the extreme case
\[|\Pbr|=\sqrt{\frac{n-2}{2(n-1)}}\]
that saturates the inequality  $|\Pbr|\le \sqrt{(n-2)/(2(n-1))}$.
These choices imply by \eqref{rmu-def} that $\rbr_0=\rbr_\Lambda=0$ and that the conformal Kasner-scalar field solution \eqref{gtau-Kasner} simplifies to
\begin{align}\label{gtau-FLRW}
\gbr = -dt\otimes dt + \delta_{\Lambda\Omega}dx^\Lambda \otimes dx^\Omega
\AND
\taubr= t.
\end{align}
We further note via a straightforward calculation that the FLRW solution (as well as all Kasner-scalar field solutions of the form \eqref{gtau-Kasner}) 
satisfies
\begin{equation} \label{FLRW-Lag}
\frac{1}{|\breve{\nabla}\taubr|_{\gbr}^2}\breve{\nabla}{}^\mu \taubr=\delta^\mu_0,
\end{equation}
which, in terminology that we will introduce below, implies that the coordinates $(x^\mu)=(t,x^\Lambda)$ used to define the FLRW metric \eqref{gtau-FLRW} are \textit{Lagrangian}. By \eqref{Kasner-wave-gauge}, these coordinates are also wave coordinates, that is, $\Box_{\gb}x^\mu =0$. Both of these gauge conditions will play a pivotal role in the proof of our stability results.

\subsection{AVTD and asymptotic pointwise Kasner behavior}
\label{sec:AVTDAPK}
In this article, we analyze solutions of the Einstein-scalar field system that are \emph{near} Kasner-scalar field solutions, specifically near the FLRW solution. Even though the spacetimes that are generated by our stability results are generically spatially inhomogeneous and without any symmetries, they do retain some of the asymptotic properties that characterize the spatially homogeneous Kasner-scalar field solutions.
For example, we show that spatial derivative terms that appear in the dynamical equations become negligible at $t=0$ in comparison to time derivative terms, the so-called \emph{velocity terms} \cite{Eardley:1972,Isenberg:1990}. This AVTD behaviour means that the dynamical equations could be approximated by ODEs close to the big bang.  In agreement with \cite{Fournodavlos_et_al:2023}, we define that a solution to the conformal Einstein-scalar field system satisfies the AVTD property provided it satisfies the VTD equations -- the equations obtained from the main evolution system by removing all spatial derivative terms and by normalising all time derivative terms -- up to an error term that is integrable in time near $t=0$. By \textit{normalising all time derivatives}, we specifically mean that the evolution equations are put into Fuchsian form as discussed in Sections~\ref{sec:FuchsianForm}~and~\ref{sec:proof_globstab_Fuchsian}.

Given that spatial inhomogeneities of AVTD solutions become irrelevant at the big bang, it is not surprising that they behave essentially like spatially homogeneous solutions and satisfy the following \emph{asymptotic pointwise Kasner} property.
\begin{Def}
\label{def:APKasner}
Given a $C^2$-solution $(M=(0,t_0]\times U, g_{\mu\nu},\tau)$, $t_0>0$, of the conformal Einstein-scalar field equations,
where $U\subset \Tbb^{n-1}$ is open and  $(x^\mu)=(t,x^\Lambda)$ are coordinates on $M$ such that $t\in (0,t_0]$, $\tau=t$ and the $(x^\Lambda)$ are periodic coordinates on $\Tbb^{n-1}$, we say that the spacetime $(M, g_{\mu\nu},\tau)$ is \textit{asymptotically pointwise Kasner on $U$} provided there exists a frame $e_i=e_i^\mu\del{\mu}$ and a continuous spatial tensor field $\kf_{I}{}^J$ on $U$ such that the following hold:
\begin{enumerate}[(i)]
\item The spatial vector fields $e_I$ are tangential to the $t=const$-surfaces, i.e. $e_I=e^\Lambda \del{\Lambda}$, and
\begin{equation}
  \label{eq:asymptptwKasner2}
  \lim_{t\searrow 0} \bigl|2t\,\Ntt(t,x)\, \Ktt_{I}{}^J (t,x)-\kf_{I}{}^J(x)\bigr|=0
\end{equation}
for each $x\in U$, where $\Ntt$ is the lapse and $\Ktt_{I}{}^{J}$ is the Weingarten map induced on the $t=\mathrm{const}$-hypersurfaces by $g_{\mu\nu}$. 
\item The tensor field $\kf_I{}^J$ satisfies $\kf_I{}^I\geq 0$ and the \emph{Kasner relation}
\begin{equation}
  \label{eq:asymptptwKasner}
  (\kf_{I}{}^{I})^2 -
  \kf_{I}{}^J \kf_{J}{}^I 
  +4\kf_{I}{}^{I}=0
\end{equation}
everywhere on $U$. At each point $x\in U$, the symmetry of $\kf_I{}^J(x)$ guarantees that $\kf_I{}^J(x)$ has $n-1$ real eigenvalues, which we denote by $r_1(x),\ldots,r_{n-1}(x)$. We refer to the functions $r_1,\ldots,r_{n-1}$ on $U$ as the \emph{Kasner exponents}\footnote{Notice that it is more customary in the literature to call the quantities $q_1,\ldots, q_{n-1}$ defined in \eqref{eq:r0qIdef} \emph{Kasner exponents} and \eqref{eq:Kasnerrel} the \emph{Kasner relations}.}.
\end{enumerate}
\end{Def}

The Kasner relation \eqref{eq:asymptptwKasner} and the Kasner exponents $r_1,\ldots,r_{n-1}$ from the above definition should be viewed as a natural generalization to inhomogeneous spacetimes of the Kasner relations \eqref{qLambda-def} and exponents from Section \ref{sec:Kasnerscalarfield}. This can be justified
by noting that \eqref{eq:asymptptwKasner} is invariant under the choice of spatial frame.  By employing an eigenframe of $\kf_{J}{}^I$ to define the Kasner exponents $r_1,\ldots,r_{n-1}$ of an asymptotically pointwise Kasner spacetime,
the Kasner relation \eqref{eq:asymptptwKasner} can be expressed as
\begin{equation}
  \label{eq:asymptptwKasnerexpo}
  \biggl(\sum_{\Lambda=1}^{n-1}r_\Lambda\biggr)^2
  - \sum_{\Lambda=1}^{n-1}r_\Lambda^2
  +4\sum_{\Lambda=1}^{n-1}r_\Lambda=0.
\end{equation}
To put \eqref{eq:asymptptwKasnerexpo} in the same form as the homogenous Kasner relations \eqref{qLambda-def}, we define functions
$r_0$ and $q_\Lambda$ on $U$ by
\begin{equation}
  \label{eq:r0qIdef}
  r_0:=\sum_{\Lambda=1}^{n-1}r_\Lambda=\kf_{I}{}^{I} \AND
  q_\Lambda:=\Ptt\sqrt{\frac{n-2}{2(n-1)}}\Bigl(r_\Lambda+\frac{2}{n-2}\Bigr),
\end{equation}
where the function $\Ptt$ will be fixed below. From these definitions and \eqref{eq:asymptptwKasner}, we see that
\begin{equation*}
  \sum_{\Lambda=1}^{n-1}q_\Lambda=\Ptt\sqrt{\frac{n-2}{2(n-1)}}\Bigl(r_0+\frac{2(n-1)}{n-2}\Bigr) \AND \sum_{\Lambda=1}^{n-1}q_\Lambda^2 =\Ptt^2\frac{n-2}{2(n-1)}\Bigl(r_0+\frac{2(n-1)}{n-2}\Bigr)^2 
     -2\Ptt^2,
\end{equation*}
and it is then clear that the $q_\Lambda$ will satisfy the ``standard'' Kasner relations
\begin{equation}\label{eq:Kasnerrel}
  \sum_{\Lambda=1}^{n-1}q_\Lambda =1 \AND 
  \sum_{\Lambda=1}^{n-1}q_\Lambda^2 = 1-2 \Ptt^2
\end{equation}
on $U$ provided we set
\begin{equation}
  \label{eq:pdef}
  \Ptt=\sqrt{\frac{2(n-1)}{n-2}}\Biggl(\sum_{\Lambda=1}^{n-1}r_\Lambda+\frac{2(n-1)}{n-2}\Biggr)^{-1}
=\frac{\sqrt{{2(n-1)}(n-2)}}{{2(n-1)} +(n-2)\sum_{\Lambda=1}^{n-1}r_\Lambda},
\end{equation}
which we note is well-defined since $\kf_I{}^I=\sum_{\Lambda=1}^{n-1}r_\Lambda\geq 0$ by assumption. From \eqref{eq:pdef}, we also observe that $\Ptt$ is bounded by  $0<\Ptt\le \sqrt{(n-2)/(2(n-1))}$, which is the same bound as in  Section~\ref{sec:Kasnerscalarfield}.

An immediate consequence of the above discussion is that the Kasner-scalar field spacetimes defined above in Section~\ref{sec:Kasnerscalarfield} are asymptotically pointwise Kasner on $\Tbb^{n-1}$ in the sense of Definition \ref{def:APKasner}. The virtue of Definition~\ref{def:APKasner} is that it continues to makes sense for inhomogeneous spacetimes. It is also worthwhile noting at this point that it is shown in Section~\ref{sec:proof_globstab} that the Kasner relation \eqref{eq:asymptptwKasner} arises from the $t\searrow 0$ limit of a rescaled version of the physical Hamiltonian constraint.

In analogy with the Kasner-scalar field solutions \eqref{gtau-Kasnerphys}, the function $\Ptt$ defined by \eqref{eq:pdef} for asymptotically pointwise Kasner spacetimes can be interpreted as the \emph{asymptotic scalar field strength} provided there exist continuous positive functions  $\bfr$ and $\nu$ such that
\begin{equation}
  \label{eq:lapselimit}
  \Bigl|t^{-\frac{1}{2}r_0(x)}\Ntt (t,x)-\bfr(x)\Bigr|\lesssim t^{\nu(x)}
\end{equation}
for all $t\in (0,t_0]$ and each $x\in U$. This property clearly holds for the Kasner-scalar field spacetimes due to  \eqref{eq:KSFlapse2ndFF}, and we will show that it also holds for the entire class of perturbations of the FLRW solutions that are generated from our stability results. 

Assuming that \eqref{eq:lapselimit} is satisfied for an asymptotically pointwise Kasner spacetime, we can introduce
a new time coordinate $\tb=\tb(t,x)$ 
to make the component $\gb_{\tb\tb}$ of the physical metric $\gb$ with respect to this new time coordinate equal to minus one by demanding that
\[\frac{\partial \tb(t,x)}{\partial t}=\Ntt(t,x)t^{\frac{1}{n-2}}
  =\bfr(x) t^{\frac{1}{n-2}+\frac{1}{2}r_0(x)} + (t^{-\frac{1}{2}r_0(x)}\Ntt(t,x)-\bfr(x)) t^{\frac{1}{n-2}+\frac{1}{2}r_0(x)}.\]
It is not difficult to verify from \eqref{eq:r0qIdef} and \eqref{eq:pdef} that there is a solution to this equation with the property
\[\Biggl|\tb(t,x)-\sqrt{\frac {2(n-2)}{n-1}}\bfr(x) {\Ptt(x)}t^{\frac 1{\Ptt(x)}\sqrt{\frac{n-1}{2(n-2)}}}\Biggr|\lesssim t^{\nu(x)+\frac 1{\Ptt(x)}\sqrt{\frac{n-1}{2(n-2)}}}\]
for all $(t,x)\in (0,t_0]\times U$.
Since this implies
\begin{align*}
  &\Biggl|\ln(\tb(t,x))-\ln\Biggl(\sqrt{\frac {2(n-2)}{n-1}}\bfr(x) {\Ptt(x)}t^{\frac 1{\Ptt(x)}\sqrt{\frac{n-1}{2(n-2)}}}\Biggr)\Biggr|
    \lesssim t^{\nu(x)},
\end{align*}
we get from \eqref{eq:conf2phys}, recalling $\tau=t$, that 
\begin{align*}
  \Biggl|\phi(t)-\Biggl(&\Ptt(x)\ln( \tb(t,x))-\Ptt(x) \ln\Biggl(\sqrt{\frac {2(n-2)}{n-1}}\bfr(x) {\Ptt(x)}\Biggr)\Biggr)\Biggr|  
  \le \Biggl|\sqrt{\frac{n-1}{2(n-2)}}\ln(t) \\
  &-\Ptt(x) \ln\biggl(\sqrt{\frac {2(n-2)}{n-1}}\bfr(x) {\Ptt(x)}t^{\frac 1{\Ptt(x)}\sqrt{\frac{n-1}{2(n-2)}}}\biggr)+\Ptt(x) \ln\biggl(\sqrt{\frac {2(n-2)}{n-1}}\bfr(x) {\Ptt(x)}\biggr)\Biggr|\\
  &\hspace{4.0cm} +\Ptt(x)\Biggl|\ln(\tb(t,x))-\ln\Bigl(\sqrt{\frac {2(n-2)}{n-1}}\bfr(x) {\Ptt(x)}t^{\frac 1{\Ptt(x)}\sqrt{\frac{n-1}{2(n-2)}}}\Bigr)\Biggr|
  \lesssim t^{\nu(x)}
\end{align*}
for all  $(t,x)\in (0,t_0]\times U$.
In analogy to the Kasner-scalar field solutions \eqref{gtau-Kasnerphys}, the function $\Ptt$ defined by \eqref{eq:pdef} can therefore be interpreted as the \textit{asymptotic scalar field strength}.

Condition \eqref{eq:asymptptwKasner2} is consistent with the behaviour of Kasner-scalar field solutions, see \eqref{eq:KSFlapse2ndFF}, and therefore supports the assertion that asymptotically pointwise Kasner solutions behave asymptotically at each spatial point like a Kasner-scalar field spacetime. The second fundamental form $\bar\Ktt_{\Lambda\Omega}$ induced on $t=const$-surfaces by the physical metric $\gb_{\mu\nu}$ is related to the one $\Ktt_{\Lambda\Omega}$ induced by the conformal metric $g_{\mu\nu}$ via
\[\bar\Ktt_\Lambda{}^{\Omega}=\frac 12 t^{-\frac {1}{n -2}} t^{-1}\Ntt^{-1}\Bigl(2 t\Ntt\Ktt _{\Lambda}{}^{\Omega}
    +\frac {2}{n-2}{\delta}_{\Lambda}{}^{\Omega}\Bigr).\]
Because of this, \eqref{eq:asymptptwKasner2} and \eqref{eq:lapselimit}  imply that both the mean curvatures associated with the physical and with the conformal metric diverge pointwise near $t=0$, except in the case of FLRW where $K_{\Lambda}{}^{\Lambda}$ vanishes while $\bar{K}_{\Lambda}{}^{\Lambda}$ diverges. Given suitable uniform bounds over the spatial domain $U$, which we establish hold in our main results, asymptotically pointwise Kasner metrics will have \emph{crushing singularities} at $t=0$  in the language of \cite{eardley1979}.

\subsection{An informal statement of the main results}

The main results of this article are given in Theorem~\ref{glob-stab-thm} and Theorem~\ref{loc-stab-thm}. The first theorem 
establishes the nonlinear stability in the contracting direction of perturbations of the FLRW solution \eqref{gtau-FLRW} to the Einstein-scalar field in $n\geq 3$ spacetime dimensions. The perturbed solutions are shown to be asymptotically pointwise Kasner and terminate in a big bang singularity. In this theorem, the initial data is prescribed on the hypersurface $\{t_0\}\times \Tbb^{n-1}$, $t_0>0$, and is assumed to be \textit{synchronized}, that is, $\tau=t_0$ on $\{t_0\}\times \Tbb^{n-1}$. The purpose of this synchronization condition is to ensure that the big bang singularity occurs at $\tau=0$. By the results of Section \ref{temp-synch}, no generality is lost by restricting our attention to synchronized initial data.
Moreover, because the initial data for Theorem~\ref{glob-stab-thm} is specified on a closed hypersuface, this theorem should be interpreted as a past global stability result. The informal statement of this theorem is as follows.
\begin{thm}[Past global stability of the FLRW solution of the Einstein-scalar field equations]
Solutions $\{g_{ij},\tau\}$ of the conformal Einstein-scalar field equations that are 
generated from sufficiently differentiable, synchronized initial data imposed on $\{t_0\}\times \Tbb^{n-1}$ that is suitably close to the FLRW data exist on the spacetime region $M \cong \bigcup_{t\in (0,t_0]}\tau^{-1}(\{t\})\cong (0,t_0]\times \Tbb^{n-1}$. Moreover, these solutions are asymptotically pointwise Kasner
and the corresponding physical solutions $\{\gb_{ij},\phi\}$ of the Einstein-scalar field equations are $C^2$-inextendible through the $\tau=0$-boundary of $M$, past timelike geodesically incomplete and terminate at a crushing big bang singularity at $\tau=0$ that is characterized by curvature blow-up.  
\end{thm}

In contrast, initial data in the second stability theorem, Theorem~\ref{loc-stab-thm}, is specified on an open centred ball $\{t_0\}\times \mathbb{B}_{\rho_0}$ in $\{t_0\}\times \Tbb^{n-1}$; see Section \ref{domains} below for a definition of $\mathbb{B}_{\rho_0}$. If the initial data is synchronized and chosen to be sufficiently close to the FLRW initial data on $\{t_0\}\times \mathbb{B}_{\rho_0}$, then  Theorem~\ref{loc-stab-thm} guarantees that the solutions generated from this initial data will be asympotically Kasner and terminate in a big bang singularity. Since the initial data is specified on an open set in $\{t_0\}\times \Tbb^{n-1}$, this stability result should be viewed as a localized version of Theorem~\ref{glob-stab-thm}. An informal statement of Theorem~\ref{loc-stab-thm} 
is given by the following.

\begin{thm}[Localized past stability of the FLRW solution of the Einstein-scalar field system]
Given $\vartheta\in (0,1)$ and an open centred ball $\mathbb{B}_{\rho_0}\subset \Tbb^{n-1}$, there exist a $t_0>0$ and $\bar{\ep}\in (0,1)$ such that solutions $\{g_{ij},\tau\}$ of the conformal Einstein-scalar field equations that are 
generated from sufficiently differentiable, synchronized initial data imposed on $\{t_0\}\times \mathbb{B}_{\rho_0}$ that is suitably close to the FLRW data exist on the spacetime region 
   $ M \cong \bigcup_{t\in (0,t_0]}\tau^{-1}(\{t\})\cong \bigcup_{t\in (0,t_0]} \{t\}\times \mathbb{B}_{\rho(t)}$
where $\rho(t) = \rho_0 + (1-\vartheta) \rho_0 \bigl( \bigl(\frac{t}{t_0}\bigr)^{1-\bar{\ep}}-1\bigr)$. Moreover, these solutions are asymptotically pointwise Kasner
and the corresponding physical solutions $\{\gb_{ij},\phi\}$ of the Einstein-scalar field equations are $C^2$-inextendible through the $\tau=0$-boundary of $M$ and terminate at a crushing big bang singularity at $\tau=0$ that is characterized by curvature blow-up. 
\end{thm}


\subsection{Prior and related work}
In the articles \cite{ABIO:2021_Royal_Soc,ABIO:2021}, the nonlinear stability in the contracting direction of Kasner solutions
to the vacuum Einstein equation within the class of polarised $\Tbb^2$-symmetric spacetimes was
established for a certain range of the Kasner exponents. The case of zero cosmological constant was considered in \cite{ABIO:2021}  and the results  of \cite{ABIO:2021} were later shown to hold for a non-vanishing cosmological constant and more general initial data in \cite{ABIO:2021_Royal_Soc}. While the past stability results, assuming a zero cosmological constant, from \cite{ABIO:2021_Royal_Soc,ABIO:2021} can be viewed, for the most part, as a special case of the more general stability results obtained in  
\cite{Fournodavlos_et_al:2023}, they are still interesting for the following reasons. First, in the restricted setting of polarised $\Tbb^2$-symmetric spacetimes, more detailed asymptotics near the big bang singularity are obtained in \cite{ABIO:2021_Royal_Soc,ABIO:2021} as compared to \cite{Fournodavlos_et_al:2023}. Second, past stability is established in \cite{ABIO:2021_Royal_Soc} for a class of initial data that would be considered large by the smallness definition employed in \cite{Fournodavlos_et_al:2023}. The third and final reason, which is directly relevant to this article, is the method used to prove the stability is different. The stability proofs from \cite{ABIO:2021_Royal_Soc,ABIO:2021} rely on a particular choice of gauge and variables that allow the Einstein equations to be formulated as a Fuchsian equation that satisfies a specific set of structural conditions. Once that is achieved, stability then follows as a direct consequence of the global existence theory for Fuchsian equations that has been developed in \cite{BOOS:2021,Oliynyk:CMP_2016}. As we discuss in more detail in Section \ref{overview}, the approach we take to establishing stability in this article is broadly the same. It is worth pointing out that the Fuchsian approach to establishing the global existence of solutions to systems of hyperbolic equations is a very general method and has recently been employed to establish a variety of stability results in the articles \cite{BeyerOliynyk:2024,Fajman_et_al:2023,FOW:2021,LeFlochWei:2021,LiuOliynyk:2018b,LiuOliynyk:2018a,LiuWei:2021,MarshallOliynyk:2023,Oliynyk:2021,OliynykOlvera:2021,Wei:2018}.


\subsection{Overview of the past stability proofs\label{overview}}
The proofs of both the past local and global stability theorems rely on two different formulations of the reduced Einstein-scalar field equations. These formulations are used for distinct purposes, which we describe below.

\subsubsection{Local-in-time-existence and continuation in Lagrangian coordinates}
The first formulation of the reduced Einstein-scalar field equations given by \eqref{tconf-ford-C.1}-\eqref{tconf-ford-C.8} below is used to establish the local-in-time existence and uniqueness of solutions
to the reduced conformal Einstein-scalar field equations in a \textit{Lagrangian coordinate system} $(x^\mu)$, which is adapted to the vector field
\begin{equation*}
    \chi^\mu = (|\nabla \tau|^2_g)^{-1}\nabla^\mu\tau,
\end{equation*}
as well as a continuation principle for these solutions. 
Here, Lagrangian coordinates mean that in the coordinate system $(x^\mu)$ the vector field $\chi^\mu$ is trivialized, that is,
\begin{equation*}
    \chi^\mu =\delta^\mu_0.
\end{equation*}
The precise definition of the Lagrangian coordinates $(x^\mu)$ can be found in  
Section \ref{Lag-coordinates}. For initial data that satisfies the gravitational and wave gauge  constraints, the system \eqref{tconf-ford-C.1}-\eqref{tconf-ford-C.8} propagates both of these constraints
and determines solutions of the conformal Einstein-scalar field equations.

An important point regarding the wave gauge constraint
\begin{equation} \label{wave-gauge-intro}
    \frac{1}{2}g^{\gamma \lambda}(2 \Dc_\mu g_{\nu\lambda}-\Dc_\lambda g_{\mu\nu}) = 0,
\end{equation}
which is propagated by \eqref{tconf-ford-C.1}-\eqref{tconf-ford-C.8}, is that the covariant derivative $\Dc_\mu$ it is not determined by a fixed Minkowski metric in the Lagrangian coordinates $(x^\mu)$, and consequently, the Lagrangian coordinates $(x^\mu)$ are \emph{not} wave coordinates, that is, generically $\Box_g x^\mu \neq 0$. Instead, the covariant derivative $\Dc_\mu$ is computed with respect to the flat metric
\begin{equation}\label{gc-def-intro}
    \gc_{\mu\nu} = \del{\mu}l^\alpha \eta_{\alpha\beta} \del{\nu}l^\beta
\end{equation}
where the Lagrangian map $l^\mu(x)$ is determined by a solution of the system \eqref{tconf-ford-C.1}-\eqref{tconf-ford-C.8}; see Section \ref{Lag-coordinates} for details.

The primary role of the Lagrangian coordinates $(x^\mu)$ is to synchronize the singularity. In these coordinates, the scalar field $\tau$ coincides with the time coordinate, that is,
\begin{equation} \label{tau=x0}
\tau = t:= x^0.
\end{equation}
For the Kasner solutions \eqref{gtau-Kasner}, the big bang singularity occurs where $\tau$ vanishes, and due to this, it is reasonable to conjecture that the big bang singularity can be synchronized by using $\tau$ as a time coordinate. As we rigorously establish in our stability theorems, $\tau$ does, in fact, synchronize the big bang singularity for nonlinear perturbations of the FLRW metric \eqref{gtau-FLRW}, which implies that the big bang singularity is located at $\tau=x^0=0$ in the Lagrangian coordinates. For additional discussions regarding the temporal synchronization of the big bang singularity, and, in particular, how this relates to the choice of initial data, see Section \ref{temp-synch}.  

The precise statement of the local-in-time existence and uniqueness result for solutions of the system \eqref{tconf-ford-C.1}-\eqref{tconf-ford-C.8} as well as the continuation principle and propagation of the wave gauge constraints is given in Proposition \ref{lag-exist-prop}. The proof of this proposition is, for the most part, straightforward and relies on standard local-in-time existence and uniqueness theory for systems of symmetric hyperbolic equations and well known wave gauge propagation results for solutions of the reduced Einstein equations. The only novel aspect of the proof is the use of Lagrangian coordinates and a wave gauge that is determined by a flat metric that is not a priori known but determined from a solution of the system \eqref{tconf-ford-C.1}-\eqref{tconf-ford-C.8}. 

\subsubsection{Fuchsian formulation and global-in-time estimates}
While the system \eqref{tconf-ford-C.1}-\eqref{tconf-ford-C.8} is useful for establishing the local-in-time existence of solutions to the reduced conformal Einstein-scalar field equations and the propagation of the wave gauge constraint \eqref{wave-gauge-intro},
it is not useful for establishing global-in-time estimates that can be used in conjunction with the continuation principle to show that solutions can be continued from some starting time $t_0>0$ all the way down to the big bang singularity at $t=0$. 
The system that we do use to establish global-in-time estimates is formulated in terms of a frame $e_i=e_i^\mu\del{\mu}$, the connection coefficients $\gamma_i{}^k{}_j$
of the flat background metric \eqref{gc-def-intro} relative to the frame $e_i$, i.e. $D_{e_i}e_j = \gamma_{i}{}^k{}_j e_j$,
and suitable combinations of the metric and scalar frame fields 
\begin{equation} \label{fields-intro}
\{g_{ijk}=\Dc_i g_{jk},g_{ijkl}=\Dc_{i}\Dc_j g_{kl},\tau_{ij}= \Dc_i\Dc_j \tau, \tau_{ijk}=\Dc_i \Dc_j\Dc_k \tau\}     
\end{equation}
where 
\begin{equation*}
g_{ij}=e_i^\mu g_{\mu\nu}e^\nu_j
\end{equation*}
is the frame representation of the conformal metric and 
$\tau=t$.

The frame $e_i^\mu$ is fixed by first setting
$e_0^\mu = (-|\chi|_g^2)^{-\frac{1}{2}}\chi^\mu$, where 
\begin{equation} \label{chi-triv}
    \chi^\mu=\delta^\mu_0
\end{equation} 
since we are using Lagrangian coordinates $(x^\mu)$. With this choice made, the spatial frame $e_I^\mu$ is then determined by using Fermi-Walker transport, which is defined by
\begin{equation} \label{FWT}
\nabla_{e_0}e_J = -\frac{g(\nabla_{e_0}e_0,e_J)}{g(e_0,e_0)} e_0,
\end{equation}
to propagate initial data $e_I^\mu|_{t=t_0}=\er^\mu_I$ that is chosen so that the frame is orthonormal at $t=t_0$. 
The orthonormality of the frame is preserved by Fermi-Walker transport, which implies, in particular, that the frame metric satisfies 
\begin{equation} \label{FON}
    g_{ij}=\eta_{ij}.
\end{equation} The use of a Fermi-Walker transported spatial frame was inspired by the work of \cite{Fournodavlos_et_al:2023} in which Fermi-Walker transported spatial frames played an essential role in the proof of the stability results established there. 

Due to \eqref{chi-triv}, the frame vector $e^\mu_0$ is determined by
\begin{equation} \label{e0-det}
e_0^\mu = \betat^{-1}\delta^\mu_0
\end{equation}
where 
\begin{equation*}
\betat=(-|\nabla\tau|^2_g)^{-\frac{1}{2}},
\end{equation*}
while it is shown in Section \ref{Fermi} that the Fermi-Walker transport equation \eqref{FWT} implies that the spatial frame components $e^\mu_I$
evolve according to
\begin{gather}
e_I^0=0  \label{eI0-det}
\intertext{and}
\del{t}e^\Lambda_I =\betat (\gamma_0{}^J{}_I - \gamma_I{}^J{}_0)e^\Lambda_J, \notag 
\end{gather}
where we note, see \eqref{for-L}, $\betat$ satisfies
\begin{equation*} 
\del{t}\betat = -\betat^3 \tau_{00}+\frac{1}{2}\betat^2
\delta^{JK}(g_{0JK}-g_{J0K}-g_{K0J}).
\end{equation*}
In Section \ref{Fermi}, we also use the Fermi-Walker transport equation \eqref{FWT} in conjunction with the orthonormality \eqref{FON} of the frame $e_i^\mu$ to determine the background connection coefficients $\gamma_0{}^j{}_k$, $\gamma_I{}^0{}_0$
and $\gamma_I{}^K{}_0$ in terms of the frame fields \eqref{fields-intro}
and the other connection coefficients $\gamma_I{}^k{}_J$ via
\begin{gather}
\gamma_{0}{}^k{}_j 
=
-\betat\delta^i_0\bigl(\delta^0_jg^{lk}+\delta^l_j\delta^k_0\bigr)
(\Dc_{i}\Dc_l\tau-\Cc_i{}^p{}_l\Dc_p\tau)  -\Cc_{0}{}^k{}_j,   \label{gamma-det-1}\\
  \gamma_I{}^0{}_0 = \frac{1}{2}\delta^i_I\delta_0^j\delta^k_0\Dc_i g_{jk} \AND
   \gamma_I{}^K{}_0 =-  \delta^{Kl} \delta_I^i\delta^j_0\Dc_i g_{jl}+ \eta^{KL}\gamma_I{}^0{}_L,  \label{gamma-det-2}
\end{gather}
where $\Cc_i{}^k{}_j$ is as defined above by \eqref{Ccdef}. As also shown in Section \ref{Fermi}, the
remaining background connection components $\gamma_I{}^k{}_J$ satisfy the transport equation
\begin{equation*}
 e_{0}(\gamma_I{}^k{}_J)=e_{I}(\gamma_0{}^k{}_J)-\gamma_I{}^l{}_J\gamma_0{}^k{}_l+\gamma_0{}^l{}_J\gamma_I{}^k{}_l+(\gamma_0{}^l{}_I-\gamma_I{}^l{}_0)\gamma_l{}^k{}_J,
\end{equation*}
which is easily derived from the vanishing of the curvature of the flat background metric $\gc$, see \eqref{gc-def-intro}, and the curvature formula
\begin{equation*}
0=\Rc_{ijk}{}^l = e_{j}(\gamma_i{}^l{}_k)-e_{i}(\gamma_j{}^l{}_k)+\gamma_i{}^m{}_k\gamma_j{}^l{}_m-\gamma_j{}^m{}_k\gamma_i{}^l{}_m-(\gamma_j{}^m{}_i-\gamma_i{}^m{}_j)\gamma_m{}^l{}_k.
\end{equation*}

The evolution of the metric and scalar frame fields
\eqref{fields-intro} is governed by the reduced conformal Einstein-scalar field equations with respect to the frame $e_i^\mu$. In Section \ref{sec:FirstOrderForm}, we derive a first order form of these equations. These first order equations when combined with the evolution equations for the frame and connections coefficients discussed above, see Sections \ref{sec:NonlinDecomp} and \ref{sec:FuchsianForm} for details, yields a system of the form
\begin{align}
  \del{t}g_{00M} &= \frac{1}{t}\bigl(2g_{00M}-\delta^{IJ}(2g_{IJM}-g_{MIJ})\bigr)+\frac{2}{t}\betat\tau_{0M} - B^{0jK}\betat g_{Kj0M}
+\betat Q^0_{0M}, \label{for-G.1.S-intro}\\
\del{t}g_{R0M} &= -\delta_{RI} B^{IjK}\betat g_{Kj0M}+\delta_{RI}\betat Q^I_{0M}, \label{for-G.2.S-intro}\\
\del{t}(g_{0LM}-g_{L0M}-g_{M0L}) &=  -\frac{1}{t}(g_{0LM}-g_{L0M}-g_{M0L}) +\frac{2}{t}\betat\tau_{(LM)}+S_{(LM)}, \label{for-H.S-intro}\\
\del{t}g_{RLM} &= -\delta_{RI} B^{IjK}\betat g_{KjLM}+
                 \delta_{RI}\betat Q^I_{LM}, \label{for-G.4.S-intro}\\
\del{t}\tau_{rl} &=
 \delta_{ri}\betat J^i_l-B^{ijK}\betat\tau_{Kj(l}\delta_{r)i},
 \label{for-F.2.S-intro}\\
\label{for-I.S-intro}
B^{ijk}\betat\Dc_k g_{qjlm} &= \frac{1}{t} \delta_0^i\delta_0^j  (g_{qljm}+ g_{qmjl}-g_{qjlm})+\frac{2}{t}\delta^i_0
\betat\tau_{q(lm)}+ \delta^i_0 \betat P_{q(lm)},\\
B^{ijk}\betat\Dc_k\tau_{qjl} &=  \betat K^i_{ql}, \label{for-J.S-intro}\\
\label{for-L.S-intro}
\del{t}\betat 
&=-\betat^{3}\tau_{00}+\frac{1}{2}\betat^2 \delta^{JK}(g_{0JK}-2g_{J0K}),\\
\del{t}e_I^\Lambda=&-\betat \Bigl(\frac{1}{2}\delta^{JK}(g_{0IK}-g_{I0K}-g_{K0I}) + \delta^{JK}\gamma_{I}{}^0{}_K \Bigr)
e^\Lambda_J, \label{for-M.1.S-intro}\\
\del{t}\gamma_I{}^k{}_J=&
-\betat\Bigl(\delta^k_0\Bigl(\betat e_I(\tau_{0J})+\frac{1}{2}e_I(g_{J00})\Bigr)+\frac{1}{2}\eta^{kl}\bigl(e_I(g_{0Jl})+e_I(g_{J0l})-e_I(g_{l0J})\bigr)\Bigr)+L_I{}^k{}_J. \label{for-M.2.S-intro}
\end{align}
It is important to note here that in the derivation of
this system, we assume that the wave gauge constraint \eqref{wave-gauge-intro} holds, which in terms of the frame fields, is given by 
\begin{equation}\label{wave-gauge-frame-intro}
g_{000}= -\delta^{JK}(g_{0JK}-2g_{J0K}) \AND g_{I00}= 2g_{00I}-\delta^{JK}(2g_{JKI}-g_{IJK}).
\end{equation}
These relations turn out to play an essential role in obtaining the above system. For example, the wave gauge constraint allows us to avoid the use of the evolution equation for the components $g_{i00}$ as these components can be determined from the components $g_{ijK}$ via \eqref{wave-gauge-frame-intro}. This is crucial to our stability proof because the evolution equations for the components $g_{i00}$ contain dangerous singular terms that are quadratic in the variables\footnote{The appearance of the quadratic terms in $\kt_{IJ}$ can be seen as a consequence of the quadratic terms in the second fundamental form 
$\Ktt_{IJ}$ that appear in the Hamiltonian constraint.}
\begin{equation*}
\kt_{IJ} = g_{0IJ}-g_{I0J}-g_{J0I},
\end{equation*}
which are closely related, c.f.~\eqref{Ktt-def}, to the second fundamental form $\Ktt_{IJ}$ associated to the $t=const$-hypersurfaces  and the conformal metric $g$,
that we are unable to control directly. The relations \eqref{wave-gauge-frame-intro} are also used at other points in the argument to show that potentially problematic singular terms can, in fact, be replaced by more favorable terms.

Now, we view \eqref{for-G.1.S-intro}-\eqref{for-M.2.S-intro} as an evolution equation for the variables\footnote{Even though we are using $\tau$ in \eqref{vars-intro.5} to denote
the collection of derivatives $\tau_{ij}=\Dc_i\Dc_j\tau$, no ambiguities will arise due to the slicing condition \eqref{tau=x0} that allows us to use the coordinate time $t$ to denote the scalar field $\tau$.} 
\begin{align}
\kt&=(\kt_{IJ}) :=(g_{0IJ}-g_{I0J}-g_{J0I}), \label{vars-intro.1}\\
\betat &= (-|d t|^2_g)^{-\frac{1}{2}}, \label{vars-intro.2}\\
\ellt&=(\ellt_{IjK}) := (g_{IjK}), \label{vars-intro.3} \\
\mt&=(\mt_{I})  := (g_{00M}), \label{vars-intro.4} \\
\tau &= (\tau_{ij}),  \label{vars-intro.5}\\
\gt&=(\gt_{Ijkl}) := (g_{Ijkl}),  \label{vars-intro.6}\\    
\taut&=(\taut_{Ijk}) := (\tau_{Ijk}),  \label{vars-intro.7}\\
e &= (e^\Lambda_I),  \label{vars-intro.8} \\
\psit&=(\psit_I{}^k{}_J):=(\gamma_I{}^k{}_J),  \label{vars-intro.9}
\end{align}
where other fields can be determined in terms of this set by the relations \eqref{tau=x0}, \eqref{e0-det}, \eqref{gamma-det-1}-\eqref{gamma-det-2}, \eqref{wave-gauge-frame-intro} and the reduced conformal Einstein-scalar field equations. The metric combination $\kt_{ij}$
plays a pivotal role in our analysis.
The property that distinguishes $\kt_{ij}$, as far as the analysis is concerned, is we have no freedom to rescale the normalized version
\begin{equation*}
k_{IJ} = t \betat \kt_{IJ}
\end{equation*}
by any power of $t$. 
Our stability proof relies on showing that $k_{IJ}$ remains bounded as $t\searrow 0$, and in fact, we show that $2k_{IJ}$ converges as $t\searrow 0$ to a, in general, non-vanishing symmetric matrix $\kf_{IJ}$ satisfying $\kf_I{}^I\geq 0$ and $(\kf_I{}^I)^2-\kf_I{}^J\kf_J{}^I+4\kf_I{}^I=0$. On the other hand, there is slack in the remaining variables  \eqref{vars-intro.2}-\eqref{vars-intro.9} in the sense that we can rescale them by certain positive powers of $t$. This freedom to rescale these variables is essential to our stability proof.

It is important to note at this point that we can construct solutions to the first order equations \eqref{for-G.1.S-intro}-\eqref{for-M.2.S-intro} for the frame variables from solutions of Lagrangian formulation of the reduced conformal Einstein-Euler equations, see \eqref{tconf-ford-C.1}-\eqref{tconf-ford-C.8}, which we know exist, at least locally in time, by Theorem \ref{lag-exist-prop}. This is because 
given a solution of \eqref{tconf-ford-C.1}-\eqref{tconf-ford-C.8},
we can solve the Fermi-Walker transport equations to obtain the orthonormal frame $e_i^\mu$, which we can then use along with the given solution to \eqref{tconf-ford-C.1}-\eqref{tconf-ford-C.8}  to obtain all the frame fields \eqref{fields-intro}.

The first order frame equations \eqref{for-G.1.S}-\eqref{for-M.2.S} are themselves not directly useful for analysing the behaviour of solutions near $t=0$, but they can be transformed into a system that is by employing the following rescaled variables
\begin{align}
k&=(k_{IJ}) :=(t\betat \kt_{IJ})
\label{vars-intro.1a}\\
\beta &=  t^{\ep_0}\betat, \label{vars-intro.2a}\\
\ell&=(\ell_{IjK}):= (t^{\ep_1}\ellt_{IjK}), \label{vars-intro.3a}\\
m&=(m_{I}) := (t^{\ep_1}\mt_I), \label{vars-intro.4a} \\
\xi&=(\xi_{ij}) := (t^{\ep_1-\ep_0}\tau_{ij}), \label{vars-intro.5a}\\
\psi &=(\psi_I{}^k{}_J) :=  (t^{\ep_1}\psit_I{}^k{}_J),\label{vars-intro.6a}\\
f&=(f_I^\Lambda) := (t^{\ep_2} e^\Lambda_I), \label{vars-intro.7a}\\
\gac&=(\gac_{Ijkl}) := (t^{1+\ep_1} \betat \gt_{Ijkl}), \label{vars-intro.8a}\\
\tauac &=(\tauac_{Ijk}):= (t^{\ep_0+2 \ep_1} \taut_{Ijk}), \label{vars-intro.9a}
\end{align}
where the constants  $\ep_0,\ep_1,\ep_2>0$ are chosen to satisfy
\begin{equation*}
    0<\ep_0<\ep_1, \quad 3\ep_0+\ep_1<1, \quad 0<\ep_2, \quad \ep_0+\ep_2<1.
\end{equation*}
Expressing the system \eqref{for-G.1.S}-\eqref{for-M.2.S} in terms of these new variables is computationally straightforward and is carried out in the Sections \ref{sec:chvar1} to \ref{sec:Fuch-form}.
The result is a Fuchsian system of equations of the form
\begin{equation} \label{Fuch-eqn-intro}
    A^0 \del{t}u +\frac{1}{t^{\ep_0+\ep_1}}A^\Lambda \del{\lambda}u = \frac{1}{t}\Ac \Pbb u + F,
\end{equation}
where
\begin{equation*}
u = \bigl(k_{LM},m_M,\ell_{R0M},\ell_{RLM},\xi_{rl},\beta,f^\Lambda_I,\psi_I{}^k{}_{J},\tauac_{Qjl}, \gac_{Qjlm}\bigr)^{\tr} - \bigl(0,0,0,0,0, t^{\ep_0}, t^{\ep_2}\delta_{I}^{\Lambda},0,0,0\bigr)^{\tr}
\end{equation*}
and $\Pbb$
is the projection matrix 
\begin{equation*}
 \Pbb = \diag\Bigl(0,\delta_{\Mt}^{ M},\delta_{\Rt}^{R}\delta_{\Mt}^{M},\delta_{\Rt}^{R} \delta_{\Lt}^{ L} \delta_{\Mt}^{M},\delta_{\rt}^{r}\delta_{\lt}^{l},1,\delta_{\It}^{I}\delta^{\Lambdat}_{\Lambda},\delta_{\It}^{I}\delta^{\kt}_{k} \delta_{\Jt}^{J},\delta_{\Qt}^{Q}\delta_{\jt}^{j}\delta_{\lt}^{l},\delta_{\Qt}^{Q}\delta_{\jt}^{j}\delta_{\lt}^{l} \delta_{\mt}^{m}\Bigr).
 \end{equation*}
It is worth noting here that the zero in the first diagonal component of $\Pbb$ is ultimately responsible for the convergence as $t\searrow 0$ of
$2k_{IJ}$ to a, generally non-vanishing, matrix $\kf_{IJ}$. On the other hand the non-zero eigenvalues of $\frac{1}{2}(\Ac+\Ac^{\tr})\Pbb$, which we show are all positive, lead to power law decay, i.e. $t^a$ with $a>0$, for each of the other variables \eqref{vars-intro.2a}-\eqref{vars-intro.9a} where decay rates\footnote{That is the $a$'s where there is a different $a$ for each of the different groups of
  variables \eqref{vars-intro.2a}-\eqref{vars-intro.9a}.} are determined by the  eigenvalues.

The virtue of the Fuchsian formulation \eqref{Fuch-eqn-intro} is that we can now appeal to the existence theory developed in the articles\footnote{The actual existence theory we apply is from \cite{BeyerOliynyk:2024}, which a slight generalization of the existence theory from \cite{BOOS:2021}. } \cite{BeyerOliynyk:2024,BOOS:2021} to conclude, for suitably small choice of initial data $u_0$ at $t=t_0>0$, that there exist a unique solution of \eqref{Fuch-eqn-intro} that is defined all the way down to $t=0$ and satisfies $u|_{t=t_0}=u_0$. The Fuchsian existence theory also yields energy and decay estimates that provide uniform control over the behaviour of solutions in the limit $t\searrow 0$. The precise statement of the global existence result for the Fuchsian equation \eqref{Fuch-eqn-intro} is given in Proposition \ref{prop:globalstability}.

On one hand Proposition \ref{prop:globalstability} yields the existence of a unique solution on $(0,t_0]\times \Tbb^{n-1}$ to the Fuchsian equation \eqref{Fuch-eqn-intro} generated from initial data\footnote{Here, the initial data for \eqref{Fuch-eqn-intro} is assumed to be derived from initial data for the reduced conformal Einstein-scalar equations that satisfies the gravitational and wave gauge constraint equations.} $u|_{t=t_0} = u_0$ that is sufficiently close to FLRW initial data defined by $u|_{t=t_0} = 0$.
On the other hand, this same initial data generates, by 
Proposition \ref{lag-exist-prop}, a local-in-time solution to the the system 
\eqref{tconf-ford-C.1}-\eqref{tconf-ford-C.8} that, after solving the Fermi-Walker transport equations for the spatial frame fields, determines a solution of the Fuchsian equation \eqref{Fuch-eqn-intro}. By uniqueness, these two solutions must be the same. The energy estimates from Proposition \ref{prop:globalstability} then allows us to conclude via the continuation principle from Proposition \ref{lag-exist-prop} that the solution $u$ of the Fuchsian equation determines a solution of the conformal Einstein-scalar field equations on $(0,t_0]\times \Tbb^{n-1}$. Asymptotic properties of the solution to the conformal Einstein-scalar field equations are then deduced from the energy and decay estimates for $u$ from Proposition \ref{prop:globalstability}. This completes the overview of the major steps involved in our proof of the past global nonlinear stability of small perturbations of the FLRW solutions to the conformal Einstein-scalar field equations.
The precise statement of this result is presented in Theorem \ref{glob-stab-thm} and the proof can be found in Section  \ref{sec:proof_globstab}.  The localized version of  Theorem \ref{glob-stab-thm}, which is stated in Theorem \ref{loc-stab-thm}, follows from the exact same reasoning except we now exploit the finite propagation speed property of our evolution equations to localize all of the arguments. The proof of Theorem \ref{loc-stab-thm} is given in Section \ref{local-sec}.  

\section{Preliminaries\label{prelim}}

\subsection{Data availability statement}

This article has no associated data.

\subsection{Coordinates, frames and indexing conventions\label{indexing}}
In the article, we will consider $n$-dimensional spacetime manifolds of the form
\begin{equation} \label{Mt1t0-def}
    M_{t_1,t_0}= (t_1,t_0]\times \Tbb^{n-1},
\end{equation}
where $t_0>0$, $t_1<t_0$, and $\Tbb^{n-1}$ is
the $(n-1)$-torus defined by
\begin{equation} \label{Tbb-def}
    \Tbb^{n-1} = [-L,L]^{n-1}/\sim
\end{equation} 
with $\sim$ the equivalence relation obtained
from identifying the sides of the box $[-L,L]^{n-1}\subset \Rbb^{n-1}$. On $M_{t_1,t_0}$,
we will always employ coordinates $(x^\mu)=(x^0,x^\Lambda)$
where the $(x^\Lambda)$ are periodic spatial coordinates
on $\Tbb^{n-1}$ and $x^0$ is a time coordinate
on the interval $(t_1,t_0]$. Lower case Greek letters, e.g. $\mu,\nu,\gamma$, will run from $0$ to $n-1$ and be used to label spacetime coordinate indices while upper case Greek letters, e.g. $\Lambda,\Omega,\Gamma$, will run from $1$ to $n-1$ and label spatial coordinate indices. Partial derivatives with respect to the coordinates $(x^\mu)$ will be denoted by $\del{\mu} = \frac{\del{}\;}{\del{}x^\mu}$.
We will often use $t$ to denote the time coordinate $x^0$, that is, $t=x^0$, and use the notion $\del{t} = \del{0}$ for the partial derivative with respect to the coordinate $x^0$.

We will use frames $e_j= e_j^\mu \del{\mu}$ 
throughout this article. Lower case Latin letter, e.g. $i,j,k$, will be used to label frame indices and they will run from $0$ to $n-1$ while spatial frame indices will be labelled by upper case Latin letter, e.g. $I,J,K$, that run from $1$ to $n-1$.

\subsection{Inner-products and matrices}
Throughout this article, we will denote the Euclidean inner-product by $\ipe{\xi}{\zeta} = \xi^{\tr} \zeta$, $\xi,\zeta \in \Rbb^N$, and use
$|\xi| = \sqrt{\ipe{\xi}{\xi}}$
to denote the Euclidean norm. The set of $N\times N$ matrices will be denoted by  $\Mbb{N}$, and we will use $\Sbb{N}$ denote the subspace of symmetric $N\times N$-matrices. 

Given $A\in \Mbb{N}$, we define
the operator norm $|A|_{\op}$ of $A$ via
\begin{equation*}
   |A|_{\op} = \sup_{\xi\in \Rbb^N_\times} \frac{|A\xi|}{|\xi|},
\end{equation*}
where $\Rbb^N_\times = \Rbb^N\setminus\{0\}$.
For any $A,B\in \Mbb{N}$, we will also employ the notation
\begin{equation*}
    A\leq B \quad \Longleftrightarrow \quad \xi^{\tr}A\xi \leq \xi^{\tr}B\xi, \quad \forall \; \xi \in \Rbb^N.
\end{equation*}

\subsection{Balls and truncated cone domains\label{domains}}
Inside $\Tbb^{n-1}$, we define, for $0<\rho<L$, the \textit{centred ball of
radius} $\rho$ by
\begin{equation*}
    \mathbb{B}_{\rho} = \{\, x\in (-L,L)^{n-1}\, |\, |x|<\rho \, \} \subset \Tbb^{n-1},
\end{equation*}
where, in line with our above choice of notation,  $|x|=\sqrt{\delta_{\Lambda\Sigma}x^\Lambda x^{\Sigma}}$ denotes the Euclidean norm on $\Rbb^{n-1}$.
Given constants $t_0,\rho_1>0$, $0<\rho_0<L$ and $0\leq \ep<1$ satisfying
\begin{equation*}
\rho_0 -\frac{\rho_1 t_0^{1-\ep}}{1-\ep}>0,    
\end{equation*}
we define the \textit{truncated cone domain} $\Omega_{t_0,\rho_0,\rho_1,\ep}$ inside  $M_{0,t_0}$ by
\begin{equation} \label{Omega-def}
    \Omega_{t_0,\rho_0,\rho_1,\ep} = \Bigl\{\, (t,x)\in (0,t_0]\times (-L,L)^{n-1}\, \Bigl| \,|x|< \frac{\rho_1(t^{1-\ep}-t_0^{1-\ep})}{1-\ep}+\rho_0 \,  \Bigr\} \subset M_{0,t_0}.
\end{equation}
The ``side'' piece of the boundary of $\Omega_{t_0,\rho_0,\rho_1,\ep}$, which we denote by 
$\Upsilon_{t_0,\rho_0,\rho_1,\ep}$, is defined by the vanishing of the function $\rf =|x|-\frac{\rho_1(t^{1-\ep}-t_0^{1-\ep})}{1-\ep}-\rho_0$.
Differentiating $\rf$ yields the $1$-form
\begin{equation} \label{n-def}
    n_\mu := \del{\mu}\rf  =  \frac{1}{|x|}\delta_{\mu\Lambda}x^\Lambda-\frac{\rho_1}{t^\ep}\delta^0_\mu,
\end{equation}
which we note determines a co-normal to $\Upsilon_{t_0,\rho_0,\rho_1,\ep}$.
The boundary of $\Omega_{t_0,\rho_0,\rho_1,\ep}$ can be decomposed as the disjoint union
\begin{equation*}
    \del{}\Omega_{t_0,\rho_0,\rho_1,\ep}= \bigl(\{t_0\}\times\mathbb{B}_{\rho_0}\bigr) \cup \Upsilon_{t_0,\rho_0,\rho_1,\ep}\cup \bigl(\{0\}\times\mathbb{B}_{\tilde{\rho}_1}\bigr)
\end{equation*}
where $\tilde{\rho}_1 =\rho_0- \frac{\rho_1 t_0^{1-\ep}}{1-\ep}$
and the balls $\{t_0\}\times\mathbb{B}_{\rho_0}$ and
$\{0\}\times\mathbb{B}_{\tilde{\rho}_1}$ cap $\Omega_{t_0,\rho_0,\rho_1,\ep}$ on top and bottom, respectively.  

\subsection{Sobolev spaces and extension operators\label{Sobolev}}
The $W^{k,p}$, $k\in \Zbb_{\geq 0}$, norm of a map $u\in C^\infty(U,\Rbb^N)$ with $U\subset \Tbb^{n-1}$ open is defined by
\begin{equation*}
\norm{u}_{W^{k,p}(U)} = \begin{cases} \begin{displaystyle}\biggl( \sum_{0\leq |\Ic|\leq k} \int_U |D^{\Ic} u|^p \, d^{n-1} x\biggl)^{\frac{1}{p}}  \end{displaystyle} & \text{if $1\leq p < \infty $} \\
 \begin{displaystyle} \max_{0\leq |\Ic| \leq k}\sup_{x\in U}|D^{\Ic} u(x)|  \end{displaystyle} & \text{if $p=\infty$}
\end{cases},
\end{equation*}
where $\Ic=(\Ic_1,\ldots,\Ic_{n-1})\in \Zbb_{\geq 0}^{n-1}$ denotes a multi-index and 
$D^\Ic = \del{1}^{\Ic_1}\del{2}^{\Ic_2}\cdots\del{n-1}^{\Ic_{n-1}}$.
The Sobolev space $W^{k,p}(U,\Rbb^N)$ is then defined to be the completion of $C^\infty(U,\Rbb^N)$ with respect to the norm
$\norm{\cdot}_{W^{k,p}(U)}$. When $N=1$ or the dimension $N$ is clear from the context, we will simplify notation and write $W^{k,p}(U)$ instead of $W^{k,p}(U,\Rbb^N)$, and we will employ the standard notation $H^k(U,\Rbb^N)=W^{k,2}(U,\Rbb^N)$ throughout.

To each centred ball $\mathbb{B}_\rho\subset \Tbb^{n-1}$, $0<\rho<L$, we
assign a (non-unique) total
extension operator
\begin{equation} \label{Ebb-def}
\Ebb_\rho \: : \: H^k(\mathbb{B}_\rho,\Rbb^N)\longrightarrow H^k(\Tbb^{n-1},\Rbb^N), \qquad k \in \Zbb_{\geq 0},
\end{equation}
that satisfies
\begin{equation}\label{Ebb-prop}
\Ebb_\rho(u)\bigl|_{\mathbb{B}_\rho} = u, \AND
\norm{\Ebb_{\rho}(u)}_{H^{k}(\Tbb^{n-1)}} \leq C\norm{u}_{H^k(\mathbb{B}_\rho)}
\end{equation}
for some constant $C=C(k,n,\rho)>0$ independent of $u\in H^k(\mathbb{B}_\rho)$. The existence of such an operator is established in
\cite{AdamsFournier:2003}; see Theorems 5.21 and 5.22, and Remark 5.23 for details.

\subsection{Constants and inequalities}
We use the standard notation $a \lesssim b$
for inequalities of the form
$a \leq Cb$
in situations where the precise value or dependence on other quantities of the constant $C$ is not required.
On the other hand, when the dependence of the constant on other inequalities needs to be specified, for
example if the constant depends on the norm $\norm{u}_{L^\infty}$, we use the notation
$C=C(\norm{u}_{L^\infty})$.
Constants of this type will always be non-negative, non-decreasing, continuous functions of their arguments.

\subsection{Curvature}
The curvature of tensor $\Rc_{ijk}{}^l$ of the background metric $\gc_{ij}$ is defined via
\begin{equation}\label{comm-Rc}
  [\Dc_i,\Dc_j]\omega_k = \Rc_{ijk}{}^l\omega_l
\end{equation}
for arbitrary $1$-forms $\omega_l$.
This definition along with $\Rc_{ik}=\Rc_{ijk}{}^j$ for
the Ricci tensor fixes the curvature conventions that will be employed for all curvature tensors appearing in this article.   

\section{Reduced conformal field equations}
In order to establish the existence of solutions to the conformal Einstein-scalar field equations, we will need to replace the conformal Einstein equations \eqref{confESFAb} or equivalently \eqref{confESFAa} with a gauge reduced version. In the following, we employ a (conformal) wave gauge defined by the constraint
\begin{equation}\label{wave-gauge}
X^k =0,
\end{equation}
where  $X^k$ is given above by \eqref{Xdef},
and consider the wave gauge reduced equations 
\begin{equation}\label{confESFF}
-2R_{ij}+2\nabla_{(i} X_{j)}=-\frac{2}{\tau}\nabla_i \nabla_j \tau \oset{\eqref{Ccdef}}{=}-\frac{2}{\tau}\bigl(
\Dc_i \Dc_j \tau - \Cc_i{}^k{}_j\Dc_k \tau  \bigr),
\end{equation}
which we will refer to as the \textit{reduced conformal Einstein equations}.

For the moment, we \emph{assume} that the wave gauge constraint \eqref{wave-gauge} holds. Later in Proposition \ref{prop-ES-constr}, it is established that the wave gauge constraint propagates, which validates this assumption. Now, by \eqref{Ccdef}, \eqref{Xdef} and \eqref{wave-gauge}, we observe that the conformal scalar field equation \eqref{confESFCb}
can be expressed as
\begin{equation} \label{confESFG}
    g^{ij}\Dc_{i}\Dc_{j}\tau= 0.
\end{equation}
Gathering \eqref{confESFF} and \eqref{confESFG} together, we have
\begin{align}
-2R_{ij}+2\nabla_{(i} X_{j)}& =-\frac{2}{\tau}\bigl(
\Dc_i \Dc_j \tau - \Cc_i{}^k{}_j\Dc_k \tau  \bigr) , 
\label{confESFFa}\\
 g^{ij}\Dc_{i}\Dc_{j}\tau    &= 0. \label{confESFGa}
\end{align}
We will refer to these equations as the \textit{reduced conformal Einstein-scalar field equations}.

For use below, we recall that the reduced Ricci tensor can be expressed as
\begin{equation} \label{red-Ricci}
-2R_{ij}+2\nabla_{(i} X_{j)}=g^{kl}\Dc_k \Dc_l g_{ij} + Q_{ij}
+2g^{kl}g_{m(i}\Rc_{j)kl}{}^m
\end{equation}
where 
\begin{equation}\label{Q-def}
Q_{ij} = \frac{1}{2}g^{kl}g^{mn}\Bigl(\Dc_i g_{mk} \Dc_j g_{n l}
+2 \Dc_{n}g_{il}\Dc_{k}g_{jm} - 2\Dc_{l}g_{in} \Dc_{k}g_{jm}
-2 \Dc_{l}g_{in}\Dc_j g_{mk} -2 \Dc_{i}g_{mk}\Dc_{l}g_{jn}\Bigr)
\end{equation}
and, as above, $\Rc_{ijk}{}^l$ denotes the curvature tensor of the background metric $\gc_{ij}$. By differentiating \eqref{confESFGa} and employing the commutator formula
$\Dc_k\Dc_i \Dc_j\tau-\Dc_i\Dc_j \Dc_k \tau
=\Rc_{kij}{}^l\Dc_l\tau$,
we also note that
\begin{equation} \label{confESFI} 
g^{ij}\Dc_i\Dc_{j} \Dc_{k}\tau = g^{il}g^{jm}\Dc_kg_{lm}  \Dc_{i} \Dc_{j}\tau
-g^{ij}\Rc_{kij}{}^l \Dc_l \tau.
\end{equation}

\section{Choice of background metric}
The background metric $\gc_{ij}$ is thus far arbitrary. Since the conformal FLRW metric \eqref{gtau-FLRW} is flat and we are interested in nonlinear perturbations of this solution, we are motivated to restrict our attention to background metrics that are flat, which by definition, means that the curvature tensor vanishes, that is,
\begin{equation}\label{curvature}
\Rc_{ijk}{}^l = 0.
\end{equation}
By the commutator formula \eqref{comm-Rc}, the vanishing of the curvature implies that
\begin{equation}\label{commutator}
[\Dc_i,\Dc_j] =0.
\end{equation}

\begin{rem} \label{rem-FLRW-to-Kasner}
In order to extend the stability results obtained in this article so that they hold for the full range of Kasner exponents where stable singularity formation is expected as in \cite{Fournodavlos_et_al:2023}, it would be natural to try replacing the flat metric with a Kasner-scalar field metric whose exponents lie within the range where stable singularity formation is expected.  The calculations and proofs below would have to be suitably modified to account for this change of background metric. We plan on carrying out this analysis in future work.
\end{rem}

\section{Local existence and continuation in Lagrangian coordinates}
\label{sec:locexist_cont}
Up to this point, we have expressed the reduced conformal Einstein-scalar field equations \eqref{confESFFa}-\eqref{confESFGa} in an arbitrary frame. In order to establish the local-in-time existence and uniqueness of solutions to these equations along with a continuation principle, we need to completely fix the frame. We do so in this section by employing a coordinate frame. It is worth noting at this point that while coordinate frames are suitable for establishing the local-in-time existence and uniqueness of solutions, they are not suitable for obtaining estimates that can be used to bound the solution near the big bang singularity. For this, we will need to employ a non-coordinate frame; see Section \ref{Fermi} for details.

To fix the coordinate frame, we introduce coordinates $(\xh^\mu)=(\xh^0,\xh^\Lambda)$ on a spacetime $M_{t_1,t_0}$ of the form
\eqref{Mt1t0-def},
and we assume that the components of the flat background metric $\gc$ in this coordinate system, which we denote by $\gchat_{\mu\nu}$, are given by
\begin{equation}\label{gct-def}
    \gchat_{\mu\nu}=\eta_{\mu\nu}:= -\delta_\mu^0\delta_\nu^0 + \delta_\mu^\Lambda \delta_\nu^\Gamma \delta_{\Gamma\Lambda}.
\end{equation}
In this coordinate frame, the Levi-Civita connection of the background metric coincides with partial differentiation with respect to the coordinates $(\xh^\mu)$, that is,
$\hat{\Dc}_\mu = \delh{\mu}$. 
From this, we find, with the help of \eqref{Ccdef}, \eqref{red-Ricci} and \eqref{curvature}, that the reduced conformal Einstein-scalar field equations \eqref{confESFFa}-\eqref{confESFGa} are given in the coordinates $(\xh^\mu)$ by
\begin{align}
\gh^{\alpha\beta}\delh{\alpha}\delh{\beta}\gh_{\mu\nu} + \Qh_{\mu\nu} &= -\frac{2}{\tauh}\bigl( \delh{\mu}\delh{\nu}\tauh - \Gammah^\gamma_{\mu\nu} \delh{\gamma}\tauh \bigr), \label{tconfESF-A.1}\\
\gh^{\alpha\beta}\delh{\alpha}\delh{\beta}\tauh &= 0,  \label{tconfESF-A.2}
\end{align}
where $\tauh$ denotes the scalar field $\tau$ viewed as a function of the coordinates $(\xh^\mu)$, $\gh_{\mu\nu}$ are the components of the conformal metric $g$ with respect to the coordinates $(\xh^\mu)$, $\Gammah^\gamma_{\mu\nu} =
     \frac{1}{2}\gh^{\gamma\lambda}\bigl(\delh{\mu}\gh_{\nu\lambda} +\delh{\nu}\gh_{\mu\lambda}-\delh{\lambda}\gh_{\mu\nu}\bigr)$
are the Christoffel symbols of $\gh_{\mu\nu}$, and
\begin{equation}
  \label{eq:defQh}
\Qh_{\mu \nu} = \frac{1}{2}\gh^{\alpha \beta}\gh^{\sigma \delta}\bigl(\delh{\mu} \gh_{\sigma \alpha} \delh{\nu} \gh_{\delta \beta}
+2 \delh{\delta}\gh_{\mu \beta}\delh{\alpha}\gh_{\nu \sigma} - 2\delh{\beta}\gh_{\mu \delta} \delh{\alpha}\gh_{\nu \sigma}
-2 \delh{\beta}\gh_{\mu \delta}\delh{\nu} \gh_{\sigma\alpha} -2 \delh{\mu }\gh_{\sigma\alpha}\delh{\beta}\gh_{\nu \delta}\bigr).
\end{equation}
We further observe from \eqref{Ccdef} and \eqref{Xdef} that the coordinate components of the wave gauge vector field $X$, which we denote by $\Xh^\gamma$, are given by
\begin{equation} \label{Xt-rep}
    \Xh^\gamma = \gh^{\mu\nu}\Gammah^\gamma_{\mu\nu}.
\end{equation}

\subsection{Initial data}
On the initial hypersurface 
\begin{equation*}
\Sigma_{t_0}=\{t_0\}\times\Tbb^{n-1},
\end{equation*}
we specify the following initial data for the reduced conformal Einstein-scalar field equations \eqref{tconfESF-A.1}-\eqref{tconfESF-A.2}:
\begin{align}
    \gh_{\mu\nu}\bigl|_{\Sigma_{t_0}} &= \gr_{\mu\nu}, \label{gt-idata} \\
    \delh{0}\gh_{\mu\nu}\bigl|_{\Sigma_{t_0}} &= \grave{g}_{\mu\nu}, \label{dt-gt-idata} \\
    \tauh\bigl|_{\Sigma_{t_0}} &= \taur, \label{taut-idata}\\
    \delh{0}\tauh \bigl|_{\Sigma_{t_0}} &= \taug.\label{dt-taut-idata}
\end{align}
Since we want solutions of the reduced conformal Einstein-scalar field equations \eqref{tconfESF-A.1}-\eqref{tconfESF-A.2} to also satisfy the conformal Einstein-scalar field equation \eqref{confESFAb}-\eqref{confESFCb},
the initial data cannot be chosen freely. Instead, the initial data will need to satisfy the constraint equations
\begin{align}
  \nh_\mu\Bigr( \Gh^{\mu\nu} - \frac{1}{\tauh}\nablah^\mu\nablah^\nu \tauh\Bigr)\Bigl|_{\Sigma_{t_0}} &=0, \qquad \text{(gravitational constraints)}
  \label{grav-constr}\\
  \Xh^\mu\bigl|_{\Sigma_{t_0}} &=0, \qquad \text{(wave gauge constraints)}
  \label{wave-constr}
\end{align}
where $\Gh^{\mu\nu}$ and $\nablah_\mu$ are the Einstein tensor and Levi-Civita connection of the conformal metric $\gh_{\mu\nu}$, respectively, and $\nh = -d\xh^0$ (i.e. $\nh_\mu = -\delta^0_\mu)$. As is shown in  
Proposition \ref{prop-ES-constr} below, solutions of the reduced conformal Einstein-scalar field equations that are generated from initial data satisfying both of these constraint equations will be guaranteed to also solve the conformal Einstein-scalar field equations.

\begin{rem} \label{idata-rem}
  The wave gauge constraint \eqref{wave-constr} does not restrict the geometry of the initial data set $\{\gr_{\mu\nu}, \grave{g}_{\mu\nu}, \taur, \taug\}$ on $\Sigma_{t_0}=\{t_0\}\times\Tbb^{n-1}$ and can always be satisfied by an appropriate choice of the gauge initial data for given geometric initial data. To see why, we recall that the constraint \eqref{grav-constr} can be formulated exclusively in terms \emph{geometric initial data} $\{\gtt,\Ktt,\taur,\taugr\}$ on $\Sigma_{t_0}$
where 
$\gtt=\gtt_{\Lambda\Omega}d\xh^\Lambda \otimes d\xh^\Omega$ is the \textit{spatial metric} and $\Ktt=\Ktt_{\Lambda\Omega}d\xh^\Lambda \otimes d\xh^\Omega$ is the \textit{second fundamental form}, which are related to $\{\gr_{\mu\nu}, \grave{g}_{\mu\nu}, \taur, \taug\}$ via
\begin{equation}
  \label{gtt-def1}
  \gtt_{\Lambda\Omega} = \gr_{\Lambda\Omega} \AND \Ktt_{\Lambda\Omega} = \frac{1}{2\Ntt}(\ggr_{\Lambda\Omega}-2\Dtt_{(\Lambda} \btt_{\Omega)}),
\end{equation}
respectively. Here, 
\begin{equation}
  \label{gtt-def2}
  \btt_\Lambda = \gr_{0\Lambda} \AND \Ntt^2 = -\gr_{00}+ \btt^\Lambda\btt_\Lambda
\end{equation}
define the \textit{shift}
$\btt=\btt_\Lambda d\xh^\Lambda$ and  \textit{lapse} $\Ntt$, respectively,  $\Dtt_\Lambda$
denotes the Levi-Civita connection of the spatial metric $\gtt_{\Lambda\Omega}$, and we have used the inverse metric $\gtt^{\Lambda\Omega}$ of $\gtt_{\Lambda\Omega}$ to raise indices, e.g.\ $\btt^\Lambda = \gtt^{\Lambda\Omega}\btt_\Omega$.
Suppose now that {geometric initial data} $\{\gtt,\Ktt,\taur,\taugr\}$ satisfying \eqref{grav-constr} is given on $\Sigma_{t_0}$. Then we can construct an equivalent initial data set $\{\gr_{\mu\nu}, \grave{g}_{\mu\nu}, \taur, \taug\}$ that satisfies \emph{both} \eqref{grav-constr}  and \eqref{wave-constr}. To see how, we observe from \eqref{gtt-def1} and \eqref{gtt-def2} that $\gr_{\mu\nu}$ and $\ggr_{\Lambda\Omega}$ are determined on $\Sigma_{t_0}$ by the given geometric initial data $\{\gtt,\Ktt,\taur,\taugr\}$ and an arbitrary choice of the \textit{gauge initial data} $\{\Ntt,\btt_\Lambda\}$ satisfying $\Ntt>0$.
Now, a straightforward calculation shows that the wave gauge constraint \eqref{wave-constr} takes the form
\begin{equation}
  \label{wave-constr_geom}
  \dot{\Ntt}-\btt^\Lambda \delh{\Lambda} \Ntt - \Ntt^2 \Ktt_\Lambda{}^\Lambda =0 \AND \dot{\btt}_{\Lambda}-2\btt^{\Omega}\bigl(\Ntt\Ktt_{\Lambda\Omega}+\Dtt_{(\Lambda} \btt_{\Omega)})-\gtt_{\Lambda\Sigma}\btt^\Omega \delh{\Omega}\btt^\Sigma +\Ntt\delh{\Lambda} \Ntt -\Ntt^2\gtt_{\Lambda\Delta} g^{\Omega\Sigma}\Gamma_{\Omega\Sigma}^\Delta =0, 
\end{equation}
where $\Gamma_{\Sigma\Omega}^\Lambda$ are the Christoffel symbols of the spatial metric $\gtt_{\Lambda\Omega}$ and the dot represent the time derivative $\delh{0}$ of the fields $\{\Ntt,\btt_\Lambda\}$ evaluated on $\Sigma_{t_0}$, which are related to $\ggr_{\Lambda0}$ and $\ggr_{00}$ by
\begin{equation}
  \label{gtt-def3}
  \dot{\btt}{}_\Lambda = \ggr_{0\Lambda} \AND   
    \dot{\Ntt} = \frac{1}{2\Ntt}\bigl(-\ggr_{00}+2\btt^\Lambda \dot{\btt}{}_\Lambda - 2\btt^\Lambda \btt^\Omega(\Ntt \Ktt_{\Lambda\Omega}+\Dtt_\Lambda \btt_\Omega)
\bigr).
\end{equation}
From \eqref{wave-constr_geom} and \eqref{gtt-def3}, it is then clear that these equations determine the remaining components $\ggr_{\Lambda0}$ and $\ggr_{00}$ of an initial data
set $\{\gr_{\mu\nu}, \grave{g}_{\mu\nu}, \taur, \taug\}$  that satisfies \emph{both} \eqref{grav-constr}  and \eqref{wave-constr}.
\end{rem}

\subsection{Constraint propagation}
Before proceeding with our analysis of the reduced conformal Einstein-scalar field equations, we first verify that solutions of this system that satisfy the constraint equations \eqref{grav-constr}-\eqref{wave-constr} also satisfy the (full)
conformal Einstein-scalar field equations.

\begin{prop} \label{prop-ES-constr}
Suppose $\gh_{\mu\nu},\tauh\in C^3(M_{t_1,t_0})$ solve the reduced conformal Einstein-scalar field equations \eqref{tconfESF-A.1}-\eqref{tconfESF-A.2} and the constraints \eqref{grav-constr}-\eqref{wave-constr}. Then  $\gh_{\mu\nu},\tauh$ satisfy the conformal Einstein-scalar field equations 
\begin{align} 
\Gh^{\mu\nu} = \frac{1}{\tauh}\nablah^\mu\nablah^\nu \tauh, \quad
\Box_{\gh}\tauh = 0,\label{tconfESF-Aa}
\end{align}
and the wave gauge constraint $\Xh^\mu = 0$
in $M_{t_1,t_0}$.
\end{prop}

\begin{proof}
Assume $\gh_{\mu\nu},\tauh\in C^3(M_{t_1,t_0})$ solves the reduced conformal Einstein-scalar field equations \eqref{tconfESF-A.1}-\eqref{tconfESF-A.2} and satisfies the constraints \eqref{grav-constr}-\eqref{wave-constr} in $\Sigma_{t_0}$. Then with the help of \eqref{Xt-rep}, we can express the reduced conformal Einstein-scalar field equations as
\begin{align} 
-2\Rh_{\mu\nu} + 2\nablah_{(\mu}\Xh_{\nu)} &= -\frac{2}{\tauh}\nablah_\mu\nablah_\nu \tauh, \label{prop-ES-constr1.1}\\
\gh^{\alpha\beta}\nablah_\alpha\nablah_\beta\tauh &= -\Xh^\alpha \nablah_\alpha \tauh.\label{prop-ES-constr1.2}
\end{align}
Taking the trace of \eqref{prop-ES-constr1.1}, we find, with the help of \eqref{prop-ES-constr1.2}, that
the scalar curvature of $\gh_{\mu\nu}$ is given by
$\Rh = \nablah^\mu\Xh_\mu - \Xh^\mu \nablah_\mu \ln(\tauh)$.
Using this expression, it is not difficult to verify that the reduced conformal Einstein equations  \eqref{prop-ES-constr1.1} are given by
\begin{equation} \label{prop-ES-constr2}
    \Gh^{\mu\nu}=\nablah^{(\mu}\Xh^{\nu)} -\frac{1}{2}\nablah_\alpha \Xh^\alpha \gh^{\mu\nu}+\frac{1}{2}\Xh^\alpha \nablah_\alpha \ln(\tauh) \gh^{\mu\nu}+\frac{1}{\tauh}\nablah^\mu\nablat^\nu \tauh.
\end{equation}
Then noting 
\begin{align}
    \nablah_\mu\Bigl(\frac{1}{\tauh}\nablah^\mu\nablat^\nu \tauh\Bigr) &= -\frac{1}{\tauh^2}\nablah_\mu\tauh \nablah^\mu\nablah^\nu \tauh + \frac{1}{\tauh}\nablah_\mu \nablah^\mu\nablah^\nu  \tauh \notag \\
    &= -\frac{1}{\tauh^2}\nablah_\mu\tauh \nablah^\mu\nablah^\nu \tauh +\frac{1}{\tauh} \nablah^\nu \nablah_\mu \nablah^\mu \tauh +\frac{1}{\tauh}\Rh^{\nu\alpha}\nablah_\alpha\tauh
    \notag \\
    &=   - 
     \Xh^\alpha \nablah^\nu\nablah_\alpha \ln(\tauh)
     - \Xh^\alpha\nablah^\nu\ln(\tauh)
     \nablah_\alpha\ln(\tauh)
  + \nablah^{[\alpha}\Xh^{\nu]} \nablah_\alpha \ln(\tauh) ,
    && \text{(by \eqref{prop-ES-constr1.1}-\eqref{prop-ES-constr1.2})}\notag
\end{align}
it follows from applying $\nablah_\mu$ to \eqref{prop-ES-constr2}
that
\begin{align*}
&\nablah_{\mu}\nablah^{(\mu}\Xh^{\nu)}
  -\frac{1}{2}\nablah^{\nu}\nablah_\mu \Xh^\mu 
      +\frac{1}{2}\nablah^{\nu}\Xh^\mu \nablah_\mu \ln(\tauh)
      -\frac{1}{2}\Xh^\mu \nablah^{\nu}\nablah_\mu \ln(\tauh)\\
      &\hspace{6.0cm}
     - \Xh^\mu\nablah^\nu\ln(\tauh)
     \nablah_\mu\ln(\tauh)
  + \nablah^{[\mu}\Xh^{\nu]} \nablah_\mu \ln(\tauh)=0,
\end{align*}
where in deriving this expression we have employed the contracted Bianchi identity $\nablah_\mu \Gh^{\mu\nu}=0$.
Re-expressing this as
\begin{align} 
 & \nablah_{\mu}\nablah^{\mu}\Xh^{\nu}     
     + 2\nablah_\mu \ln(\tauh)\nablah^{[\mu}\Xh^{\nu]} 
        +\nablah_\mu \ln(\tauh)\nablah^{\nu}\Xh^\mu
      +\bigl(\Rh^{\nu}{}_{\mu}
      -\nablah^{\nu}\nablah_\mu \ln(\tauh)      
     - 2\nablah^\nu\ln(\tauh)
     \nablah_\mu\ln(\tauh)\bigr) \Xh^\mu=0,
 \label{wave-gauge-prop-eqn}\end{align}
we see that $\Xh^\mu$ satisfies a linear wave equation on $M_{t_0,t_1}$. 
Since the constraints \eqref{grav-constr}-\eqref{wave-constr} imply, by \eqref{Xt-rep} and a well known argument, e.g. see \cite[\S 10.2]{Wald:1994}, that $\Xh^\mu$ and $\delh{0}\Xh^\mu$ vanish in $\Sigma_{t_0}$, we conclude from the uniqueness of solutions to linear wave equations that $\Xh^\mu$ must vanish in $M_{t_1,t_0}$. By \eqref{prop-ES-constr1.2}-\eqref{prop-ES-constr2}, it then follows that the pair $\{\gh_{\mu\nu},\tauh\}$ solves the conformal Einstein-scalar field equations \eqref{tconfESF-Aa}, which completes the proof.
\end{proof}

\subsection{First order formulation\label{fof}}
For the purposes of establishing the local-in-time existence of
solutions to the reduced conformal Einstein-scalar field system, we find it convenient to use a first order formulation. We begin the derivation
of the first order formulation by introducing the first order variables
\begin{gather}
    \hh_{\beta\mu\nu} = \delh{\beta}\gh_{\mu\nu}, \quad
    \vh_{\mu} = \delh{\mu}\tauh \AND
    \wh_{\mu\nu} = \delh{\mu}\delh{\nu}\tauh, \label{tvars}
\end{gather}
and we define a vector field $\chih^\mu$ via
\begin{equation}
  \chih^\mu = \frac{1}{|\vh|_{\gh}^2} \vh^{\mu}. \label{chit-def}
\end{equation}
Notice that $\chih(\tauh)=1$ by definition.  
We assume, in the following, that $\vh^\mu$ is timelike (i.e.\ $|\vh|_{\gh}^2 <0$), which will ensure that the vector field $\hat{\chi}^\mu$ remains well defined and timelike.
By \eqref{tconfESF-A.1}-\eqref{tconfESF-A.2}, it is not difficult to see that the fields $\{\gh_{\mu\nu},\tauh,\vh_{\mu}\}$ evolve according to 
\begin{align}
\gh^{\alpha\beta}\delh{\alpha}\delh{\beta}\gh_{\mu\nu} + \Qh_{\mu\nu} &= -\frac{2}{\tauh}\bigl( \delh{(\mu}\vh_{\nu)} - \Gammah^\gamma_{\mu\nu} \vh_\gamma \bigr), \label{tconfESF-B.1}\\
\gh^{\alpha\beta}\delh{\alpha}\delh{\beta}\tauh &=0,  \label{tconfESF-B.2} \\
\gh^{\alpha\beta}\delh{\alpha}\delh{\beta}\vh_\mu &= \gh^{\alpha \sigma}\gh^{\beta\delta}\delh{\mu}\gh_{\sigma\delta}\delh{\alpha}\vh_\beta. \label{tconfESF-B.3}
\end{align}

Multiplying \eqref{tconfESF-B.1} by $-\chih^\lambda$ gives
\begin{equation}\label{tconf-ford-A}
    -\chih^\lambda\gh^{\alpha\beta}\delh{\alpha}\delh{\beta}\gh_{\mu\nu}= \chih^\lambda \Bigr(\Qh_{\mu\nu}+\frac{2}{\tauh}\bigl( \delh{(\mu}\vh_{\nu)} - \Gammah^\gamma_{\mu\nu} \vh_\gamma\bigl) \Bigr).
\end{equation}
Noticing that $-\chih^\beta \gh^{\lambda\alpha} 
   \delh{\alpha}\delh{\beta}\gh_{\mu\nu}+ \gh^{\lambda\beta}\chih^\alpha 
   \delh{\alpha}\delh{\beta}\gh_{\mu\nu}=0$
holds by the symmetry of mixed partial derivatives, we get from adding this to \eqref{tconf-ford-A} that
\begin{equation*}
   \Bh^{\lambda\beta\alpha} \delh{\alpha}\hh_{\beta\mu\nu}= \chih^\lambda \Bigr(\Qh_{\mu\nu}+\frac{2}{\tauh}\bigl( \delh{(\mu}\vh_{\nu)} - \Gammah^\gamma_{\mu\nu} \vh_\gamma\bigl) \Bigr),
\end{equation*}
where 
\begin{equation}\label{Bt-def}
 \Bh^{\lambda\beta\alpha}=    -\chih^\lambda\gh^{\beta\alpha} -\chih^\beta \gh^{\lambda\alpha} 
  + \gh^{\lambda\beta}\chih^\alpha,
\end{equation}
and in deriving this we have employed the definitions
\eqref{tvars} and we are now viewing $\Qh_{\mu\nu}$ as well as $\Gammah^\gamma_{\mu\nu}$ as depending on $\gh_{\mu\nu}$ and $\hh_{\gamma\mu\nu}$.
Performing the same procedure on the other equations
\eqref{tconfESF-B.2}-\eqref{tconfESF-B.3}, it follows that
\eqref{tconfESF-B.1}-\eqref{tconfESF-B.3} can be expressed in
first order form as
\begin{align}
   \Bh^{\lambda\beta\alpha} \delh{\alpha}\hh_{\beta\mu\nu}&= \chih^\lambda \Bigr(\Qh_{\mu\nu}+\frac{2}{\tauh}\bigl( \wh_{(\mu \nu)} - \Gammah^\gamma_{\mu\nu} \vh_\gamma\bigl) \Bigr), \label{tconf-ford-B.1}\\
     \Bh^{\lambda\beta\alpha}\delh{\alpha}\wh_{\beta\mu}
     & =-\chih^\lambda\gh^{\alpha \sigma}\gh^{\beta\delta}\hh_{\mu\sigma\delta}\wh_{\alpha\beta}, \label{tconf-ford-B.2}\\
      \Bh^{\lambda\beta\alpha}\delh{\alpha}\zh_{\beta} &=0, \label{tconf-ford-B.3} \\
      \chih^\alpha \delh{\alpha}\gh_{\mu\nu}&= \chih^\alpha \hh_{\alpha\mu\nu}, \label{tconf-ford-B.4}\\
     \chih^\alpha \delh{\alpha} \vh_{\mu}&= \chih^\alpha \wh_{\alpha\mu},  \label{tconf-ford-B.5}\\
     \chih^\alpha \delh{\alpha} \tauh&= \chih^\alpha \zh_{\alpha},  \label{tconf-ford-B.6}
\end{align}
where here $\zh_\mu$ should be interpreted as being the derivative of $\tauh$, i.e. 
\begin{equation} \label{zh-def}
    \zh_\mu = \delh{\mu}\tauh.
\end{equation}
Although, $\vh_\mu$ and $\zh_\mu$ are both equal to $\delh{\mu}\tauh$ by definition, we treat them as being, a piori, independent variables and verify later that $\vh_\mu=\zh_\mu=\delh{\mu}\tauh$.
It turns out to be technically easier to show first that the constraint $\zh_\mu-\delh{\mu}\tauh=0$ propagates and then use this to show that $\vh_\mu-\delh{\mu}\tauh=0$ propagates, which together yield the equality $\vh_\mu=\zh_\mu=\delh{\mu}\tauh$.

\subsection{Lagrangian coordinates\label{Lag-coordinates}}
Rather than solving the first order system \eqref{tconf-ford-B.1}-\eqref{tconf-ford-B.6} directly, we instead introduce Lagrangian coordinates adapted to the vector field $\chih^\alpha$ and solve a transformed version of this system. The main reason for doing this is that working in Lagrangian coordinates will be essential for our stability proof. 
Letting $\Gc_s(\xh^\lambda) = (\Gc^\mu_s(\xh^\lambda))$
denote the flow of $\chih^\mu$ so that
\begin{equation*}
\frac{d\;}{ds} \Gc^\mu_s(\xh^\lambda) = \chih^\mu(\Gc_s(\xh^\lambda)) \AND
\Gc^\mu_{0}(\xh^\lambda) = \xh^\mu,
\end{equation*}
we introduce Lagrangian coordinates $(x^\mu)$ via
\begin{equation} \label{lemBa.1a}
\xh^\mu = l^\mu(x) := \Gc_{x^0-t_0}^\mu(t_0,x^\Lambda), \quad \forall\, (x^0,x^\Lambda)\in M_{t_1,t_0},
\end{equation}
and note that this map defines a diffeomorphism
\begin{equation*}
    l \: : \: M_{t_1,t_0} \longrightarrow l(M_{t_1,t_0})\subset M_{-\infty,t_0}
\end{equation*}
that satisfies $l(\Sigma_{t_0})=\Sigma_{t_0}$ so long as the vector field $\chih^\mu$ does not vanish.
In line with our coordinate conventions, we often use 
$t=x^0$
to denote the Lagrangian time coordinate.
\begin{rem}
As we show below, by formulating the reduced conformal Einstein-scalar field equations in Lagrangian coordinates, the Lagrangian map $l^\mu$ becomes an additional unknown field that needs to be solved for. The local-in-time existence theory developed in Proposition \ref{lag-exist-prop} will then guarantee that $l^\mu$ exists and is well-defined on a spacetime region of the form $M_{t_1,t_0}$ for $t_1$ sufficiently close to $t_0$. Moreover, the continuation principle from Proposition \ref{lag-exist-prop} will ensure that $l^\mu$ can be extended to domains of the form  $M_{t_1^*,t_0}$ with $t_1^*<t_1$ provided the full solution to the reduced conformal Einstein-scalar field equations in Lagrangian coordinates satisfies appropriate bounds. In this way, the local-in-time existence and continuation theory from Proposition \ref{lag-exist-prop} determines the domain of definition of the Lagrangian map $l^\mu$. 
\end{rem}

By definition, the map $l^\mu$ solves 
the IVP
\begin{align}
    \del{0}l^\mu &= \chihu^\mu, \label{l-ev.1} \\
    l^\mu(t_0,x^\Lambda) &= \delta^\mu_0 t_0 +\delta^\mu_\Lambda x^\Lambda, \label{l-ev.2} 
\end{align}
where, here and below, we use the notation
\begin{equation}\label{fu-def}
    \underline{f}=f\circ l
\end{equation}
to denote the pull-back of scalars by the Lagrangian map $l$.
In the following, symbols without a ``hat'' will denote the geometric pull-back by the Lagrangian map $l$. For example,
\begin{equation}\label{chi-lag}
\chi^\mu = \Jcch_\nu^\mu \chihu^\nu \oset{\eqref{l-ev.1}}{=} \Jcch_\nu^\mu \Jc_0^\nu =\delta^\mu_0,
\end{equation}
\begin{equation} \label{tau-lag}
\tau = \underline{\tauh},
\end{equation}
and
\begin{equation} \label{g-lag}
    g_{\mu\nu} = \Jc^\alpha_\nu \Jc^\beta_\mu\ghu_{\alpha\beta},
\end{equation}
where
\begin{equation} \label{Jc-def}
\Jc^\mu_\nu = \del{\nu}l^\mu
\end{equation}
is the Jacobian matrix of the map $l^\mu$ and 
\begin{equation} \label{Jcch-def}
    (\Jcch^\mu_\nu):= (\Jc^\mu_\nu)^{-1}
\end{equation}    
is its inverse.

Differentiating $\Jc^\mu_\nu$ with respect to $t=x^0$,
it then follows from \eqref{chit-def},  \eqref{l-ev.1}, \eqref{Jc-def} and the transformation law
\begin{equation} \label{pd-trans}
 \underline{\delh{\alpha} f} = \Jcch^\gamma_\alpha \del{\gamma}\underline{f}
\end{equation}
for partial derivatives that
\begin{align*}
    \del{0}\Jc^\mu_\nu = \del{\nu}\chihu^\mu &= \Jc_\nu^\lambda \underline{\delh{\lambda}\chih^\mu} =\frac{1}{|\vhu|^2_{\ghu}} \Jc_\nu^\lambda  \biggl(
    \underline{\delh{\lambda}\vh^\mu}-\frac{1}{|\vhu|^2_{\ghu}}\underline{\delh{\lambda}|\vh|_{\gh}^2}\vhu^\mu \biggr) \\
    &=\frac{1}{|\vhu|^2_{\ghu}} \Jc_\nu^\lambda  \biggl(
   \ghu^{\mu\sigma}\underline{\delh{\lambda}\vh_\sigma}-\ghu^{\mu\tau}\ghu^{\sigma\omega}\underline{\delh{\lambda} \gh_{\tau\omega}}\vhu_\sigma -\frac{1}{|\vhu|^2_{\ghu}}\bigl(-\ghu^{\alpha\tau}\ghu^{\beta\omega}\underline{\delh{\lambda} \gh_{\tau\omega}}\vhu_\alpha\vhu_\beta+2\ghu^{\alpha\beta}\vhu_\alpha \underline{\delh{\lambda}\vh_{\beta}}\bigr)\ghu^{\sigma\mu}\vhu_\sigma \biggr).
\end{align*}
By \eqref{tvars} and \eqref{l-ev.1}-\eqref{l-ev.2}, we
observe that $\Jc^\mu_\nu$ satisfies the IVP
\begin{align}
    \del{0}\Jc^\mu_\nu &= \Jc^\lambda_\nu \Jsc_\lambda^\mu, \label{J-ev.1} \\
    \Jc^\mu_\nu(t_0,x^\Lambda) &= \delta^0_\nu \chih^\mu(t_0,x^\Lambda)+\delta_\nu^\Lambda\delta^\mu_\Lambda, \label{J-ev.2} 
\end{align}
where
\begin{equation} \label{Jsc-def}
    \Jsc_\lambda^\mu =\frac{1}{|\vhu|^2_{\ghu}} \biggl(
   \ghu^{\mu\sigma}\whu_{\lambda\sigma}-\ghu^{\mu\tau}\ghu^{\sigma\omega} \hhu_{\lambda\tau\omega}\vhu_\sigma -\frac{1}{|\vhu|^2_{\ghu}}\bigl(-\ghu^{\alpha\tau}\ghu^{\beta\omega} \hhu_{\lambda\tau\omega}\vhu_\alpha\vhu_\beta+2\ghu^{\alpha\beta}\vhu_\alpha \wh_{\lambda\beta}\bigr)\ghu^{\sigma\mu}\vhu_\sigma \biggr).
\end{equation}

Putting everything together, we find from \eqref{l-ev.1}, \eqref{chi-lag}, \eqref{tau-lag}, \eqref{Jc-def}, \eqref{pd-trans} and \eqref{J-ev.1}
that the Lagrangian representation of the system \eqref{tconf-ford-B.1}-\eqref{tconf-ford-B.6} is given by 
\begin{align}
   \Bhu^{\lambda\beta\alpha} \Jcch_\alpha^\gamma \del{\gamma}\hhu_{\beta\mu\nu}&= \Jc_0^\lambda \Bigr(\Qhu_{\mu\nu}+\frac{2}{\tau}\bigl( \whu_{(\mu \nu)} - \Gammahu^\gamma_{\mu\nu} \vhu_\gamma\bigl) \Bigr), \label{tconf-ford-C.1}\\
     \Bhu^{\lambda\beta\alpha} \Jcch_\alpha^\gamma \del{\gamma}\whu_{\beta\mu},
     &=-\Jc^\lambda_0\ghu^{\alpha \sigma}\ghu^{\beta\delta}\hhu_{\mu\sigma\delta}\whu_{\alpha\beta}, \label{tconf-ford-C.2}\\
      \Bhu^{\lambda\beta\alpha} \Jcch_\alpha^\gamma \del{\gamma}\zhu_{\beta} &=0, \label{tconf-ford-C.3} \\
       \del{0}\ghu_{\mu\nu}&= \Jc_0^\alpha \hhu_{\alpha\mu\nu}, \label{tconf-ford-C.4}\\
     \del{0} \vhu_{\mu}&= \Jc_0^\alpha \whu_{\alpha\mu},  \label{tconf-ford-C.5}\\
     \del{0} \tau&= \Jc^\alpha_0\zhu_{\alpha},  \label{tconf-ford-C.6}\\
     \del{0} \Jc^\mu_\nu &=  \Jc^\lambda_\nu \Jsc^\mu_\lambda,\label{tconf-ford-C.7}\\
    \del{0}l^\mu &= \chihu^\mu, \label{tconf-ford-C.8}\end{align}
where we now view the first order variables $\{\hhu_{\beta\mu\nu},\whu_{\beta\mu},\zhu_{\beta},\ghu_{\mu\nu},\vhu_\mu,\tau,\Jc^\mu_\nu,l^\mu\}$ as being 
independent, that is, not a priori related by \eqref{tvars}, \eqref{zh-def} and \eqref{Jc-def}, $\Bhu^{\lambda\beta\alpha}$ is defined by \eqref{Bt-def}, $\chihu^\mu$ is defined by \eqref{chit-def}, $\Qh_{\mu\nu}$ and $\Gammah^\gamma_{\mu\nu}$ are given by \eqref{eq:defQh} and the standard Christoffel symbol formula, respectively (interpreting $\underline{\delh{\mu} \gh_{\sigma \alpha}}$ as $\hhu_{\mu\sigma\alpha}$), $\ghu^{\mu\nu}$ is the inverse of $\ghu_{\mu\nu}$, and $\Jcch_\alpha^\gamma$ as the inverse of $\Jc_\alpha^\gamma$. 

\subsection{Lagrangian initial data}
In the Lagrangian representation, we see, with the help of
 \eqref{tvars}, \eqref{zh-def}, \eqref{l-ev.2}, \eqref{Jc-def} and \eqref{J-ev.2}, that the reduced conformal Einstein-scalar field initial data
\eqref{gt-idata}-\eqref{dt-taut-idata} generates the following
Lagrangian initial data:
\begin{align}
    l^\mu \bigl|_{\Sigma_{t_0}} &= \lr^\mu, \label{l-idata}\\
    \Jc^\mu_\nu \bigl|_{\Sigma_{t_0}} &= \Jcr^\mu_\nu, \label{Jc-idata}\\
    \tau\bigl|_{\Sigma_{t_0}} &= \taur, \label{tauh-idata} \\ 
    \vhu_\mu\bigl|_{\Sigma_{t_0}} &= \delta^0_\mu \taug +\delta_\mu^\Lambda \del{\Lambda} \taur, \label{vhu-idata} \\
    \whu_{\Lambda\Omega}\bigl|_{\Sigma_{t_0}} &= \del{\Lambda}\del{\Omega} \taur, \label{whu-idata-1}\\
    \whu_{0\Omega}\bigl|_{\Sigma_{t_0}} &= \del{\Omega} \taug, \label{whu-idata-2}\\
     \whu_{\Lambda 0}\bigl|_{\Sigma_{t_0}} &= \del{\Lambda} \taug, \label{whu-idata-3}\\
      \whu_{0 0}\bigl|_{\Sigma_{t_0}} &= -\frac{1}{\gr^{00}}\bigl(2\gr^{0\Lambda}\del{\Lambda}\taug+
      \gr^{\Lambda\Omega}\del{\Lambda}\del{\Omega}\taur \bigr), \label{whu-idata-4}\\
    \zhu_\mu\bigl|_{\Sigma_{t_0}} &= \delta^0_\mu \taug +\delta_\mu^\Lambda \del{\Lambda} \taur, \label{zhu-idata} \\
    \ghu_{\mu\nu}\bigl|_{\Sigma_{t_0}} &= \gr_{\mu\nu}, \label{ghu-idata} \\
    \hhu_{\alpha\mu\nu}\bigl|_{\Sigma_{t_0}} &= \delta_\alpha^0 \grave{g}_{\mu\nu}+\delta_\alpha^\Lambda\del{\Lambda}\gr_{\mu\nu}, \label{hhu-idata}
\intertext{where}
    \lr^\mu &= \delta^\mu_0 t_0 + \delta^\mu_\Lambda x^\Lambda, \label{lr-def} \\
    \vr^\mu &= \gr^{\mu\nu}(\delta^0_\nu \taug +\delta_\nu^\Lambda\del{\Lambda}\taur),\label{vr-def}\\
    \chir^\mu &= \frac{1}{|\vr|_{\gr}^2}\vr^{\mu}, \label{chir-def}\\
    \Jcr^\mu_\nu &= \delta^0_\nu \chir^\mu + \delta^\Lambda_\nu \delta^\mu_\Lambda, \label{Jcr-def}
\end{align}
and we note that the inverse of $\Jcr^\mu_\nu$ is given by
\begin{equation}
    (\Jcr^{-1})^\mu_\nu = \frac{1}{\chir^0}\delta^0_\nu(\delta^\mu_0 -\delta^{\mu}_\Lambda \chir^\Lambda) + \delta^\Lambda_\nu \delta^\mu_\Lambda. \label{Jcr-inv}
\end{equation}

\begin{rem}\label{rem-Lag-idata}
On the initial hypersurface $\Sigma_{t_0}=\{t_0\}\times \Tbb^{n-1}$,  our choice of initial data
implies that 
\begin{equation} \label{dt-tau-idata}
    \del{0}\tau\bigl|_{\Sigma_{t_0}} = \Jc^\mu_0 \zh_\mu\bigl|_{\Sigma_{t_0}} = \chir^\mu \vr_\mu =1, 
\end{equation}
\begin{equation}\label{g-idata}
 g_{\mu\nu} \bigl|_{\Sigma_{t_0}} = \Jc^\alpha_\mu \Jc^\beta_\nu \gh_{\alpha\beta} \bigl|_{\Sigma_{t_0}} 
 = \Jcr^\alpha_\mu \gr_{\alpha\beta}\Jcr^\beta_\nu,
\end{equation}
\begin{equation} \label{dt-g-idata}
    \del{0}g_{\mu\nu} \bigl|_{\Sigma_{t_0}} = 
     \Jcr^\gamma_0\Jcr_\mu^\alpha \Jcr^\beta_\nu \bigl(\delta_\gamma^0 \grave{g}_{\alpha\beta}+\delta_\gamma^\Lambda\del{\Lambda}\gr_{\alpha\beta}\bigr) +  \gr_{\alpha\beta}\bigl(
     \Jcr_\mu^\lambda \mathring{\Jsc}{}^\alpha_\lambda\Jcr^\beta_\nu+\Jcr^\alpha_\mu \Jcr_\nu^\lambda \mathring{\Jsc}{}^\beta_\lambda\bigr)
\end{equation}
and
\begin{equation} \label{chi-idata}
    \chi^\mu \bigl|_{\Sigma_{t_0}} 
    \overset{\eqref{chi-lag}}{=} \delta^\mu_0,
\end{equation}
where $\mathring{\Jsc}{}^\mu_\nu= \Jsc^\mu_\nu|_{\Sigma_{t_0}}$.
\end{rem}

\begin{rem} \label{FLRW-idata-rem-A}
By \eqref{Kasner-wave-gauge} and \eqref{gtau-FLRW}, it is clear that the FLRW solution 
determines the initial data
\begin{equation*}
    \bigl\{g_{\mu\nu}\bigl|_{\Sigma_{t_0}},\del{0}g_{\mu\nu}\bigl|_{\Sigma_{t_0}}, \tau\bigl|_{\Sigma_{t_0}},\del{0}\tau\bigl|_{\Sigma_{t_0}}\bigr\}_{\textrm{FLRW}}=
    \bigl\{\eta_{\mu\nu},0,t_0,1\bigr\}
\end{equation*}
on $\Sigma_{t_0}$ and this initial data satisfies both the gravitational and wave gauge constraints. Furthermore, since the FLRW solution satisfies \eqref{FLRW-Lag}, it is already in the Lagrangian representation.
\end{rem}

\subsection{Local-in-time existence}
In the following proposition, we establish the existence and uniqueness of solutions to the first order system \eqref{tconf-ford-C.1}-\eqref{tconf-ford-C.8} and we develop a continuation principle for these solutions. We also show that solutions of this system that are generated from the initial
data \eqref{l-idata}-\eqref{hhu-idata} determine solutions of the conformal Einstein-scalar field equations provided the initial data satisfies the gravitational and wave gauge constraint equations.  

\begin{prop} \label{lag-exist-prop}
Suppose $k>(n-1)/2+1$, $t_0>0$, and the initial data $\taur\in H^{k+2}(\Tbb^{n-1})$, $\taugr\in H^{k+1}(\Tbb^{n-1})$, $\gr_{\mu\nu}\in H^{k+1}(\Tbb^{n-1},\Sbb{n})$ and $\ggr_{\mu\nu}\in H^{k}(\Tbb^{n-1},\Sbb{n})$ is chosen so that 
the inequalities $\det(\gr_{\mu\nu})<0$ and $|\vr|_{\gr}^2 <0$
are satisfied where $\vr^\mu$ is defined by \eqref{vr-def}.
Then there exists a $t_1<t_0$ and a unique solution
\begin{equation}
  \label{eq:Wreg}
W \in \bigcap_{j=0}^{k}C^j\bigl((t_1,t_0], H^{k-j}(\Tbb^{n-1})\bigr),
\end{equation}
where
\begin{equation}
  \label{eq:Wdef}
    W=(\hhu_{\beta\mu\nu},\whu_{\beta\nu}, \zhu_\beta,\ghu_{\mu\nu},\vhu_\mu,\tau,\Jc^\mu_\nu,\ell^\mu ),
\end{equation}
on $M_{t_1,t_0}$ to the IVP consisting of the evolution
equations \eqref{tconf-ford-C.1}-\eqref{tconf-ford-C.8} and
the initial conditions \eqref{l-idata}-\eqref{hhu-idata}. Moreover, the following properties hold:
\begin{enumerate}[(a)]
\item Letting $W_0 = W|_{\Sigma_{t_0}} \in H^{k}(\Tbb^{n-1})$
denote the initial data, there exists for each $t_*\in (t_1,t_0)$ a $\delta>0$ such that if $\Wt_0 \in  H^{k}(\Tbb^{n-1})$ satisfies $\norm{\Wt_0-W_0}_{H^k(\Tbb^{n-1})}<\delta$,    
then there exists a unique solution 
$\Wt \in \bigcap_{j=0}^{k}C^j\bigl((t_*,t_0], H^{k-j}(\Tbb^{n-1})\bigr)$
of the evolution equations \eqref{tconf-ford-C.1}-\eqref{tconf-ford-C.8} on $M_{t_*,t_0}$  that agrees with the initial data $\Wt_0$ on the initial hypersurface $\Sigma_{t_0}$. 
\item The relations 
\begin{equation} \label{eq:Lag-constraints}
\del{\alpha}\ghu_{\mu\nu}=\Jc^\beta_\alpha \hhu_{\beta\mu\nu},\quad  \del{\alpha}\vhu_{\mu}=\Jc^\beta_\alpha \whu_{\beta\mu}, \quad \del{\alpha}\tau=\Jc^\beta_\alpha \zhu_{\beta},\quad \vhu_\mu=\zhu_\mu \AND \Jc^\mu_\nu = \del{\nu} l^\mu    
\end{equation}
hold in $M_{t_1,t_0}$.
\item The pair $\{g_{\mu\nu}=\del{\mu}l^\alpha\ghu_{\alpha\beta}\del{\nu}l^\beta,\tau\}$
determines a strong solution 
of the reduced conformal Einstein-scalar field equations
\begin{equation} \label{lag-redeqns}
    -2R_{\mu\nu}+2\nabla_{(\mu} X_{\nu)}=-\frac{2}{\tau}\nabla_\mu \nabla_\nu \tau, \quad g^{\alpha\beta}\Dc_\alpha\Dc_\beta \tau=0,
\end{equation}
on $M_{t_1,t_0}$ of
regularity
\begin{equation} \label{reg-loss} 
g_{\mu\nu} \in \bigcap_{j=0}^k C^j\bigl((t_1,t_0], H^{k-j}(\Tbb^{n-1})\bigr) 
\AND \tau \in \bigcap_{j=0}^{k+1} C^j\bigl((t_1,t_0], H^{k+1-j}(\Tbb^{n-1})\bigr)
\end{equation}
that satisfies the initial conditions \eqref{tauh-idata} and \eqref{dt-tau-idata}-\eqref{dt-g-idata}, where $X^\gamma = \frac{1}{2}g^{\mu\nu}g^{\gamma\lambda}
    (2\Dc_{\mu}g_{\nu\lambda}-\Dc_\lambda g_{\mu\nu})$
and $\Dc_\mu$ is the Levi-Civita connection of the flat metric
$\gc_{\mu\nu}=\del{\mu}l^\alpha \eta_{\alpha\beta} \del{\nu}l^\beta$
on $M_{t_1,t_0}$. 
\item The scalar field $\tau$ is given by
\begin{equation} \label{tau-synch} 
    \tau = t-t_0 + \taur  
\end{equation}
in $M_{t_1,t_0}$ while the vector field
\begin{equation} \label{chi-def}
    \chi^\mu = \frac{1}{|\nabla\tau|^2_g}\nabla^\mu\tau 
\end{equation}
satisfies 
\begin{equation} \label{Lagrangian}
\chi^\mu = \delta^\mu_0
\end{equation}
in $M_{t_1,t_0}$.
\item If the initial data $\{\gr_{\mu\nu},\ggr_{\mu\nu},\taur,\taugr\}$ also satisfies the
constraint equations \eqref{grav-constr}-\eqref{wave-constr}
on $\Sigma_{t_0}$, then the pair $\{g_{\mu\nu},\tau\}$ solves
the conformal Einstein-scalar field equations
\begin{align} 
G^{\mu\nu} = \frac{1}{\tau}\nabla^\mu\nabla^\nu \tau, \quad
\Box_{g}\tau = 0,\label{lag-confeqns}
\end{align}
and satisfies the 
wave gauge constraint 
\begin{equation} \label{lag-wave-gauge}
    X^\gamma :=\frac{1}{2}g^{\mu\nu}g^{\gamma\lambda}(2\Dc_\mu g_{\nu\lambda}-\Dc_{\lambda}g_{\mu\nu} )=0  
\end{equation}
in $M_{t_1,t_0}$.
\item If 
  \begin{equation}
    \label{eq:cont_crit1}
\max\biggl\{\sup_{M_{t_1,t_0}}\!\det(g_{\mu\nu}), \sup_{M_{t_1,t_0}}\!|\nabla\tau|_{g}^2\biggl\} <0
\end{equation}
and
\begin{equation}
  \label{eq:cont_crit2} \sup_{t_1<t<t_0}\Bigl(\norm{g_{\mu\nu}(t)}_{W^{2,\infty}(\Tbb^{n-1})}+\norm{\del{t}g_{\mu\nu}(t)}_{W^{1,\infty}(\Tbb^{n-1})}
  +\norm{\Dc_\nu \chi^\lambda(t)}_{W^{2,\infty}(\Tbb^{n-1})}+\norm{\del{t}(\Dc_\nu \chi^\lambda)(t)}_{W^{1,\infty}(\Tbb^{n-1})}\Bigr)<\infty,
\end{equation}
then there exists a $t_1^*<t_1$ such that the solution $W$ can be uniquely continued to the time interval $(t_1^*,t_0]$.
\end{enumerate}
\end{prop}

\begin{proof}
$\;$\\
\noindent \underline{Existence - first order equations:} We first note from \eqref{Bt-def} and  \eqref{fu-def} that $\Bhu^{\lambda \beta \alpha}$ satisfies $\Bhu^{\lambda \beta \alpha} = \Bhu^{\beta \lambda \alpha}$
and that we can write
$\Bhu^{\lambda \beta \alpha}\Jcch^0_\alpha$ as 
\begin{equation}\label{Bhu0}
    \Bhu^{\lambda \beta \alpha}\Jcch^0_\alpha
    = -\frac{\chihu^\sigma \Jcch^0_\sigma}{|\chihu|_{\ghu}^2}
    \chihu^\beta\chihu^\lambda - \chihu^\lambda \pi^{\beta\sigma}\Jcch^0_\sigma - \chihu^\beta \pi^{\lambda\sigma}\Jcch^0_\sigma + \pi^{\lambda \beta}\chihu^\sigma \Jcch^0_\sigma
\end{equation}
where 
\begin{equation}\label{pi-def} 
\pi^\gamma_\lambda = \delta^\gamma_\lambda  -\frac{\chihu^\gamma \chihu_\lambda}{|\chihu|_{\ghu}^2}
\end{equation}
defines the projection onto the $\ghu$-orthogonal subspace to $\chihu^\mu$ and we have set $\pi^{\gamma\lambda}=\ghu^{\lambda\sigma}\pi^\gamma_\sigma$.
Now, it is not difficult to verify from 
\eqref{Bhu0} that $\Bhu^{\lambda \beta \alpha}\Jcch^0_\alpha$
will remain positive definite as long as  the co-vector field $\Jcch^0_\alpha$ remains future directed and timelike, that is, 
\begin{equation}\label{Bhu0-pos}
\chihu^\sigma \Jcch^0_\sigma >0 \AND \ghu^{\alpha\beta}\Jcch^0_\alpha \Jcch^0_\beta < 0.
\end{equation}
We further note that the coefficients of the system \eqref{tconf-ford-C.1}-\eqref{tconf-ford-C.8} depend smoothly on the fields $\{\ghu_{\mu\nu},\vhu_{\mu}, \Jc^\mu_\nu\}$ as long as
\begin{equation} \label{coeff-bnds}
    \det(\ghu_{\mu\nu})<0, \quad \det(\Jc^\mu_\nu) > 0 \AND |\vhu|_{\gh}^2 <0.
\end{equation}
Taken together, these observations imply that the system
\eqref{tconf-ford-C.1}-\eqref{tconf-ford-C.8} is symmetric hyperbolic.

Since $k \in \Zbb_{>(n-1)/2+1}$ and the initial data
for \eqref{tconf-ford-C.1}-\eqref{tconf-ford-C.8} is
determined from \eqref{l-idata}-\eqref{hhu-idata} and the
intial data 
$\taur\in H^{k+2}(\Tbb^{n-1})$, $\taugr\in H^{k+1}(\Tbb^{n-1})$, $\gr_{\mu\nu}\in H^{k+1}(\Tbb^{n-1},\Sbb{n})$, and $\ggr_{\mu\nu}\in H^{k}(\Tbb^{n-1},\Sbb{n})$, which, by assumption, are chosen so that the inequalities \eqref{Bhu0-pos} and \eqref{coeff-bnds} are satisfied on the initial hypersurface $\Sigma_{t_0}$, we can appeal to standard local-in-time existence and uniqueness results for systems of symmetric hyperbolic equations, e.g. see \cite[Thm.~10.1]{BenzoniSerre:2007},
to deduce the existence of a time $t_1 < t_0$ such that
the IVP consisting of
the evolution equations \eqref{tconf-ford-C.1}-\eqref{tconf-ford-C.8} and 
initial conditions \eqref{l-idata}-\eqref{hhu-idata} admits a unique solution\footnote{While Theorem 10.1 from \cite{BenzoniSerre:2007} only guarantees the existence and uniqueness of a solution $W \in \bigcap_{j=0}^1 C^j\bigl((t_1,t_0], H^{k-j}(\Tbb^{n-1})\bigr)$, the regularity \eqref{sol-reg} of the solution $W$ can be established by first using the evolution equations \eqref{tconf-ford-C.1}-\eqref{tconf-ford-C.8} to express the time derivatives $\del{t}^jW$, $1\leq j\leq k$, in terms of $W$ and its spatial derivatives. Using these representations for $\del{t}^j W$, it is then not difficult to verify that \eqref{sol-reg} is a consequence of
$W\in \bigcap_{j=0}^1 C^j\bigl((t_1,t_0], H^{k-j}(\Tbb^{n-1})\bigr)$
and the calculus inequalities from Appendix \ref{calc}.}
\begin{equation} \label{sol-reg}
W \in \bigcap_{j=0}^k C^j\bigl((t_1,t_0], H^{k-j}(\Tbb^{n-1})\bigr),
\end{equation}
where $W$ is defined by \eqref{eq:Wdef},
that satisfies the bounds \eqref{Bhu0-pos}-\eqref{coeff-bnds} 
on $M_{t_1,t_0}$. By the continuation principle satisfied by symmetric hyperbolic systems, e.g, see \cite[Thm.~10.3]{BenzoniSerre:2007}, if the solution \eqref{sol-reg} satisfies
\begin{equation} \label{sol-cont-A}
    \sup_{t_1<t<t_0}\norm{W(t)}_{W^{1,\infty}(\Tbb^{n-1})} < \infty 
\end{equation}
and 
\begin{equation}
\max\biggl\{-\sup_{M_{t_1,t_0}}\!\chihu^\sigma \Jcch^0_\sigma, \sup_{M_{t_1,t_0}}\!\ghu^{\alpha\beta}\Jcch^0_\alpha \Jcch^0_\beta ,\sup_{M_{t_1,t_0}}\!\det(\ghu_{\mu\nu}), -\sup_{M_{t_1,t_0}}\!\det(\Jc^\mu_\nu), \sup_{M_{t_1,t_0}}\!|\vhu|_{\gh}^2\biggl\} < 0, \label{sol-cont-B}
\end{equation}
then there exists a time $t_1^* <t_1$ such that the solution
can be uniquely continued to the time interval $(t_1^*,t_0]$. Moreover, by the continuous dependence on initial conditions property that solutions of symmetric hyperbolic systems enjoy, e.g. see \cite[Theorem~III]{Kato:1975},  for each $t_*\in (t_1,t_0)$
there exists a $\delta>0$ such that if $\Wt_0 \in  H^{k}(\Tbb^{n-1})$ satisfies
$\norm{\Wt_0-W_0}_{H^k(\Tbb^{n-1})}<\delta$, then there exists a unique solution $\Wt \in \bigcap_{j=0}^k C^j\bigl((t_*,t_0], H^{k-j}(\Tbb^{n-1})\bigr)$
on $M_{t_*,t_0}$ of the evolution equations \eqref{tconf-ford-C.1}-\eqref{tconf-ford-C.8} that agrees with the initial data $\Wt_0$ on the initial hypersurface. The geometric version \eqref{eq:cont_crit1}-\eqref{eq:cont_crit2} of the continuation criterion \eqref{sol-cont-A}-\eqref{sol-cont-B} is derived below.

\bigskip

\noindent \underline{Propagation of the constraints \eqref{eq:Lag-constraints}:}
In order to guarantee that solutions to \eqref{tconf-ford-C.1}-\eqref{tconf-ford-C.8} define solutions to the Einstein-scalar field equations, we need to show that the constraints \eqref{eq:Lag-constraints} used to define the first order Lagrangian variables propagate. To this end,  
we observe,  by \eqref{chit-def} and \eqref{fu-def}, that $\chihu^\mu$ depends smoothly on $(\vhu_\alpha,\ghu_{\alpha\beta})$ and by \eqref{Jsc-def} that 
\begin{equation} \label{lag-cprop-0}
    \Jsc^\mu_\lambda = \frac{\del{}\chihu^\mu}{\del{}\ghu_{\alpha\beta}}\hhu_{\lambda\alpha\beta}+
    \frac{\del{}\chihu^\mu}{\del{}\vhu_\alpha}\whu_{\lambda\alpha}.
\end{equation}
With the help of this expression, we see from \eqref{tconf-ford-C.4}-\eqref{tconf-ford-C.5} and 
\eqref{tconf-ford-C.7}-\eqref{tconf-ford-C.8} that
\begin{equation*}
    \del{0}(\del{0}l^\mu - \Jc^\mu_0) = 
    \frac{\del{}\chihu^\mu}{\del{}\ghu_{\alpha\beta}}\bigl(\del{0}\ghu_{\alpha\beta}-\Jc^\lambda_0\hhu_{\lambda\alpha\beta}\bigr)+
    \frac{\del{}\chihu^\mu}{\del{}\vhu_\alpha}\bigl(\del{0}\vhu_\alpha-\Jc^\lambda_0\whu_{\lambda\alpha}\bigr) = 0,
\end{equation*}
which, in turn, implies 
\begin{equation} \label{lag-cprop-1}
     \Jc^\mu_0 = \del{0}l^\mu =\chihu^\mu \quad \text{in $M_{t_1,t_0}$}
\end{equation}
since $(\del{0}l^\mu - \Jc^\mu_0)|_{\Sigma_{t_0}}=0$ by assumption.

Now, applying the projection \eqref{pi-def} to $\Bhu^{\lambda\beta\alpha}$ yields
\begin{equation*}
    \pi^\gamma_\lambda \Bhu^{\lambda\beta\alpha}\oset{\eqref{Bt-def}}{=} -\pi^{\gamma\alpha}\chihu^\beta+\pi^{\gamma\beta}\chihu^\alpha  = - \ghu^{\gamma\alpha}\chihu^\beta+\ghu^{\gamma\beta}\chihu^{\alpha}.
\end{equation*}
Then applying the projection \eqref{pi-def} to \eqref{tconf-ford-C.1}, we find, with the help of the above identity and
\eqref{lag-cprop-1}, that 
\begin{equation*}
    (- \ghu^{\gamma\alpha}\chihu^\beta+\ghu^{\gamma\beta}\chihu^{\alpha})\Jcch_\alpha^\sigma\del{\sigma}\hhu_{\beta\mu\nu}=0
\end{equation*}
from which
\begin{equation}\label{lag-cprop-2}
    \del{0}\hhu_{\alpha\mu\nu}= \chihu^\beta \Jcch^\sigma_\alpha \del{\sigma}\hhu_{\beta\mu\nu}
\end{equation}
follows since \eqref{lag-cprop-1} implies that $\chihu^\beta \Jcch_\beta^\sigma = \delta^\sigma_0$. By similar arguments, we find also that 
\begin{equation}\label{lag-cprop-3}
    \del{0}\whu_{\alpha\mu}= \chihu^\beta \Jcch^\sigma_\alpha \del{\sigma}\whu_{\beta\mu} \AND
    \del{0}\zhu_{\alpha}= \chihu^\beta \Jcch^\sigma_\alpha \del{\sigma}\zhu_{\beta}
\end{equation}
are a consequence of the evolution equations \eqref{tconf-ford-C.2}-\eqref{tconf-ford-C.3}.

By \eqref{tconf-ford-C.7} and \eqref{lag-cprop-0}-\eqref{lag-cprop-1}, we observe that
\begin{align}
    2\del{[\alpha} \Jc^\beta_{0]} = \del{\alpha}\chihu^\beta-\Jcch^\lambda_\alpha \Jsc^\beta_\lambda =
    \frac{\del{}\chihu^\beta}{\del{}\ghu_{\sigma\gamma}}(\del{\alpha}\ghu_{\sigma\gamma}-\Jc^\lambda_\alpha\hhu_{\lambda\sigma\gamma})+
    \frac{\del{}\chihu^\beta}{\del{}\vhu_\sigma}(\del{\alpha}\vhu_{\sigma}-\Jc^\lambda_\alpha\whu_{\lambda\sigma}).
    \label{lag-cprop-7}
\end{align}
We also observe, with the help of \eqref{tconf-ford-C.4}, \eqref{lag-cprop-1} and \eqref{lag-cprop-2}, that
\begin{align*}
    \del{0}(\del{\alpha}\ghu_{\mu\nu}-\Jc^\beta_\alpha \hhu_{\beta\mu\nu}) &= \del{\alpha}\del{0}\ghu_{\mu\nu}-\del{0}\Jc^\beta_\alpha \hhu_{\beta\mu\nu}-\Jc^\beta_\alpha \del{0}\hhu_{\beta\mu\nu} \\
    &= \del{\alpha}(\Jc_0^\beta\hhu_{\beta\mu\nu})-\del{0}\Jc^\beta_\alpha \hhu_{\beta\mu\nu}-\Jc^\sigma_\alpha \Jc^\beta_0 \Jcch_\sigma^\tau \del{\tau}\hhu_{\beta\mu\nu}\\
    &= 2\del{[\alpha} \Jc^\beta_{0]} \hhu_{\beta\mu\nu}.
\end{align*}
Combining this with \eqref{lag-cprop-7} yields
\begin{equation}\label{lag-cprop-4}
  \del{0}(\del{\alpha}\ghu_{\mu\nu}-\Jc^\beta_\alpha \hhu_{\beta\mu\nu})  =\biggl(\frac{\del{}\chihu^\beta}{\del{}\ghu_{\sigma\gamma}}(\del{\alpha}\ghu_{\sigma\gamma}-\Jc^\lambda_\alpha\hhu_{\lambda\sigma\gamma})+
    \frac{\del{}\chihu^\beta}{\del{}\vhu_\sigma}(\del{\alpha}\vhu_{\sigma}-\Jc^\lambda_\alpha\whu_{\lambda\sigma})\biggr)  \hhu_{\beta\mu\nu} .  
\end{equation}
Moreover, employing similar arguments, we find that  
\begin{align}
     \del{0}(\del{\alpha}\vhu_{\mu}-\Jc^\beta_\alpha \whu_{\beta\mu}) = \biggl(\frac{\del{}\chihu^\beta}{\del{}\ghu_{\sigma\gamma}}(\del{\alpha}\ghu_{\sigma\gamma}-\Jc^\lambda_\alpha\hhu_{\lambda\sigma\gamma})+
    \frac{\del{}\chihu^\beta}{\del{}\vhu_\sigma}(\del{\alpha}\vhu_{\sigma}-\Jc^\lambda_\alpha\whu_{\lambda\sigma})\biggr) \whu_{\beta\mu}\label{lag-cprop-5}
     \intertext{and}
      \del{0}(\del{\alpha}\tau-\Jc^\beta_\alpha \zhu_{\beta}) = \biggl(\frac{\del{}\chihu^\beta}{\del{}\ghu_{\sigma\gamma}}(\del{\alpha}\ghu_{\sigma\gamma}-\Jc^\lambda_\alpha\hhu_{\lambda\sigma\gamma})+
    \frac{\del{}\chihu^\beta}{\del{}\vhu_\sigma}(\del{\alpha}\vhu_{\sigma}-\Jc^\lambda_\alpha\whu_{\lambda\sigma})\biggr) \zhu_{\beta}.\label{lag-cprop-6}
\end{align}
But $\del{\alpha}\ghu_{\mu\nu}-\Jc^\beta_\alpha \hhu_{\beta\mu\nu}$, $\del{\alpha}\vhu_{\mu}-\Jc^\beta_\alpha \whu_{\beta\mu}$ and $\del{\alpha}\tau-\Jc^\beta_\alpha \zhu_{\beta}$ all vanish on the initial hypersurface $\Sigma_{t_0}$ by the choice of initial data, and so, we deduce from the uniqueness of solutions to the system of transport equations
\eqref{lag-cprop-4}-\eqref{lag-cprop-6} that
\begin{equation} \label{lag-cprop-8}
   \del{\alpha}\ghu_{\mu\nu}=\Jc^\beta_\alpha \hhu_{\beta\mu\nu},\quad  \del{\alpha}\vhu_{\mu}=\Jc^\beta_\alpha \whu_{\beta\mu} \AND \del{\alpha}\tau=\Jc^\beta_\alpha \zhu_{\beta} 
\end{equation}
in $M_{t_1,t_0}$. This, in turn, implies via \eqref{lag-cprop-1} and \eqref{lag-cprop-7} that \begin{equation*}
       \del{0}(\Jc_\nu^\mu-\del{\nu}\ell^\mu) = \del{\nu}\Jc^\mu_0 - \del{\nu}\del{0}l^\mu = \del{\nu}\Jc^\mu_0-\del{\nu} \Jc^\mu_0 = 0
\end{equation*}
in $M_{t_1,t_0}$. But $\Jc^\mu_\nu-\del{\nu} l^\mu$ vanishes on $\Sigma_{t_0}$, and so, the equality
\begin{equation} \label{lag-cprop-9}
    \Jc^\mu_\nu = \del{\nu} l^\mu
\end{equation}
must also hold in $M_{t_1,t_0}$.

Using \eqref{lag-cprop-8} and \eqref{lag-cprop-9}, a straightforward calculation shows that the evolution equations
\eqref{tconf-ford-C.2} and \eqref{tconf-ford-C.3} imply
that $\tau$ and $\vhu_\mu$ satisfy the wave equations 
\begin{align}
\ghu^{\alpha\beta}\Jcch_\alpha^\sigma\del{\sigma}(\Jcch_\beta^\gamma\del{\gamma}\tau) &=0,  \label{lag-tau-wave}\\
\ghu^{\alpha\beta}\Jcch_\alpha^\sigma\del{\sigma}(\Jcch_\beta^\gamma\del{\gamma}\vhu_\mu) &=\ghu^{\alpha \sigma}\ghu^{\beta\delta}\Jcch_\mu^\lambda \del{\lambda}\ghu_{\sigma\delta}\Jcch_\alpha^\omega \del{\omega}\vhu_{\beta}, \label{lag-v-wave}  
\end{align}
in $M_{t_1,t_0}$. Applying $\Jcch^\omega_\mu \del{\omega}$
to \eqref{lag-tau-wave}, we find with the help of
the commutation relationship $[\Jcch^\omega_\mu \del{\omega},\Jcch^\lambda_\nu \del{\lambda}]=0$, which is a consequence of \eqref{lag-cprop-9}, that 
$\Jch^\omega_\mu \del{\omega}\tau$ satisfies
\begin{equation*}
   \ghu^{\alpha\beta}\Jcch_\alpha^\sigma\del{\sigma}(\Jcch_\beta^\gamma\del{\gamma}(\Jcch^\omega_\mu \del{\omega}\tau)) = \ghu^{\alpha \sigma}\ghu^{\beta\delta}\Jcch_\mu^\lambda \del{\lambda}\ghu_{\sigma\delta}\Jcch_\alpha^\omega \del{\omega}(\Jcch_\beta^\gamma \del{\gamma}\tau).
\end{equation*}
Subtracting this from \eqref{lag-v-wave} shows that
\begin{align*}
   \ghu^{\alpha\beta}\Jcch_\alpha^\sigma\del{\sigma}(\Jcch_\beta^\gamma\del{\gamma}(\vhu_\mu-\Jcch^\omega_\mu \del{\omega}\tau)) =& \ghu^{\alpha \sigma}\ghu^{\beta\delta}\Jcch_\mu^\lambda \del{\lambda}\ghu_{\sigma\delta}\Jcch_\alpha^\omega \del{\omega}(\vhu_\beta-\Jcch_\beta^\gamma \del{\gamma}\tau)
\end{align*}
in $M_{t_1,t_0}$. By our choice of initial data,
$(\vhu_\mu-\Jcch^\omega_\mu \del{\omega}\tau)\bigl|_{\Sigma_{t_0}} =  \del{0}(\vhu_\mu-\Jcch^\omega_\mu \del{\omega}\tau)\bigl|_{\Sigma_{t_0}} = 0$,
and so, we deduce
\begin{equation} \label{lag-cprop-10}
    \vhu_\mu=\Jcch^\omega_\mu \del{\omega}\tau\quad \text{in $M_{t_1,t_0}$}
\end{equation}
from the uniqueness of solutions to linear wave equations.

\bigskip

\noindent \underline{Existence - reduced conformal Einstein-scalar field equations:}  Now  \eqref{sol-reg}, \eqref{sol-cont-B}, \eqref{lag-cprop-9} and Sobolev's inequality (Theorem \ref{Sobolev}, Appendix \ref{calc}) imply that 
the Lagrangian map $\xh^\mu =l^\mu(x)$ defines a local $C^2$-diffeomorphism\footnote{For $t$ close enough to $t_0$, $l=(l^\mu)$ will define a diffeomorphism from $(t,t_0)\times \Tbb^{n-1}$ onto its image, but this property may not continue to hold as the difference between $t$ and $t_0$ grows.}. Moreover, from  \eqref{sol-reg}, 
\eqref{lag-cprop-8} and Sobolev's inequality, we note that $\ghu_{\mu\nu}$ and $\tau$ are $C^2$ on $(t_1,t_0)\times \Tbb^{n-1}$, and that $\hhu_{\beta\mu\nu}$ and $\whu_{\mu\nu}$
are $C^1$ on $(t_1,t_0)\times \Tbb^{n-1}$. Fixing a point $p \in (t_1,t_0)\times \Tbb^{n-1}$ and letting $\Uc$ be an open neighbourhood of $p$ in spacetime on which $l=(l^\mu)$ defines a diffeomorphism from $\Uc$ to its image, we set
\begin{equation*}
\gh_{\mu\nu} = \ghu_{\mu\nu}\circ l^{-1}\AND \tauh = \tau\circ l^{-1}.
\end{equation*}
By \eqref{fu-def}, \eqref{pd-trans},  \eqref{lag-cprop-8}, \eqref{lag-cprop-9} and \eqref{lag-cprop-10}, the relations \eqref{tvars} and \eqref{eq:Lag-constraints} hold on $l(\Uc)$ and 
\begin{equation*}
\delh{\beta}\gh_{\mu\nu} = \hhu_{\beta\mu\nu}\circ l^{-1}
\AND
\delh{\mu}\delh{\nu}\tauh = \whu_{\mu\nu}\circ l^{-1}.
\end{equation*}
From this, we deduce $\gh_{\mu\nu}$ and $\tauh$ are $C^2$ and $C^3$ on $l(\Uc)$, respectively. Moreover, from the computations carried out in Section \ref{fof}, it follows that $\{\gh_{\mu\nu},\tauh\}$ satisfies the reduced conformal Einstein-scalar field equations 
\eqref{tconfESF-A.1}-\eqref{tconfESF-A.2} on $l(\Uc)$. Using the Lagrangian map $l$ to pull-back $\{\gh_{\mu\nu},\tauh\}$, we conclude
that
\begin{equation} \label{lag-vars}
\{g_{\mu\nu}:=\del{\mu}l^\alpha\ghu_{\alpha\beta}\del{\nu}l^\beta,\tau\}
\end{equation}
defines a solution of the reduced conformal Einstein-scalar field equations \eqref{lag-redeqns}
on $\Uc$. Noting that \eqref{lag-vars} is independent of the particular local inverse of $l$, it
defines a solution of the reduced conformal Einstein-scalar field equations \eqref{lag-redeqns} on $M_{t_1,t_0}$, which, by construction, satisfies the initial conditions
\eqref{tauh-idata} and \eqref{dt-tau-idata}-\eqref{dt-g-idata} on the initial hypersurface $\Sigma_{t_0}$. It is worth noting that since $l$ will not  be $C^3$ on $(t_1,t_0)\times \Tbb^{n-1}$ unless we assume that $k>(n-1)/2+2$, the solution \eqref{lag-vars} will not, in general, be a classical solution of the reduced conformal Einstein-scalar field solution, that is, a $C^2$ solution. However, due to \eqref{sol-reg}, it is a consequence of the product inequalities (Theorem \ref{Product}.(ii), Appendix \ref{calc}) and \eqref{lag-cprop-8} that
\eqref{lag-vars} defines a  strong solution of regularity
\begin{equation*} 
g_{\mu\nu} \in \bigcap_{j=0}^k C^j\bigl((t_1,t_0], H^{k-j}(\Tbb^{n-1})\bigr) 
\AND \tau \in \bigcap_{j=0}^{k+1} C^j\bigl((t_1,t_0], H^{k+1-j}(\Tbb^{n-1})\bigr) .
\end{equation*}

\bigskip

\noindent \underline{Continuation of solutions:} 
We now turn to proving a more useful geometric formulation of
the continuation criteria \eqref{sol-cont-A}-\eqref{sol-cont-B}.
We start by noting from \eqref{chit-def}, \eqref{tconf-ford-C.8}, \eqref{lag-cprop-9} and \eqref{lag-cprop-10}-\eqref{lag-vars} that the vector field
\begin{align}\label{lag-chi-A} 
\chi^\mu := \Jcch^\mu_\nu\chihu^\nu= \frac{\nabla^\mu\tau}{|\nabla \tau|_g^2}
\end{align}
satisfies
\begin{equation}\label{lag-chi-B}
    \chi^\mu = \delta^\mu_0,
\end{equation}
and that
\begin{equation} \label{sol-cont-C}
    \max\bigl\{|\vhu|^2_{\ghu}, \ghu^{\alpha\beta}\Jcch^0_\alpha \Jcch^0_\beta\bigr\} < 0 \quad \Longleftrightarrow \quad |\nabla\tau|_g^2 < 0.
\end{equation}
An important consequence of \eqref{lag-chi-A}-\eqref{lag-chi-B}
is that $\del{0}\tau = 1$, and hence, after integrating in time, that
\begin{equation} \label{lag-tau}
    \tau = t-t_0 + \taur \quad \text{in $M_{t_1,t_0}$.}
\end{equation}

Next, letting $\Dc_\mu$ denote the Levi-Civita connection of the flat metric
\begin{equation} \label{lag-gc}
    \gc_{\mu\nu} = \del{\mu}l^\alpha\eta_{\alpha\beta}\del{\nu}l^\beta,
\end{equation}
we note by \eqref{pd-trans}, \eqref{tconf-ford-C.8}, \eqref{lag-cprop-9} and   \eqref{lag-chi-A},
that
\begin{align*}
    \Dc_\nu \chi^\mu =  \Jc^\alpha_\nu \underline{\delh{\alpha}\chih^\beta}\Jcch^\mu_\beta 
    =  \del{\nu}\chihu^\beta\Jcch_\beta^\mu
    = \del{\nu} \del{0}l^\beta\Jcch_\beta^\mu
    =  \del{0} \Jc^\beta_\nu \Jcch_\beta^\mu, 
\end{align*}
and hence, after rearranging, that
\begin{equation} \label{sol-cont-F}
  \del{0} \Jc^\mu_\nu  =  \Dc_\nu \chi^\lambda \Jc_\lambda^\mu.
\end{equation}
Integrating this in time, we deduce that
\begin{equation} \label{sol-cont-G}
      \sup_{t_1<t<t_0}\norm{\Dc_\nu \chi^\lambda(t)}_{W^{1,\infty}(\Tbb^{n-1})}<\infty\quad \Longrightarrow \quad   \sup_{t_1<t<t_0}\norm{\Jc^\mu_\nu(t)}_{W^{1,\infty}(\Tbb^{n-1})} < \infty.
\end{equation}
Then differentiating \eqref{sol-cont-F} to get
$\del{0} \del{\alpha}\Jc^\mu_\nu  =  (\del{\alpha}\Dc_\nu \chi^\lambda) \Jc_\lambda^\mu +  \Dc_\nu \chi^\lambda \del{\alpha}\Jc_\lambda^\mu$,
we find after integrating in time and employing \eqref{sol-cont-G} that
\begin{align}
      \sup_{t_1<t<t_0}\bigl(\norm{\Dc_\nu \chi^\lambda(t)}_{W^{2,\infty}(\Tbb^{n-1})}+&\norm{\del{t}(\Dc_\nu \chi^\lambda)(t)}_{W^{1,\infty}(\Tbb^{n-1})}\bigr)<\infty \notag \\
      & \Longrightarrow \quad   \sup_{t_1<t<t_0}\bigl(\norm{\Jc^\mu_\nu(t)}_{W^{1,\infty}(\Tbb^{n-1})} + \norm{\del{\alpha}\Jc^\mu_\nu(t)}_{W^{1,\infty}(\Tbb^{n-1})} \bigr)  < \infty.  \label{sol-cont-H}
\end{align}
Furthermore, we conclude from integrating $\del{0}l^\mu=\Jc^\mu_0$ in time that
\begin{equation} \label{sol-cont-I}
      \sup_{t_1<t<t_0}\norm{\Dc_\nu \chi^\lambda(t)}_{W^{1,\infty}(\Tbb^{n-1})}<\infty\quad \Longrightarrow \quad   \sup_{t_1<t<t_0}\norm{l^\mu(t)}_{W^{1,\infty}(\Tbb^{n-1})} < \infty.
\end{equation}

Next, from \eqref{sol-cont-F}  and Lemma \ref{diff-matrix-lem}, we observe that the determinant of $\Jc^\mu_\nu$
can be expressed as
\begin{equation*}
\det(\Jc^\mu_\nu(t,x)) = e^{-\int_{t}^{t_0} (\Dc_\nu \chi^\nu)(s,x)\,ds} \det(\Jc^\mu_\nu(0,x)), \quad x\in \Tbb^{n-1},
\end{equation*}
which, since $\det(\Jc^\mu_\nu)$ is strictly positive on $\Sigma_{t_0}$, allows us to conclude that
\begin{equation} \label{sol-cont-J}
\sup_{t_1<t<t_0}\norm{\Dc_\nu \chi^\lambda(t)}_{W^{1,\infty}(\Tbb^{n-1})}<\infty \quad \Longrightarrow \quad    \inf_{M_{t_1,t_0}}\det(\Jc^\mu_\nu) > 0.  
\end{equation}
With the help of \eqref{sol-cont-C}, \eqref{sol-cont-H}, \eqref{sol-cont-I} and
\eqref{sol-cont-J}, it is then not difficult to verify from \eqref{lag-cprop-8}, \eqref{lag-cprop-10}, \eqref{lag-tau}, \eqref{lag-gc}  and 
$\ghu_{\mu\nu}= \Jcch_\mu^\alpha g_{\alpha\beta}\Jcch_\nu^\beta  $, see \eqref{lag-cprop-9} and \eqref{lag-vars}, that the continuation criteria \eqref{sol-cont-A}-\eqref{sol-cont-B} 
will be satisfied provided that
\eqref{eq:cont_crit1}-\eqref{eq:cont_crit2} are satisfied.

\bigskip
\noindent \underline{Existence - conformal Einstein-scalar field equations:} If the initial data also satisfies the gravitational and gauge constraints, c.f.\ \eqref{grav-constr}-\eqref{wave-constr}, it is then a consequence of
Proposition \ref{prop-ES-constr} that the pair $\{\tau, g_{\mu\nu}\}$, which solves the reduced conformal Einstein-scalar field equations, will also solve the conformal Einstein-scalar field equations \eqref{lag-confeqns} and satisfy the wave gauge constraint \eqref{lag-wave-gauge}.
\end{proof}

\begin{rem} \label{rem:reg-loss}
The lower regularity of $g_{\mu\nu}$, see \eqref{reg-loss}, compared to the (non-geometric) Lagrangian variables \eqref{eq:Wdef} is expected whenever Lagrangian coordinates, which are adapted to a vector field that is part of the evolution system, are employed. In general, the Lagrangian representation of geometric fields lose one order of differentiablity compared to fields whose components are treated as scalars under transformation to Lagrangian coordinates, e.g. $\ghu_{\mu\nu}$. With that said, in certain situations, it can be shown that this loss of regularity does not occur; for example, it is shown in  \cite{Oliynyk:PRD_2012} that there is no regularity loss for solutions of the Einstein-Euler equations for a particular choice of Lagrangian coordinates. We  are not certain if the regularity loss encountered here can be avoided in a similar fashion, but there is hope that it can since solutions of the Einstein-scalar field system can be interpreted as irrotational solutions of the Einstein-Euler equations with a stiff equation of state as long as the gradient of the scalar field stays timelike, which is the case in our setting.

We further note that there is a loss of one order of differentiability in the solutions of $\tau$ compared to the differentiability of its initial data. This is due to our use of
the conformal Einstein-scalar field system. Since this loss of regularity does not occur for the Einstein-scalar field system, it is likely that this loss of regularity can be also be avoided.
\end{rem}

\subsection{Temporal synchronization of the singularity\label{temp-synch}}
In general, the temporal synchronization of a big bang singularity can be achieved through the introduction of a time coordinate whose level set at a particular time, say $0$, coincides with the spacelike singular hypersurface that defines the big bang singularity. From the structure of the reduced conformal Einstein-scalar field equations \eqref{lag-redeqns}, it is clear that these equations become formally singular when $\tau$ vanishes, and because of this, it is reasonable to guess that the location of the singular hypersurface is determined by the vanishing of the scalar field $\tau$. This motivates us to use the scalar field $\tau$ as a time coordinate to synchronize the singularity. However, due to \eqref{tau-synch}, this would force us use to assume that the initial data $\taur$ for the scalar field $\tau$ is constant on the initial hypersurface $\Sigma_{t_0}$, and in fact, would be given by $\taur=t_0$ there.

Now, in general, we cannot assume that the scalar field is constant on the initial hypersurface if we want our results to apply to an open set of geometric initial data. To remedy this, we proceed as follows: if $\tau$ is not constant on the initial hypersurface $\Sigma_{t_0}$, but is close to constant, say $\tau = t_0 + \rhor$ in $\Sigma_{t_0}$ with $\rhor$ a sufficiently small function, then we evolve $\tau$ for a short amount of time to obtain a solution $\{g_{\mu\nu},\tau\}$ of the conformal Einstein-scalar field equations on $M_{t_1,t_0}$ for some $t_1<t_0$ with $t_1$ close to $t_0$. We then find a level surface of $\tau^{-1}(t^*_0)$
for some $t^*_0\in (t_1,t_0)$ that satisfies $\tau^{-1}(t^*_0) \subset (t_1,t_0)\times \Tbb^{n-1}$ and $\tau^{-1}(t^*_0)\cong \Tbb^{n-1}$. 
By replacing $\Sigma_{t_0}$ with $\tau^{-1}(t^*_0)$, we then obtain a hypersurface $\tau^{-1}(t^*_0)\cong \Tbb^{n-1}$ on which
$\tau$ is constant as desired. This construction is made precise in the following proposition.

\begin{prop} \label{synch-prop}
Suppose $k>(n-1)/2+2$, $t_0>0$, the initial data $\taur=t_0+\rhor$, $\rhor\in H^{k+2}(\Tbb^{n-1})$, $\taugr\in H^{k+1}(\Tbb^{n-1})$, $\gr_{\mu\nu}\in H^{k+1}(\Tbb^{n-1},\Sbb{n})$ and $\ggr_{\mu\nu}\in H^{k}(\Tbb^{n-1},\Sbb{n})$ is chosen so that 
the inequalities $\det(\gr_{\mu\nu})<0$ and $|\vr|_{\gr}^2 <0$ hold
and the constraint equations \eqref{grav-constr}-\eqref{wave-constr} are satisfied, and
let $\{\gtt_{\Lambda\Gamma},\Ktt_{\Lambda\Gamma},\taur,\taugr\}$
denote the geometric initial data on $\Sigma_{t_0}=\{t_0\}\times \Tbb^{n-1}$ that is determined from the initial data $\{\gr_{\mu\nu},\ggr_{\mu\nu},\taur=t_0+\rhor,\taugr\}$ via \eqref{gtt-def1}.
Then for any $\tilde{\delta}>0$, there exists a $\delta>0$ and times $t^*_1<t^*_0<t_0$ such that if $\norm{\rhor}_{H^{k+2}(\Tbb^{n-1})}<\delta$, then:
\begin{enumerate}[(a)]
    \item The solution
$W$ from Proposition \ref{lag-exist-prop} to the IVP consisting of the evolution
equations \eqref{tconf-ford-C.1}-\eqref{tconf-ford-C.8} and
the initial conditions \eqref{l-idata}-\eqref{hhu-idata} exists on $M_{t_1^*,t_0}$.
\item The pair $\{g_{\mu\nu}=\del{\mu}l^\alpha\ghu_{\alpha\beta}\del{\nu}l^\beta,\tau\}$
determined from the solution $W$ via \eqref{eq:Wdef} defines a solution to the conformal Einstein-scalar field equations \eqref{lag-confeqns} of regularity \eqref{reg-loss}.
\item The map
\begin{equation*}
    \Psi \: : \: (t_1^*,t_0)\times \Tbb^{n-1}  \longrightarrow  \Rbb \times \Tbb^{n-1}\: :\: (t,x)\longmapsto (\ttl,\xt) = (t+\rhor(x),x)
\end{equation*}
defines a $H^{k+2}$-diffeomorphism onto its image and the push-forward
\begin{equation*} 
(\gt_{\mu\nu},\taut):=((\Psi_*g)_{\mu\nu},\Psi_*\tau)
\in H^k(\Mt)\times
H^{k+1}(\Mt)
\end{equation*}
of the solution $\{g_{\mu\nu},\tau\}$ by this map 
determines geometric initial data 
$\{\gttt_{\Lambda\Sigma},\Kttt_{\Lambda\Sigma},\mathring{\taut},\grave{\taut}\}$ on the hypersurface $\Sigma_{t^*_0}=\{t^*_0\}\times\Tbb^{n-1}$ satisfying $\mathring{\taut}=t^*_0$,
\begin{equation*}
(\gttt_{\Lambda\Sigma},
\Kttt_{\Lambda\Sigma},\grave{\taut}) \in  H^{k}(\Tbb^{n-1})\times H^{k-1}(\Tbb^{n-1})\times H^{k}(\Tbb^{n-1})
\end{equation*}
and
\begin{gather*}
\norm{\gttt_{\Lambda\Sigma}-\gtt_{\Lambda\Sigma}}_{H^{k}(\Tbb^{n-1})}+\norm{\Kttt_{\Lambda\Sigma}-\Ktt_{\Lambda\Sigma}}_{H^{k-1}(\Tbb^{n-1})}+\norm{t^*_0-\taur}_{H^{k+1}(\Tbb^{n-1})}+\norm{\grave{\taut}-\taugr}_{H^{k}(\Tbb^{n-1})} <\tilde{\delta}.    
\end{gather*}
\end{enumerate}
\end{prop}
\begin{proof}
 Given $t_0>0$, let us first set $\taur=t_0$ and suppose the rest of the initial data $\taugr\in H^{k+1}(\Tbb^{n-1})$, $\gr_{\mu\nu}\in H^{k+1}(\Tbb^{n-1},\Sbb{n})$ and $\ggr_{\mu\nu}\in H^{k}(\Tbb^{n-1},\Sbb{n})$ are chosen so that 
the inequalities $\det(\gr_{\mu\nu})<0$ and $|\vr|_{\gr}^2 <0$
are satisfied. Then since $k>(n-1)/2+2$, it follows from Proposition~\ref{lag-exist-prop} that there exists a $t_1<t_0$ and a unique solution, see \eqref{eq:Wdef},
$\Wt \in \bigcap_{j=0}^{k}C^j\bigl((t_1,t_0], H^{k-j}(\Tbb^{n-1})\bigr)$
of the evolution
equations \eqref{tconf-ford-C.1}-\eqref{tconf-ford-C.8}
on $M_{t_1,t_0}$
for which $\Wt_0 = \Wt|_{\Sigma_{t_0}}$
satisfies the initial conditions \eqref{l-idata}-\eqref{hhu-idata}.
Next, we keep the initial data the same as above
except for the change $\taur=t_0 \longmapsto \taur=t_0 + \rhor$,
where $\rhor\in H^{k+2}(\Tbb^{n-1})$, and let
$W_0$ denote the corresponding initial data for the system \eqref{tconf-ford-C.1}-\eqref{tconf-ford-C.8} that is determined by
\eqref{eq:Wdef} and \eqref{l-idata}-\eqref{hhu-idata}.
Then it is not difficult to verify, with the help of the calculus inequalities from Appendix \ref{calc}, that there exists a
$\delta>0$ such that if $\norm{\rhor}_{H^{k+2}(\Tbb^{n-1})}<\delta$, then
$\norm{W_0-\Wt_0}_{H^k(\Tbb^{n-1})} = C(\delta)\delta$.
By Proposition~\ref{lag-exist-prop}.(a), it then follows for any fixed $t_*\in (t_1,t_0)$ that there exists a $\delta>0$ such that the solution
$W$ of \eqref{tconf-ford-C.1}-\eqref{tconf-ford-C.8} generated
by the initial data $W_0$ will also exist for all $t\in (t_*,t_0]$ and satisfies
\begin{equation} \label{W-reg}
W \in \bigcap_{j=0}^{k}C^j\bigl((t_*,t_0], H^{k-j}(\Tbb^{n-1})\bigr).
\end{equation}
As a consequence, we can take the time $t_*$ to be independent of the choice of $\rhor$ so long as $\norm{\rhor}_{H^{k+2}(\Tbb^{n-1})}<\delta$.
Moreover, by Proposition~\ref{lag-exist-prop}.(d), we know that the scalar field $\tau$ determined by this solution satisfies 
\begin{equation}\label{tau-synch-prop-A}
  \tau(t,x) = t+\rhor(x), \quad (t,x)\in (t_*,t_0]\times \Tbb^{n-1}.
\end{equation}

Letting $C_{\text{Sob}}$ denote the Sobolev constant (i.e.\
$\norm{\cdot}_{C^{2,\sigma}(\Tbb^{n-1})}\leq C_{\text{Sob}} \norm{\cdot}_{H^{k+2}(\Tbb^{n-1})}$, $0<\sigma<1$), then for any and
$t_1^*, t^*_0$ satisfying $t_*<t_1^*<t^*_0<t_0$, we observe that
\begin{equation*}
    \tau(t_1^*,x) \leq t_1^*+C_{\text{Sob}}\delta
    < t_0^* < t_0-C_{\text{Sob}}\delta \leq \tau(t_0,x),\quad x\in \Tbb^{n-1},
\end{equation*}
provided that $\delta$ also satisfies
\begin{equation} \label{tau-synch-prop-B}
    \delta < C_{\text{Sob}}^{-1}\max\{t_0^*-t_1^*,t_0-t_0^*\}, 
\end{equation}
which we can always arrange by shrinking $\delta>0$ if necessary. We also note, since $\rhor\in H^{k+2}(\Tbb^{n-1})$ and $k>(n-1)/2+2$, that 
\begin{equation} \label{Psi-synch-prop}
    \Psi \: : \: (t_1^*,t_0)\times \Tbb^{n-1}  \longrightarrow  \Rbb \times \Tbb^{n-1}\: :\: (t,x)\longmapsto (\ttl,\xt) = (t+\rhor(x),x)
\end{equation}
defines a $H^{k+2}$-diffeomorphism, see \cite[\S 2]{EbinMarsden:1970}, onto its image
 image $\Mt  = \Psi((t_1^*,t_0)\times \Tbb^{n-1}) \subset \Rbb \times \Tbb^{n-1}$
with inverse given by
\begin{equation*}
    \Psi^{-1} \: : \:\Mt \longrightarrow   (t_1^*,t_0)\times \Tbb^{n-1} \: :\:  (\ttl,\xt)\longmapsto (t,x) = (\ttl-\rhor(\xt),\xt).
\end{equation*}  
Furthermore, by \eqref{tau-synch-prop-B} and \eqref{Psi-synch-prop}, we observe that the level set
$\tau^{-1}(t_0^*)\subset (t_1^*,t_0)\times \Tbb^{n-1}$ is non-empty, and in fact, is given by
$\tau^{-1}(t_0^*) = \Psi^{-1}\bigl(\Sigma_{t_0^*}\bigr)$
where $\Sigma_{t_0^*}= \{t_0^*\}\times \Tbb^{n-1}$.
Now, by assumption, the initial data $\{\gr_{\mu\nu},\gr_{\mu\nu},\taur=t_0+\rhor, \taugr\}$
solves the constraint equations \eqref{grav-constr}-\eqref{wave-constr}, and consequently, 
by Proposition~\ref{lag-exist-prop}.(e), the
pair $\{ g_{\mu\nu}:= \del{\mu}l^\alpha \ghu_{\alpha\beta}\del{\nu}l^\beta,\tau\}$,
which is derived from \eqref{W-reg} using  \eqref{eq:Wdef}, satisfies
\begin{align*} 
g_{\mu\nu}&\in \bigcap_{j=0}^k C^j\bigl([t_1^*,t_0], H^{k-j}(\Tbb^{n-1})\bigr) \subset H^k\bigl((t_1^*,t_0)\times \Tbb^{n-1}\bigr), 
\\ 
\tau&\in \bigcap_{j=0}^{k+1} C^j\bigl([t_1^*,t_0], H^{k+1-j}(\Tbb^{n-1})\bigr) \subset H^{k+1}\bigl((t_1^*,t_0)\times \Tbb^{n-1}\bigr)
,
\end{align*}
and 
solves the conformal Einstein-scalar field equations \eqref{lag-confeqns} and the wave gauge constraint \eqref{lag-wave-gauge} in $(t_1^*,t_0)\times \Tbb^{n-1}$. 
Since $\Psi$ is a $H^{k+2}$-diffeomorphism, it maps $H^{s}$ functions to $H^{s}$ functions for any $0\leq s\leq k+2$ as can be verified directly using the calculus inequalities or by referring to \cite[\S 2]{EbinMarsden:1970}. Using this property,
we can push-forward the solution $\{g_{\mu\nu},\tau\}$ 
 to $\Mt$ using $\Psi$ to obtain a solution\footnote{
 This way of stating the regularity is not optimal, but it has the benefit of being succinct and sufficient for our purposes. In the end, we are only interested in the regularity of the initial data obtained from this solution, which can be determined directly from the regularity of the solution \eqref{W-reg}, the special form of the diffeomorphism \eqref{Psi-synch-prop}, and the definition of the push-forward \eqref{gt-taut-def}. } \begin{equation} \label{gt-taut-def}
(\gt_{\mu\nu},\taut):=((\Psi_*g)_{\mu\nu},\Psi_*\tau)
\in H^k(\Mt)\times
H^{k+1}(\Mt)
\end{equation}
of the conformal Einstein-scalar field equations on $\Mt$ that also satisfies the wave gauge constraint. 
As a consequence, the restriction
$\{\gt_{\mu\nu}|_{\Sigma_{t_0^*}},\delt{0}\gt_{\mu\nu}|_{\Sigma_{t_0^*}}, \taut|_{\Sigma_{t_0^*}}, \delt{0}\taut|_{\Sigma_{t_0^*}}\}$
of this solution, which we observe from \eqref{tau-synch-prop-A} and  \eqref{Psi-synch-prop} satisfies $\taut|_{\Sigma_{t_0^*}}=t^*_0$,
determines geometric initial data $\{\gttt_{\Lambda\Sigma},\Kttt_{\Lambda\Sigma},\mathring{\taut},\grave{\taut}\}$ on the hypersurface $\Sigma_{t^*_0}$
that solves the gravitational constraint equations.
It is worth noting here that the wave gauge constraint is now determined by the background metric $\tilde{\gc}_{\mu\nu}=(\Psi^*\gc)_{\mu\nu}$, where
$\gc_{\mu\nu}= \del{\mu}l^\alpha \eta_{\alpha\beta}\del{\nu}l^\beta$. While this background metric is still flat, it no longer has constant coefficients as did the starting background metric, c.f. \eqref{gct-def}.

Finally, since we can naturally identify $\Sigma_{t_0}$ and $\Sigma_{t^*_0}$ with $\Tbb^{n-1}$, it then is not difficult to verify from the regularity of the solution \eqref{W-reg}, the special form \eqref{tau-synch-prop-A} of the diffeomorphism $\Psi$, the definition \eqref{gt-taut-def} of $\gt_{\mu\nu}$ and $\taut$, and the calculus inequalities from Appendix \ref{calc}
that, for any given $\tilde{\delta} >0$, we can, by shrinking, if necessary,  $\delta$ and $t_0-t^*_0$ while observing \eqref{tau-synch-prop-B}, ensure that
\begin{equation*}
(\gttt_{\Lambda\Sigma},
\Kttt_{\Lambda\Sigma},\grave{\taut}) \in  H^{k}(\Tbb^{n-1})\times H^{k-1}(\Tbb^{n-1})\times H^{k}(\Tbb^{n-1})
\end{equation*}
and
\begin{equation*}
\norm{\gttt_{\Lambda\Sigma}-\gtt_{\Lambda\Sigma}}_{H^{k}(\Tbb^{n-1})}+\norm{\Kttt_{\Lambda\Sigma}-\Ktt_{\Lambda\Sigma}}_{H^{k-1}(\Tbb^{n-1})}+\norm{t^*_0-\taur}_{H^{k+1}(\Tbb^{n-1})}+\norm{\grave{\taut}-\taugr}_{H^{k}(\Tbb^n)} <\tilde{\delta}.
\end{equation*}
It is worth noting that a weaker version of this results with a loss of a
$1/2$ order of differentiablity (i.e.\ $k$ replaced by $k-1/2$) can be obtained from Trace Theorem by restricting the solution \eqref{gt-taut-def} to the hypersurface $\Sigma_{t_0^*}$.
\end{proof}

\begin{rem} \label{synch-rem}
Given the geometric initial data 
$\{\gttt_{\Lambda\Sigma},\Kttt_{\Lambda\Sigma},\mathring{\taut},\grave{\taut}\}$ on $\Sigma_{t^*_0}$ from Proposition \ref{synch-prop}, we can always solve the wave gauge constraint on $\Sigma_{t^*_0}$ by an appropriate choice of the free initial data\footnote{It is worthwhile noting that this choice of free initial data will, in general, be different from the lapse-shift pair computed from restricting the conformal Einstein-scalar field solution $\{\gt_{\mu\nu},\taut\}$ from Proposition \ref{synch-prop} to $\Sigma_{t^*_0}$.} $\{\Nttt,\bttt_\Lambda,\dot{\Nttt},\dot{\bttt}{}_\Lambda\}$; see Remark \ref{idata-rem} for details. Because of this, we lose no generality, as far as our stability results are concerned, by assuming that the initial data \eqref{gt-idata}-\eqref{dt-taut-idata} satisfies the gravitational and wave gauge constraint equations
\eqref{grav-constr}-\eqref{wave-constr}
along with the synchronization condition $\taur = t_0$ on the initial hypersurface $\Sigma_{t_0}$,
which by \eqref{tau-synch} implies that
\begin{equation} \label{tau-fix}
\tau = t \quad \text{in $M_{t_1,t_0}$}.
\end{equation}
\end{rem}

\section{A Fermi-Walker transported frame\label{Fermi}}
In the calculations carried out in this section, we will assume that $\{g_{\mu\nu},\tau\}$ is a solution, which is guaranteed to exist by Proposition \ref{lag-exist-prop}, of the conformal Einstein-scalar field equations \eqref{lag-confeqns} in the Lagragian coordinates $(x^\mu)$ that satisfies the wave gauge constraint \eqref{lag-wave-gauge} along with the slicing condition \eqref{tau-fix}. A difficulty with the Lagrangian coordinate representation of the conformal metric is that it is not suitable for obtaining estimates that are well behaved near the big bang singularity, which is located at $t=0$ in these coordinates. In order to obtain estimates that are well behaved in the limit $t\searrow 0$, we will instead use a frame representation of the conformal metric given by \begin{equation*}
g_{ij} = e_i^\mu g_{\mu\nu} e^\nu_j,
\end{equation*}
where $e_i = e^\mu_i \del{\mu}$ is a frame that is to be determined. In the following, we take all frame indices as being expressed relative to this frame, and so in particular, the frame components
$e_i^j$ of the frame vector fields $e_i$ are
\begin{equation} \label{frame-def}
e_i^j = \delta^j_i.
\end{equation}

To proceed, we need to fix the frame. We do this by first fixing $e_0$ via
\begin{equation}\label{e0-def}
e_0 = (-|\chi|_g^2)^{-\frac{1}{2}}\chi \oset{\eqref{chi-def}}{=}-\betat \nabla \tau
\end{equation}
where $\betat$ is defined by
\begin{equation}\label{alpha-def}
    \betat = (-|\nabla\tau|^2_g)^{-\frac{1}{2}}.
\end{equation}
From this definition, it follows that $e_0$ is normalized according to
\begin{equation*} 
g_{00} = g(e_0,e_0) =-1.
\end{equation*}
We then complete $e_0$ to a frame by propagating the spatial vector fields $e_J$ using Fermi-Walker transport, which is defined by
\begin{equation} \label{Fermi-A}
\nabla_{e_0}e_J = -\frac{g(\nabla_{e_0}e_0,e_J)}{g(e_0,e_0)} e_0.
\end{equation}
\begin{rem}
Our use of a Fermi-Walker transported spatial frame is inspired by \cite{Fournodavlos_et_al:2023} where this type of spatial frame played a crucial role in the stability results established there. It is interesting to note, however, that the time slicing that we use here is different from the constant mean curvature time foliation employed in \cite{Fournodavlos_et_al:2023}. As a consequence,
the frame used here will, in general, be different even though the spatial part is determined by Fermi-Walker transport.
\end{rem}

Now, a short calculation involving \eqref{Fermi-A} reveals that $e_0(g(e_0,e_J)) = 0$ from which we deduce that the condition $g_{0J}=g(e_0,e_J)=0$ is preserved, that is, if we assume it holds on the initial surface $\Sigma_{t_0}$, then it hold for all time. Assuming that $g_{0J}=0$ holds, it then follows from \eqref{Fermi-A} and a short calculation that
$e_0(g(e_I,e_J)) = 0$, and hence that the condition
$g_{IJ}=g(e_I,e_J)=\delta_{IJ}$
is also preserved. Consequently, by choosing initial data for the spatial frame $e_J^\mu$ on the initial hypersurface $\Sigma_{t_0}$ so that the frame $e_i^\mu$ is orthonormal, we can ensure that the frame remains orthonormal throughout the evolution, that is,
\begin{equation}\label{orthog}
g_{ij}=g(e_i,e_j) = \eta_{ij}:= -\delta_i^0\delta_j^0+ \delta_i^I\delta_j^J\delta_{IJ}.
\end{equation}
In the following calculations, we will always assume that we have chosen the initial data for the frame so that \eqref{orthog} holds.

By definition, the connection coefficients $\Gamma_i{}^k{}_j$ of the metric $g_{ij}$ are determined from the $g$-orthonormal frame $e_i^\mu$ via
\begin{equation} \label{Gamma-def}
\nabla_{e_i} e_j = \Gamma_{i}{}^k{}_{j} e_k,
\end{equation}
while the background connection coefficients $\gamma_{i}{}^k{}_j$
are determined by
\begin{equation*} 
\Dc_{e_i} e_j = \gamma_{i}{}^k{}_{j} e_k.
\end{equation*}
From this point onward, all frame components will be relative to the $g$-orthonormal frame $e_i^\mu$ determined by Fermi-Walker transport.

Using \eqref{Fermi-A} and \eqref{orthog}, we osbserve that
\begin{equation}\label{Fermi-B}
\Gamma_{0}{}^k{}_J =\delta_{IJ}\Gamma_{0}{}^I{}_0 \delta^k_0.
\end{equation}
On the other hand, applying $\nabla_{e_0}$ to \eqref{e0-def} yields $\nabla_{e_0}e_0 = -\betat\bigl( \nabla_{e_0} \nabla\tau + g(\nabla_{e_0}\nabla \tau, e_0)e_{0}\bigr)$
from which we deduce
\begin{equation*}
g(\nabla_{e_0}e_0,e_j)=  -\betat\bigl( g(\nabla_{e_0} \nabla\tau,e_j) + g(\nabla_{e_0}\nabla \tau, e_0)g(e_{0},e_j)\bigr).
\end{equation*}
Using \eqref{Gamma-def} and \eqref{orthog}, we can express this as $\eta_{jk}\Gamma_{0}{}^k{}_0 = -\betat\bigl( \delta^i_0\nabla_i\nabla_j \tau + \delta_0^i\delta^k_0\nabla_i\nabla_k \tau\eta_{0j}\bigr)$
or equivalently,
\begin{equation}\label{Gamma-0k0-A}
\Gamma_{0}{}^k{}_0 =-\betat\bigl(\delta^i_0\eta^{jk}+\delta^i_0\delta^j_0\delta^k_0\bigr)\nabla_{i}\nabla_j\tau.
\end{equation}
Substituting this into \eqref{Fermi-B} then gives
\begin{equation}\label{Fermi-C}
\Gamma_{0}{}^k{}_J = -\betat\delta_0^k \delta^i_0\delta_J^j \nabla_i\nabla_j \tau.
\end{equation}

Next, with the help of \eqref{Ccdef} and \eqref{orthog}, we can combine \eqref{Gamma-0k0-A} and \eqref{Fermi-C} together to get
\begin{align}
\gamma_{0}{}^k{}_j  &
=
-\betat\delta^i_0\bigl(\delta^0_jg^{lk}+\delta^l_j\delta^k_0\bigr)
(\Dc_{i}\Dc_l\tau-\Cc_i{}^p{}_l\Dc_p\tau)  -\Cc_{0}{}^k{}_j.\label{gamma-0kj}
\end{align}
Moreover, by \eqref{orthog}, we have
$\Dc_i g_{jk}= -\gamma_i{}^l{}_j \eta_{lk}- \gamma_i{}^l{}_k\eta_{jl}$,
or equivalently, after rearranging,
\begin{equation*} 
 \gamma_i{}^k{}_j =-  \eta^{km} \Dc_i g_{jm}- \eta^{km}\eta_{jl}\gamma_i{}^l{}_m.
\end{equation*}
Setting $(i,j,k)=(I,0,0)$ and $(i,j,k)=(I,0,K)$ in the above
expression shows that
\begin{equation}\label{gamma-I00-IK0}
  \gamma_I{}^0{}_0 = \frac{1}{2}\delta^i_I\delta_0^j\delta^k_0\Dc_i g_{jk} \AND
   \gamma_I{}^K{}_0 =-  \delta^{Kl} \delta_I^i\delta^j_0\Dc_i g_{jl}+ \eta^{KL}\gamma_I{}^0{}_L,
\end{equation}
respectively.

In the following, we will view \eqref{gamma-0kj} and \eqref{gamma-I00-IK0} as determining the background
connection coefficients $\gamma_{0}{}^k{}_j$, $ \gamma_I{}^0{}_0$
and  $\gamma_I{}^K{}_0$.  We will determine the remaining background connection coefficients
$\gamma_{I}{}^k{}_J$ using a transport equation. To derive the transport equation, we recall that the background curvature is determined via
\begin{equation*}
\Rc_{ijk}{}^l = e_{j}(\gamma_i{}^l{}_k)-e_{i}(\gamma_j{}^l{}_k)+\gamma_i{}^m{}_k\gamma_j{}^l{}_m-\gamma_j{}^m{}_k\gamma_i{}^l{}_m-(\gamma_j{}^m{}_i-\gamma_i{}^m{}_j)\gamma_m{}^l{}_k \oset{\eqref{curvature}}{=}0.
\end{equation*}
Rearranging this expression yields 
\begin{equation}\label{gamma-IkJ}
 e_{0}(\gamma_I{}^k{}_J)=e_{I}(\gamma_0{}^k{}_J)-\gamma_I{}^l{}_J\gamma_0{}^k{}_l+\gamma_0{}^l{}_J\gamma_I{}^k{}_l+(\gamma_0{}^l{}_I-\gamma_I{}^l{}_0)\gamma_l{}^k{}_J,
\end{equation}
which we view as a transport equation for the connection coefficients $\gamma_{I}{}^k{}_J$. By \eqref{chi-def}, \eqref{Lagrangian} and  \eqref{e0-def}, we observe that the coordinate components of $e_0$ are given by
\begin{equation}\label{e0-mu}
e_0^\mu = \betat^{-1}\chi^\mu =\betat^{-1}\delta^\mu_0,
\end{equation}
which we use to express \eqref{gamma-IkJ} as
\begin{equation}\label{gamma-Ikj-A}
 \del{t}\gamma_I{}^k{}_J=\betat\bigl(e_I(\gamma_0{}^k{}_J)-\gamma_I{}^l{}_J\gamma_0{}^k{}_l+\gamma_0{}^l{}_J\gamma_I{}^k{}_l+(\gamma_0{}^l{}_I-\gamma_I{}^l{}_0)\gamma_l{}^k{}_J\bigr).
\end{equation}

Next, using \eqref{e0-mu} to express the Lie bracket
$[e_0,e_I] = (\gamma_0{}^j{}_I - \gamma_{I}{}^j{}_0)e_j$ as
\begin{equation} \label{[e0,eI]-A}
\betat^{-1}\del{t}e_I^\mu- e_I(\betat^{-1})\delta^\mu_0 = \betat^{-1}(\gamma_0{}^0{}_I - \gamma_{I}{}^0{}_0)\delta^\mu_0+ (\gamma_0{}^J{}_I - \gamma_{I}{}^J{}_0)
e^\mu_J,
\end{equation}
we view this expression as a transport equation for the coordinate components $e^\mu_I$ of the spatial frame field. Noting that 
\begin{equation} \label{e0I-fix}
    e_I^0=e_I(t) = e_I(\tau) = g(e_I,\nabla \tau) = \betat^{-1}g(e_0,e_I) = 0
\end{equation}
by \eqref{tau-fix}, \eqref{e0-def} and \eqref{orthog}, it follows that \eqref{[e0,eI]-A} simplifies to
\begin{equation} \label{[e0,eI]-C}
\del{t}e_I^\Lambda=\betat (\gamma_0{}^J{}_I - \gamma_{I}{}^J{}_0)
e^\Lambda_J.
\end{equation}

Finally, for use below, we note that the derivative of \eqref{alpha-def} is
given by 
\begin{equation}\label{grad-alpha}
    e_i(\betat)=\nabla_i\betat =\betat^3
    \nabla^k\tau \nabla_i\nabla_k \tau =
    -\betat^2 \delta^k_0 (\Dc_i\Dc_k\tau
    -\Cc_i{}^l{}_k \Dc_l\tau)
\end{equation}
where in deriving the third equality we 
used \eqref{Ccdef}, \eqref{frame-def} and \eqref{e0-def}.

\section{First order form}
\label{sec:FirstOrderForm}

We now turn to formulating the frame representation of the reduced conformal Einstein-scalar field equations in first order form. Ultimately, we are seeking to obtain a Fuchsian formulation of the evolution equations 
that will be suitable for establishing the existence of solutions all the way to the big bang singularity at $t=0$. 

\subsection{Primary fields}
The first step in deriving a first order system for the frame representation of the reduced conformal Einstein-scalar system is to derive a set of first order evolution equations for the \emph{primary fields}
\begin{equation}\label{p-fields}
g_{ijk}=\Dc_i g_{jk}, 
\AND \tau_{ij}=\Dc_i\Dc_j \tau .
\end{equation}
After deriving these equations, we will turn to deriving evolution equations for the differentiated fields
\begin{equation}\label{d-fields}
    g_{ijkl}=\Dc_ig_{jkl} \AND \tau_{ijk}=\Dc_i\tau_{jk}.
\end{equation}
 
We begin the derivation of the first order equations for the primary fields by using \eqref{red-Ricci} and \eqref{tau-fix} to express
\eqref{confESFF} as
\begin{equation}\label{confESSFJ}
g^{kj}\Dc_k \Dc_j g_{lm} =-\frac{2}{t}\bigl(
\Dc_l \Dc_m \tau - \Cc_l{}^p{}_m\Dc_p \tau  \bigr)- Q_{lm},
\end{equation}
where by \eqref{Ccdef},  \eqref{frame-def}, \eqref{e0-def} and \eqref{orthog}, we
note that
\begin{equation}\label{grad-tau}
    \Dc_i\tau= -\betat^{-1}g_{ij}e^j_0 = \betat^{-1}\delta_i^0 
\end{equation}
and
\begin{equation}\label{Cc-dt}
\Cc_l{}^p{}_m\Dc_p \tau 
=-\frac{1}{2}\betat^{-1}e_0^j\bigl(\Dc_l g_{jm}
+\Dc_m g_{jl} - \Dc_j g_{lm}\bigr).
\end{equation}
Multiplying \eqref{confESSFJ} by $-e_0^i$ yields
\begin{equation} \label{for-A}
-e^i_0 g^{kj}\Dc_k \Dc_j g_{lm} =\frac{2}{t}e^i_0
\bigl(
\Dc_l \Dc_m \tau - \Cc_l{}^p{}_m\Dc_p \tau  \bigr) + e^i_0 Q_{lm}.
\end{equation} 
But $-g^{ik}e^j_0\Dc_k\Dc_j g_{lm} + g^{ij}e^k_0 \Dc_k \Dc_j g_{lm} 
=0$ by \eqref{commutator},
and so, adding this to \eqref{for-A} yields with the help of \eqref{frame-def}, \eqref{orthog} and \eqref{Cc-dt}, the first order formulation of the frame formulation of the reduced conformal Einstein equations given by
\begin{equation}\label{for-B}
B^{ijk}\Dc_kg_{jlm} =
\frac{1}{t}\betat^{-1}\delta^i_0\delta^j_0(g_{ljm}+g_{mjl}-g_{jlm})+\frac{2}{t}\delta^i_0
 \tau_{lm}+ \delta^i_0 Q_{lm}
\end{equation}
where
\begin{align}
B^{ijk} &= -\delta^i_0 \eta^{jk}-\delta^j_0\eta^{ik}+ \eta^{ij}\delta^k_0.\label{B-def}
\end{align}
Using similar arguments, it is not difficult to verify that the conformal scalar field equation
\eqref{confESFI} can be written in first order
form as
\begin{equation}\label{for-D}
B^{ijk}\Dc_k \tau_{jl} = -\delta^i_0\eta^{kp}\eta^{jq}g_{lpq}\tau_{kj}.
\end{equation}

In the following, we will view \eqref{for-B} and \eqref{for-D} as transport equations for $g_{jlm}$ and  $\tau_{jl}$ by expressing them as
\begin{align}
\del{t}g_{rlm} &=
\frac{1}{t}\delta^0_r\delta^j_0(g_{ljm}+g_{mjl}-g_{jlm})+\frac{2}{t}\delta_r^0\betat
 \tau_{lm}-\delta_{ri}B^{ijK}\betat g_{Kjlm}+ \delta_{ri}\betat Q^i_{lm} \label{for-F.1}
\intertext{and}
 \del{t}\tau_{rl} &=
 \delta_{ri}\betat J^i_l-\delta_{ri}B^{ijK}\betat\tau_{Kjl},
 \label{for-F.2}
\end{align}
respectively, where in deriving these expression we employed \eqref{e0-mu}, \eqref{d-fields} and
\eqref{B-def} and have set
\begin{align}
Q^i_{lm} &= \delta^i_0 Q_{lm}+\delta^{ij}\bigl(  \gamma_0{}^p{}_j  g_{plm} +\gamma_0{}^p{}_l  g_{jpm} +\gamma_0{}^p{}_m  g_{jlp}\bigr)\label{Qilm-def}
\intertext{and}
J^i_l&=-\delta^i_0\eta^{kp}\eta^{jq}g_{lpq}\tau_{kj}+\delta^{ij}\bigl(  \gamma_0{}^p{}_j  \tau_{pl} +\gamma_0{}^p{}_l  \tau_{jp}\bigr).\label{J-def}
\end{align}

\begin{rem}
An important point going forward is that we will not use equation \eqref{for-F.1} with $l=m=0$ in any of our subsequent arguments. This is because the wave
gauge condition \eqref{wave-gauge} can be used to determine $g_{i00}$ in terms of the other metric variables $g_{ilM}$. To see how, we use \eqref{Xdef},
\eqref{orthog} and \eqref{p-fields} to express the wave gauge constraint \eqref{wave-gauge} as
\begin{equation*}
2X_0 = -g_{000}+\delta^{JK}(2g_{JK0}-g_{0JK})=0 \AND 2 X_I = -(2g_{00I}-g_{I00})+\delta^{JK}(2g_{JKI}-g_{IJK})=0.
\end{equation*}
Rearranging and using the symmetry $g_{IJ0}=g_{I0J}$, this becomes
\begin{equation} \label{gi00}
g_{000}= -\delta^{JK}(g_{0JK}-2g_{J0K}) \AND g_{I00}= 2g_{00I}-\delta^{JK}(2g_{JKI}-g_{IJK}),
\end{equation}
which verifies the claim.
\end{rem}

Now, separating \eqref{for-F.1} into the components $(r,l,m)=(0,0,M)$ and $(r,l,m)=(R,0,M)$, $(r,l,m)=(0,L,M)$ and
$(r,l,m)=(R,L,M)$, we see
with the help of \eqref{gi00} that
\begin{align}
\del{t}g_{00M} &= \frac{1}{t}\bigl(2g_{00M}-\delta^{IJ}(2g_{IJM}-g_{MIJ})\bigr)+\frac{2}{t}\betat\tau_{0M} - B^{0jK}\betat g_{Kj0M}
+\betat Q^0_{0M}, \label{for-G.1}\\
\del{t}g_{R0M} &= -\delta_{RI}\betat B^{IjK}g_{Kj0M}+\delta_{RI}\betat Q^I_{0M}, \label{for-G.2}\\
\del{t}g_{0LM} &= \frac{1}{t}(g_{L0M}+g_{M0L}-g_{0LM})+\frac{2}{t}\betat\tau_{LM}- B^{0jK}\betat g_{KjLM}
+\betat Q^0_{LM} \label{for-G.3}
\intertext{and}
\del{t}g_{RLM} &= -\delta_{RI}\betat B^{IjK}g_{KjLM}+ \delta_{RI}\betat Q^I_{LM}. \label{for-G.4}
\end{align}
Rather than using \eqref{for-G.3}, we find it advantageous to use \eqref{for-G.2} to recast this
equation as
\begin{equation}
\del{t}(g_{0LM}-g_{L0M}-g_{M0L}) =
-\frac{1}{t}(g_{0LM}-g_{L0M}-g_{M0L})
+\frac{2}{t}\betat\tau_{LM}+S_{LM} \label{for-H}
\end{equation}
where
\begin{align}
    S_{LM}=&-\betat B^{0jK}g_{KjLM}
+\betat Q^0_{LM} +\delta_{MI}\betat B^{IjK}g_{Kj0L} \notag \\
&-\delta_{MI}\betat Q^I_{0L} 
+\delta_{LI}\betat B^{IjK}g_{Kj0M}-\delta_{LI}\betat Q^I_{0M}.\label{S-def}
\end{align}
It is worth noting that the metric combination $g_{0LM}-g_{L0M}-g_{M0L}$, which appears above in \eqref{for-H},
plays a dominant role in our analysis.
It is related to the second fundamental form $\Ktt_{IJ}$ associated to the $t=const$-hypersurfaces and the conformal metric $g$ via the formula
\begin{equation}\label{Ktt-def}
\Ktt_{IJ} = \nabla_K \ntt_{(J}\delta_{I)}^K
  {=} \frac{1}{2}(g_{0IJ}-g_{I0J}-g_{J0I}) +\gamma_{(I}{}^0{}_{J)},
\end{equation}
where $\ntt_i= -(-|\nabla t|_g^2)^{-\frac{1}{2}} e_i(t)$
and in deriving the second equality we
used \eqref{Ccdef} and the identity
$\nabla_I n_J = g(\nabla_{e_I} e_0, e_J) = -g(e_0,\nabla_{e_I}e_J) =\Gamma_{I}{}^{0}{}_J$, which holds by \eqref{alpha-def}, \eqref{orthog} and \eqref{e0-mu}.

\subsection{Differentiated fields}
We now turn to deriving evolution equations for the differentiated fields \eqref{d-fields}. Applying $\Dc_q$ to \eqref{confESSFJ}, we find with the
help of \eqref{commutator}, \eqref{tau-fix}, \eqref{frame-def}, \eqref{orthog}, \eqref{p-fields}-\eqref{d-fields} and \eqref{grad-tau}-\eqref{Cc-dt} that
 \begin{align}
g^{kj}\Dc_k \Dc_q \Dc_j g_{lm} =&-\frac{1}{t}\betat^{-1} e_0^j (\Dc_q\Dc_l g_{jm}+\Dc_q\Dc_m g_{jl}-
\Dc_q\Dc_j g_{lm}) -\frac{2}{t}\Dc_q\Dc_l\Dc_m\tau -P_{qlm} \label{confESSFL}
\end{align}
where
\begin{align}
P_{qlm}=&  \frac{1}{t}\bigl(\Dc_q(\betat^{-1}) \delta_0^j+\betat^{-1}  \gamma_q{}^j{}_0\bigr)  (g_{ljm}+g_{mjl}-
g_{jlm}) - \eta^{kr}\eta^{js} g_{qrs}g_{kjlm} \notag\\
 &+ \Dc_q Q_{lm}
+\biggl(- \frac{1}{t^2}\betat^{-1} \delta_0^j (g_{ljm}+g_{mjl}-
g_{jlm})-\frac{2}{t^2}
 \tau_{lm}\biggr)\betat^{-1}\delta_q^0 \label{P-def}
\end{align}
and in deriving \eqref{P-def} we have used
\begin{equation} \label{De0}
 \Dc_q e^j_0=\Dc_q \delta^j_0 = \gamma_q{}^j{}_0.
\end{equation}
We also note by  \eqref{grad-alpha}, \eqref{Cc-dt} and
\eqref{gi00} that
\begin{equation}\label{grad-alpha-A}
\Dc_q\betat 
= e_q(\betat)=-\delta_q^0\Bigl(\betat^{2}\tau_{00}-\frac{1}{2}\betat \delta^{JK}(g_{0JK}-2g_{J0K}) \Bigr)
-\delta_q^P\Bigl(\betat^{2}\tau_{P0}+\frac{1}{2}\betat\bigl(2g_{00P}-\delta^{JK}(2g_{JKP}-g_{PJK})\bigr)\Bigr).
\end{equation}

On the other hand, we have by \eqref{commutator} that
\begin{equation*}
-g^{ik}e^j_0\Dc_k\Dc_q\Dc_j g_{lm} + g^{ij}e^k_0 \Dc_k \Dc_q \Dc_j g_{lm} =0,
\end{equation*}
and so adding this to \eqref{confESSFL} yields the
first order equation
\begin{equation}\label{for-I}
B^{ijk}\betat\Dc_k g_{qjlm} = \frac{1}{t} \delta_0^i\delta_0^j  (g_{qljm}+ g_{qmjl}-g_{qjlm})+\frac{2}{t}\delta^i_0
\betat\tau_{qlm}+ \delta^i_0 \betat P_{qlm}
\end{equation}
for the differentiated fields $g_{qjlm}$. 
Similarly, by applying $\Dc_q$ to the scalar field equation \eqref{confESFI}, it is not difficult to verify using the above arguments that the resulting equation can be put into the the first order form
\begin{align} 
B^{ijk}\betat\Dc_k\tau_{qjl} &=  \betat K^i_{ql} \label{for-J}
\end{align}
where
\begin{equation} \label{K-def}
K^i_{ql}= \delta^{i}_0\bigl(
-\eta^{jr}\eta^{ks}g_{qrs}\tau_{kjl}
-\bigl(\eta^{jr}\eta^{ks}g_{qlrs}-2\eta^{jm}\eta^{rn}\eta^{ks}g_{qmn}
g_{lrs}\bigr)\tau_{jk} - \eta^{jr}\eta^{ks}g_{lrs}\tau_{qjk}\bigr).
\end{equation}

\begin{rem}
In the following arguments, we can ignore the $q=0$ component of the equations
\eqref{for-I} and \eqref{for-J}. This is because we can obtain $g_{0jlm}$ and
$\tau_{0jl}$ from the equations \eqref{for-B} and \eqref{for-D}; specifically,
\begin{align}
g_{0rlm} &= \frac{1}{t}\betat^{-1}\delta_r^0\delta_0^j(g_{ljm}+g_{mjl}-g_{jlm})+\frac{2}{t}
\delta_r^0\tau_{lm} - \delta_{ri}B^{ijK}g_{Kjlm}+\delta_r^0 Q_{lm} \label{for-K.1}
\intertext{and}
\tau_{0rl} &=-\delta_{ri}B^{ijK}\tau_{Kjl}-\delta_r^0 \eta^{jp}\eta^{kq} g_{lpq} \tau_{kj}. \label{for-K.2}
\end{align}
\end{rem}

\subsection{Frame and connection coefficient transport equations}
Setting $q=0$ in \eqref{grad-alpha-A}, we observe, with the help of \eqref{e0-mu} that $\betat$ evolves according to
\begin{equation}\label{for-L}
\del{t}\betat 
=-\betat^{3}\tau_{00}+\frac{1}{2}\betat^2 \delta^{JK}(g_{0JK}-2g_{J0K}).
\end{equation}
Also by \eqref{Ccdef}, \eqref{frame-def}, \eqref{orthog},
 \eqref{p-fields}-\eqref{d-fields} and \eqref{Cc-dt}, we note
that the connection coefficients
\eqref{gamma-0kj} can be expressed as
\begin{equation}\label{gamma-0kj-A}
    \gamma{}_0{}^k{}_j
    =-(\delta^0_j\eta^{kl}+\delta^l_j\delta^k_0)\Bigl(\betat\tau_{0l}+\frac{1}{2}g_{l00} \Bigr)-\frac{1}{2}\eta^{kl}(g_{0jl}+g_{j0l}-g_{l0j}).
\end{equation}
Setting $(k,j)=(0,0)$, $(k,j)=(K,0)$, $(k,j)=(0,J)$
and $(k,j)=(K,J)$
in the above expression gives
\begin{align}
     \gamma{}_0{}^0{}_0 =&
     \frac{1}{2}g_{000}\oset{\eqref{gi00}}{=}-\frac{1}{2}g^{PQ}(g_{0PQ}-2g_{P0Q}),\label{gamma-000}\\
     \gamma{}_0{}^K{}_0 =& -\delta^{KL}(\betat \tau_{0L}+g_{00L}), \label{gamma-0K0}\\
     \gamma_0{}^0{}_J=&-\betat \tau_{0J}, \label{gamma-00J}
     \intertext{and}
      \gamma_0{}^K{}_J=& -\frac{1}{2}\delta^{KL}(g_{0JL}+g_{J0L}-g_{L0J}). \label{gamma-0KJ}
\end{align}
We further observe from \eqref{gamma-I00-IK0}, \eqref{p-fields} and \eqref{gi00} that the connection coefficients $\gamma_I{}^0{}_0$ and $\gamma_{I}{}^J{}_0$ are given by
\begin{align}
    \gamma_I{}^0{}_0 &= g_{00I}-\frac{1}{2}\delta^{JK}(2 g_{JKI}-g_{IJK}) \label{gamma-I00}
    \intertext{and}
    \gamma_I{}^J{}_0 &= -\delta^{JK}g_{I0K} + \delta^{JK}\gamma_{I}{}^0{}_K. \label{gamma-IJ0}
\end{align}

Applying $e_I$ to \eqref{gamma-0kj-A}, we find that
\begin{equation}\label{gamma-0kj-B}
    e_I(\gamma{}_0{}^k{}_j)
    =-(\delta^0_j\eta^{kl}+\delta^l_j\delta^k_0)\Bigl( e_I(\betat)\tau_{0l}+ \betat e_I(\tau_{0l})+\frac{1}{2}e_I(g_{l00}) \Bigr)-\frac{1}{2}\eta^{kl}(e_I(g_{0jl})+e_I(g_{j0l})-e_I(g_{l0j})).
\end{equation}
With the help of \eqref{grad-alpha-A}, \eqref{gamma-000}-\eqref{gamma-0kj-B},   
it then follows from \eqref{gamma-Ikj-A} and
\eqref{[e0,eI]-C} that frame components $e^\Lambda_I$ and the connection coefficients
$\gamma_I{}^k{}_J$ evolve according to
\begin{align}
\del{t}e_I^\Lambda=&-\betat \Bigl(\frac{1}{2}\delta^{JK}(g_{0IK}-g_{I0K}-g_{K0I}) + \delta^{JK}\gamma_{I}{}^0{}_K \Bigr)
e^\Lambda_J \label{for-M.1}
\intertext{and}
\del{t}\gamma_I{}^k{}_J=&
-\betat\Bigl(\delta^k_0\Bigl(\betat e_I(\tau_{0J})+\frac{1}{2}e_I(g_{J00})\Bigr)+\frac{1}{2}\eta^{kl}\bigl(e_I(g_{0Jl})+e_I(g_{J0l})-e_I(g_{l0J})\bigr)\Bigr)+L_I{}^k{}_J, \label{for-M.2}
\end{align}
respectively,
where
\begin{align}
    L_I{}^k{}_J =& \betat\Bigl(\delta^k_0\Bigl(\betat^2 \tau_{I0}+\frac{1}{2}\betat \bigl(2g_{00I}-\delta^{KL}(2g_{KLI}-g_{IKL})\bigr)\Bigr)\tau_{0J} -\gamma_I{}^l{}_J\gamma_0{}^k{}_l+\gamma_0{}^l{}_J\gamma_I{}^k{}_l+ (\gamma_0{}^l{}_I-\gamma_I{}^l{}_0)\gamma_l{}^k{}_J\Bigr). \label{L-def}
\end{align}

\section{Nonlinear decompositions}
\label{sec:NonlinDecomp}
In this section, we will examine more closely the nonlinear terms appearing in the first order equations that were derived in the previous section. To assist with this task, we
make the following definitions:
\begin{align}
\kt&=(\kt_{IJ}) :=(g_{0IJ}-g_{I0J}-g_{J0I}),
\label{kt-def} \\
\ellt&=(\ellt_{IjK}) := (g_{IjK}), &&  \label{ellt-def}\\
\mt&=(\mt_{I})  := (g_{00M}), \label{mt-def}\\
\tau &= (\tau_{ij}), \label{tau-def}\\
\gt&=(\gt_{Ijkl}) := (g_{Ijkl}), \label{gt-def}\\    
\taut&=(\taut_{Ijk}) := (\tau_{Ijk}), \label{taut-def}\\
\intertext{and}
\psit&=(\psit_I{}^k{}_J):=(\gamma_I{}^k{}_J).   \label{psit-def}
\end{align}
Even though we are using $\tau$ in \eqref{tau-def} to denote
the collection of derivatives $\tau_{ij}=\Dc_i\Dc_j\tau$, no ambiguities will arise due to the slicing condition \eqref{tau-fix} that allows us to use the coordinate time $t$ to denote the scalar field $\tau$.

\subsection{$*$-notation}
In the following, we will use the following $*$-notation to denote multilinear maps for which it is not necessary
for subsequent arguments to know the exact values of their constant coefficients. For example,
$\kt*\ellt$ can be used to denote tensor fields of the form
\begin{align*}
[\kt*\ellt]_{ij}= C_{ij}^{KLMpQ}\kt_{KL}\ellt_{MpQ}
\end{align*}
where the coefficients $C_{ij}^{KLMpQ}$ are constants. We will also use the notation
\begin{equation*}
    (1+\betat)\kt*\ellt= \kt*\ellt +\betat (\kt*\ellt)
\end{equation*}
where on the right hand side the two $\kt*\ellt$ terms correspond, in general, to distinct bilinear maps, e.g.
\begin{equation*}
    [(1+\betat)\kt*\ellt]_{ij}:=  C_{ij}^{KLMpQ}\kt_{KL}\ellt_{MpQ}+\betat  \tilde{C}_{ij}^{KLMpQ}\kt_{KL}\ellt_{MpQ}.
\end{equation*}
More generally, the $*$ will function as a product that distributes over sums of terms where terms like
$\lambda \kt*\ellt$, $\lambda \in \Rbb\setminus\{0\}$, and $\ellt *\kt$ are identified. For example,
\begin{equation*}
    \ellt*(\kt+\mt)= (\kt+\mt)*\ellt= \ellt*\kt+\ellt*\mt.
\end{equation*}
With this notation, we would have
\begin{equation*}
    (\mt+\psit)*(\ellt+\kt) = \mt*\ellt + \mt*\kt + \psit*\ellt + \psit*\kt 
\AND
    (\kt+\ellt)*(\kt+\ellt)= \kt*\kt + \kt*\ellt + \ellt*\ellt.
\end{equation*}

\subsection{$Q_{ij}$ and $Q^k_{ij}$}

\begin{lem} \label{Qij-lem}
\begin{gather}
Q_{00}=  -\frac{3}{2}\bigl(\delta^{RS}\kt_{RS}\bigr)^2 + \frac{1}{2} \delta^{PQ}\delta^{RS}\kt_{PR} \kt_{QS}+
\Qf, \quad
Q_{0M}= \Qf_M \AND
Q_{LM}= \delta^{RS}\kt_{LR} \kt_{MS}+
\Qf_{LM}\label{Qij-lem.1}
\end{gather}
where  $\Qf = (\kt+\ellt+\mt)*(\ellt+\mt)$.
\end{lem}
\begin{proof}
By \eqref{Q-def}, \eqref{orthog} and \eqref{p-fields}-\eqref{d-fields}, we have
\begin{align}
Q_{ij}=&\frac{1}{2}\bigl(g_{i00}g_{j00}-2g_{0i0}g_{j00}-2g_{i00}g_{0j0}\bigr)
-\delta^{RS}\bigl(g_{i0R}g_{j0S}+g_{Ri0}g_{0jS}+g_{0iS}g_{Rj0}
\notag\\
&-g_{0iR}g_{0jS}-g_{0iR}g_{jS0}-g_{iS0}g_{0jR}-g_{Si0}g_{Rj0}-g_{Si0}g_{j0R}-g_{i0R}g_{Sj0}\bigr)\notag \\
&+ \frac{1}{2}\delta^{RS}\delta^{PQ}\bigl( g_{iPR} g_{jQS}
+2g_{QiS}g_{RjP} - 2g_{SiQ}g_{RjP}
-2g_{SiQ}g_{jPR} -2g_{iPR}g_{SjQ}\bigr). \notag 
\end{align}
Setting $(i,j)=(0,0)$, $(i,j)=(L,M)$ and $(i,j)=(0,M)$ in the
above expression while employing \eqref{gi00} yields

\begin{align}
    Q_{00} &= -\frac{3}{2}\bigl(\delta^{RS}(g_{0RS}-2g_{R0S})\bigr)^2 + \frac{1}{2} \delta^{PQ}\delta^{RS}g_{0PR} g_{0QS}
+\delta^{RS}\Bigl(
2g_{00R}g_{0S0} \notag \\
&\quad+\bigl(2g_{00S}-\delta^{PQ}(2g_{PQS}-g_{SPQ})\bigr)\bigl(2g_{00R}-\delta^{UV}(2g_{UVR}-g_{RUV})\bigr)\Bigr)\notag \\
&\quad+ \frac{1}{2}\delta^{RS}\delta^{PQ}\bigl(
2g_{Q0S}g_{R0P} - 2g_{S0Q}g_{R0P}
-2g_{S0Q}g_{0PR} -2g_{0PR}g_{S0Q}\bigr), \label{Q-00}\\
Q_{LM}&=\frac{1}{2}\Bigl[\bigl(2g_{00L}-\delta^{PQ}(2g_{PQL}-g_{LPQ})\bigr)\bigl(2g_{00M}-\delta^{RS}(2g_{RSM}-g_{MRS})\bigr)\notag \\
&\quad -2g_{0L0}\bigl(2g_{00M}-\delta^{RS}(2g_{RSM}-g_{MRS})\bigr)-2\bigl(2g_{00L}-\delta^{PQ}(2g_{PQL}-g_{LPQ})\bigr)g_{0M0}\Bigr]\notag \\
&\quad -\delta^{RS}\bigl(g_{L0R}g_{M0S}+g_{RL0}g_{0MS}+g_{0LS}g_{RM0}
-g_{0LR}g_{0MS}-g_{0LR}g_{MS0}-g_{LS0}g_{0MR}\notag \\
&\quad -g_{SL0}g_{RM0}-g_{SL0}g_{M0R}-g_{L0R}g_{SM0}\bigr)
+ \frac{1}{2}\delta^{RS}\delta^{PQ}\bigl( g_{LPR} g_{MQS}
+2g_{QLS}g_{RMP} \notag \\
&\quad - 2g_{SLQ}g_{RMP}
-2g_{SLQ}g_{MPR} -2g_{LPR}g_{SMQ}\bigr) \label{Q-LM}
\intertext{and}
Q_{0M}&=-\frac{1}{2} \delta^{UV}(2g_{U0V}-g_{0UV})\Bigl(4g_{00M}-\delta^{PQ}(2g_{PQM}-g_{MPQ})\Bigr)\notag \\
&\quad-\delta^{RS}\Bigl(\bigl(2g_{00S}-\delta^{PQ}(2g_{PQS}-g_{SPQ})\bigr)(g_{0MR}-g_{M0R}-g_{R0M})-2g_{00S}g_{0MR}\Bigr)\notag \\
&\quad + \frac{1}{2}\delta^{RS}\delta^{PQ}\bigl( g_{0PR} g_{MQS}
+2g_{Q0S}g_{RMP} - 2g_{S0Q}g_{RMP}
-2g_{S0Q}g_{MPR} -2g_{0PR}g_{SMQ}\bigr), \label{Q-0M}
\end{align}
respectively. It is then straightforward to verify
from
\eqref{Q-00}-\eqref{Q-0M} and the definitions
\eqref{kt-def}-\eqref{mt-def} that the expansions  \eqref{Qij-lem.1} hold.
\end{proof}

\begin{lem}\label{Qijk-lem}
\begin{gather}
Q^0_{0M}= \Qft_M, \quad Q^I_{0M}=\Qft^I_M, \quad
Q^I_{LM}=  \Qft^I_{LM}
\AND
Q^0_{LM}= -\frac{1}{2}\delta^{RS}\kt_{RS}\kt_{LM}+ \Qft^0_{LM} \label{Qijk-lem.1}
\end{gather}
where $\Qft =\betat(\kt+\ellt+\mt)*\tau +(\kt+\ellt+\mt)*(\ellt+\mt)$.
\end{lem}
\begin{proof}
Define
\begin{equation} \label{Psi-def}
\Psi^{i}_{lm}= \delta^{ij}\bigl(\gamma_0{}^p{}_j g_{plm}+\gamma_0{}^p{}_l g_{jpm} +\gamma_0{}^p{}_m g_{jlp}\bigr).
\end{equation}
Then setting $(i,l,m)=(0,0,M)$,  $(i,l,m)=(I,0,M)$,  $(i,l,m)=(0,L,M)$ and  $(i,l,m)=(I,L,M)$ in
\eqref{Psi-def}, we observe using \eqref{gi00} that
\begin{align*} 
\Psi^{0}_{0M}=& 2\gamma_0{}^0{}_0 g_{00M}+\gamma_0{}^P{}_0( g_{P0M}+ g_{0PM}) -\gamma_0{}^0{}_M\delta^{RS}(g_{0RS}-2g_{R0S})+\gamma_0{}^P{}_M g_{00P}, 
\\
\Psi^{I}_{0M}=& \delta^{IJ}\bigl(\gamma_0{}^0{}_J g_{00M}+\gamma_0{}^P{}_J g_{P0M}+\gamma_0{}^0{}_0 g_{J0M} +\gamma_0{}^P{}_0 g_{JPM} \notag \\
&\qquad +\gamma_0{}^0{}_M (2g_{00J}-\delta^{RS}(2g_{RSJ}-g_{JRS}))+\gamma_0{}^{P}{}_M g_{J0P}\bigr), 
\\
\Psi^{I}_{LM}=& \delta^{IJ}\bigl(\gamma_0{}^0{}_J g_{0LM}+\gamma_0{}^P{}_J g_{PLM}+\gamma_0{}^0{}_L g_{J0M} +\gamma_0{}^P{}_L g_{JPM} +\gamma_0{}^0{}_M g_{JL0}+\gamma_0{}^{P}{}_M g_{JLP}\bigr), 
\intertext{and}
\Psi^{0}_{LM}
=& -\frac{1}{2}\delta^{PQ}g_{0PQ}g_{0LM}-\delta^{PQ}g_{0LQ}g_{0MP}
+\Bigl[\delta^{PQ}g_{P0Q}g_{0LM}-\frac{1}{2}\delta^{PQ}(g_{L0Q}g_{0PM}-g_{Q0L}g_{0PM}) \notag \\
&-\frac{1}{2}\delta^{PQ}(g_{M0Q}g_{0LP}-g_{Q0M}g_{0LP})+
\gamma_0{}^P{}_0 g_{PLM}+\gamma_0{}^0{}_L g_{00M} +\gamma_0{}^0{}_M g_{0L0}\Bigr],
\end{align*}
where in deriving the last equality we used \eqref{gamma-000} and
\eqref{gamma-0KJ}.
With the help of these expansions, we see from \eqref{gamma-000}-\eqref{gamma-0KJ} and the definitions
\eqref{kt-def}-\eqref{mt-def} that
\begin{align}
\Psi^0_{0M}&= \Qft_M,\label{Qijk-lem.5} \\
\Psi^I_{0M}&= \Qft^I_{0M}, \label{Qijk-lem.6}\\
\Psi^0_{LM}&=  -\frac{1}{2}\delta^{PQ}\kt_{PQ}\kt_{LM}-\delta^{PQ}\kt_{LQ}\kt_{MP}+\Qft_{LM} \label{Qijk-lem.7}
\intertext{and}
\Psi^I_{LM}&=  \Qft^I_{LM} \label{Qijk-lem.8}
\end{align}
where $\Qft$ is as defined in the statement of the lemma.
But by \eqref{Qilm-def},
$Q^i_{lm}=\Psi^i_{lm}+\delta^i_0 Q_{lm}$,
and so, we conclude that the desired expansions are a direct consequence of \eqref{Qij-lem.1} and
\eqref{Qijk-lem.5}-\eqref{Qijk-lem.8}.
\end{proof}

\subsection{$e_I(\tau_{jk})$, $e_I(g_{jkl})$ and $L_I{}^k{}_J$}
\begin{lem} \label{lem-cov-exp}
Let $L_I{}^k{}_J$ be as defined above by \eqref{L-def}. Then
\begin{align}
     e_I(\tau_{jk})&=\tau_{Ijk} + \tf_{Ijk}, \label{lem-cov-exp.1}\\
     e_I(g_{jkl})&= g_{Ijkl}+\gf_{Ijkl},
     \label{lem-cov-exp.2}
     \intertext{and}
     L_I{}^k{}_J&=\betat\Lf_I{}^k{}_J\label{lem-cov-exp.3}
\end{align}
where
\begin{equation} \label{gf-def}
\tf= (\psit+\ellt+\mt)*\tau, \quad \gf= (\psit+\ellt+\mt)*(\kt+\ellt +\mt) 
\end{equation}
and
\begin{equation} \label{Lf-def}
     \Lf=  (\betat\tau+\psit+\kt+\ellt +\mt)*(\betat \tau+\psit+\ellt+\mt).
\end{equation}
\end{lem}
\begin{proof}
By the definition of the covariant derivative $\Dc_j$ of the background metric $\gc_{ij}$, we have
\begin{align*}
    e_I(\tau_{jk})=&\delta_I^i\bigl(\Dc_i\tau_{jk}+
    \gamma_i{}^m{}_j\tau_{mk}+\gamma_i{}^m{}_k\tau_{jm}\bigr)
    = \tau_{Ijk}+
    \gamma_I{}^m{}_j\tau_{mk}+\gamma_I{}^m{}_k\tau_{jm}
    \intertext{and}
     e_I(g_{jkl})=&\delta_I^i\bigl(\Dc_ig_{jkl}+
    \gamma_i{}^m{}_jg_{mkl}+\gamma_i{}^m{}_kg_{jml}+\gamma_i{}^m{}_l g_{jkm}\bigr)
    = g_{Ijkl}+
    \gamma_I{}^m{}_jg_{mkl}+\gamma_I{}^m{}_kg_{jml}+\gamma_I{}^m{}_l g_{jkm}. 
\end{align*}
From these expressions, it is then clear that the expansion \eqref{lem-cov-exp.1}-\eqref{lem-cov-exp.2} are a direct 
consequence of 
\eqref{gi00}, \eqref{gamma-000}-\eqref{gamma-IJ0} and
\eqref{kt-def}-\eqref{psit-def}. 

Next, we observe that
\begin{align*}
    -\gamma_I{}^l{}_J\gamma_0^k{}_l + \gamma_0{}^l{}_J\gamma_I{}^k{}_l+ (\gamma_0{}^l{}_I-\gamma_I{}^l{}_0)\gamma_l{}^k{}_J
    &=  -\gamma_I{}^0{}_J\gamma_0{}^k{}_0-\gamma_I{}^L{}_J\gamma_0{}^k{}_L + \gamma_0{}^0{}_J\gamma_I{}^k{}_0+ \gamma_0{}^L{}_J\gamma_I{}^k{}_L \\
    & \qquad +(\gamma_0{}^0{}_I-\gamma_I{}^0{}_0)\gamma_0{}^k{}_J+(\gamma_0{}^L{}_I-\gamma_I{}^L{}_0)\gamma_L{}^k{}_J. 
\end{align*}
With the help of \eqref{kt-def}-\eqref{mt-def}, \eqref{psit-def}, and \eqref{gamma-000}-\eqref{gamma-IJ0}, we can write the above expression
as 
\begin{equation*}
     -\betat \gamma_I{}^l{}_J\gamma_0{}^k{}_l + \betat \gamma_0{}^l{}_J\gamma_I{}^k{}_l + \betat(\gamma_0{}^l{}_I-\gamma_I{}^l{}_0)\gamma_l{}^k{}_J =\betat \Lf_I{}^k{}_J,
\end{equation*}
where $\Lf_I{}^k{}_j$ is as defined above by
\eqref{Lf-def}. Using this, it is not difficult to verify
that \eqref{lem-cov-exp.3} follows from \eqref{L-def} and \eqref{kt-def}-\eqref{mt-def}.
\end{proof}

Now, by \eqref{lem-cov-exp.1}-\eqref{lem-cov-exp.2}, we can express the evolution equation
\eqref{for-M.2} for the background connection coefficients $\gamma_I{}^k{}_J$
as 
\begin{align}
    \del{t}\gamma_I{}^k{}_J=
-\delta^k_0\Bigl(\betat^2 \tau_{I0J}+\frac{1}{2}\betat g_{IJ00}\Bigr)-\frac{1}{2}\eta^{kl}\bigl(\betat g_{I0Jl}+\betat g_{IJ0l}-\betat g_{Il0J}\bigr)+\betat \Lf_I{}^k{}_J\label{for-N},
\end{align}
where $\Lf$ is of the form \eqref{Lf-def}.
Furthermore, using \eqref{kt-def} and \eqref{psit-def}, it clear that the transport equations
\eqref{for-L} and \eqref{for-M.1} can be expressed as
\begin{align}
\del{t}\betat &= -\betat^3 \tau_{00}+\frac{1}{2} \delta^{JK}\betat^2\kt_{JK} \label{for-O.1}
\intertext{and}
\del{t}e^\Lambda_I &= - \Bigl(\frac{1}{2}\delta^{JL}\betat\kt_{IL}+\delta^{JK}\betat\psit_{I}{}^0{}_K\Bigr)e^\Lambda_J, \label{for-O.2}
\end{align}
respectively. For use below, we also note that from \eqref{grad-alpha-A} and
\eqref{ellt-def}-\eqref{mt-def} that
\begin{equation} \label{grad-alpha-B}
    e_I(\betat)= -\betat^2 \tau_{I0}-\frac{1}{2}\betat\bigl( 2 \mt_I -\delta^{JK}(2\ellt_{JKI}-\ellt_{IJK})\bigr).
\end{equation}

\subsection{$P_{Qlm}$, $\delta^q_Q \betat \Dc_k g_{qjlm}$ and $\delta^q_Q \betat \Dc_k \tau_{qjl}$} 

\begin{lem}\label{lem-P-exp}
Let $ P_{Qlm}$ be as defined above by \eqref{P-def}. Then
\begin{align}
   \betat P_{Qlm} &= \Pf_{Qlm},\label{lem-P-exp.1} \\
   \delta^q_Q \betat \Dc_k g_{qjlm} &=  \delta_k^0\del{t}g_{Qjlm} +\delta_k^K\betat e_K(g_{Qjlm})+ \Gf_{Qkjlm}, \label{lem-P-exp.2}
   \intertext{and}
   \delta^q_Q \betat \Dc_k \tau_{qjl}&=
   \delta_k^0\del{t}\tau_{Qjl} +\delta_k^K\betat e_K(\tau_{Qjl})+ \Tf_{Qkjl} \label{lem-P-exp.3}
\end{align}
where
\begin{align}
\Pf &= \frac{1}{t}(\betat\tau + \psit + \ellt +\mt)*(\kt+\ellt+\mt) + \betat(\kt+\ellt+\mt) *\bigl( \gt+
(\ellt+\mt)*(\kt+\ellt+\mt)\bigl), \label{Pf-def}\\
    \Gf&=(\betat \tau +\psit) *\Bigl( \frac{1}{t}\bigl(\betat \tau+\kt+\ellt+\mt)  
     +\betat(\kt+\ellt+\mt)*(\kt+\ellt+\mt)\Bigr) \notag\\
     &\qquad
     +\betat \bigl(\betat \tau+\psit+\kt+\ellt+\mt)*\gt  
    \label{Gf-def}
    \intertext{and}
    \Tf&=\betat (\betat \tau+\psit+ \kt+ \ellt +\mt)*\taut
    +\betat (\betat \tau+\psit +\ellt +\mt)*(\kt + \ellt +\mt)*\tau. \label{Tf-def}
\end{align}
\end{lem}
\begin{proof}
Differentiating \eqref{Q-def}, we find, by \eqref{orthog} and \eqref{p-fields}-\eqref{d-fields}, that
\begin{equation} \label{Dq-Qij}
    \Dc_q Q_{ij} =\Qtt^1_{qij}+ \Qtt^2_{qij}
\end{equation}
where 
\begin{align*}
\Qtt^{1}_{qij} &= \frac{1}{2}\eta^{kl}\eta^{mn}\Bigl(g_{qimk} g_{jn l}+g_{imk}g_{qjnl}
+2 g_{qnil}g_{kjm} +2g_{nil}g_{qkjm}\\ &\qquad- 2g_{qlin} g_{kjm}- 2g_{lin}g_{qkjm}
-2g_{qlin}g_{jmk}-2g_{lin}g_{qjmk} -2 g_{qimk}g_{ljn}-2g_{imk}g_{qljn}\Bigr)
\intertext{and}
    \Qtt^2_{qij} &=  -\frac{1}{2}(\eta^{kr}\eta^{ls}g^{mn}+ \eta^{kl} \eta^{mr}\eta^{ns}\bigr)g_{qrs}\Bigl(g_{imk}g_{jn l}
+2 g_{nil}g_{kjm} - 2g_{lin}g_{kjm}
-2g_{lin}g_{jmk} -2g_{imk}g_{ljn}\Bigr).
\end{align*}
Setting $q=Q$ in \eqref{Dq-Qij} gives
\begin{equation} \label{Dq-exp}
    \delta^q_Q\Dc_q Q_{ij} =\bigl[(\kt+\ellt +\mt)*\bigl( \gt+(\ellt+\mt)*(\kt+\ellt +\mt) \bigr)\bigr]_{Qij},
\end{equation}
where in deriving this we have employed  \eqref{gi00} and \eqref{kt-def}-\eqref{gt-def}.

To proceed, we let
\begin{equation*}
    \Ptt_{Qlm} =
 \frac{1}{t}\bigl(e_Q(\betat^{-1}) \delta_0^j+\betat^{-1}  \gamma_Q{}^j{}_0\bigr)  (g_{ljm}+g_{mjl}-
g_{jlm}) -\eta^{kr}\eta^{js}g_{Qrs}g_{kjlm},
\end{equation*}
which we note, using \eqref{gi00} and \eqref{ellt-def}-\eqref{mt-def}, can be expressed as
\begin{align}
    \Ptt_{Qlm} &=
 \frac{1}{t}\bigl(e_Q(\betat^{-1})+\betat^{-1}  \gamma_Q{}^0{}_0\bigr)  (g_{l0m}+g_{m0l}-
g_{0lm})
 +\frac{1}{t}\betat^{-1}  \gamma_Q{}^J{}_0 (g_{lJm}+g_{mJl}-
g_{Jlm})\notag \\
&\qquad -\bigl(2\mt_Q-\delta^{RS}(2\ellt_{RSQ}-\ellt_{QRS})\bigr)g_{00lm}+\eta^{JS} \ellt_{Q0S}g_{0Jlm}
- \delta^{KR}\eta^{js} \ellt_{QsR}g_{Kjlm}.  \label{Ptt-def}
\end{align}
Now, by \eqref{gi00}, \eqref{gamma-I00}-\eqref{gamma-IJ0},  \eqref{kt-def}-\eqref{mt-def}, \eqref{psit-def} and  
\eqref{grad-alpha-B}, we have
\begin{align}
\frac{1}{t}\bigl(e_Q(\betat^{-1})+\betat^{-1}  \gamma_Q{}^0{}_0\bigr)  (g_{l0m}&+g_{m0l}-
g_{0lm})
+\frac{1}{t}\betat^{-1}  \gamma_Q{}^J{}_0 (g_{lJm}+g_{mJl}-
g_{Jlm})\notag \\
&= \frac{1}{t}\betat^{-1}[(\betat \tau+\psit+\ellt+\mt)*(\kt+\ellt+\mt)]_{Qlm} \label{Ptt-exp-A}
\end{align}
while it is clear from \eqref{for-K.1} and \eqref{Qij-lem.1} that
\begin{equation} \label{g0rlm-exp}
    g_{0rlm} = \frac{1}{t}\betat^{-1}\delta_r^0(g_{l0m}+g_{m0l}-g_{0lm})+\frac{2}{t}
\delta_r^0\tau_{lm} - \delta_{ri}B^{ijK}g_{Kjlm}+[(\kt+\ellt+\mt)*(\kt+\ellt+\mt)]_{rlm}.
\end{equation}
By \eqref{Dq-exp}-\eqref{g0rlm-exp}, we then, with the help of the definitions \eqref{kt-def}-\eqref{mt-def}, deduce from \eqref{P-def} the validity of the expansion \eqref{lem-P-exp.1}.

Next, by the definition of the covariant derivative $\Dc_q$, we have
\begin{align*}
    \delta^q_Q \betat \Dc_k g_{qjlm} &= \betat e_k(g_{Qjlm})-\betat \gamma_k{}^p{}_{Q} g_{pjlm}
    -\betat \gamma_k{}^p{}_{j} g_{Qplm}
     -\betat \gamma_k{}^p{}_{l} g_{Qjpm}
      -\betat \gamma_k{}^p{}_{m} g_{Qjlp}\\
     &= \delta_k^0\del{t}g_{Qjlm} +\delta_k^K\betat e_K(g_{Qjlm}) -\betat \gamma_k{}^0{}_{Q} g_{0jlm}
    -\Bigl(\betat \gamma_k{}^P{}_{Q} g_{Pjlm}\\
    &\hspace{4.0cm} +\betat \gamma_k{}^p{}_{j} g_{Qplm}
     +\betat \gamma_k{}^p{}_{l} g_{Qjpm}
      +\betat \gamma_k{}^p{}_{m} g_{Qjlp}\Bigr)
\end{align*}
where in deriving the second equality we used \eqref{e0-mu}. The expansion \eqref{lem-P-exp.2} is then a straightforward consequence of \eqref{gamma-000}-\eqref{gamma-0KJ}, \eqref{g0rlm-exp} and the definitions \eqref{kt-def}-\eqref{psit-def}. Furthermore, by employing similar arguments, it is not difficult with the help of \eqref{for-K.2} to verify that the expansion \eqref{lem-P-exp.3} also holds. 
\end{proof}

\begin{lem}\label{lem-JK-exp}
Let $J^j_l$ adn $K^i_{Ql}$ be as defined above by
\eqref{J-def} and \eqref{K-def}, respectively. Then
\begin{align*}
    \betat J^j_l=   \bigl[(\betat^2 \tau + \betat\kt +\betat\ellt +\betat \mt)*\tau\bigr]^j_l
 \AND
    \betat K^i_{Ql} = \Kf^i_{Ql}
\end{align*}
where
\begin{equation*}
    \Kf = (\betat\kt+\betat\ellt+\betat\mt)*\taut + \betat\gt*\tau + (\ellt+\mt)*(\betat\kt+\betat\ellt+\betat\mt)*\tau.
\end{equation*}
\end{lem}
\begin{proof}
Using the definition \eqref{kt-def}-\eqref{tau-def}, it is not difficult to see from \eqref{J-def}, \eqref{gi00} and \eqref{gamma-000}-\eqref{gamma-0KJ} that $\betat J^j_l$ can be expressed as 
$\betat J^j_l = (\betat^2 \tau + \betat \kt +\betat \ellt +\mt)*\tau$,
which establishes the validity of the first expression. The validity of the second expression  $\betat K^i_{Ql} = \Kf^i_{Ql}$ is also straightforward to verify from \eqref{gi00}, \eqref{K-def}, \eqref{for-K.2} and the definitions
\eqref{kt-def}-\eqref{taut-def}.
\end{proof}

\section{Fuchsian form}
\label{sec:FuchsianForm}
We are now prepared to complete the transformation of the frame representation of the reduced conformal Einstein-scalar field equations into a Fuchsian form that will be suitable for establishing the existence of solutions all the way to the singularity. In particular, we
will use the Fuchsian formulation to derive energy estimates from which uniform bounds on the fields near the big bang singularity can be extracted.

\subsection{The complete first order system}
Thus far, the first order formulation
of the frame representation of the reduced conformal Einstein-scalar field equations that
was derived above in Section~\ref{sec:FirstOrderForm} consists of the following equations:
\begin{align}
  \del{t}g_{00M} &= \frac{1}{t}\bigl(2g_{00M}-\delta^{IJ}(2g_{IJM}-g_{MIJ})\bigr)+\frac{2}{t}\betat\tau_{0M} - B^{0jK}\betat g_{Kj0M}
+\betat Q^0_{0M}, \label{for-G.1.S}\\
\del{t}g_{R0M} &= -\delta_{RI} B^{IjK}\betat g_{Kj0M}+\delta_{RI}\betat Q^I_{0M}, \label{for-G.2.S}\\
\del{t}(g_{0LM}-g_{L0M}-g_{M0L}) &=  -\frac{1}{t}(g_{0LM}-g_{L0M}-g_{M0L}) +\frac{2}{t}\betat\tau_{(LM)}+S_{(LM)}, \label{for-H.S}\\
\del{t}g_{RLM} &= -\delta_{RI} B^{IjK}\betat g_{KjLM}+
                 \delta_{RI}\betat Q^I_{LM}, \label{for-G.4.S}\\
\del{t}\tau_{rl} &=
 \delta_{ri}\betat J^i_l-B^{ijK}\betat\tau_{Kj(l}\delta_{r)i},
 \label{for-F.2.S}\\
\label{for-I.S}
B^{ijk}\betat\Dc_k g_{qjlm} &= \frac{1}{t} \delta_0^i\delta_0^j  (g_{qljm}+ g_{qmjl}-g_{qjlm})+\frac{2}{t}\delta^i_0
\betat\tau_{q(lm)}+ \delta^i_0 \betat P_{q(lm)},\\
B^{ijk}\betat\Dc_k\tau_{qjl} &=  \betat K^i_{ql}, \label{for-J.S}\\
\label{for-L.S}
\del{t}\betat 
&=-\betat^{3}\tau_{00}+\frac{1}{2}\betat^2 \delta^{JK}(g_{0JK}-2g_{J0K}),\\
\del{t}e_I^\Lambda=&-\betat \Bigl(\frac{1}{2}\delta^{JK}(g_{0IK}-g_{I0K}-g_{K0I}) + \delta^{JK}\gamma_{I}{}^0{}_K \Bigr)
e^\Lambda_J, \label{for-M.1.S}\\
\del{t}\gamma_I{}^k{}_J=&
-\betat\Bigl(\delta^k_0\Bigl(\betat e_I(\tau_{0J})+\frac{1}{2}e_I(g_{J00})\Bigr)+\frac{1}{2}\eta^{kl}\bigl(e_I(g_{0Jl})+e_I(g_{J0l})-e_I(g_{l0J})\bigr)\Bigr)+L_I{}^k{}_J, \label{for-M.2.S}
\end{align}
where
\begin{align}
B^{ijk} &= -\delta^i_0 \eta^{jk}-\delta^j_0\eta^{ik}+ \eta^{ij}\delta^k_0,\label{B-def.S}
\end{align}
$Q^i_{lm}$ is determined by \eqref{Q-def} and \eqref{Qilm-def}, $S_{LM}$ and $J^i_l$ are given
 by \eqref{S-def} and \eqref{J-def}, respectively, $g_{000}$ and $g_{I00}$ are determined by \eqref{gi00}, $\tau_i=-\betat^{-1}g_{0i}$ according to \eqref{e0-def}, $P_{qlm}$ is determined by \eqref{P-def} and \eqref{grad-alpha-A}, $K^i_{ql}$ is given by \eqref{K-def}, $g_{0rlm}$ and $\tau_{0rl}$ are determined by \eqref{for-K.1}-\eqref{for-K.2}, $L_I{}^k{}_J$ is given by \eqref{L-def}, and  $\gamma{}_0{}^0{}_0$, $\gamma{}_0{}^K{}_0$, $\gamma_0{}^0{}_J$, $\gamma_0{}^K{}_J$, $\gamma_I{}^0{}_0$ and $\gamma_I{}^J{}_0$  are determined by \eqref{gamma-000}-\eqref{gamma-IJ0}. 
 \begin{rem}
\label{rem:symmetry}
 In the following, we will not always assume that solutions of the system \eqref{for-G.1.S}-\eqref{for-M.2.S} are obtained from solutions of the conformal Einstein-scalar field equations. We will also consider general solutions to this system of equations. For these solutions, we will require in our subsequent that $g_{0LM}-g_{L0M}-g_{M0L}$
 and $g_{qjlm}$ are symmetric in $L,M$ and $l,m$, respectively, which, of course would be true for a solution of \eqref{for-G.1.S}-\eqref{for-M.2.S} that was derived from a solution to the conformal Einstein-scalar field equations. Due to this requirement, we have explicitly symmetrized the equations \eqref{for-H.S} and \eqref{for-I.S} in the indices $L,M$ and $l,m$, respectively, in order to ensure that the solutions of  \eqref{for-G.1.S}-\eqref{for-M.2.S} have the desired symmetry provided the initial data does also. We will also do the same for all equations derived from these ones, c.f.~\eqref{for-H.S2} and \eqref{for-I.S2}.
 \end{rem}

\subsection{Initial data\label{frame-idata}}
Before proceeding, we first describe how the initial data $\{\gr_{\mu\nu},\ggr_{\mu\nu},\taur=t_0,\taugr\}$ for the reduced conformal Einstein-scalar field equations along with a choice of spatial frame initial data
\begin{equation} \label{frame-idata-A}
    e_I^\mu\bigl|_{\Sigma_{t_0}} = \er^\mu_I
\end{equation}
completely determines initial data for the first order system \eqref{for-G.1.S}-\eqref{for-M.2.S}. It is important to note that we do not assume here that the initial
data $\{\gr_{\mu\nu},\ggr_{\mu\nu},\taur=t_0,\taugr\}$
satisfies either of the gravitational or wave gauge constraints. 

We recall from Remark \ref{rem-Lag-idata} that the initial data set $\{\gr_{\mu\nu},\ggr_{\mu\nu},\taur=t_0,\taugr\}$ determines a corresponding Lagrangian initial data set
\begin{equation} \label{Lag-idata-set-A}
\bigl\{g_{\mu\nu}\bigl|_{\Sigma_{t_0}},\del{0}g_{\mu\nu}\bigl|_{\Sigma_{t_0}}, \tau\bigl|_{\Sigma_{t_0}}=t_0,\del{0}\tau\bigl|_{\Sigma_{t_0}}=1\bigr\}
\end{equation}
via \eqref{vr-def}-\eqref{Jcr-inv} and \eqref{dt-tau-idata}-\eqref{dt-g-idata}. This initial data, with the help of the reduced conformal Einstein-scalar field equations \eqref{lag-redeqns} and first derivatives thereof evaluated on $\Sigma_{t_0}$, then uniquely determines the following higher order partial derivatives on the initial hypersurface:
\begin{equation} \label{Lag-idata-set-B}
\bigl\{\del{\beta}g_{\mu\nu}\bigl|_{\Sigma_{t_0}},\del{\alpha}\del{\beta}g_{\mu\nu}\bigl|_{\Sigma_{t_0}},\del{\alpha}\tau\bigl|_{\Sigma_{t_0}}=\delta^0_\alpha,\del{\alpha}\del{\beta}\tau\bigl|_{\Sigma_{t_0}}, \del{\alpha}\del{\beta}\del{\gamma}\tau\bigl|_{\Sigma_{t_0}}\bigr\}.
\end{equation}
We further observe that the flat metric $\gc_{\mu\nu}=\del{\mu}l^\alpha \eta_{\alpha\beta}\del{\nu} l^\beta$ and its first and second order partial derivatives on the initial hypersurface, that is,
\begin{equation} \label{Lag-idata-set-C}
\bigl\{\gc_{\mu\nu}\bigl|_{\Sigma_{t_0}},\del{\alpha}\gc_{\mu\nu}\bigl|_{\Sigma_{t_0}},\del{\alpha}\del{\beta}\gc_{\mu\nu}\bigl|_{\Sigma_{t_0}} \bigr\},
\end{equation}
are uniquely determined from the initial data set $\{\gr_{\mu\nu},\ggr_{\mu\nu},\taur=t_0,\taugr\}$ by
\eqref{gct-def}, \eqref{Jc-def}, \eqref{vr-def}-\eqref{Jcr-def} and the
reduced conformal Einstein-scalar field equations, especially \eqref{J-ev.1}, and first derivatives thereof.
Taken together, \eqref{Lag-idata-set-A}-\eqref{Lag-idata-set-C} determine the following geometric fields on $\Sigma_{t_0}$:
\begin{equation} \label{Lag-idata-set-D}
\bigl\{g_{\mu\nu}\bigl|_{\Sigma_{t_0}},\Dc_{\beta}g_{\mu\nu}\bigl|_{\Sigma_{t_0}},\Dc_{\alpha}\Dc_{\beta}g_{\mu\nu}\bigl|_{\Sigma_{t_0}},\tau\bigl|_{\Sigma_{t_0}}=t_0,\Dc_{\alpha}\tau\bigl|_{\Sigma_{t_0}}=\delta^0_\alpha,\Dc_{\alpha}\Dc_{\beta}\tau\bigl|_{\Sigma_{t_0}},\Dc_{\alpha}\Dc_{\beta}\Dc_{\gamma}\tau\bigl|_{\Sigma_{t_0}} \bigr\}.
\end{equation}

Next, by \eqref{chi-idata} and \eqref{e0-def}, the coordinate components of the frame vector $e_0$ on the initial hypersurface are given by 
\begin{equation} \label{frame-idata-E}
e_0^\mu \bigl|_{\Sigma_{t_0}} = (-g(\del{0},\del{0})|_{\Sigma_{t_0}} )^{-\frac{1}{2}}\delta^\mu_0,   
\end{equation}
while by \eqref{e0-mu} we have 
\begin{equation} \label{frame-idata-F}
\betat\bigl|_{\Sigma_{t_0}} = (- g(dx^0,dx^0)\bigl|_{\Sigma_{t_0}})^{-\frac{1}{2}}.
\end{equation}
Using \eqref{frame-idata-E}, we can employ a Gram-Schmidt algorithm to select (non-unique) spatial frame initial data \eqref{frame-idata-A} so that that the frame metric
$g_{ij} = e_i^\mu g_{\mu\nu} e^\nu_j$
satisfies
\begin{equation*}
    g_{ij}\bigl|_{\Sigma_{t_0}}=\eta_{ij}
\end{equation*}
on $\Sigma_{t_0}$.
Furthermore, with the help of the relations
\begin{gather*}
\Dc_i g_{jk} = e_i^\beta e_j^\mu e_k^\nu \Dc_{\beta}g_{\mu\nu}, \quad \Dc_i\Dc_j g_{kl} = e_i^\alpha e_j^\beta e_k^\mu e_l^\nu \Dc_{\alpha}\Dc_{\beta}g_{\mu\nu}, \\
\Dc_i \Dc_{j}\tau = e_i^\beta e_j^\beta \Dc_{\beta}\Dc_{\beta}g_{\mu\nu}, \quad \Dc_i\Dc_j \Dc_k \tau = e_i^\alpha e_j^\beta e_k^\gamma \Dc_{\alpha}\Dc_{\beta}\Dc_{\gamma}\tau,
\end{gather*}
it follows from the definitions  \eqref{p-fields}-\eqref{d-fields} that \eqref{frame-idata-A} and \eqref{Lag-idata-set-D} determine the fields
$g_{ijk}$, $g_{ijkl}$, $\tau_{ij}$ and $\tau_{ijk}$ on
the initial hypersurface, that is,
\begin{equation}\label{frame-idata-H}
    \bigl\{g_{ijk}\bigl|_{\Sigma_{t_0}},g_{ijkl}\bigl|_{\Sigma_{t_0}},\tau_{ij}\bigl|_{\Sigma_{t_0}},\tau_{ijk}\bigl|_{\Sigma_{t_0}}\bigr\}.
\end{equation}

Using the fact that the frame $e_j^\mu$ is orthonormal with respect to the metric $g_{\mu\nu}$, it follows from a straightforward calculation that 
$\Gamma_{I}{}^0{}_J|_{\Sigma|_{t_0}} = \Ktt_{\Lambda\Omega}\er_I^\Lambda \er_J^\Omega$,
where $\Ktt=\Ktt_{\Lambda\Omega}dx^\Lambda\otimes dx^\Omega$ is the second fundamental form, c.f.\ \eqref{gtt-def1}, determined from the initial data $\{g_{\mu\nu}|_{\Sigma_{t_0}},\del{0}g_{\mu\nu}|_{\Sigma_{t_0}}\}$. From this expression and \eqref{Ccdef}, we deduce that
\begin{equation} \label{frame-idata-J}
    \gamma_I{}^0{}_J\bigl|_{\Sigma_{t_0}} = \Ktt_{\Lambda\Omega}\er_I^\Lambda \er_J^\Omega +\frac{1}{2}(g_{IJ0}+g_{JI0}-g_{0IJ})\bigl|_{\Sigma_{t_0}}.
\end{equation}
We also note that the connection coefficients 
$\Gamma_I{}^K{}_J\bigl|_{\Sigma_{t_0}}$ on the initial hypersurface are determined completely by the spatial orthonormal frame $\ett_I = \er^\Lambda_I \del{\Lambda}$
and the spatial metric
$\gtt=\gr_{\Lambda\Omega}dx^\Lambda \otimes dx^\Omega$, c.f.\ \eqref{gtt-def1}, on
$\Sigma_{t_0}$ according to
$\Gamma_I{}^K{}_J|_{\Sigma_{t_0}} = \delta^{KL}\gtt(\Dtt_{\ett_I}\ett_J,\ett_L)$,
where $\Dtt$ is the Levi-Civita connection of $\gtt$.
We can then use this to determine 
$\gamma_I{}^K{}_J\bigl|_{\Sigma_{t_0}}$ via
\eqref{Ccdef} to get
\begin{equation}  \label{frame-idata-K}
    \gamma_I{}^K{}_J \bigl|_{\Sigma_{t_0}} = \delta^{KL}\gtt(\Dtt_{\ett_I}\ett_J,\ett_L)
    - \frac{1}{2}\delta^{KL}(g_{IJL}+g_{JIL}-g_{LIJ})\bigl|_{\Sigma_{t_0}}.
\end{equation}
Together, \eqref{frame-idata-A}, \eqref{frame-idata-F},  \eqref{frame-idata-H},  \eqref{frame-idata-J} and \eqref{frame-idata-K} determine initial data on $\Sigma_{t_0}$
for the system \eqref{for-G.1.S}-\eqref{for-M.2.S}.

\begin{rem} \label{FLRW-idata-rem-B}
By a straightforward calculation, it can be verified that
the initial data for the system  \eqref{for-G.1.S}-\eqref{for-M.2.S} that is determined by 
the FLRW solution \eqref{gtau-FLRW} of the conformal Einstein-scalar field equations is given by
\begin{equation}\label{FLRW-idata-rem-B.1}  
    \bigl\{g_{ijk}\bigl|_{\Sigma_{t_0}},\tau_{ij}\bigl|_{\Sigma_{t_0}}, \betat\bigl|_{\Sigma_{t_0}}, e_I^\Lambda\bigl|_{\Sigma_{t_0}}, \gamma_I{}^k{}_J\bigl|_{\Sigma_{t_0}},\tau_{ijk}\bigl|_{\Sigma_{t_0}},g_{ijkl}\bigl|_{\Sigma_{t_0}}\bigr\}_{\textrm{FLRW}} =\{0,0,1,\delta^\Lambda_I,0,0,0\}. 
\end{equation}
By \eqref{Kasner-wave-gauge}, we note that the wave gauge constraint is also satisfied on $\Sigma_{t_0}$ for this choice of initial data.
\end{rem}

\subsection{Change of variables}
\label{sec:chvar1}
Using the variable definitions \eqref{kt-def}-\eqref{psit-def} along with \eqref{S-def}, we can express the first order system \eqref{for-G.1.S}-\eqref{for-M.2.S} as
\begin{align}
  \del{t}\mt_{M} =& \frac{1}{t}\bigl(2 \mt_{M}-\delta^{IJ}(2 \ellt_{IJM}-\ellt_{MIJ})\bigr)+\frac{2}{t}\betat\tau_{0M} - B^{0jK}\betat \gt_{Kj0M}
+\betat \Qft_M, \label{for-G.1.S2}\\
\del{t}\ellt_{R0M} =& -\delta_{RI} B^{IjK}\betat\gt_{Kj0M}+\delta_{RI}\betat \Qft^I_M, \label{for-G.2.S2}\\
\del{t}\kt_{LM} =&  -\frac{1}{t}\kt_{LM} +\frac{2}{t}\betat\tau_{(LM)}- B^{0jK}\betat\gt_{Kj(LM)}
-\frac{1}{2}\delta^{RS}\betat\kt_{RS}\kt_{LM}+ \betat\Qft^0_{(LM)} \notag \\
&
+ 2B^{IjK}\betat g_{Kj0(M}\delta_{L)I} -2\betat \Qft^I_{(L} \delta_{M)I},
\label{for-H.S2}\\
\del{t}\ellt_{RLM} =& -\delta_{RI} B^{IjK}\betat\gt_{KjLM}+
                 \delta_{RI}\betat \Qft^I_{LM}, \label{for-G.4.S2}\\
\del{t}\tau_{rl} =&
 \delta_{ri}\betat J^i_l-B^{ijK}\betat\taut_{Kj(l}\delta_{r)i},
 \label{for-F.2.S2}\\
\label{for-I.S2}
\delta^{ij}\del{t}\gt_{Qjlm} +B^{ijK}\betat e_K^\Lambda\partial_\Lambda\gt_{Qjlm}=& \frac{1}{t} \delta_0^i\delta_0^j  (\gt_{Qljm}+ \gt_{Qmjl}-\gt_{Qjlm})+\frac{2}{t}\delta^i_0
\betat\taut_{Q(lm)}+ \delta^i_0 \Pf_{Q(lm)}-B^{ijk}\Gf_{Qkj(lm)},\\
\delta^{ij}\del{t}\taut_{Qjl} +B^{ijK}\betat e_K^\Lambda\partial_\Lambda \taut_{Qjl}=&  \betat K^i_{Ql}-B^{ijk}\Tf_{Qkjl}, \label{for-J.S2}\\
\del{t}\betat =& -\betat^3 \tau_{00}+\frac{1}{2} \delta^{JK}\betat^2\kt_{JK} \label{for-O.1.S2}\\
\del{t}e^\Lambda_I =& - \Bigl(\frac{1}{2}\delta^{JL}\betat\kt_{IL}+\delta^{JK}\betat\psit_{I}{}^0{}_K\Bigr)e^\Lambda_J,\label{for-M.1.S2}\\
\del{t}\psit_I{}^k{}_J=&
-\delta^k_0\Bigl(\betat^2 \taut_{I0J}+\frac{1}{2}\betat \gt_{IJ00}\Bigr)-\frac{1}{2}\eta^{kl}\bigl(\betat \gt_{I0Jl}+\betat \gt_{IJ0l}-\betat \gt_{Il0J}\bigr)+\betat \Lf_I{}^k{}_J\label{for-N.S2}.
\end{align}
Here, $\Lf_I{}^k{}_J$ of the form \eqref{Lf-def}, $\Pf_{Qlm}$, $\Gf_{Qkjlm}$, and $\Tf_{Qkjl}$ are determined by \eqref{lem-P-exp.1}-\eqref{lem-P-exp.3}, and $\Qft_M$, $\Qft^I_M$, $\Qft^0_{LM}$ and $\Qft^I_{LM}$ are determined by \eqref{Qijk-lem.1}, and we note by \eqref{B-def.S} that
\begin{align}
  B^{ij0}&=\delta^{ij}, \label{B0-def}\\
B^{ijK} &= -\delta^i_0 \eta^{jK}-\delta^j_0\eta^{iK} \label{BK-def}.
\end{align}

In the following, it will become necessary to work with
$\betat \gt_{Qjlm}$ instead of $\gt_{Qjlm}$. Using \eqref{grad-alpha-A} and \eqref{for-O.1.S2}, it is
not difficult to see that we can write the evolution
equation \eqref{for-I.S2} in terms of this new variable
as
\begin{align}
    \delta^{ij}\del{t}(\betat \gt_{Qjlm}) +B^{ijK}\betat e_K^\Lambda\partial_\Lambda(\betat \gt_{Qjlm})=& \frac{1}{t} \delta_0^i\delta_0^j  (\betat \gt_{Qljm}+ \betat \gt_{Qmjl}-\betat \gt_{Qjlm}) +\frac{1}{2}\delta^{ij} \delta^{JK}\betat\kt_{JK}\betat\gt_{Qjlm} +\Hf^i_{Qlm}, \label{for-I.S2a}
\end{align}
where we have set
\begin{align}
    \Hf^i_{Qlm}=& -\delta^{ij}\betat^2 \tau_{00}\betat\gt_{Qjlm} + B^{ijK}\Bigl(-\betat^2 \tau_{K0}-\frac{1}{2}\betat\bigl( 2 \mt_K  
   -\delta^{JI}(2\ellt_{JIK}
   -\ellt_{KJI})\bigr)\Bigr)\betat \gt_{Qjlm} \notag \\
   &+\frac{2}{t}\delta^i_0
\betat^2\taut_{Q(lm)} 
 + \delta^i_0 \betat\Pf_{Q(lm)}-B^{ijk}\betat\Gf_{Qkj(lm)}.
 \label{Hf-def}
\end{align}
In addition, we will also need to replace $\kt_{LM}$ with $\betat \kt_{LM}$. With the help of \eqref{for-O.1.S2}, we can write the evolution equation
\eqref{for-H.S2} in terms of this new variable as
\begin{align}
    \del{t}(\betat\kt_{LM}) =& -\frac{1}{t}\betat\kt_{LM}-\betat^2 \tau_{00} \betat\kt_{LM}  +\frac{2}{t}\betat^2\tau_{(LM)}- \betat B^{0jK}\betat\gt_{Kj(LM)} \notag\\
    &+2\betat B^{IjK}\betat\gt_{Kj0(L}\delta_{M)I} 
+\Mf_{LM} \label{for-H.S2a}
\end{align}
where
\begin{align} 
    \Mf_{LM} =&  \betat^2\Qft^0_{(LM)}  -2\betat^2 \Qft^I_{(L}\delta_{M)I}.
    \label{Kf-def}
\end{align}
For use below, we note from Lemmas \ref{Qijk-lem} and
\ref{lem-P-exp} that \eqref{Hf-def} and \eqref{Kf-def} can
be expressed as
\begin{align}
    \Hf =&   (\betat^2 \tau+ \betat\psit+\betat\kt+ \betat \ellt +\betat \mt)*\betat\gt  +\frac{1}{t}\betat^2\taut \notag \\
& + (\betat\tau +\psit+\ellt+\mt) *
(\betat\kt+\betat\ellt+\betat\mt)*(\betat\kt+\betat\ellt+\betat\mt) \notag\\
& +\frac{1}{t}(\betat \tau +\psit+\ellt+\mt) *(\betat\kt+\betat\ellt+\betat\mt)+\frac{1}{t}(\betat \tau +\psit) *\betat^2 \tau   \label{Hf-exp}
    \intertext{and}
    \Mf &= (\betat\kt+\betat\ellt+\betat\mt)*(\betat^2\tau+\betat\ellt+\betat\mt),\label{Kf-exp}
\end{align}
respectively.

\subsection{Rescaled variables}
\label{sec:chvar2}
The next step in the transformation to Fuchsian form involves the introduction of the following rescaled variables:
\begin{align}
k&=(k_{IJ}) :=(t\betat \kt_{IJ})
\label{k-def}\\
\beta &=  t^{\ep_0}\betat, \label{beta-def}\\
\ell&=(\ell_{IjK}):= (t^{\ep_1}\ellt_{IjK}), \label{ell-def}\\
m&=(m_{I}) := (t^{\ep_1}\mt_I), \label{m-def} \\
\xi&=(\xi_{ij}) := (t^{\ep_1-\ep_0}\tau_{ij}), \label{xi-def}\\
\psi &=(\psi_I{}^k{}_J) :=  (t^{\ep_1}\psit_I{}^k{}_J),\label{psi-def}\\
f&=(f_I^\Lambda) := (t^{\ep_2} e^\Lambda_I), \label{f-def}\\
\gac&=(\gac_{Ijkl}) := (t^{1+\ep_1} \betat \gt_{Ijkl}), \label{gac-def}\\
\tauac &=(\tauac_{Ijk}):= (t^{\ep_0+2 \ep_1} \taut_{Ijk}), \label{tauac-def} 
\end{align}
where the constants  $\ep_0,\ep_1,\ep_2>0$ are chosen to satisfy
\begin{equation*}
    0<\ep_0<\ep_1, \quad 3\ep_0+\ep_1<1,  \quad 0<\ep_2, \quad \ep_0+\ep_2<1,
\end{equation*}
which we note imply that $\ep_0+\ep_1<3\ep_0+\ep_1<1$.

\subsection{Fuchsian formulation\label{sec:Fuch-form}}
Gathering the rescaled variables \eqref{k-def}-\eqref{tauac-def} together in the following vector\footnote{In line with Remark~\ref{rem:symmetry}, we always assume that $U$ is defined with $k_{LM}=k_{(LM)}$ and $\gac_{Qjlm}=\gac_{Qj(lm)}$.} 
\begin{equation}\label{U-def}
U =\bigl(k_{LM},m_M,\ell_{R0M},\ell_{RLM},\xi_{rl},\beta,f^\Lambda_I,\psi_I{}^k{}_{J},\tauac_{Qjl}, \gac_{Qjlm}\bigr)^{\tr},
\end{equation}
it is then straightforward to verify from
the first order equations \eqref{for-G.1.S2}-\eqref{for-G.2.S2}, \eqref{for-G.4.S2}-\eqref{for-F.2.S2}, \eqref{for-J.S2}-\eqref{for-N.S2}, \eqref{for-I.S2a} and \eqref{for-H.S2a} that $U$ satisfies the symmetric hyperbolic Fuchsian equation
\begin{equation} \label{Fuch-ev-A}
A^0\del{t}U + \frac{1}{t^{\ep_0+\ep_2}}A^\Lambda \del{\Lambda} U =
\frac{1}{t}\Ac\Pbb U + F
\end{equation}
where
\begin{equation} \label{A0-def}
    A^0 =\small \diag\bigl(\delta^{\Lt L}\delta^{\Mt M},\delta^{\Mt M},\delta^{\Rt R}\delta^{\Mt M},\delta^{\Rt R} \delta^{\Lt L} \delta^{\Mt M},\delta^{\rt r}\delta^{\lt l},1,\delta^{\It I}\delta_{\Lambdat \Lambda},\delta^{\It I}\delta_{\kt k} \delta^{\Jt J},\delta^{\Qt Q}\delta^{\jt j}\delta^{\lt l},\delta^{\Qt Q}\delta^{\jt j}\delta^{\lt l} \delta^{\mt m}\bigr),
\end{equation}
\begin{equation}\label{ALambda-def}
    A^\Lambda = \diag\bigl( 0,0,0,0,0,0,0,0,\delta^{\Qt Q}\delta^{\lt l}B^{\jt j K}\beta f^\Lambda_K,\delta^{\Qt Q}\delta^{\lt l}\delta^{\mt m}B^{\jt j K}\beta f^\Lambda_K\bigr),
\end{equation}
\begin{equation}\label{Pbb-def}
  \Pbb = \diag\Bigl(0,\delta_{\Mt}^{ M},\delta_{\Rt}^{R}\delta_{\Mt}^{M},\delta_{\Rt}^{R} \delta_{\Lt}^{ L} \delta_{\Mt}^{M},\delta_{\rt}^{r}\delta_{\lt}^{l},1,\delta_{\It}^{I}\delta^{\Lambdat}_{\Lambda},\delta_{\It}^{I}\delta^{\kt}_{k} \delta_{\Jt}^{J},\delta_{\Qt}^{Q}\delta_{\jt}^{j}\delta_{\lt}^{l},\delta_{\Qt}^{Q}\delta_{\jt}^{j}\delta_{\lt}^{l} \delta_{\mt}^{m}\Bigr),
\end{equation}
\begin{align} \label{Ac-def}
    \Ac = 
    \begin{pmatrix} \Ac_{1\,1} & 0 & 0 & 0 & 0 & 0 & 0 & 0 & 0  & 0 \\
    0 &\Ac_{2\,2} & 0 & \Ac_{2\,4} & 0 & 0 & 0 & 0 & 0  & \Ac_{2\,10} \\
    0 & 0 & \Ac_{3\,3} & 0 & 0 & 0 & 0 & 0 & 0  & \Ac_{3\,10} \\
    0 & 0 & 0 & \Ac_{4\,4} & 0 & 0 & 0 & 0 & 0  & \Ac_{4\,10} \\
    0 & 0 & 0 & 0 & \Ac_{5\,5} & 0 & 0 & 0 & 0  & 0 \\
    0 & 0 & 0 & 0 & 0 & \Ac_{6\,6} & 0 & 0 & 0  & 0 \\
    0 & 0 & 0 & 0 & 0 & 0 & \Ac_{7\,7} & 0 & 0  & 0 \\
    0 & 0 & 0 & 0 & 0 & 0 & 0 & \Ac_{8\,8} & 0  & \Ac_{8\,10} \\
    0 & 0 & 0 & 0 & 0 & 0 & 0 & 0 &\Ac_{9\,9} & 0 \\
    0 & 0 & 0 & 0 & 0 & 0 & 0 & 0 & 0  & \Ac_{10\,10}  
    \end{pmatrix}
\end{align}
with the non-zero diagonal and off diagonal 
blocks of $\Ac$ given by
\begin{gather}
    \Ac_{1\,1}=\delta^{\Lt L}\delta^{\Mt M},\quad \Ac_{2\,2}=(2+\ep_1)\delta^{\Mt M}, \quad
    \Ac_{3\, 3} = \ep_1\delta^{\Rt R}\delta^{\Mt M}, \label{Ac-diag-1}\\
    \Ac_{4\, 4} =  \ep_1\delta^{\Rt R}\delta^{\Lt L}\delta^{\Mt M}, \quad  
    \Ac_{5\, 5} = (\ep_1-\ep_0)\delta^{\rt r}\delta^{\lt l},\quad
     \Ac_{6\, 6} =  \ep_0, \label{Ac-diag-2}\\
     \Ac_{7\, 7} = \ep_2\delta^{\It I}\delta_{\Lambdat\Lambda},\quad
     \Ac_{8\, 8} = \ep_1\delta^{\It I}\delta_{\kt k}\delta^{\Jt J},\quad
     \Ac_{9\, 9} = (\ep_0+2\ep_1)\delta^{\Qt Q}\delta^{\jt j} \delta^{\lt l}, \label{Ac-diag-3}\\
     \Ac_{10\, 10} = (1+\ep_1)\delta^{\Qt Q}\delta^{\jt j} \delta^{\lt l}\delta^{\mt m}+ \delta^{\Qt Q}\delta^{\jt}_0\bigl( \delta_0^l\delta^{\lt j}\delta^{\mt m}+\delta_0^l \delta^{\lt m}\delta^{\mt j}-\delta_0^j \delta^{\lt l}\delta^{\mt m}\bigr),
       \label{Ac-diag-4}
\end{gather}
and
\begin{gather}
    \Ac_{2\,4} = \delta^{\Mt R}\delta^{LM}-2\delta^{RL}\delta^{\Mt M}, \quad 
    \Ac_{2\, 10}=-B^{0jQ}\delta_0^l\delta^{\Mt m}, \label{Ac-offdiag-1} \\
    \Ac_{3\, 10}=-B^{\Rt jQ}\delta^l_0\delta^{\Mt m}, \quad 
    \Ac_{4\, 10} =  -B^{\Rt jQ}\delta^{\Lt l}\delta^{\Mt m}, \label{Ac-offdiag-2} \\
     \Ac_{8\, 10} =
     \frac{1}{2}\delta^{\It Q}\bigl(-\delta_{\kt}^0\delta^{\Jt j}\delta^l_0\delta^m_0-\delta_{\kt k}\bigl(\delta^j_0\eta^{k m}\delta^{\Jt l}+\eta^{km}\delta^{\Jt j}\delta^l_0-\eta^{kj}\delta^{\Jt m}\delta^l_0\bigr)\bigr), \label{Ac-offdiag-3}
\end{gather}
respectively,
and
\begin{equation}\label{F-def}
F=\bigl(F_1,F_2,F_3,F_4,F_5,F_6,F_7,F_8,F_9,F_{10}\bigr)^{\tr}
\end{equation}
with
\begin{align}
F_1 =&\delta^{\Lt L}\delta^{\Mt M}\Bigl( -t^{- \ep_0-\ep_1}\beta^2 \xi_{00} k_{LM}  +2t^{- \ep_0-\ep_1}\beta^2\xi_{(LM)}-t^{-\ep_0-\ep_1} \beta B^{0jK}\gac_{Kj(LM)} \notag\\
    &+t^{-\ep_0-\ep_1}2\beta B^{IjK}\gac_{Kj0(L}\delta_{M)I}  
+t\Mf_{LM}\Bigr), \label{F1-def}\\
F_2 =& t^{-1}2\delta^{\Mt M}\beta\xi_{0M}
+ t^{\ep_1}\delta^{\Mt M}\betat\Qft_M, \label{F2-def}\\
F_3 =& \delta^{\Mt M}t^{\ep_1}\betat \Qft^{\Rt}_M, \label{F3-def}\\
F_4 =& \delta^{\Lt L}\delta^{\Mt M}t^{\ep_1}\betat \Qft^{\Rt}_{LM}, \label{F4-def}\\
F_5 =& \delta^{\lt l}
 t^{\ep_1-\ep_0}\betat J^{\rt}_l-t^{-\ep_1-3\ep_0}\delta^{\lt l}\delta^{\rt r} B^{i jK} \beta\tauac_{Kj(l}\delta_{r)i},
 \label{F5-def}\\
 F_6 =& t^{-1}\frac{1}{2}\delta^{JK}k_{JK}\beta -t^{-\ep_0-\ep_1}\beta^3 \xi_{00},\label{F6-def}\\
F_7 =& t^{-1}\delta^{\It I}\delta_{\Lambdat \Lambda}\Bigl(- \frac{1}{2}\delta^{JL}k_{IL}-\delta^{JK} t^{1-\ep_0-\ep_1}\beta\psi_{I}{}^0{}_K\Bigr)f^\Lambda_J,\label{F7-def}\\
F_8 =&
-\delta^{\It I}\delta_{\kt}^0\delta^{\Jt J}t^{-\ep_1-3\ep_0}\beta^2 \tauac_{I0J}+ \delta^{\It I}\delta_{\kt k} \delta^{\Jt J}t^{\ep_1}\betat \Lf_I{}^k{}_J,\label{F8-def}\\
F_9=& \delta^{\Qt Q}\delta^{\lt l}\bigl(t^{\ep_0+2 \ep_1}\betat K^{\jt}_{Ql}-B^{\jt jk}t^{\ep_0+2 \ep_1}\Tf_{Qkjl}\bigr), \label{F9-10}\\
  F_{10}=& \delta^{\Qt Q}\delta^{\lt l}\delta^{\mt m} t^{1+\ep_1}\Hf^{\jt}_{Qlm}
           +\frac{1}{2t} \delta^{j\jt} \delta^{\Qt Q}\delta^{\lt l}\delta^{\mt m} \delta^{JK}k_{JK}\gac_{Qjlm} .
\label{F10-def}
\end{align}

The Fuchsian equation \eqref{Fuch-ev-A} is not quite in a suitable form for us to analyse. To bring it into a suitable form, we need to subtract off the FLRW solution, which is given by
 \begin{equation}
    \label{eq:Ubgdef}
    \mathring U
    =\bigl(0,0,0,0,0, t^{\ep_0}, t^{\ep_2}\delta_{I}^{\Lambda},0,0,0\bigr).
  \end{equation}
So letting
\begin{equation} \label{u-def}
  u=U-\mathring U
\end{equation}
denote the nonlinear perturbation of the FLRW solution,
it is not difficult to verify from \eqref{Fuch-ev-A}, \eqref{A0-def}, \eqref{Ac-def}, and \eqref{Ac-diag-2}-\eqref{Ac-diag-3} that
$u$ satisfies
\begin{equation} \label{Fuch-ev-A2}
A^0\del{t}u + \frac{1}{t^{\ep_0+\ep_2}}A^\Lambda \del{\Lambda} u =
\frac{1}{t}\Ac\Pbb u + F,
\end{equation}
where $A^\Lambda$, $\Ac$ and $F$ depend on $(t,U)$ and $F$ satisfies $F|_{u=0}=0$.
It is this equation in conjunction with the local existence theory developed in Proposition \ref{lag-exist-prop} that will be used in Sections \ref{global-sec} and \ref{local-sec} below to establish the existence of solutions to the conformal Einstein-scalar field equations in a neighborhood of the big bang singularity at $t=0$ and to determine the asymptotic behavior of the physical fields near the singularity. It is also worth noting that since $\mathring{U}$ is only depends on $t$, we can and will, in the following, view the
matrices  $A^\Lambda$, $\Ac$ and
the source term $F$ as depending
on $(t,u)$.

\subsection{Source term expansion}
Before we move on to establishing the existence of solutions to the Fuchsian equation \eqref{Fuch-ev-A2}, we first examine the structure of the source term $F$ because this will be needed in the subsequent analysis. We begin by inspecting the various terms that appear in the components of $F$, see \eqref{F1-def}-\eqref{F10-def}. By Lemma \ref{Qijk-lem}, the terms $t^{\ep_1}\betat\Qft_M$, $t^{\ep_1}\betat\Qft^I_M$, $t^{\ep_1}\betat\Qft^I_{LM}$ are of the form
\begin{equation} \label{Qft-Fuch}
    t^{\ep_1}\betat \Qft = (t^{-1}k+t^{-\ep_0-\ep_1}\beta\ell+t^{-\ep_0-\ep_1}\beta m)*\beta\xi +(t^{-1}k+t^{-\ep_0-\ep_1}\beta\ell+t^{-\ep_0-\ep_1}\beta m)*(\ell+m)
\end{equation}
while we see from \eqref{Kf-exp} that $t\Mf_{LM}$ is given by
\begin{equation} \label{Kf-Fuch}
t\Mf= t^{-\ep_0-\ep_1}(k+t^{1-\ep_0-\ep_1}\beta \ell+t^{1-\ep_0-\ep_1}\beta m)*(\beta^2\xi +\beta\ell+\beta m).
\end{equation}
We further observe from Lemmas \ref{lem-cov-exp} and \ref{lem-P-exp} that
the terms $t^{\ep_0+2 \ep_1}\Tf_{QKjl}$ and
$t^{\ep_1}\betat\Lf_I{}^k_{J}$ are
of the form
\begin{align}
    t^{\ep_0+2 \ep_1}\Tf =& ( t^{-\ep_0-\ep_1}\beta^2 \xi+t^{-\ep_0-\ep_1}\beta\psi+t^{-1} k+ t^{-\ep_0-\ep_1}\beta\ell +t^{-\ep_0-\ep_1}\beta m)*\tauac \notag \\
    &+ (t^{2\ep_0-1}\beta \xi+t^{2\ep_0-1}\psi +t^{2\ep_0-1}\ell +t^{2\ep_0-1} m)*(k + t^{1-\ep_0-\ep_1}\beta\ell +t^{1-\ep_0-\ep_1}\beta m)*\xi \label{Tf-Fuch}
\end{align}
and
\begin{equation} \label{Lf-Fuch}
   t^{\ep_1}\betat\Lf = (t^{-\ep_1-\ep_0}\beta^2\xi+t^{-\ep_0-\ep_1}\beta\psi+t^{-1}k+t^{-\ep_0-\ep_1}\beta\ell +t^{-\ep_0-\ep_1}\beta  m)*( \beta\xi+\psi+\ell+m),
\end{equation}
respectively. Finally, we observe form \eqref{J-def}, \eqref{K-def} and from \eqref{Hf-exp} that
the terms $t^{\ep_1-\ep_0}\betat J^i_l$, $t^{\ep_0+2 \ep_1}\betat K^i_{Ql}$ and
$t^{1+\ep_1}\Hf^{i}_{Qlm}$ are of
the form
\begin{equation} \label{J-Fuch}
t^{\ep_1-\ep_0}\betat J =(t^{-\ep_0-\ep_1}\beta^2 \xi + t^{-1}k +t^{-\ep_0-\ep_1}\beta\ell +t^{-\ep_0-\ep_1}\beta m)*\xi,
\end{equation}
\begin{align}
t^{\ep_0+2 \ep_1}\betat K =& (t^{-1}k+t^{-\ep_0-\ep_1}\beta\ell+t^{-\ep_0-\ep_1}\beta m)*\tauac + t^{2\ep_0-1}\gac*\xi \notag\\
&+ (\ell+m)*(t^{2\ep_0-1}k+t^{\ep_0- \ep_1}\beta\ell+t^{\ep_0-\ep_1}\beta m)*\xi \label{K-Fuch}
\end{align}
and
\begin{align}
t^{1+\ep_1}\Hf =& (t^{-\ep_0-\ep_1}\beta^2 \xi+ t^{-\ep_0-\ep_1}\beta\psi+t^{-1}k+ t^{-\ep_0-\ep_1}\beta \ell +t^{-\ep_0-\ep_1}\beta m)*\gac +
t^{-\ep_1-3\ep_0}\beta^2\tauac \notag \\
& + t^{-1}(\beta\xi +\psi+\ell+m) *
(k+t^{1-\ep_0-\ep_1}\beta\ell+t^{1-\ep_0-\ep_1}\beta m)*(k+t^{1-\ep_0-\ep_1}\beta\ell+t^{1-\ep_0-\ep_1}\beta m) \notag\\
& +t^{-1}(\beta \xi +\psi+\ell+m) *(k+t^{1-\ep_0-\ep_1}\beta\ell+t^{1-\ep_0-\ep_1}\beta m)
+t^{-1}(\beta \xi +\psi) *t^{1-\ep_0-\ep_1}\beta^2\xi. \label{Hf-fuch} 
\end{align}

With the help of the expansions \eqref{Qft-Fuch}-\eqref{Hf-fuch}, it is then straightforward to verify from \eqref{F1-def}-\eqref{F10-def} that, for any $T_0>0$ and $\tilde{\ep}$ satisfying
\begin{equation}
    \label{eq:epstildecond}
     \max\{3\ep_0+\ep_1,1-\ep_0,1-\ep_2\}\leq \tilde{\ep}<1,
\end{equation} $F$ can be decomposed 
as
\begin{equation}
    \label{eq:SourceDecomp}
    F=\frac{1}{t}G(u)+\frac{1}{t^{\tilde\epsilon}} H(t,u) 
  \end{equation}
where $H\in C^0\bigl([0,T_0],C^\infty\bigl(\Rbb^{\udim},\Rbb^{\udim}\bigr)\bigr)$ satisfies
\begin{equation}
    \label{eq:SourcetermVan2}
    H(t,0)=0,
\end{equation}
and
     \begin{equation}
    \label{eq:GtDef}
    G(u)=\Gt(u)\Pbb u
    \end{equation}
for some $\Gt \in C^\infty\bigl(\Rbb^{\udim},\Mbb{\udim}\bigr)$ satisfying
\begin{equation} \label{eq:Gtprops}
      \Gt(0) = 0 \AND  [\Gt,\Pbb]=0.
\end{equation}
For use below, we also note by \eqref{Pbb-def} and \eqref{Ac-def}
that 
\begin{equation} \label{Ac-comm}
    [\Ac,\Pbb]=0.
\end{equation}

\section{Past global FLRW stability\label{global-sec}}
We are now ready to state and prove our past global FLRW stability result, which represents the first of the two
main results of this article. Note here that past stability refers to our choice of time orientation where the past is in the direction of decreasing $t$ and the fact that we establish stability on time intervals of the form $(0,t_0]$. Thus, with this choice of time orientation, all big bang singularities will occur in the past at $t=0$.

The precise statement of our first stability result is given below in Theorem \ref{glob-stab-thm}. The proof of Theorem \ref{glob-stab-thm} is carried out in two steps. The first step involves establishing the stability of the trivial solution $u=0$ of the Fuchsian equation \eqref{Fuch-ev-A2}.  Proposition \ref{prop:globalstability} contains the statement of this Fuchsian stability result and its proof is carried out in Section \ref{sec:proof_globstab_Fuchsian}. Proposition \ref{prop:globalstability} is then used in Section \ref{sec:proof_globstab} to prove Theorem \ref{glob-stab-thm}.

\subsection{The past global stability theorem}
Theorem \ref{glob-stab-thm} below establishes the nonlinear stability of the FLRW solution of the conformal Einstein-scalar field equations on $M_{0,t_0}=(0,t_0]\times \Tbb^{n-1}$ for any $t_0>0$
by guaranteeing that sufficiently small perturbations of the FLRW initial data, see Remark \ref{FLRW-idata-rem-A}, 
that also satisfy the gravitational and wave gauge constraints as well as the synchronization condition $\tau|_{\Sigma_{t_0}}=t_0$ will generate solutions of conformal Einstein-scalar field equations on $M_{0,t_0}$ that are asymptotically Kasner in the sense of Definition \ref{def:APKasner}. 
As discussed in Section \ref{temp-synch}, if the initial data does not satisfy the synchronization condition  $\tau|_{\Sigma_{t_0}}=t_0$, then it can be evolved for short amount of time so that it does, and so, we lose no generality in assuming that the initial data is synchronized. 

\begin{thm}[Past global stability of the FLRW solution of the Einstein-scalar field system]\label{glob-stab-thm}

Suppose $n\in\Zbb_{\ge 3}$, $k \in \Zbb_{>(n+3)/2}$, $\sigma>0$ and $T_0>0$. Then for every $t_0\in (0,T_0]$ there exists a $\delta_0>0$ such that for every $\delta \in (0,\delta_0]$ and  $\gr_{\mu\nu}\in H^{k+2}(\Tbb^{n-1},\mathbb{S}_n)$, 
$\ggr_{\mu\nu}\in H^{k+1}(\Tbb^{n-1},\mathbb{S}_n)$, $\taur=t_0$ and
$\taugr\in H^{k+2}(\Tbb^{n-1})$
satisfying
\begin{equation} \label{glob-stab-thm-idata-A}
  \norm{\gr_{\mu\nu}-\eta_{\mu\nu}}_{H^{k+2}(\Tbb^{n-1})}
  +\norm{\ggr_{\mu\nu}}_{H^{k+2}(\Tbb^{n-1})}
  +\norm{\taugr-1}_{H^{k+2}(\Tbb^{n-1})}< \delta,
\end{equation}
and the gravitational and wave gauge constraints \eqref{grav-constr}-\eqref{wave-constr}, there exists is a unique classical solution $W\in C^1(M_{0,t_0})$, see \eqref{eq:Wdef}, 
of the system of evolution equations \eqref{tconf-ford-C.1}-\eqref{tconf-ford-C.8} on $M_{0,t_0}=(0,t_0]\times \Tbb^{n-1}$ with regularity
\begin{equation}
  \label{eq:WregFin}
W \in \bigcap_{j=0}^{k}C^j\bigl((0,t_0], H^{k-j}(\Tbb^{n-1})\bigr)
\end{equation} 
that satisfies the initial
conditions \eqref{l-idata}-\eqref{hhu-idata} on $\Sigma_{t_0}=\{t_0\}\times\Tbb^{n-1}$, which are uniquely determined by the initial data $\{\gr_{\mu\nu},\ggr_{\mu\nu},\taur=t_0,\taugr\}$, and
the constraints \eqref{eq:Lag-constraints} in $M_{0,t_0}$.
\bigskip

\noindent Moreover, the pair $\{g_{\mu\nu}=\del{\mu}l^\alpha\ghu_{\alpha\beta}\del{\nu}l^\beta,\tau=t\}$, which is uniquely determined by $W$, defines a solution of the conformal Einstein-scalar field equations \eqref{lag-confeqns} on $M_{0,t_0}$ that satisfies the wave gauge constraint \eqref{lag-wave-gauge} and the following properties: 
\begin{enumerate}[(a)]
  \item Let $e_0^\mu=\betat^{-1}\delta_0^\mu$ with $\betat= (-g(dx^0,dx^0))^{-\frac{1}{2}}$, and $e^\mu_I$ be the unique solution of the Fermi-Walker transport equations \eqref{Fermi-A} with initial conditions $e^\mu_I |_{\Sigma_{t_0}}=\delta^\mu_\Lambda\er^\Lambda_I$ where the functions $\er^\Lambda_I \in H^k(\Tbb^{n-1})$ are chosen to satisfy $\norm{\er^\Lambda_I-\delta^\Lambda_I}_{H^{k}(\Tbb^{n-1})} < \delta$
and make the frame $e^\mu_i$ orthonormal on $\Sigma_{t_0}$.
Then $e^\mu_i$ is a well defined frame in $M_{0,t_0}$ that
satisfies
$e^0_I=0$ and $g_{ij}=\eta_{ij}$,
where $g_{ij}=e_i^\mu g_{\mu\nu}e^\nu_j$ is the
frame metric. 
\item There exists a tensor field $\kf_{IJ}$ in $H^{k-1}(\Tbb^{n-1},\Sbb{n-1})$ satisfying $\norm{\kf_{IJ}}_{H^{k-1}(\Tbb^{n-1})}\lesssim\delta$ such that 
  \begin{align}
    \label{eq:glostab-est.First}
\norm{t\betat\Dc_0 g_{00}+\delta^{JK}\kf_{JK}}_{H^{k-1}(\mathbb T^{n-1})}&\lesssim
                                                                            t^{\frac19(3-\sqrt{3}-2\sigma)},\\
    \label{eq:glostab-est.2}
\norm{t\betat\Dc_0 g_{JK}-\kf_{JK}}_{H^{k-1}(\mathbb T^{n-1})}
     &\lesssim t^{\frac19(3-\sqrt{3}-2\sigma)},\\
 \norm{\Dc_I g_{00}}_{H^{k-1}(\mathbb T^{n-1})}+\norm{\Dc_0 g_{J0}}_{H^{k-1}(\mathbb T^{n-1})}
    \hspace{0.8cm} &\notag\\
      +\norm{\Dc_I g_{J0}}_{H^{k-1}(\mathbb T^{n-1})}+\norm{\Dc_I g_{JK}}_{H^{k-1}(\mathbb T^{n-1})}&\lesssim t^{-\frac 12(\sqrt{3}-1)-\sigma},  \label{eq:glostab-est.3}\\
\label{eq:glostab-est.4}
\norm{t\betat\Dc_I\Dc_j g_{lm}}_{H^{k-1}(\mathbb T^{n-1})}&\lesssim t^{-\frac 12(\sqrt{3}-1)-\sigma},\\
\label{eq:glostab-est.5}
\norm{\Dc_i\Dc_j\tau }_{H^{k-1}(\mathbb T^{n-1})}&\lesssim t^{-(\frac 23\sqrt{3}-1)-\sigma},\\
\label{eq:glostab-est.Last}
\norm{\Dc_I\Dc_j\Dc_l\tau }_{H^{k-1}(\mathbb T^{n-1})}&\lesssim t^{-(\sqrt{3}-1)-2\sigma},
  \end{align}
for all $t\in (0,t_0]$,
where all the
fields in these estimates are expressed in terms of the frame $e_i^\mu$ and the Levi-Civita connection $\Dc$ of the flat background metric
$\gc_{\mu\nu}=\del{\mu}l^\alpha \eta_{\alpha\beta} \del{\nu}l^\beta$. 
In addition, there exist a strictly positive function $\mathfrak{b}\in H^{k-1}(\Tbb^{n-1})$, a matrix $\ef^\Lambda_J\in H^{k-1}(\Tbb^{n-1},\Mbb{n-1})$, and a constant $C>0$ such that
  \begin{align}
    \label{eq:improvbetaestimate-glob}
       \Bnorm{t^{-\kf_{J}{}^J/2}\betat
    -\mathfrak{b}}_{H^{k-1}(\Tbb^{n-1})}
 & \lesssim t^{\frac 19(3-\sqrt{3}-2\sigma)}
\intertext{and}
  \label{eq:improvframeestimate}
  \Bnorm{\exp\Bigl(\frac 12\ln(t)\kf_{J}{}^I\Bigr) e_I^\Lambda-\ef_J^\Lambda}_{H^{k-1}(\Tbb^{n-1})}    
  &\lesssim t^{\frac19(3-\sqrt{3}-2\sigma)-C\delta}
\end{align}
 for all $t\in (0,t_0]$, where $\kf_{L}{}^J=\kf_{LM}\delta^{MJ}$ and $\exp\Bigl(\frac 12\ln(t)\kf_{J}{}^I\Bigr)$ is the exponential of the matrix
 $\frac 12\ln(t)(\kf_{J}{}^I)$.
\item The second fundamental form $\Ktt_{\Lambda\Omega}$ induced on the $t=const$-hypersurfaces  by $g_{\mu\nu}$ satisfies
\begin{equation}
\label{eq:2ndFF.est}
\Bnorm{2 t\betat \Ktt_{LJ}- \kf_{LJ}}_{H^{k-1}(\mathbb T^{n-1})}
\lesssim t^{\frac19(3-\sqrt{3}-2\sigma)}
\end{equation}
for all $t\in (0,t_0]$, while the lapse, shift and the spatial metric on the $t=const$-hypersurfaces are determined by $\Ntt=\betat$,
$\btt_\Lambda=0$, and $\gtt_{\Lambda\Omega}=g_{\Lambda\Omega}$, respectively.
\item The pair $\Bigl\{\gb_{\mu\nu}=t^{\frac {2}{n -2}}g_{\mu\nu},\phi=\sqrt{\frac{n-1}{2(n-2)}}\ln(t)\Bigr\}$ defines a solution of the physical Einstein-scalar field equations \eqref{ESF.1}-\eqref{ESF.2} on $M_{0,t_0}$ that exhibits AVTD behavior and is asymptotically pointwise Kasner on $\Tbb^{n-1}$ with Kasner exponents $r_1(x),\ldots,r_{n-1}(x)$ determined by the eigenvalues of $\kf_{L}{}^J(x)$ for each $x\in \Tbb^{n-1}$. In particular, $\kf_{L}{}^L(x)\ge 0$ for all $x\in\Tbb^{n-1}$, and $\kf_{L}{}^L(x)=0$ for some $x\in \Tbb^{n-1}$ if and only if $r_1(x)=\ldots=r_{n-1}(x)=0$. The time $t=0$ represents a crushing singularity in the sense of \cite{eardley1979}. Furthermore, the function
  \begin{equation*}
    \Ptt
=\frac{\sqrt{{2(n-1)}(n-2)}}{{2(n-1)} +(n-2) \kf_{L}{}^L}  
\end{equation*}
can be interpreted as the asymptotic scalar field strength in the sense of Section~\ref{sec:AVTDAPK}.
\item The physical solution $\bigl\{\gb_{\mu\nu},\phi\bigr\}$ is past $C^2$ inextendible at $t=0$ and past timelike geodesically incomplete. The scalar curvature  $\Rb=\Rb_{\mu\nu}\gb^{\mu\nu}$ of the physical metric $\gb_{\mu\nu}$ 
satisfies
  \begin{equation}
    \label{eq:SptRicciEstimat-glob}
    \Bnorm{t^{2\frac {n-1}{n -2}+\kf_{J}{}^J}\Rb+\frac{n-1}{n-2}\mathfrak{b}^{-2}}_{H^{k-1}(\mathbb T^{n-1})}\lesssim t^{\frac 19(3-\sqrt{3}-2\sigma)}
  \end{equation}
for all $t\in (0,t_0]$, and consequently, it becomes unbounded in the limit $t\searrow 0$.
\end{enumerate}

\medskip

\noindent The implicit and explicit constants in the above estimates are independent of the choice of $\delta\in (0,\delta_0]$.
\end{thm}

\subsection{Fuchsian stability proof}
\label{sec:proof_globstab_Fuchsian}
As discussed above, the first step in the proof of Theorem \ref{glob-stab-thm} is to establish the stability of the trivial solution $u=0$ to the Fuchsian equation \eqref{Fuch-ev-A2}. Thus, we need to prove the existence of solutions to the Fuchsian global initial value problem (GIVP) 
\begin{align}
  \label{eq:givp1}
  A^0\del{t}u + \frac{1}{t^{\ep_0+\ep_2}}A^\Lambda(t,u) \del{\Lambda} u &=
\frac{1}{t}\Ac\Pbb u + \frac{1}{t} G(u)+\frac{1}{t^{\tilde\ep}}H(t,u)\hspace{1.0cm} \text{in $M_{0,t_0}=(0,t_0]\times \Tbb^{n-1}$,}\\
  \label{eq:givp2}
  u&=u_0 \hspace{4.9cm} \text{in $\Sigma_{t_0}=\{t_0\}\times \Tbb^{n-1}$},
\end{align}
for initial data $u_0$ that is sufficiently small, where we assume that  $\tilde{\ep}$ has been fixed by
\begin{equation}
 \label{eq:epstildecond2}
  \tilde\ep=\max\{3\ep_0+\ep_1,1-\ep_0,1-\ep_2\};
\end{equation}
see \eqref{eq:epstildecond}.
Notice that we have used \eqref{eq:SourceDecomp} to replace the source term $F$ from \eqref{Fuch-ev-A2} with
$t^{-1} G(u)+t^{-\tilde\ep}H(t,u)$ in \eqref{eq:givp1} and that we are using the notation $A^\Lambda(t,u)$ to denote the dependence of the matrices $A^\Lambda$ on $t$ and $u$, which is determined by \eqref{U-def}, \eqref{ALambda-def}, \eqref{eq:Ubgdef} and \eqref{u-def}.

The stability of the trivial solution $u=0$ to the GIVP \eqref{eq:givp1}-\eqref{eq:givp2} is established in Proposition \ref{prop:globalstability} below subject to the conditions on the parameters $\ep_0$, $\ep_1$, and $\ep_2$ stated there. The proof amounts to an application of the Fuchsian global existence theory established in \cite{BOOS:2021,Oliynyk:CMP_2016}. The actual result that we employ is Theorem~A.2 together with Remark~A.3 from \cite{BeyerOliynyk:2024}, which is a slight generalization of the global existence result from \cite{BOOS:2021}.

\begin{prop}
  \label{prop:globalstability}
Suppose $n\in\Zbb_{\ge 3}$, $k \in \Zbb_{>(n+1)/2}$, $\sigma>0$, $T_0>0$, and $\ep_0$, $\ep_1$, $\ep_2$ satisfy
  \begin{equation}
  \label{eq:epscond2}
    0<\ep_0,\quad \frac{1}{1+\sqrt{3}}<\ep_1,\quad 3\ep_0+\ep_1<1, \quad  0<\ep_2<1-\ep_0.
  \end{equation}
Then there exists
a $\delta_0 > 0$ such that for every $u_0\in H^k(\mathbb T^{n-1})$ with
\begin{equation}
  \label{glob-stab-thm-idata-A-Fuchsian}
 \norm{u_0}_{H^k(\mathbb T^{n-1})}< \delta_0
\end{equation}
and every $t_0\in (0,T_0]$,
the Fuchsian GIVP \eqref{eq:givp1}-\eqref{eq:givp2}  admits a unique solution 
\begin{equation*}
u \in C^0_b\bigl((0,t_0],H^k(\mathbb T^{n-1})\bigr)\cap C^1\bigl((0,t_0],H^{k-1}(\mathbb T^{n-1})\bigr)
\end{equation*}
such that $\lim_{t\searrow 0} \Pbb^\perp u(t)$, denoted $\Pbb^\perp u(0)$, exists in $H^{k-1}(\mathbb T^{n-1})$. Moreover, the solution $u$ satisfies the energy estimate
\begin{equation}
  \label{eq:resenergy}
\norm{u(t)}_{H^k(\mathbb T^{n-1})}^2 + \int^{t_0}_t \frac{1}{s} \norm{\Pbb u(s)}_{H^k(\mathbb T^{n-1})}^2\, ds  \lesssim \norm{u_0}^2_{H^k(\mathbb T^{n-1})}
\end{equation}
and decay estimates
\begin{align}
\label{eq:PbbuDecay}
  \norm{\Pbb u(t)}_{H^{k-1}(\mathbb T^{n-1})} &\lesssim t^p+t^{\kappat-\sigma},\\
  \label{eq:PbbuDecay2}
\norm{\Pbb^\perp u(t) - \Pbb^\perp u(0)}_{H^{k-1}(\mathbb T^{n-1})} &\lesssim
 t^{p}+ t^{2(\kappat-\sigma)},
\end{align}
for all $t\in(0,t_0]$,  
where
\begin{align}
  p=\min\{1-3\ep_0-\ep_1,1-\ep_0-\ep_2,\ep_0,\ep_2\}, \quad
    \kappat=\min\biggl\{\ep_0,\ep_2,\ep_1-\frac{1}{1+\sqrt{3}}\biggr\},  \label{eq:defkappat}
\end{align}
and $\Pbb^\perp = \id -\Pbb$.

\medskip

\noindent The implicit constants in the energy and decay estimates are independent of the choice of $t_0\in (0,T_0]$ and $u_0$ satisfying \eqref{glob-stab-thm-idata-A-Fuchsian}.
\end{prop}

\noindent The proof of Proposition~\ref{prop:globalstability}  makes use of  two technical lemmas, Lemmas~\ref{lem-Bc-lbnd} and \ref{lem:posdef1}, which we discuss first. 

\begin{lem}
  \label{lem-Bc-lbnd}
  Suppose $\ep_1>0$ and let 
\begin{align} 
\Nc_{\gac} &=\gac_{Pqrs} \delta^{PQ}\delta^{rl}\delta^{sm} \bigl((1+\ep_1)\delta^{qj}\gac_{Qjlm} + \delta_0^q\delta_0^j  (\gac_{Qljm}+ \gac_{Qmjl}-\gac_{Qjlm})\bigr). \label{Ncgac-def}
\end{align}
Then 
\begin{equation}
  \label{eq:comparenorms}
     \Nc_{\gac} \geq \Bigl(\ep_1-\frac{1}{1+\sqrt{3}}\Bigr)|\gac|^2
\end{equation}
for all $\gac_{Qmjl}\in \Rbb^{(n-1)n^3}$ satisfying $\gac_{Qjlm}=\gac_{Qjml}$ where
$|\gac|^2 =  \delta^{PQ}\delta^{qj}\delta^{rl}\delta^{sm} \gac_{Pqrs}\gac_{Qjlm}$
is the Euclidean norm.
\end{lem}
\begin{proof}
First, by exploiting the symmetry $\gac_{Qjlm}=\gac_{Qjml}$,
we notice that the Euclidean norm $|\gac|^2$
can be expressed as 
\begin{align}
  |\gac|^2  
   &=\delta^{IJ}\gac_{I000}\gac_{J000}
    + 2\delta^{IJ}\delta^{KL}\gac_{I00K}\gac_{J00L}
    + \delta^{IJ}\delta^{KL}\delta^{MN}\gac_{I0KM}\gac_{J0LN}\notag \\
    &+ \delta^{IJ}\delta^{KL}\gac_{IK00}\gac_{JL00}
    +2\delta^{IJ}\delta^{KL}\delta^{MN}\gac_{IK0M}\gac_{JL0N} +\delta^{IJ}\delta^{KL}\delta^{MN}\delta^{PQ}\gac_{IKMP}\gac_{JLNQ}.
    \label{gac-norm-alt}
\end{align}
Next, we set
\begin{equation}\label{varep-def}
    \varep = 1+\ep_1,
  \end{equation}
  and observe, by a straightforward computation involving the symmetry $\gac_{Qjlm}=\gac_{Qjml}$, that we can
express \eqref{Ncgac-def} as
\begin{align}
    \Nc_{\gac} =&(\varep+1) \delta^{IJ}\gac_{I000}\gac_{J000}
    + 2 \varep \delta^{IJ}\delta^{KL}\gac_{I00K}\gac_{J00L}
    +(\varep-1) \delta^{IJ}\delta^{KL}\delta^{MN}\gac_{I0KM}\gac_{J0LN}\notag \\
    &+\varep \delta^{IJ}\delta^{KL}\gac_{IK00}\gac_{JL00}
    +2\varep\delta^{IJ}\delta^{KL}\delta^{MN}\gac_{IK0M}\gac_{JL0N} +\varep\delta^{IJ}\delta^{KL}\delta^{MN}\delta^{PQ}\gac_{IKMP}\gac_{JLNQ} \notag \\
    &+ 2\delta^{IJ}\delta^{KL}\gac_{I00K}\gac_{JL00}
    +2\delta^{IJ}\delta^{KL}\delta^{MN}\gac_{I0KM}\gac_{JL0N}.
    \notag
\end{align}
Then, with the help of the Cauchy-Schwartz and Young's inequalities (e.g. $2|ab| \leq \zeta^{-1} a^2 + \zeta b^2$, $\zeta>0$),
we can bound $\Nc_{\gac}$ below by 
\begin{align} 
    \Nc_{\gac} \geq &(\varep+1) \delta^{IJ}\gac_{I000}\gac_{J000}
    + (2\varep-\zeta_1^{-1}) \delta^{IJ}\delta^{KL}\gac_{I00K}\gac_{J00L}
    +(\varep-1-\zeta_2^{-1}) \delta^{IJ}\delta^{KL}\delta^{MN}\gac_{I0KM}\gac_{J0LN}\notag \\
    &+(\varep-\zeta_1) \delta^{IJ}\delta^{KL}\gac_{IK00}\gac_{JL00}
    +(2\varep-\zeta_2)\delta^{IJ}\delta^{KL}\delta^{MN}\gac_{IK0M}\gac_{JL0N}\notag\\ &+\varep\delta^{IJ}\delta^{KL}\delta^{MN}\delta^{PQ}\gac_{IKMP}\gac_{JLNQ} \label{Ncgac-lbnd-A}
\end{align}
for any $\zeta_1,\zeta_2>0$. The aim now is to optimize 
the constants $\zeta_1,\zeta_2$ in \eqref{Ncgac-lbnd-A} to get the best possible inequality of the form \eqref{eq:comparenorms}.
With the help of \eqref{gac-norm-alt}, we conclude that the optimal choice is determined by the conditions $\varep-\zeta_1^{-1}/2=\varep-\zeta_1$ and $ \varep-1-\zeta_2^{-1}=\varep-\zeta_2/2$, which we can solve for $\zeta_1,\zeta_2$ to get
$\zeta_1= \frac{1}{\sqrt{2}}$ and $\zeta_2 = 1+\sqrt{3}$.
We then have by \eqref{varep-def} and \eqref{Ncgac-lbnd-A} that
\begin{align} 
    \Nc_{\gac} \geq &(\ep_1+2) \delta^{IJ}\gac_{I000}\gac_{J000}
    + 2\Bigl(\ep_1+1-\frac{1}{\sqrt{2}}\Bigr) \delta^{IJ}\delta^{KL}\gac_{I00K}\gac_{J00L}\notag\\
    &+\Bigl(\ep_1-\frac{1}{1+\sqrt{3}}\Bigr) \delta^{IJ}\delta^{KL}\delta^{MN}\gac_{I0KM}\gac_{J0LN}
    +\Bigl(\ep_1+1-\frac{1}{\sqrt{2}}\Bigr) \delta^{IJ}\delta^{KL}\gac_{IK00}\gac_{JL00} \notag \\
    &+2\Bigl(\ep_1-\frac{1}{1+\sqrt{3}}\Bigr)\delta^{IJ}\delta^{KL}\delta^{MN}\gac_{IK0M}\gac_{JL0N} +(\ep_1+1)\delta^{IJ}\delta^{KL}\delta^{MN}\delta^{PQ}\gac_{IKMP}\gac_{JLNQ}. \notag 
\end{align}
The required inequality \eqref{eq:comparenorms} therefore follows.
\end{proof}

Before we state the second technical lemma, we first define an alternative formulation of the Fuchsian equation \eqref{eq:givp1} that will be employed in our subsequent analysis.  To this end, we consider positive constants $\sigma_1,\ldots,\sigma_{10}>0$ and
set
\begin{equation} \label{Asc-def}
  \begin{split}
  \Asc= \small \diag\bigl(&\sigma_1\delta^{\Lt L}\delta^{\Mt M},
  \sigma_2\delta^{\Mt M},
  \sigma_3\delta^{\Rt R}\delta^{\Mt M},
  \sigma_4\delta^{\Rt R} \delta^{\Lt L} \delta^{\Mt M},
  \sigma_5\delta^{\rt r}\delta^{\lt l},\\
  &\sigma_6,
  \sigma_7\delta^{\It I}\delta_{\Lambdat \Lambda},
  \sigma_8\delta^{\It I}\delta_{\kt k} \delta^{\Jt J},
  \sigma_9\delta^{\Qt Q}\delta^{\jt j}\delta^{\lt l},
  \sigma_{10}\delta^{\Qt Q}\delta^{\jt j}\delta^{\lt l} \delta^{\mt
    m}\bigr).
\end{split}
\end{equation}
Then we set
  \begin{align}    
    \Bsc^0=&\Asc A^0=\Asc, \label{Bsc0-def}\\
    \Bsc^\Lambda(t,u)=&\frac{1}{t^{\ep_0+\ep_2}}\Asc  A^\Lambda(t,u),\label{BscLambda-def}\\
    \Bc(u)=&\Asc\Ac+\Asc\Gt(u) \label{eq:defBc}
    \intertext{and}
    \Hsc(t,u)=&\frac{1}{t^{\tilde\ep}}\Asc H(t,u). \label{Hsc-def}
  \end{align}
  Using theses definitions, a short calculation shows that \eqref{eq:givp1} can be expressed as 
  \begin{equation} \label{Fuch-ev-A3}
\Bsc^0\del{t}u + \Bsc^\Lambda(t,u) \del{\Lambda} u =
\frac{1}{t}\Bc(u)\Pbb u + \Hsc(t,u),
\end{equation}
where in deriving this we have used \eqref{eq:GtDef}.

\begin{lem}
  \label{lem:posdef1}  
  Suppose $\ep_0,\ep_1,\ep_2$ satisfy \eqref{eq:epscond2} and $\kappat$ is given by \eqref{eq:defkappat}. Then there exist a $\bc>0$, and  for each $\eta>0$, constants $\sigma_i=\sigma_i(\eta)$, $1\leq i\leq 10$, such that
  \begin{equation*}
    \Asc\Ac\geq(\kappat-\bc\eta)\Bsc^0.
  \end{equation*}
\end{lem}
\begin{proof}
The proof is a straightforward consequence of Lemma~\ref{matrix-lem} in the appendix.
To see why this is the case, we identify the matrix $\Ac$ in Lemma~\ref{matrix-lem} with the matrix \eqref{Ac-def}. The matrix \eqref{Ac-def} is then partitioned by choosing $N=10$ and by identifying its diagonal blocks with the blocks $\Ac_{1\,1}$, \ldots, $\Ac_{10\,10}$ in Lemma~\ref{matrix-lem}, and 
similarly for the off-diagonal blocks. Noticing that $A^0$ in \eqref{A0-def} is the identity matrix, the 
inequality $ \Asc\Ac\geq (\kappat-\bc\eta)\Bsc^0$
is then a direct consequence of \eqref{Ac-diag-1}-\eqref{Ac-offdiag-3}, Lemma~\ref{lem-Bc-lbnd} and Lemma~\ref{matrix-lem} provided $\kappat$ is defined by
  \begin{equation*}
    \kappat=\min\biggl\{1,2+\ep_1,\ep_1,\ep_1-\ep_0,\ep_0,\ep_2,\ep_0+2\ep_1,
    \ep_1-\frac{1}{1+\sqrt{3}}\biggr\}.
  \end{equation*}
However, 
\begin{equation*}
    \ep_1-\ep_0 - \biggl(\ep_1 - \frac{1}{1+\sqrt{3}}\biggr) = \frac{1}{1+\sqrt{3}} - \ep_0 = \biggl(\frac{1}{1+\sqrt{3}} -\frac{1}{3}\biggr) 
    +\biggl(\frac{1}{3} - \ep_0\biggr)> 0
\end{equation*}
by \eqref{eq:epscond2}, and so, with the help of this inequality, it is clear from \eqref{eq:epscond2} that
$\kappat$, defined above, agrees with \eqref{eq:defkappat} and completes the proof.
\end{proof}

With the help of the above two technical lemmas, we are now in the position to prove Proposition~\ref{prop:globalstability}.

\subsubsection{Proof of Proposition~\ref{prop:globalstability}}
 The proof of Proposition~\ref{prop:globalstability} will follow from
an application of Theorem~A.2 and Remark~A.3 from \cite{BeyerOliynyk:2024} once we have verified that the Fuchsian equation \eqref{Fuch-ev-A3} satisfies the coefficient assumptions (1)-(5) from Section~A.1 of \cite{BeyerOliynyk:2024}.
Thus, to complete the proof, we need to verify these assumptions. To this end,
we let $n$, $k$, $\sigma$, $T_0$, $\ep_0$, $\ep_1$, $\ep_2$ and $\kappat$ be chosen as in the statement of Proposition~\ref{prop:globalstability} and notice
that the matrix $\Pbb$ defined by \eqref{Pbb-def} satisfies
\begin{equation} \label{Pbb-props}
\Pbb^2 = \Pbb,  \quad  \Pbb^{\tr} = \Pbb, \quad \del{t}\Pbb =0 \AND \del{\Lambda} \Pbb =0.
\end{equation}

By \eqref{Bsc0-def}, it is clear that the constant matrix $\Bsc^0$ is positive definite and symmetric, and we note from \eqref{Pbb-def}, \eqref{eq:GtDef}, \eqref{eq:Gtprops}, \eqref{Ac-comm}, \eqref{Asc-def} and \eqref{eq:defBc} that
the matrix-valued map $\Bc(u)$ is smooth, that is
$\Bc \in C^\infty(\Rbb^{\udim},\Mbb{\udim})$, and
satisfies 
\begin{equation} \label{Bc(0)}
    [\Pbb,\Bc]=0 \AND 
    \Bc(0) = \Asc \Ac.
\end{equation}
From \eqref{Bc(0)} and the smooth dependence of $\Bc(u)$ on $u$, we have
\begin{equation} \label{Bc-exp}
    \Bc(u) = \Asc \Ac + \Ord(u).
\end{equation}
Now, fix
\begin{equation}\label{kappa-fix}
    \kappa \in (0,\kappat).
\end{equation} 
Then by \eqref{Bc-exp}
and Lemma \ref{lem:posdef1}, it follows from choosing $\eta$ small enough that there exist constants $\sigma_i>0$, $1\leq i\leq 10$, which we used to define $\Asc$ via \eqref{Asc-def}, and constants $\gamma_1>0$, $\gamma_2>0$ such that
\begin{equation}\label{Bsc0-Bc-bnds}
\frac{1}{\gamma_1}\id \leq \Bsc^0 \leq \frac{1}{\kappa} \Bc(u) \leq \gamma_2 \id
\end{equation}
for all $u\in B_R(\Rbb^{\udim})$ provided $R>0$ is chosen sufficiently small. We further observe from \eqref{BscLambda-def} that $\Bsc^\Lambda$ can be expressed as
\begin{equation} \label{Bsc-exp}
    \Bsc^\Lambda(t,u) = \frac{1}{t^{1-p}}\Bsc_0^\Lambda(t,u)
\end{equation}
where
\begin{equation*}
    \Bsc_0^\Lambda(t,u)=t^{1-\ep_0-\ep_2-p}\Asc A^\Lambda(t,u).
\end{equation*}
By \eqref{U-def}, \eqref{ALambda-def}, \eqref{eq:Ubgdef}, \eqref{u-def} and \eqref{Asc-def}, 
the matrices $\Bsc_0^\Lambda$ are symmetric, and
\begin{equation}\label{Bsc-smooth}
    \Bsc^{\Lambda}_0 \in C^0\Bigl([0,T_0],C^\infty
    \bigl(\Rbb^{\udim},\mathbb{M}_{\udim}\bigr)\Bigr)
\end{equation}
provided that
\begin{equation} \label{p-rest-A}
   p\leq  1-\ep_0 -\ep_2.
\end{equation}

It is also clear from \eqref{Hsc-def} that $\Hsc(t,u)$ can be written as
\begin{equation} \label{Hsc-rescale}
    \Hsc(t,u) = \frac{1}{t^{1-p}}\Hsc_0(t,u) 
\end{equation}
where
\begin{equation*}
\Hsc_0(t,u) = t^{1-p-\tilde{\ep}}\Asc H(t,u).
\end{equation*}
Then by \eqref{eq:SourcetermVan2}, we have
\begin{equation} \label{Hsc-vanish}
    \Hsc_0(t,0) =0
\end{equation}
and 
\begin{equation} \label{Hsc-smooth}
\Hsc_0\in C^0\bigl([0,T_0],C^\infty(\Rbb^{\udim})\bigr)
\end{equation}
provided $p$
satisfies 
\begin{equation} \label{p-rest-B}
   p\leq  1-\tilde{\ep}
\end{equation}
and $\tilde{\ep}$ is defined by \eqref{eq:epstildecond2}.

The final step needed to verify the coefficient assumptions from Section~A.1 in \cite{BeyerOliynyk:2024} is to analyse, c.f. item (4) from Definition~2.1 in \cite{BeyerOliynyk:2024},  the 
 map\footnote{Here, $D_u \Bc^\Lambda(t,u)$ denotes the partial derivative operator defined by $D_a\Bsc^\Lambda(t,u) v = \frac{d\;}{ds}\bigl|_{s=0}\Bsc^\Lambda(t,u+sv)$. } 
\begin{equation*}
\Div\Bsc(t,u,w):=\del{t}\Bsc^0+\del{\Lambda}(\Bsc^\Lambda(t,u))\bigl|_{w:=(w_\Lambda)=(\del{\Lambda} u)}
= D_u\Bsc^\Lambda(t,u) w_\Lambda.
\end{equation*}
 By \eqref{Bsc-exp}, it is clear that 
\begin{align} \label{divB-1}
\Div\Bsc(t,u,w)= \frac{1}{t^{1-p}} \Csc(t,u,w)
\end{align}
where, for some $\theta>0$, $\Csc(t,u,w)$ satisfies
\begin{equation} \label{divB-2}
|\Csc(t,u,w)| \leq \theta
\end{equation}
for all $(t,u,w)\in  [0,T_0]\times B_{R}(\Rbb^{\udim}) \times B_{R}\bigl(\Rbb^{(n-1)\times (\udim)}\bigr)$ provided that $p$ satisfies \eqref{p-rest-A}.
Now choosing $\kappa$ according to \eqref{kappa-fix}, which we note can be chosen as close to $\kappat$ as we like, we fix
$p$ by setting
\begin{equation} \label{p-def}
    p = \min\{\ep_0,\ep_2,1-3\ep_0-\ep_1,1-\ep_0-\ep_2\},
\end{equation}
and observe by \eqref{eq:epstildecond2} that it satisfies \eqref{p-rest-A} and \eqref{p-rest-B}.

We now claim that the coefficient assumptions \cite[\S A.1 (1)-(5)]{BeyerOliynyk:2024}
are satisfied.\footnote{Here, the spaces $\Zc_1$ and $\Zc_2$ that appear in Section A.1 of \cite{BeyerOliynyk:2024} are not needed and we omit them, which corresponds to taking them to be trivial, i.e. $\Zc_1=\Zc_2=\{0\}$.} Indeed, assumption \cite[\S A.1 (1)]{BeyerOliynyk:2024} is satisfied by \eqref{Pbb-props}, while assumption \cite[\S A.1 (2)]{BeyerOliynyk:2024} holds on account of \eqref{Bc(0)}, \eqref{Bc-exp}, \eqref{Bsc0-Bc-bnds}, and
the fact $\Bsc^0$ is a constant matrix that is symmetric and positive definite. Here, the constants $\kappa$, $\gamma_1$, 
$\gamma_2$ and $p$ defined by \eqref{kappa-fix}, \eqref{Bsc0-Bc-bnds} and \eqref{p-def} can be taken to be the same as the constants $\kappa$, $\gamma_1$, $\gamma_2$ and $p$ appearing in \cite[\S A.1]{BeyerOliynyk:2024}. 
Continuing on, assumption  \cite[\S A.1 (3)]{BeyerOliynyk:2024} holds by
virtue \eqref{Hsc-def}, \eqref{Hsc-vanish}, \eqref{Hsc-smooth} and \eqref{p-rest-B} (recall $\tilde{\ep}>0$ by assumption). Note
also that our source term $\Hsc$, which is called $F$ in \cite[\S A.1 (3)]{BeyerOliynyk:2024}, corresponds to 
a map $F$ that can be expanded as
$F=t^{-(1-p)}\Ft+t^{-(1-p)}F_0+
t^{-(1-\frac{p}{2})}F_1+t^{-1}F_2$
with $\Ft=F_1=F_2=0$ and $F_0=\Hsc_0$. Because of this, we can set the constants $\lambda_1$, $\lambda_2$ and $\lambda_3$ appearing in  \cite[\S A.1 (3)]{BeyerOliynyk:2024} to be zero.
We further observe that assumption
\cite[\S A.1 (4)]{BeyerOliynyk:2024} 
is satisfied as a consequence of
\eqref{Bsc-exp}, \eqref{Bsc-smooth}, \eqref{p-rest-A} (recall $\ep_0+\ep_2>0$ by assumption), and the symmetry of the matrices $\Bsc^\Lambda$. We also note that
the matrices $\Bsc^\Lambda$, which are denote by $B=(B^i)$ in \cite[\S A.1 (4)]{BeyerOliynyk:2024}, corresponds to 
a map $B$ that can be expanded as
$B=t^{-(1-p)}B_0+
t^{-(1-\frac{p}{2})}B_1+t^{-1}B_2$
with $B_1=B_2=0$ and $B_0=(\Bsc^\Lambda_0)$. As a consequence, we can take the constant $\alpha$ appearing in  \cite[\S A.1 (4)]{BeyerOliynyk:2024} to be zero. Finally, assumption
\cite[\S A.1 (5)]{BeyerOliynyk:2024} 
is satisfied due to \eqref{divB-1} and \eqref{divB-2}. Here, the constant $\theta$ can be taken to be the same as in \cite[\S A.1 (5)]{BeyerOliynyk:2024}
while the other constants, $\beta_i$, $i=0,1,\ldots,7$, from \cite[\S A.1 (5)]{BeyerOliynyk:2024} can be set to zero.

Given that the Fuchsian system \eqref{Fuch-ev-A3} verifies all the coefficient assumptions from  \cite[\S A.1]{BeyerOliynyk:2024}, we obtain
from an application\footnote{Since there is no $t^{-1}$ singular term in any of the coefficient matrices $\Bsc^\Lambda$, the constant $\btt$ appearing in theorem \cite[Thm~A.2]{BeyerOliynyk:2024} vanishes and the remark \cite[Rem~A.3]{BeyerOliynyk:2024} applies.} of Theorem~A.2 and  Remark~A.3 from
 \cite{BeyerOliynyk:2024} the existence of a constant
$\delta_0 > 0$ such that if $\norm{u_0}_{H^k(\mathbb \Tbb^{n-1})}< \delta_0$,
then there exists a unique solution 
\begin{equation*}
u \in C^0_b\bigl((0,t_0],H^k(\mathbb \Tbb^{n-1})\bigr)\cap C^1\bigl((0,t_0],H^{k-1}(\mathbb \Tbb^{n-1})\bigr)
\end{equation*}
of \eqref{Fuch-ev-A3} (and hence \eqref{eq:givp1}) that satisfies the initial
condition \eqref{eq:givp2}. Moreover, the limit $\lim_{\searrow 0} \Pbb^\perp u(t)$, denoted $\Pbb^\perp u(0)$, of this solution exists in $H^{k-1}(\mathbb \Tbb^{n-1})$. 

It is important to observe that the constant $\delta_0>0$ does not depend on the choice of $t_0\in (0,T_0]$. This is because, for any $t_0\in (0,T_0)$, the coefficient assumptions are obviously satisfied on $(0,t_0]$ with \textit{same choice of constants} since, as discussed above, they hold on the larger time interval $(0,T_0]$. Using the same reasoning,
 it is straightforward to verify that all the
 explicit and implicit constants from the estimates used to establish the proof of \cite[Theorem~A.2]{BeyerOliynyk:2024} (and
 \cite[Theorem~3.8]{BOOS:2021} on which it relies) can taken to be independent of $t_0 \in (0,T_0)$. Since the size of $\delta_0$ is
 determined by the constants, which are independent of $t_0$, it follows that the $\delta_0>0$ that yields existence on $(0,T_0]$ will also yield existence on any smaller interval $(0,t_0]$ for the same choice of $\delta_0$, which is, perhaps, not surprising. 
 As a final remark, we observe by \cite[Theorem~A.2]{BeyerOliynyk:2024} that $u$ satisfies the energy and decay estimates given by
\eqref{eq:resenergy} and \eqref{eq:PbbuDecay}-\eqref{eq:PbbuDecay2}, respectively.  

\subsection{Proof of Theorem~\ref{glob-stab-thm}}
\label{sec:proof_globstab}
Having established Proposition~\ref{prop:globalstability}, we can now use it to complete the proof of Theorem~\ref{glob-stab-thm}. We begin with fixing $n\in\Zbb_{\ge 3}$, $k \in \Zbb_{>(n+3)/2}$, $T_0>0$, $t_0\in (0,T_0]$, and $\sigma>0$. For a given $\delta>0$, we assume that the (synchronized) Einstein-scalar field initial data $\gr_{\mu\nu}\in H^{k+2}(\Tbb^{n-1},\mathbb{S}_n)$, 
$\ggr_{\mu\nu}\in H^{k+1}(\Tbb^{n-1},\mathbb{S}_n)$, $\taur=t_0$, and
$\taugr\in H^{k+2}(\Tbb^{n-1})$ satisfies
\begin{equation}\label{glob-stab-thm-idata-A-proof}
  \norm{\gr_{\mu\nu}-\eta_{\mu\nu}}_{H^{k+2}(\Tbb^{n-1})}
  +\norm{\ggr_{\mu\nu}}_{H^{k+2}(\Tbb^{n-1})}
  +\norm{\taugr-1}_{H^{k+2}(\Tbb^{n-1})}
  <\delta
\end{equation}
as well as the constraint equations \eqref{grav-constr}-\eqref{wave-constr} and the inequalities  $\det(\gr_{\mu\nu})<0$ and $|\vr|_{\gr}^2 <0$, where $\vr^\mu$ is defined by \eqref{vr-def}. 

Recalling that $\{\gr_{\mu\nu},\ggr_{\mu\nu},\taur=t_0,\taugr\}$ determines initial data $\{g_{\mu\nu}|_{\Sigma_{t_0}},\del{0}g_{\mu\nu}|_{\Sigma_{t_0}},\tau=t_0,\del{0}\tau =1\}$ for the metric $g_{\mu\nu}$ and scalar field $\tau$ in Lagrangian coordinates on $\Sigma_{t_0}$ via \eqref{dt-tau-idata}-\eqref{dt-g-idata},
we set 
\begin{equation*}
e_0^\mu=(-|\chi|_g^2)^{-\frac{1}{2}}\chi^\mu    
\end{equation*} 
and note that it can be computed from the Lagrangian initial data on $\Sigma_{t_0}$ by \eqref{chi-idata}. We further fix spatial frame initial data $e^\mu_I|_{\Sigma_{t_0}}=\delta^\mu_I
\er^\Lambda_\mu$, where the functions
$\er^\Lambda_I \in H^k(\Tbb^{n-1})$ are assumed to satisfy
\begin{equation*} 
\norm{\er^\Lambda_I-\delta^\Lambda_I}_{H^{k}(\Tbb^{n-1})} < \delta
\end{equation*}
and make the frame $e^\mu_i$ orthonormal on $\Sigma_{t_0}$ with respect to the metric $g_{\mu\nu}$ given there, see \eqref{g-idata}. 

\bigskip

\noindent \underline{Fuchsian stability:} To proceed, we assume that $\ep_0$, $\ep_1$, and $\ep_2$ satisfy \eqref{eq:epscond2}, and $p$ and $\kappat$ are defined by \eqref{eq:defkappat}. Then by the discussion in Section \ref{frame-idata}, the variable definitions \eqref{k-def}-\eqref{tauac-def}, \eqref{U-def}, and \eqref{eq:Ubgdef}-\eqref{u-def},  it is not difficult, with the help of the calculus inequalities from Appendix \ref{calc},
to verify from \eqref{glob-stab-thm-idata-A-proof} that initial data $\{\gr_{\mu\nu},\ggr_{\mu\nu},\taur=t_0,\taugr,\er^\Lambda_I\}$ determines initial data $u_0\in H^k(\Tbb^{n-1})$ for the Fuchsian equation \eqref{eq:givp1} satisfying
\begin{equation} \label{glob-stab-thm-idata-D}
    \norm{u_0}_{H^k(\Tbb^{n-1})} \leq C_0\delta,
\end{equation}
where 
\begin{equation}\label{C0-def}
    C_0=C_0(t_0,t_0^{-1},\delta).
\end{equation}
Then, by Proposition \ref{prop:globalstability} and the bound \eqref{glob-stab-thm-idata-D}, it follows that there exists a $\delta_0>0$, such that if we choose $\delta>0$ small enough so that
\begin{equation} \label{C0delta<delta0}
    C_0\delta <\delta_0, 
\end{equation}   
which we can always do for any given $t_0\in (0,T_0]$,
then there exists a solution    
\begin{equation}\label{eq:symhypreg_appl}
u \in C^0_b\bigl((0,t_0],H^k(\Tbb^{n-1})\bigr)\cap C^1\bigl((0,t_0],H^{k-1}(\Tbb^{n-1})\bigr)
\end{equation}
of the Fuchsian GIVP \eqref{eq:givp1}-\eqref{eq:givp2} that extends continuously in $H^{k-1}(\Tbb^{n-1})$ to $t=0$
and
satisfies the energy
\begin{equation} \label{eq:resenergy_appl}
\norm{u(t)}_{H^k(\mathbb T^{n-1})}^2 + \int^{t_0}_t \frac{1}{s} \norm{\Pbb u(s)}_{H^k(\mathbb T^{n-1})}^2\, ds  \lesssim \norm{u_0}^2
\end{equation}
and decay estimates
\begin{align}
\label{eq:PbbuDecay_appl}
  \norm{\Pbb u(t)}_{H^{k-1}(\mathbb T^{n-1})} &\lesssim t^p+t^{\kappat-\sigma},\\
\label{eq:PbbuDecay2_appl}
\norm{\Pbb^\perp u(t) - \Pbb^\perp u(0)}_{H^{k-1}(\mathbb T^{n-1})} &\lesssim
 t^{p}+ t^{2(\kappat-\sigma)},
\end{align}
for all $t\in (0,t_0]$, where the implied constant in the above estimates are independent of $\delta$ provided it is chosen small enough to satisfy \eqref{C0delta<delta0}.

By definition $\Pbb^\perp=\id-\Pbb$, and so we have by \eqref{Pbb-def} that
  \begin{equation}\label{Pbb-perp-def}
  \Pbb^\perp = \diag\Bigl(\delta_{\Lt}^L\delta_{\Mt}^M,0,0,0,0,0,0,0,0,0\Bigr).
\end{equation}
The existence of the limit $\Pbb^\perp u(0)=\lim_{t\searrow 0} \Pbb^\perp u(t)$ in $H^{k-1}(\Tbb^{n-1})$ implies that
\begin{equation} \label{Pbb-perp-u(0)}
\Pbb^\perp u(0)=\bigl(\kf_{IJ},0,0,0,0,0,0,0,0,0\bigr)
\end{equation}
for some $\kf_{IJ}\in H^{k-1}(\mathbb T^{n-1},\Sbb{n-1})$.
The useful estimate
\begin{equation}
  \label{eq:kijsmallness}
  \norm{\kf_{IJ}}_{H^{k-1}(\mathbb T^{n-1})}\le C_0 \delta
\end{equation}
is a direct consequence of
\eqref{eq:resenergy_appl} and \eqref{glob-stab-thm-idata-D}.

In the following, a variable that is independent of
the parameters
\begin{equation*}
    \epv=(\ep_0,\ep_1,\ep_2)
\end{equation*}
will be said to be \textit{$\epv$-independent}. These
variables will be either equal to or a limit of a variable that is determined from the  
the solution of an initial value problem that
is independent of $\ep_0,\ep_1,\ep_2$.
For example, $\kf_{IJ}$ is determined from a limit of
$k_{IJ}= t\betat \kt_{IJ}$ by \eqref{k-def}, \eqref{U-def}, \eqref{u-def} and \eqref{Pbb-perp-def}.
Noting that $\betat$ and $\kt_{IJ}$
satisfy the $\epv$-independent system of equations
given by \eqref{for-G.1.S2}-\eqref{for-N.S2}, which are equivalent to the Fuchsian equation \eqref{eq:givp1}, and that the initial 
data for this system is $\epv$-independent, it follows that $\kf_{IJ}$ does not depend on the choice of
the parameters $\ep_0$, $\ep_1$ and $\ep_2$, and hence, is
$\epv$-independent.

\bigskip

\noindent \underline{Frame stability and limits:} The frame $e_i^\mu$, which is $\epv$-independent, is obtained from the components of $u$ via \eqref{beta-def}, \eqref{f-def}, \eqref{U-def}, \eqref{eq:Ubgdef} and \eqref{u-def} and from the relations 
\begin{equation} \label{e0i-fix}
    e_0^\mu=\betat^{-1}\delta_0^\mu \AND e_I^0=0, 
\end{equation}
which we note hold by \eqref{e0-mu} and \eqref{e0I-fix}.
Since $\betat$ is initially positive at $t=t_0$
and the evolution equation \eqref{for-O.1.S2} for $\betat$ implies that $\betat$ cannot cross zero, it follows $\betat$
must continue to be positive on $M_{0,t_0}$. 
While \eqref{eq:PbbuDecay_appl} yields an estimate for $\betat$, we can obtain a stronger estimate by
noting that \eqref{for-O.1.S2} can be expressed as
  \begin{equation*}
    \del{t}\bigl(t^{-\frac{1}{2}\kf_{J}{}^J}\betat\bigr) = \Bigl(-\beta^2 t^{-\ep_0-\ep_1}\xi_{00}+\frac{1}{2} t^{-1} (k_{J}{}^J-\kf_{J}{}^J)\Bigr) t^{-\frac{1}{2}\kf_{J}{}^J}\betat.
  \end{equation*}
Since $\betat>0$, we can integrate this in time to get
  \begin{equation*}
    \ln\bigl(t^{-\frac{1}{2}\kf_{J}{}^J}\betat(t)\bigr)
    -\ln\bigl({\tilde t}^{-\frac{1}{2}\kf_{J}{}^J}\betat(\tilde t)\bigr)
    = \int_{\tilde t}^t\Bigl(-s^{-\ep_0-\ep_1}
    \beta^2 \xi_{00}(s)
       +\frac{1}{2} s^{-1} (k_{J}{}^J(s)-\kf_{J}{}^J)\Bigr)ds,    
  \end{equation*}
which holds for all $t,\tilde t\in (0,t_0]$. From the  calculus inequalities from Appendix \ref{calc}, we have
  \begin{align*}
    \bnorm{\ln\bigl(t^{-\frac{1}{2}\kf_{J}{}^J}\betat(t)\bigr)
    -\ln\bigl({\tilde t}^{-\frac{1}{2}\kf_{J}{}^J}\betat(\tilde t)\bigr)}_{H^{k-1}(\Tbb^{n-1})}
    \le& \int_{\tilde t}^t\Bigl(s^{-\ep_0-\ep_1}\norm{\beta(s)}_{H^{k-1}(\Tbb^{n-1})}^2 \norm{\xi_{00}(s)}_{H^{k-1}(\Tbb^{n-1})}\\
       &\qquad+\frac{1}{2} s^{-1} \norm{k_{J}{}^J(s)-\kf_{J}{}^J}_{H^{k-1}(\Tbb^{n-1})}\Bigr)ds
  \end{align*}
  provided $0<\tilde t\le t\le t_0$. Using the estimates \eqref{eq:PbbuDecay_appl}, \eqref{eq:PbbuDecay2_appl} and \eqref{eq:kijsmallness} along with the fact that $p$ and $\kappat$ are positive and that $\sigma$ can be further shrunk if necessary, we conclude from the above estimate that $\ln(t^{-\frac{1}{2}\kf_{J}{}^J}\betat(t))$ converges in $H^{k-1}(\Tbb^{n-1})$ as $t\searrow 0$ to a limit, denoted $\ln(\mathfrak{b})$, and satisfies 
\begin{equation*}
    \bnorm{\ln\bigl(t^{-\frac{1}{2}\kf_{J}{}^J}\betat(t)\bigr)
    -\ln(\mathfrak{b})}_{H^{k-1}(\Tbb^{n-1})}
    \lesssim t^{1-\ep_0-\ep_1}(t^p+t^{\kappat-\sigma})^3
       + t^{p}+ t^{2(\kappat-\sigma)}    
\end{equation*}
for $0<t\leq t_0$.
From this inequality, it follows that the limit $\mathfrak{b}$ is a strictly positive function in $H^{k-1}(\Tbb^{n-1})$, and since $\kf{}_I{}^J$ and $\betat$ are $\epv$-independent, the same is true for $\mathfrak{b}$. By choosing $\sigma$ sufficiently small  and recalling that $1-\ep_0-\ep_1>0$, we note that the above estimate can be simplified to
\begin{equation}
       \label{eq:improvbetaestimate.pre}
       \bnorm{\ln\bigl(t^{-\frac{1}{2}\kf_{J}{}^J}\betat(t)\bigr)
    -\ln(\mathfrak{b})}_{H^{k-1}(\Tbb^{n-1})}
  \lesssim t^{p}+ t^{2(\kappat-\sigma)}.   
\end{equation}

The left hand side of the inequality \eqref{eq:improvbetaestimate.pre} is  $\epv$-independent while the right hand side of it is not, and consequently, we can use our freedom to choose $\ep_0$, $\ep_1$, $\ep_2$ to obtain an optimal decay rate on the right side of \eqref{eq:improvbetaestimate.pre}.
Assuming that $\sigma>0$ is sufficiently small, it follows from \eqref{eq:defkappat} that the exponent on the right side of \eqref{eq:improvbetaestimate.pre} is given by 
\begin{equation}
  \label{eq:defxi1}
  \xi_1:=\min\{p,2(\kappat-\sigma)\}
  =\min\bigl\{\xit_1,1-\ep_0-\ep_2,\ep_2 \bigr\}
\end{equation}
where $\xit_1$ is defined by
\begin{equation}\label{eq:defxit1}
\xit_1=\min\biggl\{1-3\ep_0-\ep_1,\ep_0, 2\biggl(\ep_1-\frac{1}{1+\sqrt{3}}-\sigma\biggr)\biggr\}.
\end{equation}
Both $\xi_1$ and $\xit_1$ are continuous piecewise linear functions of $\ep_0$, $\ep_1$ and $\ep_2$ on the closure $\bar E$ of the open domain 
\begin{equation} \label{E-def}
    E=\biggl\{ \, (\ep_0,\ep_1,\ep_2)\in \Rbb^3\, \biggl| \, 
    0<\ep_0,\quad \frac{1}{1+\sqrt{3}}<\ep_1,\quad 3\ep_0+\ep_1<1, \quad  0<\ep_2<1-\ep_0. \, \biggr\}.
\end{equation}
By definition $\xi_1\le\xit_1$ everywhere on $\bar E$. We also find easily that $\xit_1$ attains its absolute maximum value $\frac 29\Bigl(1-\frac{1}{1+\sqrt{3}}-\sigma\Bigr)$ on $\bar{E}$ for
\begin{equation}
  \label{eq:9dufsljdsf}
  \ep_0=\frac 29\Bigl(1-\frac{1}{1+\sqrt{3}}-\sigma\Bigr),\quad \ep_1=\frac 89\Bigl(\frac{1}{8}+\frac{1}{1+\sqrt{3}}+\sigma\Bigr).
\end{equation}
Since, $\xi_1=\xit_1$ when $\ep_0$ and $\ep_1$ are given by \eqref{eq:9dufsljdsf} and 
\begin{equation} \label{eq:9dufsljdsf-a}
    \ep_2=\ep_0,
\end{equation}
we conclude that 
\begin{equation}\label{xi1-max}
    \frac 29\Bigl(1-\frac{1}{1+\sqrt{3}}-\sigma\Bigr)
\end{equation} 
is the absolute maximum value of $\xi_1$ over $\bar E$, which is achieved in $E$.
Given these parameter choices, the estimate \eqref{eq:improvbetaestimate-glob} is then a direct consequence of
\eqref{eq:improvbetaestimate.pre}.

Considering now the spatial frame components $e^\Lambda_I$, we see, with the help of \eqref{k-def}, \eqref{beta-def} and \eqref{psi-def},
that the evolution equation \eqref{for-M.1.S2} for the spatial frame components can be expressed in matrix form as
\begin{equation*} 
\del{t} e^\Lambda = \frac{1}{t}(\Ksc + t\Lsc)e^\Lambda    
\end{equation*}
where 
\begin{gather} \label{frame-lim-A}
  e^\Lambda=(e^\Lambda_I), \quad    \Ksc=\Bigl(-\frac{1}{2}\delta^{JL}k_{IL}(0)\Bigr)
  \intertext{and}
 \Lsc =-\Bigl(\frac{1}{2}t^{-1}\delta^{JL}(k_{IL}-k_{IL}(0))+t^{-\ep_1+k_{J}{}^J(0)/2} \delta^{JL}(t^{-k_{J}{}^J(0)/2}\betat)\psi_{I}{}^0{}_L\Bigr). \label{frame-lim-A2}
\end{gather}
Letting\footnote{Given a square matrix $A$, we frequently use the notation $t^A$ instead of $\exp(\ln(t)A)$.}
\begin{equation} \label{frame-lim-B} 
\ec^\Lambda = (\ec^\Lambda_\mu) :=  t^{-\Ksc} e^\Lambda,
\end{equation}
a short calculation then shows that $\ec^\Lambda$ satisfies
\begin{equation} \label{frame-lim-C}
    \del{t} \ec^\Lambda = \Msc \ec^\Lambda
\end{equation}
where
\begin{equation*} \label{frame-lim-D}
    \Msc = t^{-\Ksc}\Lsc t^{\Ksc}.
\end{equation*}
By differentiating \eqref{frame-lim-C} repeatedly in space, we obtain from standard $L^2$-energy estimates and the calculus inequalities the differential energy inequality
\begin{equation*}
    \del{t}\norm{\ec^\Lambda(t)}_{H^{k-1}(\Tbb^{n-1})}^2 \lesssim
    \norm{\Msc(t)}_{H^{k-1}(\Tbb^{n-1})}\norm{\ec^\Lambda(t)}^2_{H^{k-1}(\Tbb^{n-1})},
  \end{equation*}
 which in turn, yields
 \begin{equation*}
    \del{t}\norm{\ec^\Lambda(t)}_{H^{k-1}(\Tbb^{n-1})} \lesssim
    \norm{\Msc(t)}_{H^{k-1}(\Tbb^{n-1})}\norm{\ec^\Lambda(t)}_{H^{k-1}(\Tbb^{n-1})}.
  \end{equation*}
Applying  Gr\"onwall's lemma to this differential inequality gives
  \[
    \norm{\ec^\Lambda(t)}_{H^{k-1}(\Tbb^{n-1})}\lesssim \norm{\ec^\Lambda(t_0)}_{H^{k-1}(\Tbb^{n-1})}e^{-\frac 12\int_{t_0}^t \norm{\Msc(s)}_{H^{k-1}(\Tbb^{n-1})}ds},\quad 0<t\leq t_0.
  \]
Integrating \eqref{frame-lim-C} in time, we see, with the help of the above inequality, that
\begin{equation}
  \label{eq:frame_convest}
  \norm{\ec^\Lambda(t)-\ec^\Lambda(\tilde t)}_{H^{k-1}(\Tbb^{n-1})}    
  \lesssim \norm{\ec^\Lambda(t_0)}_{H^{k-1}(\Tbb^{n-1})}e^{\frac
    12\int_{t_0}^0
    \norm{\Msc(s)}_{H^{k-1}(\Tbb^{n-1})}ds}\int_{\tilde t}^t
  \norm{\Msc(s)}_{H^{k-1}(\Tbb^{n-1})}ds
\end{equation}
for all $0<\tilde t\le t\le t_0$.

Next by \eqref{eq:PbbuDecay_appl}, \eqref{eq:PbbuDecay2_appl} and \eqref{frame-lim-A2}, we observe that the matrix $\Lsc$ is bounded by
\begin{equation} \label{frame-lim-G}
    \norm{\Lsc(t)}_{H^{k-1}(\Tbb^{n-1})} \lesssim t^{-1}(t^{p}+ t^{2(\kappat-\sigma)})+t^{-\ep_1}(t^p+t^{\kappat-\sigma}) 
  \end{equation}
for all $0<t\leq t_0$. Also, since $t^{\pm \Ksc}=e^{\pm \ln(t)\Ksc}$, the estimate
\begin{equation} \label{frame-lim-H}
\norm{t^{\pm\Ksc}}_{H^{k-1}(\Tbb^{n-1})}\le C e^{C \norm{\ln(t)\Ksc}_{H^{k-1}(\Tbb^{n-1})}} \lesssim t^{-C \norm{\kf_{IJ}}_{H^{k-1}(\Tbb^{n-1})}}, \quad 0<t\leq t_0,
\end{equation}
is a direct consequence of the analyticity of the exponential $e^{X}$, the definition \eqref{frame-lim-A}
of $\Ksc$, and the fact that $H^{k-1}(\Tbb^{n-1})$ is a Banach algebra by virtue of the assumption that  $k-1>(n+1)/2$.  
By \eqref{frame-lim-D}, \eqref{frame-lim-G}, and \eqref{frame-lim-H}, we can then bound the matrix
$\Msc$ by
\begin{equation}
  \label{frame-lim-Last}
\norm{\Msc(t)}_{H^{k-1}(\Tbb^{n-1})} \lesssim \Bigl(t^{-1}(t^{p}+ t^{2(\kappat-\sigma)})+t^{-\ep_1}(t^p+t^{\kappat-\sigma})\Bigr) t^{-C \norm{\kf_{IJ}}_{H^{k-1}(\Tbb^{n-1})}}
\end{equation}
for all $0<t\leq t_0$. By choosing $\delta$ and $\sigma$ sufficiently small, it follows from \eqref{eq:epscond2}, \eqref{eq:defkappat} and \eqref{frame-lim-Last} that $\Msc$ is $H^{k-1}(\Tbb^{n-1})$-integrable in time, that is, $\int_{0}^{t_0}\norm{\Msc(s)}_{H^{k-1}(\Tbb^{n-1})}\, ds <\infty$. We then deduce from this integrability and the inequality \eqref{eq:frame_convest} that $\ec^\Lambda(t)$ converges as $t\searrow 0$ in $H^{k-1}(\Tbb^{n-1})$ to a limit, which we will denote by $\ef^\Lambda_I$. The estimate \eqref{eq:improvframeestimate} will be established below.

\bigskip

\noindent \underline{Conformal Einstein-scalar field stability:} Proposition~\ref{lag-exist-prop} implies, for some $t_1\in (0,t_0]$, the existence of a unique solution $W$ with regularity \eqref{eq:Wreg} on $M_{t_1,t_0}=(t_1,t_0]\times \Tbb^{n-1}$ of 
the system
\eqref{tconf-ford-C.1}-\eqref{tconf-ford-C.8} that satisfies
the initial conditions \eqref{l-idata}-\eqref{hhu-idata}. Moreover, since the conformal Einstein-scalar field initial data is assumed to satisfy the gravitational and wave gauge constraint equations, it follows from  Proposition~\ref{lag-exist-prop} that this solution satisfies the constraints \eqref{eq:Lag-constraints}, and determines a solution
$\{g_{\mu\nu},\tau\}$
of the conformal Einstein scalar field equations \eqref{lag-confeqns} in Lagrangian coordinates satisfying the wave gauge constraint
\eqref{lag-wave-gauge} and where $\tau$ is
given by
\begin{equation} \label{tau=t}
\tau=t.
\end{equation}
By construction of the Fuchsian system \eqref{eq:givp1}, the solution $W$ determines a $\epv$-dependent solution $\ut$ on $M_{t_1,t_0}$ of the Fuchsian IVP \eqref{eq:givp1}-\eqref{eq:givp2} with the same $\epv$-dependent initial data $u_0$ as above. By the uniqueness statement of Proposition~\ref{prop:globalstability}, we conclude
that 
\begin{equation} \label{ut=u}
    \ut = u|_{M_{t_1,t_0}}
\end{equation}
provided the parameters $\ep_0$, $\ep_1$, and $\ep_2$ are chosen to be the same for both solutions. 

From the equality \eqref{ut=u}, it follows that the energy estimate \eqref{eq:resenergy_appl} together with the Sobolev inequality yields the bound
\begin{equation*} 
    \sup_{t_1<t<t_0}\norm{\ut(t)}_{W^{2,\infty}(\Tbb^{n-1})} < \infty.
\end{equation*}
From this bound on $\ut$, we see, with the help of  \eqref{p-fields}, \eqref{gi00},   \eqref{kt-def}-\eqref{tau-def}, \eqref{psit-def}, \eqref{k-def}-\eqref{xi-def}, \eqref{psi-def}, \eqref{f-def}, \eqref{eq:Ubgdef}-\eqref{u-def} and \eqref{e0i-fix}, that 
\begin{align}
\sup_{t_1<t<t_0}\Bigl(\norm{e_j^\mu(t)}_{W^{2,\infty}(\Tbb^{n-1})}&+\norm{\Dc_i g_{jk}(t)}_{W^{2,\infty}(\Tbb^{n-1})}+
\norm{\betat(t)}_{W^{2,\infty}(\Tbb^{n-1})}+
\norm{ \kt_{IJ}(t)}_{W^{2,\infty}(\Tbb^{n-1})}\notag \\
+&
\norm{\psit_I{}^k{}_J(t)}_{W^{2,\infty}(\Tbb^{n-1})}+\norm{\gamma_I{}^k{}_J(t)}_{W^{2,\infty}(\Tbb^{n-1})}
+\norm{\Dc_i\Dc_j\tau}_{W^{2,\infty}(\Tbb^{n-1})}\Bigr)< \infty.
\label{glob-cont-bnd-A}
\end{align}
Using this bound, it then follows from the evolution equations
\eqref{for-O.1.S2}-\eqref{for-M.1.S2} for $\betat$ and $e^\Lambda_I$, and Lemma \ref{diff-matrix-lem} that
\begin{equation} \label{glob-cont-bnd-B}
    \inf_{M_{t_1,t_0}}\bigr\{\betat,\det(e^\Lambda_I)\bigl\} > 0,
\end{equation}
which in turn, implies via
 \eqref{e0i-fix} that
\begin{equation}\label{glob-cont-bnd-C}
    \inf_{M_{t_1,t_0}}\det(e^\mu_j) > 0.
\end{equation}

Since the frame $e_i^\mu$ is orthonomal by construction, the components of the conformal metric in the Lagrangian coordinates are determined by 
\begin{equation} \label{g-Lag-components}
    g_{\mu\nu}=e_\mu^i \eta_{ij} e_\nu^j,
\end{equation}
and so by \eqref{glob-cont-bnd-A}, we have
\begin{equation} \label{glob-cont-bnd-D}
\sup_{t_1<t<t_0}\norm{g_{\mu\nu}(t)}_{W^{2,\infty}(\Tbb^{n-1})} < \infty.
\end{equation}
From the calculation\footnote{Here, we using $\Ld_X$ to denote the Lie derivative along a vector field $X$. }
\begin{align*}
    e^\mu_i e^\nu_j \del{t}g_{\mu\nu} &= \Ld_{\del{t}}g_{ij}
    \oset{\eqref{e0i-fix}}{=}\Ld_{\betat e_0}g_{ij}\\
    &=\betat e_0^k \Dc_{k}g_{ij}+\Dc_i (\betat e_0^k)\eta_{kj}
    + \Dc_j (\betat e_0^k)\eta_{jk} \\
    & =\betat \delta_0^k \Dc_{k}g_{ij}+ e_i(\betat)\eta_{0j}+\betat \gamma_i{}^k{}_0 \eta_{kj}
    + e_j(\betat)\eta_{0i}+\betat \gamma_j{}^k{}_0\eta_{ik},
\end{align*}
we also observe that the bound
\begin{equation} \label{glob-cont-bnd-E}
\sup_{t_1<t<t_0}\norm{\del{t}g_{\mu\nu}(t)}_{W^{1,\infty}(\Tbb^{n-1})} < \infty
\end{equation}
is a direct consequence
\eqref{glob-cont-bnd-A}-\eqref{glob-cont-bnd-C},
the relations  \eqref{p-fields}, \eqref{gamma-000}-\eqref{gamma-0K0}, \eqref{gamma-I00}-\eqref{gamma-IJ0} and \eqref{e0i-fix}, and the evolution
equation \eqref{for-O.1.S2}. Moreover, by employing similar arguments, it is also not difficult to 
verify that $\Dc_\nu \chi^\mu$, where $\chi^\mu$ is defined by \eqref{chi-def}, satisfies
\begin{equation} \label{glob-cont-bnd-F}
\sup_{t_1<t<t_0}\Bigl(\norm{\Dc_\nu \chi^\mu(t)}_{W^{2,\infty}(\Tbb^{n-1})}+\norm{\del{t}(\Dc_\nu \chi^\mu)(t)}_{W^{1,\infty}(\Tbb^{n-1})}\Bigr) < \infty.
\end{equation}

Together,  \eqref{g-Lag-components} and \eqref{glob-cont-bnd-D}-\eqref{glob-cont-bnd-F} imply that the solution $W$ satisfies the continuation criteria \eqref{eq:cont_crit1}-\eqref{eq:cont_crit2}. Hence, by Proposition \ref{lag-exist-prop}, the solution $W$ can be continued beyond $t_1$. This implies that $t_1=0$ and the solution $W$ exists on $M_{0,t_0}$. This solution continues to satisfy the constraints \eqref{eq:Lag-constraints} and determine a solution
$\{g_{\mu\nu},\tau\}$
of the conformal Einstein scalar field equations \eqref{lag-confeqns} in Lagrangian coordinates that verifies the wave gauge constraint \eqref{lag-wave-gauge}.

\bigskip

\noindent \underline{Second fundamental form estimate:} With the help of \eqref{alpha-def}, \eqref{e0I-fix} and \eqref{tau=t},  it follows easily from the formula \eqref{g-Lag-components} for the conformal metric that on the $t=const$-hypersurfaces  the lapse $\Ntt$ is given by
\begin{equation} \label{Ntt=betat}
    \Ntt=\betat,
\end{equation} the shift vanishes, and the induced $(n-1)$-spatial metric is $\gtt_{\Lambda\Omega}=g_{\Lambda\Omega}$. Furthermore, by \eqref{Ktt-def}, \eqref{kt-def}, \eqref{k-def}, and \eqref{psi-def}, the 
second fundamental form induced on the $t=const$-hypersurfaces by the conformal metric  is 
\begin{equation*}
  \Ktt_{LJ}
  =\frac 12t^{-1}\betat^{-1} k_{LJ}+t^{-\ep_1} \psi_{(L}{}^0{}_{J)}.
\end{equation*}       
Using \eqref{beta-def}, we can express this as
\begin{equation*}
  2t\betat \Ktt_{LJ}=k_{LJ}
  +2t^{1-\ep_1-\ep_0}\beta \psi_{(L}{}^0{}_{J)}
\end{equation*}
and the bound
\begin{equation*}
  \norm{ 2t\betat \Ktt_{LJ}(t)-\kf_{LJ}}_{H^{k-1}(\mathbb T^{n-1})} 
  \lesssim t^{p}+ t^{2(\kappat-\sigma)} 
  +t^{1-\ep_0-\ep_1}(t^p+t^{\kappat-\sigma})^2 \lesssim t^{p}+ t^{2(\kappat-\sigma)}
\end{equation*}
is then a direct consequence of the above expression for the second fundamental form,  the energy and decay estimates \eqref{eq:PbbuDecay_appl}-\eqref{eq:PbbuDecay2_appl}, and the definitions of the constants $p$, $\kappat$ given by \eqref{eq:defkappat}.
Since the right hand side of that above estimate equals that of \eqref{eq:improvbetaestimate.pre}, the same arguments used to obtain the optimal decay exponent that lead to \eqref{eq:improvbetaestimate-glob} yield the estimate \eqref{eq:2ndFF.est}.

\bigskip

\noindent \underline{Asymptotic pointwise Kasner property:} Since $\{g_{\mu\nu},\tau\}$ is a solution of the conformal Einstein-scalar field equations, it follows from \eqref{eq:conf2phys} and \eqref{tau=t} that the pair
\begin{equation}
  \label{eq:conf2phys-a}
  \biggl\{\gb_{ij}=t^{\frac {2}{n -2}}g_{ij},\,\phi=\sqrt{\frac{n-1}{2(n-2)}}\ln(t)\biggr\}
\end{equation}
 determines a solution of the physical Einstein-scalar field system \eqref{ESF.1}-\eqref{ESF.2}.  As a consequence, the $n-1$-spatial metric
  \begin{equation}
    \label{eq:spatmetrrel}
    \bar{\gtt}_{\Lambda\Omega}=t^{\frac {2}{n -2}}\gtt_{\Lambda\Omega}
  \end{equation}
  and second fundamental form 
  \begin{equation}
    \label{eq:phys2ndff}
    \bar{\Ktt}_{\Lambda\Omega}=t^{\frac {1}{n -2}}\Bigl(\Ktt_{\Lambda\Omega}
    +\frac {1}{n -2}t^{-1}\betat^{-1}{\gtt}_{\Lambda\Omega}\Bigr),
  \end{equation}
  induced by the physical metric $\gb_{\mu\nu}$ on the$t=const$-hypersurfaces  must satisfy the Hamiltonian constraint, which we write in the rescaled form
  \begin{equation}
         \label{eq:Hamilton1.1}
         t^{\frac {2}{n -2}}\betat^{2}t^{2}\overline\Rtt+t^{\frac {2}{n -2}}\betat^{2}t^{2}\bigl((\bar{\Ktt}_\Lambda{}^\Lambda)^2-\bar{\Ktt}_\Lambda{}^\Sigma \bar{\Ktt}_\Sigma{}^\Lambda\bigr)-\Tb=0
  \end{equation}
  where $\overline\Rtt$ is the scalar curvature of the spatial metric $\bar{\gtt}_{\Lambda\Omega}$ and
  \begin{align}
    \Tb&= -2t^{\frac{2}{n-2}}\betat^2 t^2 \frac{\nablab_i t \nablab_j t }{|\nablab t|^2_{\gb}} \Tb^{ij} \notag \\
    &=2\betat^{2}t^{2}\Bigl((\nabla_0\phi)^2+\delta^{IJ}\nabla_I\phi\nabla_J\phi\Bigr) && \text{(by \eqref{tau=t}, \eqref{g-Lag-components}  \& \eqref{eq:conf2phys-a})}\notag \\
    &=\frac {n-1}{n-2} && \text{(by \eqref{e0i-fix} \& \eqref{eq:conf2phys-a})}\notag  
  \end{align}
  is the rescaled double normal component of the scalar field energy momentum tensor
  $\Tb^{ij} = 2\nablab^i \phi \nablab^j\phi -|\nablab \phi|_{\gb}^2 \gb^{ij}$. 
  We also find after a straightforward calculation that
  \begin{equation}
    \label{eq:HamSecFFTrafo}
    (\bar{\Ktt}_\Lambda{}^\Lambda)^2-\bar{\Ktt}_\Lambda{}^\Sigma \bar{\Ktt}_\Sigma{}^\Lambda 
    =
       t^{-\frac {2}{n -2}}\Bigl( (\Ktt _{I}{}^{I})^2
       -\Ktt _{I}{}^{J}\Ktt _{J}{}^{I}
       +2\Ktt _{I}{}^{I}t^{-1}\betat^{-1}
       +\frac {n-1}{n -2}t^{-2}\betat^{-2}                                                   
       \Bigr)
  \end{equation}
  as a consequence of \eqref{eq:spatmetrrel} and \eqref{eq:phys2ndff}. The rescaled Hamiltonian constraint \eqref{eq:Hamilton1.1} therefore takes the form
  \begin{equation}
         \label{eq:Hamilton1}
         t^{\frac {2}{n -2}}\betat^{2}t^{2}\overline\Rtt
         +\betat^{2}t^{2}\Bigl( (\Ktt _{I}{}^{I})^2
       -\Ktt _{I}{}^{J}\Ktt _{J}{}^{I}
       +2\Ktt _{I}{}^{I}t^{-1}\betat^{-1}                                                         
       \Bigr)
         =0.
  \end{equation}
  
  Since the conformal factor $t^{\frac {2}{n -2}}$ in \eqref{eq:spatmetrrel} is constant on the $t$=const-surfaces, it follows from  \eqref{tau=t} and \eqref{eq:spatmetrrel} that
  \begin{equation*}
  t^{\frac {2}{n -2}}\overline\Rtt=\Rtt
  \end{equation*} where $\Rtt$ is the scalar curvature of the $n-1$-spatial conformal metric $\gtt_{\Lambda\Omega}$.
  Noting that 
 \begin{equation*} 
 \Rtt
    = e_J(\Gamma_I{}^J{}_K)\delta^{IK}
    - e_I(\Gamma_J{}^J{}_K) \delta^{IK}
    +\delta^{IK}\Gamma_I{}^M{}_K\Gamma_J{}^J{}_M
    -\Gamma_J{}^M{}_K \delta^{IK}\Gamma_I{}^J{}_M
+2 \delta^{IK}\Gamma_{[I}{}^M{}_{J]}\Gamma_M{}^J{}_K,
\end{equation*}
where the $\Gamma_M{}^J{}_K$ are the spatial components of the connection coefficients of the conformal metric $g_{\mu\nu}$ with respect to the frame $e_i^\mu$,
we find from \eqref{Ccdef}, \eqref{p-fields}-\eqref{d-fields},  \eqref{kt-def}-\eqref{mt-def}, \eqref{gt-def}, \eqref{psit-def}, \eqref{k-def}-\eqref{m-def}, \eqref{psi-def}-\eqref{gac-def}, \eqref{e0i-fix}, the formula
\begin{align*}
g_{IJKL} = e_I(g_{JKL})-\gamma_I{}^0{}_Jg_{0KL}-\gamma_I{}^M{}_J g_{MKL}-\gamma_I{}^0{}_Kg_{J0L}-\gamma_I{}^M{}_K g_{JML}-\gamma_I{}^0{}_Lg_{JK0}-\gamma_I{}^M{}_K g_{JKM}
\end{align*}
for the covariant derivative $g_{IJKL}=\Dc_I g_{JKL}=\Dc_I\Dc_J g_{KL}$,  and a straightforward calculation that
\begin{align*}
  \Rtt
  =& e*\del{}\psit + \gt +(\ellt+\psit)* (\kt+\psit+\ellt) \hspace{2.0cm} \text{( $e=(e^\Lambda_I)$ \& $\del{}=(\del{\Lambda})$ )}\\
  =& t^{-\ep_2-\ep_1}f*\del{}\psi + t^{-1-\ep_1}\betat^{-1}\gac
  +t^{-1-\ep_1}\betat^{-1} (\ell+\psi)* (k+t^{1-\ep_1}\betat\psi+t^{1-\ep_1}\betat\ell),
\end{align*}
or equivalently
\begin{align*}
  \betat^{2}t^{2}\Rtt
  = t^{2-\ep_2-\ep_1-2\ep_0}\beta^{2} f*\del{}\psi
      + t^{1-\ep_1-\ep_0}\beta\gac
   +t^{1-\ep_1-\ep_0}\beta (\ell+\psi)* (k+t^{1-\ep_1-\ep_0}\beta\psi+t^{1-\ep_1-\ep_0}\beta\ell).
\end{align*}
It is thus an immediate consequence of \eqref{eq:epscond2} and \eqref{eq:symhypreg_appl} that 
\begin{equation*}
\lim_{t\searrow 0}\norm{\betat^{2}t^{2}\Rtt}_{H^{k-1}(\Tbb^{n-1})} = 0.
\end{equation*} 
Given that the Hamiltonian constraint \eqref{eq:Hamilton1} holds
everywhere in $M_{0,t_0}$, we can apply the Sobolev and the Cauchy
Schwarz inequality to conclude that
\begin{equation}
\label{eq:asymptptwKasner-a-pre}
\lim_{t\searrow 0}\Bigl| 4\Ntt^{2}(t,x)t^{2} (\Ktt _{I}{}^{I})^2
       -4\Ntt^{2}(t,x)t^{2}\Ktt _{I}{}^{J}\Ktt _{J}{}^{I}
       +8 \Ntt(t,x)t\Ktt _{I}{}^{I}                                                        
       \Bigr|=0 \quad \text{for each $x\in \Tbb^{n-1}$},
\end{equation}
where we have also used \eqref{Ntt=betat}.
Since 
\begin{equation}
  \label{eq:asymptptwKasner2-a}
  \lim_{t\searrow 0} \bigl|2t\,\Ntt(t,x)\, \Ktt_{I}{}^J (t,x)-\kf_{I}{}^J(x)\bigr|=0 \quad \text{for each $x\in \Tbb^{n-1}$}
\end{equation}
as a consequence of \eqref{eq:2ndFF.est}, \eqref{Ntt=betat} and the Sobolev inequality, it is not difficult to verify
\begin{equation}
  \label{eq:asymptptwKasner-a}
  (\kf_{I}{}^{I})^2 -
  \kf_{I}{}^J \kf_{J}{}^I 
  +4\kf_{I}{}^{I}=0 \quad \text{in $\Tbb^{n-1}$}
\end{equation}
from \eqref{eq:asymptptwKasner-a-pre}.
Solving \eqref{eq:asymptptwKasner-a} for $\kf_I{}^I$, we obtain two solutions $\kf_I{}^I=\pm\sqrt{4+\kf_{I}{}^J \kf_{J}{}^I}
  -2$.
But, by \eqref{eq:kijsmallness} and the Sobolev inequality, we can choose $\delta$ small enough to ensure that 
$\norm{\kf_I{}^I}_{L^\infty(\Tbb^{n-1})}< 4$.
Doing so, we conclude
that $\kf_I{}^I$ must satisfy $\kf_I{}^I=\sqrt{4+\kf_{I}{}^J \kf_{J}{}^I}
  -2$, which in particular, implies that
\begin{equation} \label{kfII}
\kf_I{}^I \geq 0 \quad \text{in $\Tbb^{n-1}$}.
\end{equation}
Together, \eqref{eq:asymptptwKasner2-a}-\eqref{kfII} imply that the solution $\{g_{\mu\nu},\tau\}$  verifies all the conditions of Definition \ref{def:APKasner}, and hence, is asymptotically pointwise Kasner.

\bigskip

\noindent \underline{Crushing singularity:}
Taking the trace of \eqref{eq:phys2ndff} with respect to the physical metric \eqref{eq:spatmetrrel}, we observe that the physical mean curvature can be expressed as
  \begin{equation}
    \label{eq:physM}
    \bar\Ktt_\Lambda{}^\Lambda=\frac{1}{2\betat t^{\frac {n-1}{n -2}}}\Bigl(2t\betat\Ktt_I{}^I
    +\frac{2(n-1)}{n -2}\Bigr).
  \end{equation}
Recalling that $\betat>0$, we deduce from \eqref{eq:improvbetaestimate-glob}, \eqref{eq:2ndFF.est}, \eqref{kfII} and Sobolev's inequality the existence of a constant $C>0$ and $t_1\in (0,t_0)$
such that the pointwise estimates
\begin{equation*}
0<\betat(t,x) \leq \frac{1}{C} \AND 2t\betat\Ktt_I{}^I + \frac{2(n-1)}{n-2} \geq \frac{n-1}{n-2} 
\end{equation*}
hold for all $(t,x)\in M_{0,t_1}=(0,t_1]\times\Tbb^{n-1}$. These estimates together with \eqref{eq:physM} imply the pointwise lower bound
\begin{equation*}
\bar{\Ktt}_\Lambda{}^\Lambda \geq \frac{C(n-1)}{2(n-2)}\frac{1}{t^{\frac{n-1}{n-2}}}
\end{equation*}
on $M_{0,t_1}$, which implies, in particular, that $\bar{\Ktt}_\Lambda{}^\Lambda$ blows up uniformly as $t\searrow 0$. By definition, see  \cite{eardley1979}, this uniform blow up of the physical mean curvature implies that the hypersurface $t=0$ is a crushing singularity.

\bigskip

\noindent \underline{Geometric estimates for the conformal metric and scalar field:} We now turn to deriving the estimates \eqref{eq:glostab-est.First}-\eqref{eq:glostab-est.Last}. Before doing so, we first summarise the relationships between the following $\epv$-independent geometric variables 
\begin{equation*}
    \bigl\{ t\betat\Dc_0 g_{00},\Dc_I g_{00},\Dc_0 g_{J0},\Dc_I g_{J0},t\betat\Dc_0 g_{JK},\Dc_I g_{JK}, t \betat\Dc_I\Dc_j g_{kl},\Dc_i\Dc_j\tau,\Dc_I\Dc_j\Dc_k\tau\bigr\}
\end{equation*}
and the corresponding $\epv$-dependent Fuchsian variables determined from the components of $u$. 
From  \eqref{p-fields}-\eqref{d-fields}, \eqref{gi00}, \eqref{kt-def}-\eqref{psit-def} and \eqref{k-def}-\eqref{tauac-def}, it straightforward to check that
  \begin{align*}    
    t\betat\Dc_0 g_{00}=&-\delta^{JK} k_{JK},\\
    \Dc_I g_{00}=&2t^{-\ep_1}m_{I}-t^{-\ep_1}\delta^{JK}(2\ell_{JKI}-\ell_{IJK}),\\
    \Dc_0 g_{J0}=&t^{-\ep_1}m_{J},\\
    \Dc_I g_{J0}=&t^{-\ep_1}\ell_{I0J},\\
    t\betat\Dc_0 g_{JK}=&k_{JK}+ t^{1-\ep_1+\frac{1}{2}\kf_{J}{}^J}(t^{-\frac{1}{2}\kf_{J}{}^J}\betat)(\ell_{K0J}+\ell_{J0K}),\\
    \Dc_I g_{JK}=&t^{-\ep_1}\ell_{IJK},\\    
    t \betat\Dc_I\Dc_j g_{kl}=&t^{-\ep_1}\gac_{Ijkl},\\        
    \Dc_i\Dc_j\tau=&t^{\ep_0-\ep_1}\xi_{ij},\\    
    \Dc_I\Dc_j\Dc_k\tau=&t^{-\ep_0-2\ep_1}\tauac_{Ijk}.
  \end{align*}  
From these relations, it then follows via the definitions 
\eqref{U-def}, \eqref{Pbb-def}, \eqref{Pbb-perp-def}  as well as the 
estimates \eqref{eq:PbbuDecay_appl}-\eqref{eq:PbbuDecay2_appl} and \eqref{kfII}
that the esimates
  \begin{align}
    \label{eq:glostab-est.First.pre}
      \norm{t\betat\Dc_0 g_{00}+\delta^{JK}\kf_{JK}}_{H^{k-1}(\mathbb T^{n-1})}&\lesssim
                                                                                  t^{p}+ t^{2(\kappat-\sigma)}, \\
    \label{eq:glostab-est.2.pre}
      \norm{t\betat\Dc_0 g_{JK}-\kf_{JK}}_{H^{k-1}(\mathbb T^{n-1})}
      &\lesssim t^{p}+ t^{2(\kappat-\sigma)} +t^{1-\ep_1}(t^p+t^{\kappat-\sigma}),\\   
    \norm{\Dc_I g_{00}}_{H^{k-1}(\mathbb T^{n-1})}+\norm{\Dc_0 g_{J0}}_{H^{k-1}(\mathbb T^{n-1})} \quad &\notag\\
    \label{eq:glostab-est.3.pre}
                                                                                +\norm{\Dc_I g_{J0}}_{H^{k-1}(\mathbb T^{n-1})}+\norm{\Dc_I g_{JK}}_{H^{k-1}(\mathbb T^{n-1})}&\lesssim  t^{p-\ep_1}+t^{\kappat-\sigma-\ep_1},\\      
\label{eq:glostab-est.4.pre}
    \norm{t \betat\Dc_I\Dc_j g_{kl}}_{H^{k-1}(\mathbb T^{n-1})}&\lesssim t^{p-\ep_1}+t^{\kappat-\sigma-\ep_1},\\    
\label{eq:glostab-est.5.pre}
    \norm{\Dc_i\Dc_j\tau }_{H^{k-1}(\mathbb T^{n-1})}&\lesssim t^{\ep_0-\ep_1}(t^p+t^{\kappat-\sigma}),\\
\label{eq:glostab-est.Last.pre}
    \norm{\Dc_I\Dc_j\Dc_k\tau }_{H^{k-1}(\mathbb T^{n-1})}&\lesssim t^{-\ep_0-2\ep_1}(t^p+t^{\kappat-\sigma}),
  \end{align}
hold for $0<t\leq t_0$.

The right hand sides of the estimates \eqref{eq:glostab-est.First.pre}-\eqref{eq:glostab-est.Last.pre} are $\epv$-dependent, while the left hand sides are not. We
therefore proceed as before and optimise\footnote{``Optimising'' here means finding the largest exponents within the possible ranges given by \eqref{eq:glostab-est.First.pre}-\eqref{eq:glostab-est.Last.pre} and \eqref{eq:epscond2}. It is likely that with additional work these estimates could be further optimised by means of a more detailed analysis of the evolution equations. We leave this to future work.} the decay exponents in the estimates \eqref{eq:glostab-est.First.pre}-\eqref{eq:glostab-est.Last.pre} by making appropriate choices for $\ep_0$, $\ep_1$, $\ep_2$ subject to \eqref{eq:epscond2}. It is important to note here that multiple parameter sets  $\{\ep_0,\ep_1,\ep_2\}$ will be employed in order to optimise the decay rate in each of the estimates \eqref{eq:glostab-est.First.pre}-\eqref{eq:glostab-est.Last.pre}. Since $\delta$ and many of the implicit constants in \eqref{eq:glostab-est.First.pre}-\eqref{eq:glostab-est.Last.pre} depend on $\ep_0$, $\ep_1$, and $\ep_2$, once we have selected suitable parameter sets $\{\ep_0,\ep_1,\ep_2\}$ that optimise each of the estimates, we will pick the smallest of the corresponding $\delta$ constants (guaranteed to be positive) and the largest one of all the other constants (guaranteed to be finite).

We start with optimising the decay rate for the estimate \eqref{eq:glostab-est.First.pre}. Noticing that the exponent on the right-hand side is the same as in \eqref{eq:improvbetaestimate.pre}, it follows that the decay rate will be optimised for the same choice of parameters used to optimize \eqref{eq:improvbetaestimate.pre}, namely \eqref{eq:9dufsljdsf}-\eqref{eq:9dufsljdsf-a}. This yields the estimate \eqref{eq:glostab-est.First}. 
Considering next the estimate \eqref{eq:glostab-est.2.pre}, we observe that the exponent on the right hand side is
\begin{equation*}
  \xi_2=\min\{\xi_1, 1-\ep_1+\kappat-\sigma\}
\end{equation*}
where the function $\xi_1$ is defined in \eqref{eq:defxi1} and we have exploited the fact that $1-\ep_1\geq 0$ in $\bar{E}$, see \eqref{E-def}.
It follows that $\xi_2\le\xi_1$ everywhere on $\bar E$. A simple calculation shows that $\xi_2=\xi_1$ for $\ep_0$, $\ep_1$ and $\ep_2$ given by \eqref{eq:9dufsljdsf}-\eqref{eq:9dufsljdsf-a}, that is, exactly where $\xi_1$, defined above by \eqref{eq:defxit1}, attains its global maximum value. We conclude that $\xi_2$ has the same absolute maximum value \eqref{xi1-max}, which is attained in $E$. By \eqref{eq:glostab-est.2.pre}, it is then clear that the estimate \eqref{eq:glostab-est.2} follows.

Next, we consider \eqref{eq:glostab-est.3.pre} and \eqref{eq:glostab-est.4.pre} whose exponents are determined by
\begin{equation*}
  \xi_3:=\min\{p,\kappat-\sigma\}-\ep_1
  \oset{\eqref{eq:defkappat}}{=}\min\biggl\{1-3\ep_0-2\ep_1,1-\ep_0-\ep_2-\ep_1,
  \ep_0-\ep_1-\sigma,\ep_2-\ep_1-\sigma,-\frac{1}{1+\sqrt{3}}-\sigma\biggr\}.
\end{equation*}
By definition $\xi_3\le -\frac{1}{1+\sqrt{3}}-\sigma$, and because this value is achieved for the choice of parameters
\[\ep_0= \frac 15\biggl(1-\frac{1}{1+\sqrt{3}}+\sigma\biggr),\quad \ep_1=\ep_0+\frac{1}{1+\sqrt{3}},\quad \ep_2=\ep_0,\]
the absolute maximum of
$\xi_3$ on $\bar{E}$ is $-\frac{1}{1+\sqrt{3}}-\sigma$.
This implies by  \eqref{eq:glostab-est.3.pre} and \eqref{eq:glostab-est.4.pre} that the estimates \eqref{eq:glostab-est.3} and \eqref{eq:glostab-est.4} hold.

Turning out attention to the estimate \eqref{eq:glostab-est.5.pre}, we consider
\[\xi_4:=\min\{p,\kappat-\sigma\}-\ep_1+\ep_0 \oset{\eqref{eq:defkappat}}{=}\min\biggl\{1-2\ep_0-2\ep_1,1-\ep_2-\ep_1,
  2\ep_0-\ep_1-\sigma,\ep_0+\ep_2-\ep_1-\sigma,\ep_0-\frac{1}{1+\sqrt{3}}-\sigma\bigr\}.\]
On $\bar E$, $\xi_4$ satisfies $\xi_4\le\xit_4$ where
\[\xit_4:=\xi_4\bigl|_{\ep_1=\frac{1}{1+\sqrt{3}}}=\min\biggl\{1-2\ep_0-2 \frac{1}{1+\sqrt{3}},1-\ep_2-\frac{1}{1+\sqrt{3}},      
  \ep_0-\frac{1}{1+\sqrt{3}}-\sigma\biggr\}.\]
The absolute maximum value of $\xit_4$ over $\bar{E}$ is  
\begin{equation}\label{xi4-max}
    \frac13\biggl(1-4\frac{1}{1+\sqrt{3}}-2\sigma\biggr),
\end{equation} 
which is achieved for the parameter values
\begin{equation} \label{xi4-max-point}
    \ep_0=\frac 13(1-\frac{1}{1+\sqrt{3}}+\sigma),\quad \ep_2=1-\frac{1}{1+\sqrt{3}}-\ep_0,
\end{equation}
and any choice of $\ep_1$. We immediately conclude from this that \eqref{xi4-max} is also the absolute maximum value of $\xi_4$ over $\bar{E}$ and that $\xi_4$ achieves this maximum for the parameter values \eqref{xi4-max-point} and $\ep_1=\frac{1}{1+\sqrt{3}}$. Taking into account that
this choice of parameters lies on the boundary of $E$, we observe that estimate \eqref{eq:glostab-est.5} follows
from \eqref{eq:glostab-est.5.pre}.

Finally, considering the last estimate \eqref{eq:glostab-est.Last.pre}, we define
\[\xi_5:=\min\{p,\kappat-\sigma\}-2\ep_1-\ep_0\oset{\eqref{eq:defkappat}}{=}\min\biggl\{1-4\ep_0-3\ep_1,
1-\ep_2-2\ep_1-2\ep_0,      -2\ep_1-\sigma,-\ep_0+\ep_2-2\ep_1-\sigma,
-\ep_0-\ep_1-\frac{1}{1+\sqrt{3}}-\sigma\biggr\}.\]
From the definition \eqref{E-def} of $E$ and the fact that $\xi_5$ is a monotonically decreasing function of both $\ep_0$ and $\ep_1$, it is clear that $\xi_5\le -2\frac{1}{1+\sqrt{3}}-\sigma$ on $\bar E$. In fact, $-2\frac{1}{1+\sqrt{3}}-\sigma$ is the absolute maximum value of $\xi_5$ over $\bar E$ as is easily verified by evaluating $\xi_5$ at
\[\ep_0=0,\quad \ep_1=\frac{1}{1+\sqrt{3}},\]
and any choice of $\ep_2$. Taking again into account that these parameter values lie on the boundary $\partial E$ of $E$, it is clear that the estimate  \eqref{eq:glostab-est.Last} follows directly from \eqref{eq:glostab-est.Last.pre}.

We now in a position to derive the estimate \eqref{eq:improvframeestimate}. Indeed, taking the limit $\tilde t\searrow 0$ of \eqref{eq:frame_convest},
we see, with the help of \eqref{frame-lim-B} and \eqref{frame-lim-Last}, that, for $\delta$ and $\sigma$ chosen sufficiently small, the estimate 
\begin{equation}
  \label{eq:improvframeestimate.pre}
  \norm{e^{-\ln(t)\Ksc}e^\Lambda(t)-\ec^\Lambda(0)}_{H^{k-1}(\Tbb^{n-1})}    
  \lesssim \Bigl(t^{p}+ t^{2(\kappat-\sigma)}+t^{1-\ep_1}(t^p+t^{\kappat-\sigma})\Bigr) t^{-C \norm{\kf_{JM}}_{H^{k-1}(\Tbb^{n-1})}},
\end{equation}
holds for all $0<t\leq t_0$. Given that the exponent on the right hand side is, up to the $\epv$-independent term $-C \norm{\kf_{JM}}_{H^{k-1}(\Tbb^{n-1})}$, the same as that on the right hand side of \eqref{eq:glostab-est.2.pre}, the bound \eqref{eq:kijsmallness} along with the same choices of parameter  $\ep_0$, $\ep_1$ and $\ep_2$ used to derive \eqref{eq:glostab-est.2} from \eqref{eq:glostab-est.2.pre} imply that the estimate \eqref{eq:improvframeestimate} is a consequence of
\eqref{eq:improvframeestimate.pre}.

\bigskip

\noindent\underline{AVTD property:} The solution $\{g_{\mu\nu},\tau\}$ to the conformal Einstein-scalar field equations on $M_{0,t_0}$ constructed above 
satisfies the AVTD property in the sense of Section~\ref{sec:AVTDAPK}. To see why, we observe that the VTD equation corresponding to the Fuchsian equation \eqref{eq:givp1} is given by
\[A^0\del{t}u =
  \frac{1}{t}\Ac\Pbb u + \frac{1}{t} G(u)+\frac{1}{t^{\tilde\ep}}H(t,u),\]
and hence, any solution $u$ of \eqref{eq:givp1} is a solution of this VTD equation up to the error term ${t^{-\ep_0-\ep_2}}A^\Lambda(t,u) \del{\Lambda} u$. Given the definition \eqref{ALambda-def}, it then follows immediately from \eqref{eq:symhypreg_appl} and $\ep_0+\ep_2<1$, see \eqref{eq:epscond2}, that
\begin{equation*}
\int^{t_0}_0\bnorm{{s^{-\ep_0-\ep_2}}A^\Lambda(s,u(s)) \del{\Lambda} u(s)}_{H^{k-1}(\Tbb^{n-1})}\,ds <\infty,
\end{equation*}
which establishes the AVTD property. 

\bigskip

\noindent \underline{$C^2$-inextendibility of the physical metric:} Next, we establish the $C^2$-inextendibility of the physical metric by  verifying that the scalar curvature $\Rb=\gb^{\mu\nu}\Rb_{\mu\nu}$ of the physical spacetime metric $\gb_{\mu\nu}$ blows up as $t\searrow 0$. To this end, we observe by \eqref{ESF.1}, \eqref{tau=t}, \eqref{eq:conf2phys-a} and the properties \eqref{e0i-fix} of the orthonormal frame $e_i^\mu$ that the scalar curvature can be expressed as
\begin{align*}
 \Rb&=2\gb^{ij}\nablab_i\phi\nablab_j \phi 
                    =2t^{\frac {-2}{n -2}}\eta^{ij}\nabla_i\phi\nabla_j \phi
                    =-2 t^{\frac {-2}{n -2}}\betat^{-2}(\del{t}\phi)^2
         =-\frac{n-1}{n-2}(t^{-\frac{1}{2}\kf_{J}{}^J}\betat)^{-2}t^{\frac {-2}{n -2}-2-\kf_{J}{}^J}.
\end{align*}
Recalling that the function $\mathfrak{b}$ in
\eqref{eq:improvbetaestimate-glob}
is stricly positive on $\Tbb^{n-1}$, it then not difficult to verify from the above formula and the calculus inequalities from Appendix \ref{calc}
that estimate \eqref{eq:SptRicciEstimat-glob} is a direct consequence of \eqref{eq:improvbetaestimate-glob}. 

\bigskip

\noindent \underline{Past timelike geodesic incompleteness}
To complete the proof, recall from \eqref{eq:spatmetrrel} and \eqref{eq:phys2ndff} that the physical mean curvature $\bar\Ktt_\Lambda{}^\Lambda$ is related to the conformal mean curvature $\Ktt_\Lambda{}^\Lambda$
by the formula \eqref{eq:physM}.
In the case that the initial data  fields $\gr_{\mu\nu}$, 
$\ggr_{\mu\nu}$, $\taur=t_0$ and
$\taugr$ agree exactly with those of the FLRW solution (Section~\ref{sec:Kasnerscalarfield}), the value of $\Ktt_\Lambda{}^\Lambda$ is zero on the inital hypersurface $\Sigma_{t_0}=\{t_0\}\times \Tbb^{n-1}$ and hence the value of $\bar\Ktt_\Lambda{}^\Lambda$ is $\frac {n-1}{n -2} t_0^{-\frac {1}{n -2}-1}$ there.  By choosing $\delta>0$ sufficiently small, we can ensure that $\bar\Ktt_\Lambda{}^\Lambda$  as close, in a pointwise sense,
as we like to  $\frac {n-1}{n -2} t_0^{-\frac {1}{n -2}-1}$ everywhere on $\Sigma_{t_0}$ by \eqref{glob-stab-thm-idata-A} and the Sobolev's inequality. In particular, for $\delta>0$ small enough, we have that $\Ktt_\Lambda{}^\Lambda>{n-1}{2(n -2)} t_0^{-\frac {1}{n -2}-1}$ on $\Sigma_{t_0}$.
Past timelike geodesic incompleteness is then a consequence of Hawking's singularity theorem \cite[Chapter~14, Theorem~55A]{oneill1983a}, i.e., all past directed timelike geodesics starting on the $\Sigma_{t_0}$ reach $\{0\}\times\Tbb^{n-1}$ after finite proper time.

This concludes the proof of  Theorem~\ref{glob-stab-thm}.

\section{Localized stability\label{local-sec}}

We turn now to establishing a version of Theorem \ref{glob-stab-thm} that is localized to a truncated cone domain\footnote{See \eqref{Omega-def} for a definition of the truncated cone domain $\Omega_{t_0,\rho_0,\rho_1,\bar{\ep}}$.}
$\Omega_{t_0,\rho_0,\rho_1,\bar{\ep}}$ for appropriately chosen parameters $t_0>0$, $0<\rho_0 < L$, $0<\vartheta<1$ and $0<\bar{\ep}<1$, where 
\begin{equation}\label{rho1-def}
\rho_1 = t_0^{\bar{\ep}-1}(1-\bar{\ep})(1-\vartheta)\rho_0.
\end{equation}
The precise statement of the local stability result is given in Theorem \ref{loc-stab-thm} below. In order to simplify the presentation of the local stability result and its proof, we have not stated the most general result possible in terms of the choice of the parameters $t_0$, $\rho_0$, $\vartheta$, and $\bar{\ep}$ that determine $\Omega_{t_0,\rho_0,\rho_1,\bar{\ep}}$. 
We instead content ourselves with fixing $\rho_0 \in (0,L)$, which sets the size of the centred ball $\{t_0\}\times \mathbb{B}_{\rho_0}$ that caps $\Omega_{t_0,\rho_0,\rho_1,\bar{\ep}}$ from above, and the size of the centred ball $\{0\}\times \mathbb{B}_{\vartheta \rho_0}$, here $\vartheta\in (0,1)$ can be chosen arbitrarily, that caps $\Omega_{t_0,\rho_0,\rho_1,\bar{\ep}}$ from below. These choices along with the size of the deviation from the FLRW initial data, characterized by the parameter $\delta$, implicitly determine the size of $t_0>0$. We do not determine an explicit value for $t_0$ in our proof; we only note there exist a $T_0>0$ and $\bar{\ep}\in (0,1)$ such that for each $t_0 \in (0,T_0]$, it is possible to establish the stability of nonlinear perturbations of FLRW solutions on
$\Omega_{t_0,\rho_0,\rho_1,\bar{\ep}}$ provided that $\delta$ is chosen sufficiently small. With more work, it would be possible to establish a quantitative relationship between the parameters  $T_0$, $\rho_0$, $\vartheta$, $\bar{\ep}$, and $\delta$. However, without an application in mind that would require such a result, we will leave this extension to future work.

The proof of Theorem \ref{loc-stab-thm} follows from a straightforward adaptation of the proof of Theorem \ref{glob-stab-thm}. The only significant difference is that the proof of Theorem \ref{loc-stab-thm} exploits the finite propagation speed of the reduced conformal Einstein-scalar field equations. Specifically, the proof uses the finite propagation speed of the Fuchsian equation \eqref{Fuch-ev-A2} to obtain uniform estimates on $\Omega_{t_0,\rho_0,\rho_1,\bar{\ep}}$, the finite propagation speed of the  system \eqref{tconf-ford-C.1}-\eqref{tconf-ford-C.8} to obtain the local-in-time existence of solutions inside  $\Omega_{t_0,\rho_0,\rho_1,\bar{\ep}}$ to the reduced conformal Einstein-scalar field equations, i.e.\ a localized version of Proposition \ref{lag-exist-prop}, and the finite propagation speed of the wave gauge propagation equation \eqref{wave-gauge-prop-eqn} to establish that the constraints propagate on $\Omega_{t_0,\rho_0,\rho_1,\bar{\ep}}$, i.e.\ a localized version of Proposition \ref{prop-ES-constr}.

\begin{thm}[Localized past stability of the FLRW solution of the Einstein-scalar field system] \label{loc-stab-thm}
  Suppose $n\in\Zbb_{\ge 3}$, $k \in \Zbb_{>(n+3)/2}$, $\sigma>0$, $0<\rho_0<L$ and $0<\vartheta<1$. There exists constants $T_0>0$ and $\bar{\ep}\in (0,1)$ 
  such that for every $t_0\in (0,T_0]$ there exists a $\delta_0>0$ such that for every $\delta \in (0,\delta_0]$ and  $\gr_{\mu\nu}\in H^{k+2}(\mathbb{B}_{\rho_0},\mathbb{S}_n)$, 
$\ggr_{\mu\nu}\in H^{k+1}(\mathbb{B}_{\rho_0},\mathbb{S}_n)$, $\taur=t_0$ and
$\taugr\in H^{k+2}(\mathbb{B}_{\rho_0})$
satisfying
\begin{equation} \label{loc-stab-thm-idata-A}
  \norm{\gr_{\mu\nu}-\eta_{\mu\nu}}_{H^{k+2}(\mathbb{B}_{\rho_0})}
  +\norm{\ggr_{\mu\nu}}_{H^{k+2}(\mathbb{B}_{\rho_0})}
  +\norm{\taugr-1}_{H^{k+2}(\mathbb{B}_{\rho_0})}
  <\delta,
\end{equation}
the gravitational and wave gauge constraints \eqref{grav-constr}-\eqref{wave-constr} on $\{t_0\}\times \mathbb{B}_{\rho_0}$, and the inequalities $\det(\gr_{\mu\nu})<0$ and $|\vr|_{\gr}^2 <0$, there exists a unique classical solution $W\in C^1(\Omega_{t_0,\rho_0,\rho_1,\bar{\ep}})$, see \eqref{eq:Wdef}, of the system of evolution equations \eqref{tconf-ford-C.1}-\eqref{tconf-ford-C.8}
on the truncated cone domain $\Omega_{t_0,\rho_0,\rho_1,\bar{\ep}}$, where
$\rho_1$ is given by \eqref{rho1-def}, that satisfies the initial
conditions \eqref{l-idata}-\eqref{hhu-idata} on $\{t_0\}\times\mathbb{B}_{\rho_0}$, which are uniquely determined by the initial data $\{\gr_{\mu\nu},\ggr_{\mu\nu},\taur=t_0,\taugr\}$, and
the constraints \eqref{eq:Lag-constraints}.
Letting
\begin{equation*}
\rho(t) = \rho_0 + (1-\vartheta) \rho_0 \Bigl( \Bigl(\frac{t}{t_0}\Bigr)^{1-\bar{\ep}}-1\Bigr),
\end{equation*}
the regularity of the solution $W$ is such that $\del{t}^jW(t) \in H^{k-j}(\mathbb{B}_{\rho(t)})$ for every $0\leq j\leq k$ and $t\in (0,t_0]$.
\bigskip

\noindent Moreover, the pair $\{g_{\mu\nu}=\del{\mu}l^\alpha\ghu_{\alpha\beta}\del{\nu}l^\beta,\tau=t\}$, which is uniquely determined by $W$, defines a solution of the conformal Einstein-scalar field equations \eqref{lag-confeqns} on $\Omega_{t_0,\rho_0,\rho_1,\bar{\ep}}$ that satisfies the wave gauge constraint \eqref{lag-wave-gauge}
and the following
properties:
\begin{enumerate}[(a)]
\item Let $e_0^\mu=\betat^{-1}\delta_0^\mu$ with $\betat= (-g(dx^0,dx^0))^{-\frac{1}{2}}$, and $e^\mu_I$ be the unique solution of the Fermi-Walker transport equation \eqref{Fermi-A} with initial conditions $e^\mu_I |_{\mathbb{B}_{\rho_0}}=\delta^\mu_\Lambda\er^\Lambda_I$ where the functions $\er^\Lambda_I \in H^k(\mathbb{B}_{\rho_0})$ are chosen to satisfy
$\norm{\er^\Lambda_I-\delta^\Lambda_I}_{H^{k}(\mathbb{B}_{\rho_0})} < \delta$
and make the frame $e^\mu_i$ orthonormal on $\{t_0\}\times \mathbb{B}_{\rho_0}$. 
Then $e^\mu_i$ is a well defined frame in $\Omega_{t_0,\rho_0,\rho_1,\bar{\ep}}$ that
satisfies
$e^0_I=0$ and $g_{ij}=\eta_{ij}$,
where $g_{ij}=e_i^\mu g_{\mu\nu}e_j^\nu$ is the frame metric.

\item There exists a tensor field $\kf_{IJ}\in 
H^{k-1}(\mathbb{B}_{\rho_0},\Sbb{n-1})$ satisfying $\norm{\kf_{IJ}}_{H^{k-1}(\mathbb{B}_{\rho_0})}\lesssim\delta$ and a constant $C>0$ such that 
  \begin{align*}
\norm{t\betat\Dc_0 g_{00}+\delta^{JK}\kf_{JK}}_{H^{k-1}(\mathbb{B}_{\rho(t)})}&\lesssim              t^{\frac19(3-\sqrt{3}-2\sigma)},\\
\norm{t\betat\Dc_0 g_{JK}-\kf_{JK} }_{H^{k-1}(\mathbb{B}_{\rho(t)})}
     &\lesssim t^{\frac19(3-\sqrt{3}-2\sigma)-C\delta},\\
    \norm{\Dc_I g_{00}}_{H^{k-1}(\mathbb{B}_{\rho(t)})}+\norm{\Dc_0 g_{J0}}_{H^{k-1}(\mathbb{B}_{\rho(t)})}
    \hspace{0.8cm} &\\
      +\norm{\Dc_I g_{J0}}_{H^{k-1}(\mathbb{B}_{\rho(t)})}+\norm{\Dc_I g_{JK}}_{H^{k-1}(\mathbb{B}_{\rho(t)})}&\lesssim t^{-\frac 12(\sqrt{3}-1)-\sigma},\\
\norm{t\betat\Dc_I\Dc_j g_{lm}}_{H^{k-1}(\mathbb{B}_{\rho(t)})}&\lesssim t^{-\frac 12(\sqrt{3}-1)-\sigma},\\
\norm{\Dc_i\Dc_j\tau }_{H^{k-1}(\mathbb{B}_{\rho(t)})}&\lesssim t^{-(\frac 23\sqrt{3}-1)-\sigma},\\
\norm{\Dc_I\Dc_j\Dc_l\tau}_{H^{k-1}(\mathbb{B}_{\rho(t)})} &\lesssim t^{-(\sqrt{3}-1)-2\sigma},
\intertext{and}
\norm{t\betat\Dc_0 g_{JK}-\kf_{JK} }_{H^{k-1}(\mathbb{B}_{\vartheta\rho_0})}
     &\lesssim t^{\frac19(3-\sqrt{3}-2\sigma)}
  \end{align*}
for all $t\in (0,t_0]$,
where all the
fields in these estimates are expressed in terms of the frame $e_i^\mu$ and the Levi-Civita connection $\Dc$ of the flat background metric
$\gc_{\mu\nu}=\del{\mu}l^\alpha \eta_{\alpha\beta} \del{\nu}l^\beta$. 
In addition, there exist a strictly positive function $\mathfrak{b}\in H^{k-1}(\mathbb{B}_{\rho_0})$ and a matrix $\ef^\Lambda_J\in H^{k-1}(\mathbb{B}_{\rho_0},\Mbb{n-1})$ such that
  \begin{align*}
       \Bnorm{t^{-\kf_{J}{}^J/2}\betat
    -\mathfrak{b}}_{H^{k-1}(\mathbb{B}_{\rho(t)})}
 & \lesssim t^{\frac 19(3-\sqrt{3}-2\sigma)}
\intertext{and} 
  \Bnorm{\exp\Bigl(\frac 12\ln(t)\kf_{J}{}^I\Bigr) e_I^\Lambda-\ef_J^\Lambda}_{H^{k-1}(\mathbb{B}_{\rho(t)}))}    
  &\lesssim t^{\frac19(3-\sqrt{3}-2\sigma)-C\delta}
\end{align*}
 for all $t\in (0,t_0]$, where $\kf_{L}{}^J=\kf_{LM}\delta^{MJ}$  and $\exp\Bigl(\frac 12\ln(t)\kf_{J}{}^I\Bigr)$ is the exponential of the matrix
 $\frac 12\ln(t)(\kf_{J}{}^I)$.
\item The second fundamental form $\Ktt_{\Lambda\Omega}$ induced on the constant time surface $\{t\}\times \mathbb{B}_{\rho(t)}$  by $g_{\mu\nu}$ satisfies
\begin{equation*}
\Bnorm{ 2t\betat \Ktt_{LJ}-\kf_{LJ}}_{H^{k-1}(
\mathbb{B}_{\vartheta\rho_0})}
\lesssim t^{\frac19(3-\sqrt{3}-2\sigma)},
\end{equation*}
for all $t\in (0,t_0]$, while the lapse, shift and the spatial metric on $\{t\}\times \mathbb{B}_{\rho(t)}$ are determined by $\Ntt=\betat$,
$\btt_\Lambda=0$, and $\gtt_{\Lambda\Omega}=g_{\Lambda\Omega}$, respectively.

\item The pair $\Bigl\{\gb_{\mu\nu}=t^{\frac {2}{n -2}}g_{\mu\nu},\phi=\sqrt{\frac{n-1}{2(n-2)}}\ln(t)\Bigr\}$ defines a solution of the physical Einstein-scalar field equations \eqref{ESF.1}-\eqref{ESF.2} on $\Omega_{t_0,\rho_0,\rho_1,\bar{\ep}}$ that exhibits AVTD behavior and is asymptotically pointwise Kasner on $\mathbb{B}_{\vartheta\rho_0}$ with Kasner exponents $r_1(x),\ldots,r_{n-1}(x)$ determined by the eigenvalues of $\kf_{L}{}^J(x)$ for each $x\in \mathbb{B}_{\vartheta\rho_0}$. In particular, $\kf_{L}{}^L(x)\ge 0$ for all $x\in\mathbb{B}_{\vartheta\rho_0}$, and $\kf_{L}{}^L(x)=0$ for some $x\in \mathbb{B}_{\vartheta\rho_0}$ if and only if $r_1(x)=\ldots=r_{n-1}(x)=0$. The time $t=0$ represents a crushing singularity. Furthermore, the function
  \begin{equation*}
    \Ptt
  =\frac{\sqrt{{2(n-1)}(n-2)}}{{2(n-1)} +(n-2) \kf_{L}{}^L}
\end{equation*}
can be interpreted as the asymptotic scalar field strength in the sense of Section~\ref{sec:AVTDAPK}.
\item The physical solution $\bigl\{\gb_{\mu\nu},\phi\bigr\}$ is past $C^2$ inextendible at $t=0$. The scalar curvature  $\Rb=\Rb_{\mu\nu}\gb^{\mu\nu}$ of the physical metric $\gb_{\mu\nu}$ 
satisfies
  \begin{equation*}
    \Bnorm{t^{2\frac {n-1}{n -2}+\kf_{J}{}^J}\Rb+\frac{n-1}{n-2}\mathfrak{b}^{-2}}_{H^{k-1}(\mathbb{B}_{\vartheta\rho_0})}\lesssim t^{\frac 19(3-\sqrt{3}-2\sigma)}
  \end{equation*}
for all $t\in (0,t_0]$, and consequently, it becomes unbounded in the limit $t\searrow 0$.
\end{enumerate}

\medskip

\noindent The implicit and explicit constants in the above estimates are independent of the choice of $\delta\in (0,\delta_0]$.

\end{thm}

We remark that one can also adapt the arguments in the proof of Theorem~\ref{glob-stab-thm} to show a version of past timelike geodesic incompleteness for truncated code domains, namely that all past directed timelike geodesics starting from the initial hypersurface $t=t_0$ either hit the singularity after finite proper time or leave through the non-singular boundary of the truncated cone domain $\Omega_{t_0,\rho_0,\rho_1,\bar\ep}$.

\begin{proof}
Fix $T_0>0$, and for given $\rho_0\in (0,L]$ and $t_0\in (0,T_0]$, suppose that the initial data
$\gr_{\mu\nu}\in H^{k+2}(\mathbb{B}_{\rho_0},\mathbb{S}_n)$, and
$\ggr_{\mu\nu}\in H^{k+1}(\mathbb{B}_{\rho_0},\mathbb{S}_n)$, $\taur=t_0$,
and $\taugr\in H^{k+2}(\mathbb{B}_{\rho_0})$
for the conformal Einstein-scalar field equations on the
initial hypersurface $\{t_0\}\times \mathbb{B}_{\rho_0}$ satisfies the gravitational and wave gauge constraints  \eqref{grav-constr}-\eqref{wave-constr} there\footnote{That is,  \eqref{gt-idata}-\eqref{dt-taut-idata} with $\Sigma_{t_0}$ replaced by $\{t_0\}\times \mathbb{B}_{\rho_0}$.}. We
further assume that this initial data satisfies
\eqref{loc-stab-thm-idata-A},
which we note by Remark \ref{FLRW-idata-rem-A} implies that it is a small perturbation of the FLRW initial data on $\{t_0\}\times \mathbb{B}_{\rho_0}$.
Then we can extend the initial data $\{\gr_{\mu\nu},\ggr_{\mu\nu},\taur,\taugr\}$ using
the extension operator \eqref{Ebb-def} to obtain initial
data\footnote{It is worth noting that the extended initial data will, in general, no longer solve the constraint equations \eqref{grav-constr}-\eqref{wave-constr} outside of $\{t_0\}\times \mathbb{B}_{\rho_0}$.} $\{\gr^*_{\mu\nu},\ggr^*_{\mu\nu},\taur_*,\taugr_*\}$
for the reduced conformal Einstein-scalar field equations on $\Sigma_{t_0}=\{t_0\}\times \Tbb^{n-1}$
by setting
\begin{equation*}
    \gr^*_{\mu\nu}= \eta_{\mu\nu}+ \Ebb_{\rho_0}(\gr_{\mu\nu}-\eta_{\mu\nu}),
    \quad \ggr^*_{\mu\nu}= \eta_{\mu\nu}+ \Ebb_{\rho_0}(\gr_{\mu\nu}), \quad \taur_*=t_0
    \quad \AND \quad \taugr_*=1+\Ebb_{\rho_0}(\taugr-1).
\end{equation*}
We recall that this initial data determines initial data $\{g_{\mu\nu}|_{\Sigma_{t_0}},\del{0}g_{\mu\nu}|_{\Sigma_{t_0}},\tau=t_0,\del{0}\tau =1\}$ for the metric $g_{\mu\nu}$ and scalar field $\tau$ in Lagrangian coordinates on $\Sigma_{t_0}$ via \eqref{dt-tau-idata}-\eqref{dt-g-idata}. We also set $e_0^\mu=(-|\chi|_g^2)^{-\frac{1}{2}}\chi^\mu$ and recall that it can be computed from the Lagrangian initial data on $\Sigma_{t_0}$ by \eqref{chi-idata}. We then 
specify the spatial frame initial data $e^\mu_I|_{\Sigma_{t_0}}=\delta^\mu_\Lambda\er^\Lambda_\mu$ where
the functions $\er^\Lambda_I \in H^k(\Tbb^{n-1})$ are chosen to satisfy
\begin{equation} \label{loc-stab-thm-idata-B.1}
\norm{\er^\Lambda_I-\delta^\Lambda_I}_{H^{k}(\Tbb^{n-1})} < \delta
\end{equation}
and make the frame $e^\mu_i$ orthonormal on $\Sigma_{t_0}$ with
respect to the metric $g_{\mu\nu}$ given there, see \eqref{g-idata}. 

By \eqref{Ebb-prop}, \eqref{loc-stab-thm-idata-A} and \eqref{loc-stab-thm-idata-B.1}, the extended initial data is bounded by
\begin{equation} \label{loc-stab-thm-idata-C}
    \norm{\gr^*_{\mu\nu}-\eta_{\mu\nu}}_{H^{k+2}(\Tbb^{n-1})}+\norm{\ggr^*_{\mu\nu}}_{H^{k+2}(\Tbb^{n-1})}+\norm{\taugr_*-1}_{H^{k+2}(\Tbb^{n-1})} + \norm{\er^\Lambda_I-\delta^\Lambda_I}_{H^{k}(\Sigma_{t_0})}
    \leq C \delta
    \end{equation}
for some constant $C>0$ independent of
the initial data $\{\gr_{\mu\nu},\ggr_{\mu\nu},\taur,\taugr\}$, and in particular, of the constant $\delta$. 
Now, suppose $\ep_0$, $\ep_1$, and $\ep_2$ are constants satisfying \eqref{eq:epscond2}. Then by the discussion in Section \ref{frame-idata}, the variable definitions \eqref{k-def}-\eqref{tauac-def}, \eqref{U-def}, and \eqref{eq:Ubgdef}-\eqref{u-def},  it is not difficult, with the help of the Sobolev and Moser inequalities (see Propositions 2.4., 3.7.~and 3.9.~from   \cite[Ch.~13]{TaylorIII:1996}),
to verify from \eqref{loc-stab-thm-idata-C} that the extended initial data $\{\gr^*_{\mu\nu},\ggr^*_{\mu\nu},\taur_*,\taugr_*,\er^\Lambda_I\}$ determines initial data $u^*_0\in H^k(\Tbb^{n-1})$ for the Fuchsian equation \eqref{eq:givp1} satisfying
\begin{equation} \label{loc-stab-thm-idata-D}
    \norm{u^*_0}_{H^k(\Tbb^{n-1})} \leq C_0\delta,
\end{equation}
where the constant is of the form \eqref{C0-def}.
Then by Proposition \ref{prop:globalstability} and the bound \eqref{glob-stab-thm-idata-D}, it follows that there exists a $\delta_0>0$, such that if we choose $\delta>0$ small enough so that
\begin{equation} \label{C0delta<delta0-loc}
C_0\delta <\delta_0,
\end{equation}
which we can always do for any given $t_0\in (0,T_0]$,
then there exists a solution    
\begin{equation*}
u^* \in C^0_b\bigl((0,t_0],H^k(\Tbb^{n-1})\bigr)\cap C^1\bigl((0,t_0],H^{k-1}(\Tbb^{n-1})\bigr)
\end{equation*}
of the IVP \eqref{eq:givp1}-\eqref{eq:givp2} that 
extends continuously in $H^{k-1}$ to $t=0$ and
satisfies the energy
\begin{equation} \label{loc-stab-thm-energy-A}
\norm{u^*(t)}_{H^k(\mathbb T^{n-1})}^2 + \int^{t_0}_t \frac{1}{s} \norm{\Pbb u^*(s)}_{H^k(\mathbb T^{n-1})}^2\, ds  \lesssim \norm{u_0^*}^2
\end{equation}
and decay estimates
\begin{align}
\label{loc-stab-thm-decay-A}
  \norm{\Pbb u^*(t)}_{H^{k-1}(\mathbb T^{n-1})} &\lesssim t^p+t^{\kappat-\sigma},\\
\label{loc-stab-thm-decay-B}
\norm{\Pbb^\perp u^*(t) - \Pbb^\perp u^*(0)}_{H^{k-1}(\mathbb T^{n-1})} &\lesssim
 t^{p}+ t^{2(\kappat-\sigma)},
\end{align}
for $t\in (0,t_0]$, where
$p$ and $\kappat$ are defined by \eqref{eq:defkappat}, respectively, where the implied constant in the above estimates are independent of $\delta$ provided it is chosen small enough to satisfy \eqref{C0delta<delta0-loc}.

As in the proof of Theorem \ref{glob-stab-thm}, see \eqref{Pbb-perp-u(0)}, we
write the limit $\Pbb^\perp u^*(0)$ as
\begin{equation*}
    \Pbb^\perp u^*(0)=\bigl(\kf_{IJ},0,0,0,0,0,0,0,0,0\bigr)
\end{equation*}
where $\kf_{IJ}\in H^{n-1}(\Tbb^{n-1},\Sbb{n-1})$ and
we observe from \eqref{loc-stab-thm-idata-D} and \eqref{loc-stab-thm-energy-A} that
\begin{equation*}
    \norm{\kf_{IJ}}_{H^{k}(\Tbb^{n-1})} \leq C_0\delta
\end{equation*}
for a constant $C_0$ of the form \eqref{C0-def}. 
With the help of Sobolev's inequality, we also observe
from the energy estimate
\eqref{loc-stab-thm-energy-A} and the inequalities \eqref{loc-stab-thm-idata-D}
that $u^*(t)$ is bounded in the $L^\infty$-norm by
\begin{equation} \label{loc-stab-thm-sol-A} 
\norm{u^*}_{L^\infty((0,t_0)\times \Tbb^{n-1})}\leq C_0\delta 
\end{equation}
where, as above, the constant $C_0$ is a constant of the form \eqref{C0-def}. By \eqref{eq:Ubgdef} and \eqref{u-def}, this bound implies that
\begin{equation} \label{loc-stab-thm-sup-bnd}
    \norm{\beta}_{L^\infty((0,t_0)\times \Tbb^{n-1})}= t_{0}^{\ep_0}+C_0\delta  \AND
    \norm{f^\Lambda_I}_{L^\infty((0,t_0)\times \Tbb^{n-1})} \leq t_{0}^{\ep_2}+C_0\delta.
\end{equation}

Next, we want to show that for $T_0>0$ chosen sufficiently small and each $t\in (0,T_0]$, there exists a $\delta$ such that the solutions \eqref{loc-stab-thm-sol-A} to the IVP \eqref{eq:givp1}-\eqref{eq:givp2} can be localized to a truncated cone domain $\Omega_{t_0,\rho_0,\rho_1,\bar{\ep}}$ where
\begin{equation} \label{ep-rho1-def}
    \bar{\ep}=\ep_0+\ep_2,\quad \rho_1 = \frac{(1-\vartheta)\rho_0(1-\bar{\ep})}{t_0^{1-\bar{\ep}}},
\end{equation}
and $\vartheta \in (0,1)$. Here, by localized, we mean that the restriction of $u^*$ to  $\Omega_{t_0,\rho_0,\rho_1,\bar{\ep}}$ is uniquely determined as a solution of \eqref{eq:givp1} by the initial data $u_0|_{\mathbb{B}_{\rho_0}}$, which we note is, in turn, uniquely determined by the conformal Einstein-scalar field initial data $\{\gr_{\mu\nu},\ggr_{\mu\nu},\taur,\taugr\}$ on $\{t_0\}\times \mathbb{B}_{\rho_0}$. Setting
\begin{equation*}
    \At^\mu = \delta^\mu_0A^0 + \frac{1}{t^{\bar{\ep}}}\delta_\Lambda^\mu A^\Lambda
\end{equation*}
and recalling from \eqref{n-def} that $n_0<0$, 
we find from \eqref{A0-def}, \eqref{ALambda-def} and \eqref{ep-rho1-def} that
\begin{align}
(u^*)^{\tr}n_\mu \At^\mu u^* \leq& -|n_0|(|\tauac|^2+|\gac|^2)
- \frac{2}{t^{\bar{\ep}}} \beta n_\Lambda f^\Lambda_K
\delta^{\Qt Q}\delta^{\lt l}\delta^{JK}\bigl(\tauac_{\Qt 0 \lt}\tauac_{QJl}+\delta^{\mt m}\gac_{\Qt 0 \lt \mt}\gac_{QJl m}\bigr), \label{Wspacelike-1}
\end{align}
where
\begin{equation*}
    |\tauac|^2= \delta^{\Qt Q}\delta^{\jt h}\delta^{\lt l}\tauac_{\Qt \jt \lt}\tauac_{Qjl} \AND  |\gac|^2= \delta^{\Qt Q}\delta^{\jt h}\delta^{\lt l}\delta^{\mt m}\gac_{\Qt \jt \lt \mt}\gac_{Qjlm},
\end{equation*}
and we note from an application of the Cauchy-Schwartz and Young's inequalities that
\begin{align}
     2n_\Lambda f^\Lambda_K
\delta^{\Qt Q}&\delta^{\lt l}\delta^{JK}\bigl(\tauac_{\Qt 0 \lt}\tauac_{QJl}+\delta^{\mt m}\gac_{\Qt 0 \lt \mt}\gac_{QJl m}\bigr)
\leq 2\bigl(\delta^{\Qt Q}\delta^{\lt l}\tauac_{\Qt 0 \lt}\tauac_{Q 0 l}\bigr)^{\frac{1}{2}}\bigl(\delta^{\Qt Q}\delta^{\lt l}\delta^{\Kt \Jt}n_{\Lambdat} f^{\Lambdat}_{\Kt}\tauac_{\Qt \Jt \lt} \delta^{KJ}n_\Lambda f^\Lambda_K\tauac_{Q J l}\bigr)^{\frac{1}{2}}\notag\\
& + 2\bigl(\delta^{\Qt Q}\delta^{\lt l}\delta^{\mt m}\gac_{\Qt 0 \lt\mt}\gac_{Q 0 l m}\bigr)^{\frac{1}{2}}\bigl(\delta^{\Qt Q}\delta^{\lt l}\delta^{\mt m}\delta^{\Kt \Jt}n_{\Lambdat} f^{\Lambdat}_{\Kt}\gac_{\Qt \Jt \lt \mt} \delta^{KJ}n_\Lambda f^\Lambda_K\gac_{Q J l m}\bigr)^{\frac{1}{2}} \notag \\
\leq& \frac{1}{\zeta}\bigl(\delta^{\Qt Q}\delta^{\lt l}\tauac_{\Qt 0 \lt}\tauac_{Q 0 l}+\delta^{\Qt Q}\delta^{\lt l}\delta^{\mt m}\gac_{\Qt 0 \lt\mt}\gac_{Q 0 l m}\bigr)+
\zeta\bigl(\delta^{\Qt Q}\delta^{\lt l}\delta^{\Kt \Jt}n_{\Lambdat} f^{\Lambdat}_{\Kt}\tauac_{\Qt \Jt \lt} \delta^{KJ}n_\Lambda f^\Lambda_K\tauac_{Q J l}\notag \\
&+\delta^{\Qt Q}\delta^{\lt l}\delta^{\mt m}\delta^{\Kt \Jt}n_{\Lambdat} f^{\Lambdat}_{\Kt}\gac_{\Qt \Jt \lt \mt} \delta^{KJ}n_\Lambda f^\Lambda_K\gac_{Q J l m}\bigr)
 \label{Wspacelike-2}
\end{align}
for any positive function $\zeta>0$. From another application of the Cauchy-Schwartz inequality, we note also that
\begin{align*}
\delta^{\Qt Q}\delta^{\lt l}\delta^{\Kt \Jt}n_{\Lambdat}& f^{\Lambdat}_{\Kt}\tauac_{\Qt \Jt \lt} \delta^{KJ}n_\Lambda f^\Lambda_K\tauac_{Q J l}\notag \\
\leq&\delta^{\Qt Q}\delta^{\lt l} \bigl(\delta^{K \Kt}n_{\Lambdat} f^{\Lambdat}_{\Kt}n_{\Lambda} f^{\Lambda}_{K}\bigr)^{\frac{1}{2}}\bigl(\delta^{\Jt J}\tauac_{\Qt \Jt \lt}\tauac_{\Qt J \lt}
\bigr)^{\frac{1}{2}}\bigl(\delta^{\Lt L}n_\Omega f^\Omega_{\Lt} n_{\tilde{\Omega}} f^{\tilde{\Omega}}_L\bigr)^{\frac{1}{2}}\bigl(\delta^{\Mt M}\tauac_{Q \Mt l}\tauac_{Q M l}\bigr)^{\frac{1}{2}} \\
=& |n_\Lambda f^\Lambda|^2 \delta^{\Qt Q}\delta^{\Jt J}\delta^{\lt l}\tauac_{\Qt \Jt \lt}\tauac_{Q J l}
\end{align*}
where 
\begin{equation*}
    |n_\Lambda f^\Lambda|^2= \delta^{K \Kt}n_{\Lambdat} f^{\Lambdat}_{\Kt}n_{\Lambda} f^{\Lambda}_{K},
\end{equation*}
and by similar arguments that
\begin{equation*}
    \delta^{\Qt Q}\delta^{\lt l}\delta^{\Kt \Jt}n_{\Lambdat} f^{\Lambdat}_{\Kt}\tauac_{\Qt \Jt \lt} \delta^{KJ}n_\Lambda f^\Lambda_K\tauac_{Q J l}
    \leq |n_\Lambda f^\Lambda|^2 \delta^{\Qt Q}\delta^{\Jt J}\delta^{\lt l}\delta^{\mt m}\gac_{\Qt \Jt \lt\mt}\gac_{Q J l m}.
\end{equation*}
From these inequalities, \eqref{Wspacelike-1} and \eqref{Wspacelike-2}, we deduce that
\begin{align}
(u^*)^{\tr}n_\mu \At^\mu u^* \leq& |n_0|\biggl( -(|\tauac|^2+|\gac|^2)
+ \biggl[\bigl(\delta^{\Qt Q}\delta^{\lt l}\tauac_{\Qt 0 \lt}\tauac_{Q 0 l}+\delta^{\Qt Q}\delta^{\lt l}\delta^{\mt m}\gac_{\Qt 0 \lt\mt}\gac_{Q 0 l m}\bigr) \notag \\
&+\biggl(\frac{|\beta|\,|f|_{\op}||\mathbf{n}|}{t^{\bar{\ep}}|n_0|}\biggr)^2 \bigl(\delta^{\Qt Q}\delta^{\Jt J}\delta^{\lt l}\tauac_{\Qt \Jt \lt}\tauac_{Q J l}+\delta^{\Qt Q}\delta^{\Jt J}\delta^{\lt l}\delta^{\mt m}\gac_{\Qt \Jt \lt\mt}\gac_{Q J l m}\bigr)\biggr]\biggr),
 \label{Wspacelike-3}
\end{align}
where in deriving this we have set $\zeta=|\beta|/(|n_0|t^{\ep_0+\ep_1})$ and employed the inequality
\begin{equation*}
    |n_\Lambda f^\Lambda| \leq |f|_{\op}|\mathbf{n}|, \quad |\mathbf{n}| = (\delta^{\Lambda\Omega}n_\Lambda n_\Omega)^{\frac{1}{2}}.
\end{equation*}
But by \eqref{n-def}, \eqref{loc-stab-thm-sup-bnd} and  \eqref{ep-rho1-def}, we have that
\begin{equation*}
    \frac{|\beta||f|_{\op}|\mathbf{n}|}{t^{\bar{\ep}}|n_0|}=
    \frac{|\beta||f|_{\op}t_0^{1-\bar{\ep}}}{\vartheta \rho_0 (1-\bar{\ep})} \leq  \frac{1}{(1-\vartheta)\rho_0 (1-\bar{\ep})}\bigl(t_0 +t_0^{1-\bar{\ep}}C_0\delta\bigr)\leq \frac{T_0}{(1-\vartheta)\rho_0 (1-\bar{\ep})}\bigl(1 +t_0^{-\bar{\ep}}C_0\delta\bigr) ,    
\end{equation*}
where again the constant $C_0$ is of the form \eqref{C0-def}.
By first choosing $T_0=(1-\vartheta) \rho_0 (1-\bar{\ep})/2$ and then choosing small enough $\delta$ to ensure that $C_0\delta < t_0^{\bar{\ep}}$, which is always possible for fixed $t_0\in (0,T_0]$, it follows that we can arrange that
$\frac{|\beta||f|_{\op}|\mathbf{n}|}{t^{\bar{\ep}}|n_0|} < 1$.
But this implies by \eqref{Wspacelike-3} that $n_\mu \At^\mu \leq 0$,
and hence by definition, see \cite[\S 4.3]{Lax:2006}, that the boundary hypersurface $\Upsilon_{t_0,\rho_0,\rho_1,\bar{\ep}}$
of $\Omega_{t_0,\rho_0,\rho_1,\bar{\ep}}$ is weakly spacelike
for the symmetric hyperbolic Fuchsian system \eqref{eq:givp1}. The importance of this fact is that it guarantees that 
\begin{equation*}
    u:= u^*\bigl|_{\Omega_{t_0,\rho_0,\rho_1,\bar{\ep}}} 
\end{equation*}
is the unique solution of \eqref{eq:givp1} on $\Omega_{t_0,\rho_0,\rho_1,\bar{\ep}}$ satisfying the intial condition
\begin{equation*}
 u_0:=u^*_0|_{\{t_0\}\times \mathbb{B}_{\rho_0}},
 \end{equation*}
which, we note, by the way $u_0^*$ was constructed, satisfies
$\norm{u_0^*}_{H^k(\Tbb^{n-1})} \lesssim \norm{u_0}_{H^k(\mathbb{B}_{\rho_0})}$.
Setting 
\begin{equation*}
    \rho(t) = \rho_0 + \frac{\rho_1(t^{1-\bar{\ep}}-t_0^{1-\bar{\ep}})}{1-\bar{\ep}} = \rho_0 + (1-\vartheta) \rho_0 \Bigl( \Bigl(\frac{t}{t_0}\Bigr)^{1-\bar{\ep}}-1\Bigr),
\end{equation*}
we then have that $(\{t\}\times\Tbb^{n-1}) \cap \Omega_{t_0,\rho_0,\rho_1,\bar{\ep}} =
    \{t\} \times \mathbb{B}_{\rho(t)}$
and so it follows from \eqref{loc-stab-thm-energy-A}-\eqref{loc-stab-thm-decay-B} that the solution $u$ to the Fuchsian system \eqref{eq:givp1} on $\Omega_{t_0,\rho_0,\rho_1,\bar{\ep}}$ satisfies  
\begin{equation} \label{loc-stab-thm-energy-B}
\norm{u(t)}_{H^k(\mathbb{B}_{\rho(t)})}^2 + \int^{t_0}_t \frac{1}{s} \norm{\Pbb u(s)}_{H^k(\mathbb{B}_{\rho(s)})}^2\, ds  \lesssim \norm{u_0}_{H^k(\mathbb{B}_{\rho_0})}^2
\end{equation}
and
\begin{align}
\label{loc-stab-thm-decay-C}
  \norm{\Pbb u(t)}_{H^{k-1}(\mathbb{B}_{\rho(t)})} &\lesssim t^p+t^{\kappat-\sigma},\\
\label{loc-stab-thm-decay-D}
\norm{\Pbb^\perp u(t) - \Pbb^\perp u^*(0)}_{H^{k-1}(\mathbb{B}_{\rho(t)})} &\lesssim
 t^{p}+ t^{2(\kappat-\sigma)},
\end{align}
for all $t\in (0,t_0]$, where the implied constants in the above estimates are independent of $\delta$ provided it is chosen small enough to satisfy \eqref{C0delta<delta0-loc}.

To complete the proof, we note that it is not difficult to verify that the wave gauge constraint propagation result contained in Proposition \ref{prop-ES-constr} and the local existence and uniqueness results contain in Proposition \ref{lag-exist-prop} can be localized to the truncated cone domain $\Omega_{t_0,\rho_0,\rho_1,\bar{\ep}}$. These results rely on the finite propagation speed of the wave gauge propagation equation \eqref{wave-gauge-prop-eqn} and the first order system \eqref{tconf-ford-C.1}-\eqref{tconf-ford-C.8} using the same argument as above to deduce that boundary hypersurface $\Upsilon_{t_0,\rho_0,\rho_1,\bar{\ep}}$
of $\Omega_{t_0,\rho_0,\rho_1,\bar{\ep}}$ is weakly spacelike for \eqref{eq:givp1}. With these localized results in hand as well as the global estimates \eqref{loc-stab-thm-energy-B}-\eqref{loc-stab-thm-decay-D}
for the solution $u$ of the Fuchsian system \eqref{eq:givp1} on  $\Omega_{t_0,\rho_0,\rho_1,\bar{\ep}}$, the remainder of the proof follows from the same sequence of arguments used in the proof of Theorem \ref{glob-stab-thm}, where, of course, the obvious modification to take account the restriction to 
$\Omega_{t_0,\rho_0,\rho_1,\bar{\ep}}$ have been made.
It is important to note that because the arguments employed in the proof of Theorem  \ref{glob-stab-thm} involve optimising the decay exponents by selecting a finite number of distinct parameter sets $\{\ep_0,\ep_1,\ep_2\}$, the smallest truncated cone domain $\Omega_{t_0,\rho_0,\rho_1,\bar{\ep}}$, which is the one with the maximal value\footnote{Note that 
 \begin{equation*}\rho_0 + (1-\vartheta) \rho_0 \Bigl( \Bigl(\frac{t}{t_0}\Bigr)^{1-\bar{\ep}}-1\Bigr)\leq  \rho_0 + (1-\vartheta) \rho_0 \Bigl( \Bigl(\frac{t}{t_0}\Bigr)^{1-\bar{\ep}'}-1\Bigr), \quad 0<t\leq t_0,
 \end{equation*}
 for $0<\bar{\ep}'\leq \bar{\ep}<1$ since $\rho_0>0$ and $0<\vartheta <1$ by assumption.
 } of $\bar{\ep}=\ep_0+\ep_2$, must be used in all the arguments.   
\end{proof}

\subsection*{Acknowledgements}
We thank the referee for their detailed and insightful comments and
criticisms, which led to a significant improvement in the content and exposition of this article.

\appendix

\section{Matrix inequalities}

\begin{lem} \label{matrix-lem}
Suppose $\Ac=(\Ac_{ij})$,  $1\leq i,j\leq N$, is a square matrix that is partitioned into $N\times N$-block matrices $\Ac_{ij}$ and satisfies the following:
\begin{enumerate}[(a)]
\item It is block upper-triangular, that is, $\Ac_{ij}=0$ for $1\leq j < i \leq N$.
\item The diagonal blocks $\Ac_{ii}$ are square and there exist constants $\kappa_i>0$ such that
\begin{equation*}
\Ac_{ii} \geq \kappa_i \id
\end{equation*}
for $1\leq i\leq N$.
\end{enumerate}
Then for any $\eta>0$, there exists positive constants $\sigma_i=\sigma_i(\eta)>0$ such that square block diagonal matrix
\begin{equation} \label{matrix-lem-1}
\Asc = \emph{diag}(\sigma_1\id\!,\sigma_2\id\!,\ldots,\sigma_N \id\!),
\end{equation}
whose
diagonal blocks $\sigma_i\id$ are of the same size as the diagonal block $\Ac_{ii}$ of $\Ac$, satisfies
\begin{equation}
  \label{eq:matrixlemmares}
    \Asc\Ac \geq (\kappat-\bc \eta)\Asc
\end{equation}
where $\kappat = \min\{\,\kappa_i\,|\, 1\leq i \leq N\,\}$ and $\bc = \max\{\,|\Ac_{ij}|_{\op}\,|\, 1\leq i \leq N, \; i+1\leq j\leq N\,\}$.
\end{lem}
\begin{proof}
Let $v=(v_j)$, $1\leq j\leq N$, be a N-block column vector that is partitioned so that the matrix multiplications $\Ac_{ij}v_j$ are well defined. Then
\begin{align}
 v^{\tr}\Asc\Ac v =& \sum_{i=1}^N \sigma_i v_i^{\tr}\Ac_{ii}v_i + \sum_{i=1}^{N-1}\sum_{j=i+1}^N\sigma_i v_i^{\tr}\Ac_{ij}v_j \notag \\
 \geq& \sum_{i=1}^N \sigma_i\kappa_i |v_i|^2
 + \sum_{i=1}^{N-1}\sum_{j=i+1}^N\sigma_i v_i^{\tr}\Ac_{ij}v_j && \text{(since $\Ac_{ii}\geq \kappa_i \id\!$)}\notag \\
 \geq& \sum_{i=1}^N \sigma_i\kappa_i |v_i|^2
 - \sum_{i=1}^{N-1}\sum_{j=i+1}^N \sigma_i|v_i| |\Ac_{ij}v_j|
 && \text{(by the Cauchy-Schwartz inequality)} \notag \\
  \geq& \sum_{i=1}^N \sigma_i\kappa_i |v_i|^2
 - \sum_{i=1}^{N-1}\sum_{j=i+1}^N \sigma_i |\Ac_{ij}|_{\op}|v_i||v_j| \notag \\
 \geq& \kappat \sum_{i=1}^N \sigma_i |v_i|^2
 -\bc \sum_{i=1}^{N-1}\sum_{j=i+1}^N \sigma_i |v_i||v_j| \notag
\end{align}
where $\kappat$ and $\bc$ are as defined in the statement of the lemma.
Then with the help of Young's inequality, we find from the above inequality that
\begin{align}
 v^{\tr}\Asc\Ac v\geq& \kappat \sum_{i=1}^N \sigma_i |v_i|^2
 -\frac{\bc}{2} \sum_{i=1}^{N-1}\sum_{j=i+1}^N \sigma_i (\gamma_i\gamma_j|v_i|^2 + \gamma_i^{-1}\gamma_j^{-1} |v_j|^2) \notag \\
 = & \sum_{i=1}^N \sigma_i\Bigl(\kappat -\frac{\bc}{2}\Gamma_i\gamma_i\Bigr)|v_i|^2 -
 \frac{\bc}{2} \sum_{i=1}^{N-1}\sum_{j=i+1}^N \sigma_i \gamma_i^{-1}\gamma_j^{-1} |v_j|^2\notag \\
 = & \sum_{i=1}^N \sigma_i\Bigl(\kappat -\frac{\bc}{2}\Gamma_i \gamma_i\Bigr)|v_i|^2 -
 \frac{\bc}{2} \sum_{j=2}^N\sum_{i=1}^{j-1} \sigma_i \gamma_i^{-1}\gamma_j^{-1} |v_j|^2 \notag \\
 =& \sum_{i=1}^N \Bigl[\sigma_i\Bigl(\kappat -\frac{\bc}{2}\Gamma_i \gamma_i\Bigr) -\frac{\bc}{2}\Gammat_i\gamma_i^{-1}\Bigr)\Bigr]|v_i|^2
  \label{matrix-lem-2}
\end{align}
for any choice of constants $\gamma_i>0$, $1\leq i\leq N$, where
\begin{equation*}
\Gamma_i=\begin{cases}\dsp \sum_{j={i+1}}^N \gamma_j &
\text{ if  $1 \leq i \leq N-1$}\\
0 & \text{ if $i=N$}
\end{cases}
\AND
\Gammat_i= \begin{cases} 0 & \text{if $i=1$}\\
\dsp  \sum_{j=1}^{i-1}\sigma_j\gamma_j^{-1}  & \text{if  $2 \leq i \leq N$}
\end{cases}.
\end{equation*}

Now, for a given $\eta>0$, we want to solve the equations
\begin{align}
    \Gamma_i\gamma_i& =(\delta^1_i+1)\eta,\quad 1\leq i \leq N-1,  \label{matrix-lem-3.a}\\
    \Gammat_i &= (\delta^N_i+1)\eta\sigma_i\gamma_i, \quad 2\leq i\leq N, \label{matrix-lem-3.b}
\end{align}
for the constants $\gamma_i$ and $\sigma_i$. Noting that \eqref{matrix-lem-3.a} only involves the constants $\gamma_i$, we first solve these equations. To see that it can be done, we proceed by induction. First, we observe for $i=N-1$ that \eqref{matrix-lem-3.a} reduces to $\gamma_N \gamma_{N-1}=(\delta^1_{N-1}+1)\eta$, and we can solve this by setting $\{\gamma_N,\gamma_{N-1}\}=\{1,\eta (\delta^1_{N-1}+1)\}$. Next, we assume that $\{\gamma_N,\gamma_{N-1},\ldots,\gamma_i\}\subset \Rbb_{>0}$ solves $\Gamma_j\gamma_j=(\delta^1_j+1)\eta=\eta$, $i\leq j\leq N-1$, for $1<M\leq i\leq N-1$. Then noting that 
\begin{equation*}
    \Gamma_{i-1}= \gamma_i + \Gamma_{i} = \gamma_i +\frac{(\delta^1_i+1)\eta}{\gamma_i},
\end{equation*}
where the second equality is a consequence of the induction hypothesis, we observe that the equation $\Gamma_{i-1}\gamma_{i-1}=(\delta^1_{i-1}+1)\eta$ can be written as
\begin{equation*}
\gamma_{i-1}\Bigl(\gamma_i + \frac{(\delta^1_i+1)\eta}{\gamma_i}\Bigr) = (\delta^1_{i-1}+1)\eta, 
\end{equation*}
which we can solve to get
\begin{equation*}
    \gamma_{i-1} = \frac{(\delta^1_{i-1}+1)\eta}{\gamma_i+\frac{(\delta^1_i+1)\eta}{\gamma_i}} > 0.
\end{equation*}
This complete the induction step, and we conclude that there exists a solutions $\{\gamma_1,\gamma_2,\ldots,\gamma_N\}\subset \Rbb_{>0}$
of \eqref{matrix-lem-3.a} for $1\leq i\leq N-1$.

Next, we turn to solving the equations \eqref{matrix-lem-3.b} for the constants $\sigma_i$. We first note that for $i=2$ that
\eqref{matrix-lem-3.b} reduces to 
$\sigma_1 \gamma_1^{-1} = (\delta_2^N+1)\eta \sigma_2 \gamma_2$, which we can solve by setting $\{\sigma_1,\sigma_2\}=\{1,((\delta_2^N+1)\gamma_1 \gamma_2\eta)^{-1}\}$. We proceed by induction and assume now that $\{\sigma_1,\sigma_2,\ldots, \sigma_i\}\subset \Rbb_{>0}$ solves $\Gammat_j=(\delta_j^N+1)\eta\sigma_j \gamma_j$, $2\leq j\leq i$, for $2\leq i \leq M<N$. Observing that
\begin{equation*}
\Gammat_{i+1}= \Gammat_{i}+ \sigma_{i}\gamma_i^{-1} =
(\delta^N_i+1)\eta\sigma_i \gamma_i+\sigma_i \gamma_i^{-1}=\sigma_i\bigl((\delta^N_i+1)\eta \gamma_i+\gamma_i^{-1}\bigr),
\end{equation*}
where the second equality follows from the induction hypothesis, we can express $\Gammat_{i+1}=(\delta^N_{i+1}+1)\eta\sigma_{i+1}\gamma_{i+1}$ as
\begin{equation*}
\sigma_i\bigl((\delta^N_i+1)\eta \gamma_i+\gamma_i^{-1}\bigr)=(\delta^1_{i+1}+1)\eta\sigma_{i+1}\gamma_{i+1}.
\end{equation*}
We solve this for $\sigma_{i+1}$ to get
\begin{equation*}
\sigma_{i+1}=\frac{\sigma_i\bigl((\delta^N_i+1)\eta \gamma_i+\gamma_i^{-1}\bigr)}{(\delta^1_{i+1}+1)\eta\gamma_{i+1}}>0.
\end{equation*}
This completes the induction step, and hence, there exists a solution $\{\sigma_1,\sigma_2,\ldots,\sigma_N\}\subset \Rbb_{>0}$
of \eqref{matrix-lem-3.b} for $2\leq i\leq N$.

Having solved \eqref{matrix-lem-3.a}-\eqref{matrix-lem-3.b}, it then follows from \eqref{matrix-lem-2} that we
can bound $v^{\tr}\Asc \Ac v$ below by
\begin{equation*}
    v^{\tr}\Asc\Ac v \geq (\kappat-\bc\eta)\sum_{i=1}^N \sigma_i |v_i|^2 \oset{\eqref{matrix-lem-1}}{=}  (\kappat-\bc\eta) v^{\tr}\Asc v.
\end{equation*}
Since the vector $v$ was chosen arbitrarily, we conclude that $\Asc\Ac \geq(\kappat-\bc\eta)\Asc$ and the proof is complete.
\end{proof}

\begin{lem} \label{diff-matrix-lem}
Suppose $t_1 < t_0$, $N\in \Zbb_{\geq 1}$, $B \in C^0((t_1,t_0],\Mbb{N})$, $A_0\in \Mbb{N}$, $\det(A_0)>0$, and
$A \in C^1((t_1,t_0],\Mbb{N})$ solves the initial value problem
\begin{equation*}
    \frac{dA}{dt}=BA, \quad A(t_0)=A_0.
\end{equation*}
Then 
\begin{equation*}
    \det(A(t)) = e^{-\int_{t}^{t_0} \mathrm{Tr}(B(s))\,ds} \det(A_0) >0
\end{equation*}
for all $t\in (t_1,t_0]$.
\end{lem}
\begin{proof}
Since the determinant of $A$ is initially positive at time $t=t_0$, it follows from Jacobi's formula $\frac{d\;}{dt}\det(A)=\mathrm{Tr}(A^{-1}\frac{d\;}{dt}A)\det(A)$, the cyclic propery of the trace, i.e, $\mathrm{Tr}(ABC)=\mathrm{Tr}(BCA)=\mathrm{Tr}(CAB)$, and the differential equation that
\begin{equation*}
    \frac{d\;}{dt}\det(A)=\mathrm{Tr}(A^{-1}BA)\det(A)=\mathrm{Tr}(BAA^{-1})\det(A)=\mathrm{Tr}(B)\det(A),
\end{equation*}
which holds for as long as $\det(A)$ remains positive. 
Writing this as
\begin{equation*}
    \frac{d\;}{dt}\ln(\det(A)) = \mathrm{Tr}(B),
\end{equation*}
we find, after integrating in time, that 
\begin{equation*}
    \ln\biggl(\frac{\det(A_0)}{\det(A(t))}\biggr)=
    \int_t^{t_0}\mathrm{Tr}(B(s))\, ds.
\end{equation*}
Exponentiating and rearranging yields
\begin{equation*}
    \det(A(t))=
    e^{\int_t^{t_0}\mathrm{Tr}(B(s))\, ds}\det(A_0).
\end{equation*}
Since $\det(A_0)>0$ by assumption, we conclude that $\det(A(t))>0$ for all 
$t\in (t_1,t_0]$.
\end{proof}

\section{\label{calc}Calculus inequalities}
Here, we collect, for the convenience of the reader, a number of calculus inequalities that we employ throughout this article. The proofs of the following inequalities are well known and may be found, for example, in 
the books \cite[Ch.~4]{AdamsFournier:2003}, \cite[Ch.~13, \S 1-3]{TaylorIII:1996} and \cite[Ch.~IV, \S 3]{Choquet_et_al:2000}. 

\begin{thm}{\emph{[Sobolev's inequality]}} \label{Sobolev} Suppose 
 $k\in \Zbb_{\geq 0}$ and $0<\alpha \leq k-n/2\leq 1$. Then
$H^{k}(\Tbb^{n})\subset C^{0,\alpha}(\Tbb^n)$ and 
\begin{equation*}
\norm{u}_{L^\infty}\lesssim \norm{u}_{C^{0,\alpha}}  \lesssim  \norm{u}_{H^{k}}
\end{equation*}
for all $u\in H^{k}(\Tbb^{n})$.
\end{thm}

\begin{thm}{\emph{[Product and commutator inequalities]}} \label{Product} $\;$

\begin{enumerate}[(i)]
\item
Suppose $k\in \Zbb_{\geq 1}$ and $|\alpha|=k$.
Then
\begin{align*}
\norm{D^\alpha (uv)}_{L^2} \lesssim \norm{u}_{H^{k}}\norm{v}_{L^{\infty}} + \norm{u}_{L^{\infty}}\norm{v}_{H^{k}} \label{clacpropB.2.1}
\intertext{and}
\norm{[D^\alpha,u]v}_{L^2} \lesssim \norm{D u}_{L^{\infty}}\norm{v}_{H^{k-1}} + \norm{D u}_{
H^{k-1}}\norm{v}_{L^{\infty}}
\end{align*}
for all $u,v \in C^\infty(\Tbb^{n})$.
\item[(ii)]  Suppose $k_1,k_2,k_3\in \Zbb_{\geq 0}$, $\;k_1,k_2\geq k_3$, and $k_1+k_2-k_3 > n/2$. Then
\begin{equation*}
\norm{uv}_{H^{k_3}} \lesssim \norm{u}_{H^{k_1}}\norm{v}_{H^{k_2}}
\end{equation*}
for all $u\in H^{k_1}(\Tbb^{n})$ and $v\in H^{k_2}(\Tbb^{n})$.
\end{enumerate}
\end{thm}

\begin{thm}{\emph{[Moser's inequality]}}  \label{Moser}
Suppose   $k\in \Zbb_{\geq 1}$, $0\leq s\leq k$, $|\alpha|=s$, $f\in C^k(U)$, where
$U$ is open and bounded in $\Rbb$ and contains $0$, and $f(0)=0$. Then
\begin{equation*}
\norm{D^\alpha f(u)}_{L^{2}} \leq C\bigl(\norm{f}_{C^k(\overline{U})}\bigr)(1+\norm{u}^{k-1}_{L^\infty})\norm{u}_{H^{k}}
\end{equation*}
for all $u \in C^0(\Tbb^{n})\cap L^\infty(\Tbb^{n})\cap H^{k,p}(\Tbb^{n})$ with
$u(x) \in U$ for all $x\in \Tbb^{n}$.
\end{thm}

\bibliographystyle{amsplain}
\bibliography{refs}

\end{document}